\documentclass[amssymb]{revtex4}

\usepackage{array}
\usepackage{scalerel}
\usepackage{tabularx}
\usepackage{multirow}
\usepackage{amsmath}
\usepackage{mathtools}
\usepackage{stmaryrd}
\usepackage{amsthm}
\usepackage{graphicx} 
\usepackage{amssymb}
\usepackage{mathrsfs}
\usepackage{subeqnarray}
\usepackage[all]{xy}
\usepackage{accents}
\usepackage{subfigure}

\usepackage[pdftex, bookmarks, colorlinks=true, plainpages = false, citecolor = red, urlcolor = black, filecolor = blue]{hyperref} 
\usepackage[usenames,dvipsnames]{color}
\usepackage[shortlabels]{enumitem}
\setenumerate{listparindent=\parindent}

\usepackage[titletoc]{appendix}

\newcommand{\itemcolor}[1]{
  \renewcommand{\makelabel}[1]{\color{#1}\hfil ##1}}

\DeclareFontFamily{U}{mathb}{\hyphenchar\font45}
\DeclareFontShape{U}{mathb}{m}{n}{
      <5> <6> <7> <8> <9> <10> gen * mathb
      <10.95> mathb10 <12> <14.4> <17.28> <20.74> <24.88> mathb12
      }{}
\DeclareSymbolFont{mathb}{U}{mathb}{m}{n}
\DeclareMathSymbol{\curvearrowright}{3}{mathb}{'361}

\usepackage{calligra}
\usepackage{tikz-cd}
\usepackage{thm-restate}
\DeclareMathAlphabet{\mathcalligra}{T1}{calligra}{m}{n}
\DeclareFontShape{T1}{calligra}{m}{n}{<->s*[2.2]callig15}{}

\def\dlceil{\left\lceil\kern-4.75pt\left\lceil}
\def\drceil{\right\rceil\kern-4.75pt\right\rceil}

\usepackage{etoolbox}
\patchcmd{\appendices}{\quad}{. }{}{}

\global\long\def\bZ{\mathbb{Z}}

\global\long\def\bZpos{\mathbb{Z}_{>0}}

\global\long\def\bZnn{\mathbb{Z}_{\geq 0}}

\global\long\def\bC{\mathbb{C}}

\global\long\def\ii{\mathfrak{i}}

\global\long\def\sF{\mathcal{F}}

\newcommand{\rad}{\textnormal{rad}\,}
\newcommand{\Span}{\textnormal{span}\,}

\newcommand{\End}{\textnormal{End}\,}

\newcommand{\LS}{\mathsf{L}}

\newcommand{\Gen}{U}
\newcommand{\ValGenWJ}{U}
\newcommand{\ValGenMWJ}{V}

\newcommand{\LP}{\mathsf{LP}}

\newcommand{\LD}{\mathsf{LD}}

\newcommand{\ThetaNet}{\Theta}

\newcommand{\TetraNet}{\textnormal{Tet}}

\newcommand{\smin}{s_{\textnormal{min}}}
\newcommand{\smax}{s_{\textnormal{max}}}
\newcommand{\pmin}{\mathfrak{p}}
\newcommand{\ppmin}{\bar{\pmin}}
\newcommand{\TL}{\mathsf{TL}}

\newcommand{\WJ}{\mathsf{JW}}

\newcommand{\PS}{\mathsf{P}}
\newcommand{\PP}{\PS\mathsf{P}}
\newcommand{\PD}{\PS\mathsf{D}}

\newcommand{\BarAction}{\left\bracevert\phantom{A}\hspace{-9pt}\right.}

\newcommand{\cheque}{{\scaleobj{0.85}{\vee}}}

\newcommand{\Gram}{\mathscr{G}}

\newcommand{\np}{d}

\newcommand{\Summed}{n}
\newcommand{\Defect}{u}

\newcommand{\sIndex}{s}

\newcommand{\DefectSet}{\mathsf{E}}
\newcommand{\multii}{\varsigma}
\newcommand{\multiii}{\varpi}

\newcommand{\fds}{{\underset{\vspace{1pt} \scaleobj{1.4}{\check{}}}{\multii}}}
\newcommand{\lds}{{\hat{\multii}}}
\newcommand{\flds}{{\underset{\vspace{1pt} \scaleobj{1.4}{\check{}}}{\hat{\multii}}}}

\newcommand{\super}[1]{^{\scaleobj{0.85}{(#1)}}}
\newcommand{\sub}[1]{_{\scaleobj{0.85}{(#1)}}}
\newcommand{\superscr}[1]{^{\scaleobj{0.85}{#1}}}

\newcommand{\BiForm}[2]{(#1 \BarAction #2)}

\newcommand{\Dim}{D}

\newcommand{\WJProj}{P}
\newcommand{\WJproj}{\WJProj}

\newcommand{\one}{\mathbf{1} \hspace*{-.25em} \textnormal{l}}

\newcommand{\id}{\textnormal{id}}

\newcommand{\ProjBox}{\vcenter{\hbox{\includegraphics[scale=0.275]{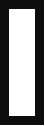}}}}

\newcommand{\const}{\mathrm{const.}}

\newsavebox\CBox
\newcommand\hcancel[2][0.5pt]{%
  \ifmmode\sbox\CBox{$#2$}\else\sbox\CBox{#2}\fi%
  \makebox[0pt][l]{\usebox\CBox}%
  \rule[0.5\ht\CBox-#1/2]{\wd\CBox}{#1}}
  
\newcommand{\Mod}[1]{\ (\mathrm{mod}\ #1)}

\definecolor{amber}{rgb}{1.0, 0.75, 0.0}

\newcommand{\blue}{\textcolor{blue}}

\newcommand{\red}{\textcolor{red}}

\newcommand{\be}{\begin{equation}}
\newcommand{\ee}{\end{equation}}
\newcommand{\bea}{\begin{eqnarray}}
\newcommand{\eea}{\end{eqnarray}}
\newcommand{\bean}{\begin{eqnarray*}}
\newcommand{\eean}{\end{eqnarray*}}

\global\long\def\Poset{\Lambda}
\global\long\def\poset{\lambda}
\global\long\def\posetprime{\lambda'}
\global\long\def\Cell{M}
\global\long\def\CellMod{\mathsf{M}}
\global\long\def\CellMap{\iota}

\global\long\def\Lpar{(\hspace*{-1mm}(}
\global\long\def\Rpar{)\hspace*{-1mm})}

\theoremstyle{plain}
\newtheorem*{theorem*}{Theorem}
\newtheorem{theorem}{Theorem}
\numberwithin{theorem}{section} 
\newtheorem{prop}[theorem]{Proposition}
\numberwithin{prop}{section} 
\newtheorem{cor}[theorem]{Corollary}
\numberwithin{cor}{section} 
\newtheorem{lem}[theorem]{Lemma}
\numberwithin{lem}{section} 
\newtheorem{conj}[theorem]{Conjecture}
\numberwithin{conj}{section} 

\numberwithin{quest}{section} 

\newtheorem{claim}[theorem]{Claim}
\newtheorem{InductAssump}[theorem]{Induction Hypothesis}
\numberwithin{InductAssump}{section} 

\numberwithin{comment}{section}

\theoremstyle{definition}
\newtheorem{remark}[theorem]{Remark}
\numberwithin{remark}{section} 
\newtheorem{recipe}[theorem]{Recipe}
\numberwithin{recipe}{section} 
\newtheorem{defn}[theorem]{Definition}
\numberwithin{defn}{section} 
 
\numberwithin{example}{section} 


\newcounter{parentnumber}

\def\clap#1{\hbox to 0pt{\hss#1\hss}}
\def\mathclap{\mathpalette\mathclapinternal}

\def\mathclapinternal#1#2{%
\clap{$\mathsurround=0pt#1{#2}$}}

\makeatletter
\DeclareRobustCommand{\cev}[1]{%
  \mathpalette\do@cev{#1}%
}
\newcommand{\do@cev}[2]{%
  \fix@cev{#1}{+}%
  \reflectbox{$\m@th#1\vec{\reflectbox{$\fix@cev{#1}{-}\m@th#1#2\fix@cev{#1}{+}$}}$}%
  \fix@cev{#1}{-}%
}
\newcommand{\fix@cev}[2]{%
  \ifx#1\displaystyle
    \mkern#23mu
  \else
    \ifx#1\textstyle
      \mkern#23mu
    \else
      \ifx#1\scriptstyle
        \mkern#22mu
      \else
        \mkern#22mu
      \fi
    \fi
  \fi
}

\numberwithin{equation}{section}

\renewcommand{\thesection}{\arabic{section}} 

\makeatother

\allowdisplaybreaks

\begin{document}
\title{Generators, projectors, and the Jones-Wenzl algebra \vspace*{.5cm}}


\author{\bf Steven M. Flores}
\affiliation{\blue{\tt \small steven.miguel.flores@gmail.com} \\ 
Department of Mathematics and Systems Analysis, \\ 
P.O. Box 11100, FI-00076, Aalto University, Finland}

\author{\bf Eveliina Peltola}
\affiliation{\blue{\tt \small eveliina.peltola@unige.ch} \\ 
Section de Math\'{e}matiques, Universit\'{e} de Gen\`{e}ve, \\
2--4 rue du Li\`{e}vre,  C.P. 64, 1211 Gen\`{e}ve, Switzerland
\vspace*{.5cm}}


\begin{abstract}
\begingroup
\setlength{\parindent}{1.5em}
\setlength{\parskip}{.5em}
We investigate a subalgebra of the Temperley-Lieb algebra called the Jones-Wenzl algebra,
which is obtained by action of certain Jones-Wenzl projectors. 
This algebra arises naturally in applications to conformal field theory and statistical physics.
It is 
also the commutant (centralizer) algebra of the Hopf algebra $U_q(\mathfrak{sl}_2)$ on its type-one modules
--- this fact is a generalization of the $q$-Schur-Weyl duality of Jimbo. 
In this article, we find two minimal generating sets for the Jones-Wenzl algebra.
In special cases, we also find all of the independent relations satisfied by these generators.

\endgroup
\end{abstract}

\maketitle

\vspace*{-1.5cm}
\renewcommand{\tocname}{}
{\hypersetup{linkcolor=black}
\tableofcontents
}

\begingroup
\setlength{\parindent}{1.5em}
\setlength{\parskip}{.5em}

\newpage

\section{Introduction} \label{Intro}

In this article, we consider a subalgebra of the Temperley-Lieb algebra that appears naturally 
in applications to conformal field theory and statistical 
physics~\cite{pms2, pms, bdmn, dgp, dge, nic, dgn, prt, mrr, fp0, fp3, fp1}.
We call 
it the ``Jones-Wenzl algebra.'' It is obtained from the Temperley-Lieb algebra 
by action of certain Jones-Wenzl projectors. 
We construct two minimal generating sets for the Jones-Wenzl algebra, and find relations that these generators satisfy. 
These results generalize similar properties of the Temperley-Lieb algebra.


The Temperley-Lieb algebra was originally discovered by H.~Temperley and E.~Lieb~\cite{tl},
and independently by V.~Jones~\cite{vj, vj2}. 
In the 1970s, Temperley and Lieb found this algebra from its connections to transfer matrices 
in integrable statistical mechanics models~\cite{tl, pen, pm, bax}. 
On the other hand, in the 1980s Jones used the Temperley-Lieb algebra to construct new knot invariants. 
Furthermore, results from~\cite{vj, vj2} manifested a close relationship between the Temperley-Lieb algebra and quantum 
groups~\cite{mj2, lk2, cp, ck, gras, krt, vt}. 
Indeed, Jimbo~\cite{mj2} noticed that the Temperley-Lieb algebra is in ``quantum Schur-Weyl duality''
with the Hopf algebra $U_q(\mathfrak{sl}_2)$ on tensor products of the fundamental $U_q(\mathfrak{sl}_2)$-modules,
see also~\cite{ppm, mma}, and~\cite{fks, cs} and references therein.
In~\cite{fp3}, we discuss a concrete generalization of such a duality, where the Temperley-Lieb algebra is replaced by 
the Jones-Wenzl algebra and the fundamental $U_q(\mathfrak{sl}_2)$-modules with modules of higher spin.


Simple special cases of the Jones-Wenzl algebra are closely related to 
the ``one-boundary Temperley-Lieb algebra'' 
and the ``two-boundary Temperley-Lieb algebra,'' that appear in the literature~\cite{pms2, bdmn, mrr}.
The one-boundary Temperley-Lieb algebra, also termed the ``blob algebra,''
was introduced by P.~Martin and H.~Saleur~\cite{pms2, pms} 
to study transfer matrices in Potts models with toroidal boundary conditions.
Its relation to other critical planar statistical mechanics models was revealed in~\cite{prt}.
M.~Batchelor, J.~de Gier, S.~Mitra, and B.~Nienhuis
introduced the two-boundary Temperley-Lieb algebra in the context of the dense loop $O(1)$ model~\cite{bdmn}.
It was also later found to be naturally related also, e.g., to the six-vertex model with integrable boundary terms~\cite{dgp, nic, dge, dgn}. 
A quotient of the one-boundary Temperley-Lieb algebra, called the ``boundary seam algebra,''
was introduced and investigated recently by A.~Morin-Duchesne, J.~Rasmussen, and D.~Ridout in~\cite{mrr}.
It is a special case of the Jones-Wenzl algebra with only one non-trivial projector box.


%
%

\subsection{Temperley-Lieb algebra}

For each $n \in \bZnn$, we define an \emph{$n$-link diagram} to be any planar geometric object comprising two vertical lines, 
$n$ distinct marked points (\emph{nodes}) on each line,
and $n$ simple, nonintersecting, planar curves (\emph{links}) between the lines, joining the nodes pairwise.
The links can be \emph{crossing links} or \emph{turn-back links},
\begin{align}
\vcenter{\hbox{\includegraphics[scale=0.275]{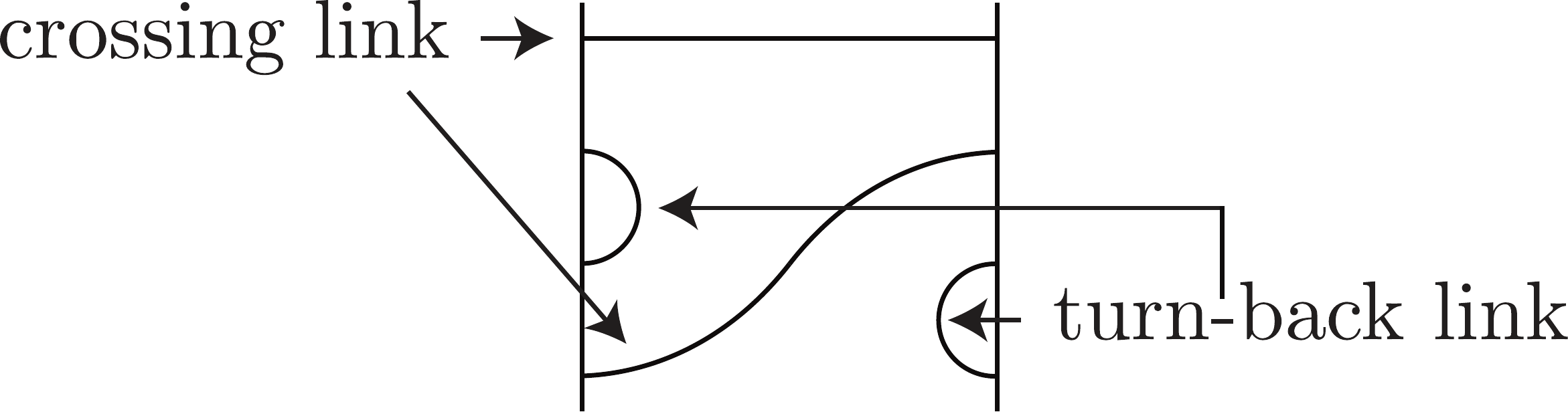} ,}}
\end{align}
and they are determined up to homotopy.  We denote the set of $n$-link diagrams by $\LD_n$, and
we denote by $\TL_n$ the complex vector space of all \emph{tangles}, that is, formal linear combinations of $n$-link diagrams.
We define a natural multiplication of link diagrams via concatenation:
\begin{align}
\label{TLmult1} 
& \vcenter{\hbox{\includegraphics[scale=0.275]{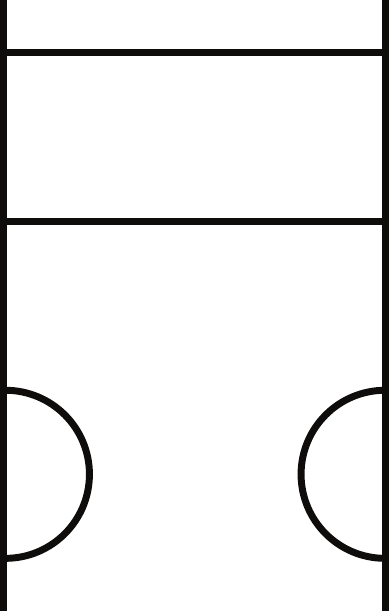}}} \quad 
\vcenter{\hbox{\includegraphics[scale=0.275]{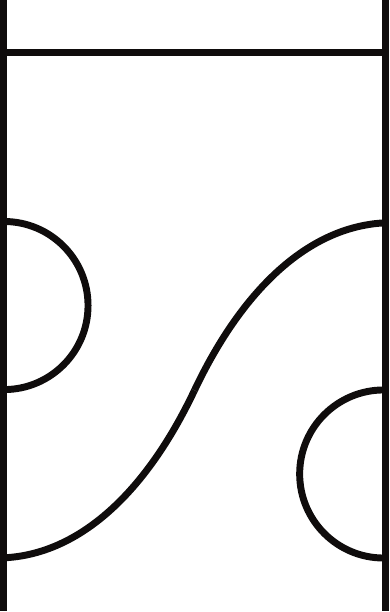}}} \quad := \quad 
\vcenter{\hbox{\includegraphics[scale=0.275]{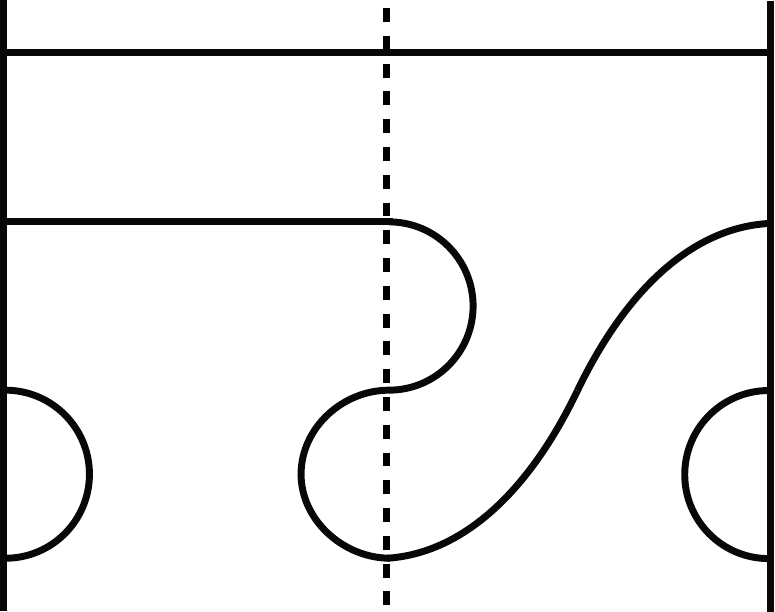}}} \quad = \quad 
\vcenter{\hbox{\includegraphics[scale=0.275]{e-TLalgebra1.pdf} .}}
\\[1em] 
\label{TLmult2}
& \vcenter{\hbox{\includegraphics[scale=0.275]{e-TLalgebra2.pdf}}} \quad 
\vcenter{\hbox{\includegraphics[scale=0.275]{e-TLalgebra1.pdf}}} \quad := \quad 
\vcenter{\hbox{\includegraphics[scale=0.275]{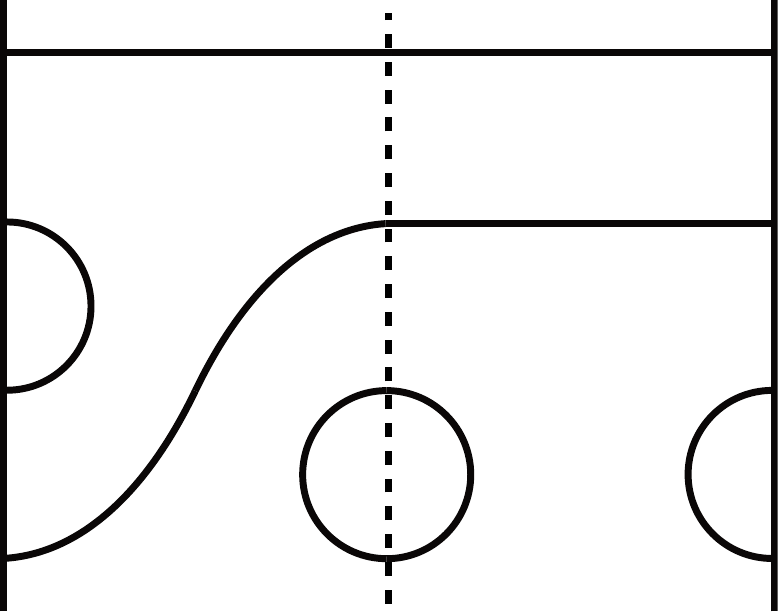}  .}}
\end{align}
The concatenation forms a number $k \in \bZnn$ of internal loops. We remove the loops and multiply the resulting tangle by $\nu^k$,
where $\nu$ is a given complex number, called the \emph{loop fugacity}. 
Thus, for instance,~\eqref{TLmult2} becomes
\begin{align}
\label{TLmult4}
& \vcenter{\hbox{\includegraphics[scale=0.275]{e-TLalgebra2.pdf}}} \quad 
\vcenter{\hbox{\includegraphics[scale=0.275]{e-TLalgebra1.pdf}}} \quad := \quad 
\vcenter{\hbox{\includegraphics[scale=0.275]{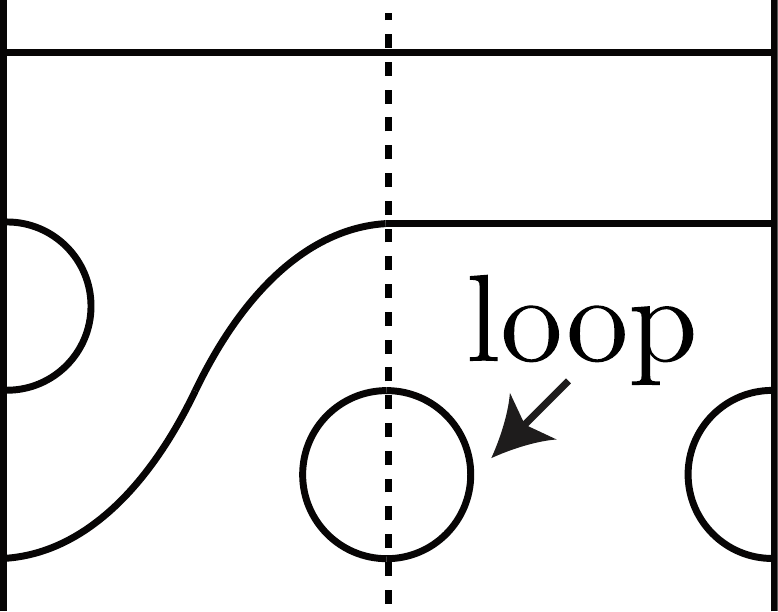}}} \quad = \quad \nu \,\, \times \,\, 
\vcenter{\hbox{\includegraphics[scale=0.275]{e-TLalgebra2.pdf}  .}}
\end{align}
This concatenation recipe endows the vector space $\TL_n$ with the structure of an associative, unital algebra, 
the \emph{Temperley-Lieb algebra}, denoted by $\TL_n(\nu)$, for fixed $\nu \in \bC$.
Its dimension is obtained by counting all $n$-link diagrams (see, e.g.,~\cite[section~\red{2}]{rsa}),
and the result is the $n$:th Catalan number
\begin{align} \label{Dim34}
\dim \TL_n(\nu) = C_n = \frac{1}{n+1} \binom{2n}{n} .
\end{align}

It is well-known~\cite{vj, lk, rsa} that the diagram algebra $\TL_n(\nu)$ 
is isomorphic to the abstract associative algebra with generating set $\{ \mathbf{1}, \Gen_1, \Gen_2, \ldots, \Gen_{n-1} \}$,
where $\mathbf{1}$ is the unit and the other generators satisfy 
the relations
\begin{alignat}{2}
\label{WordRelations1} 
\Gen_i \Gen_{i \pm 1} \Gen_i &= \Gen_i, \qquad  &&\text{if $1 \leq i\pm1 \leq n-1$}, \\ 
\label{WordRelations2} 
\Gen_i^2 &= \nu \Gen_i, \qquad && \\
\label{WordRelations3} 
\Gen_i \Gen_j &= \Gen_j \Gen_i, \qquad  &&\text{if $|i-j| > 1$},
\end{alignat} 
for all $i,j \in \{ 1, 2, \ldots, n - 1 \}$, and no other relations.  
Following the pioneering work~\cite{vj} of V.~Jones, 
an elementary proof of this fact appears in~\cite[theorem~\red{2.4}]{rsa}.  
Diagrammatically, the generating set reads
\begin{align}
\label{Units} 
\hphantom{\vcenter{\hbox{\includegraphics[scale=0.275]{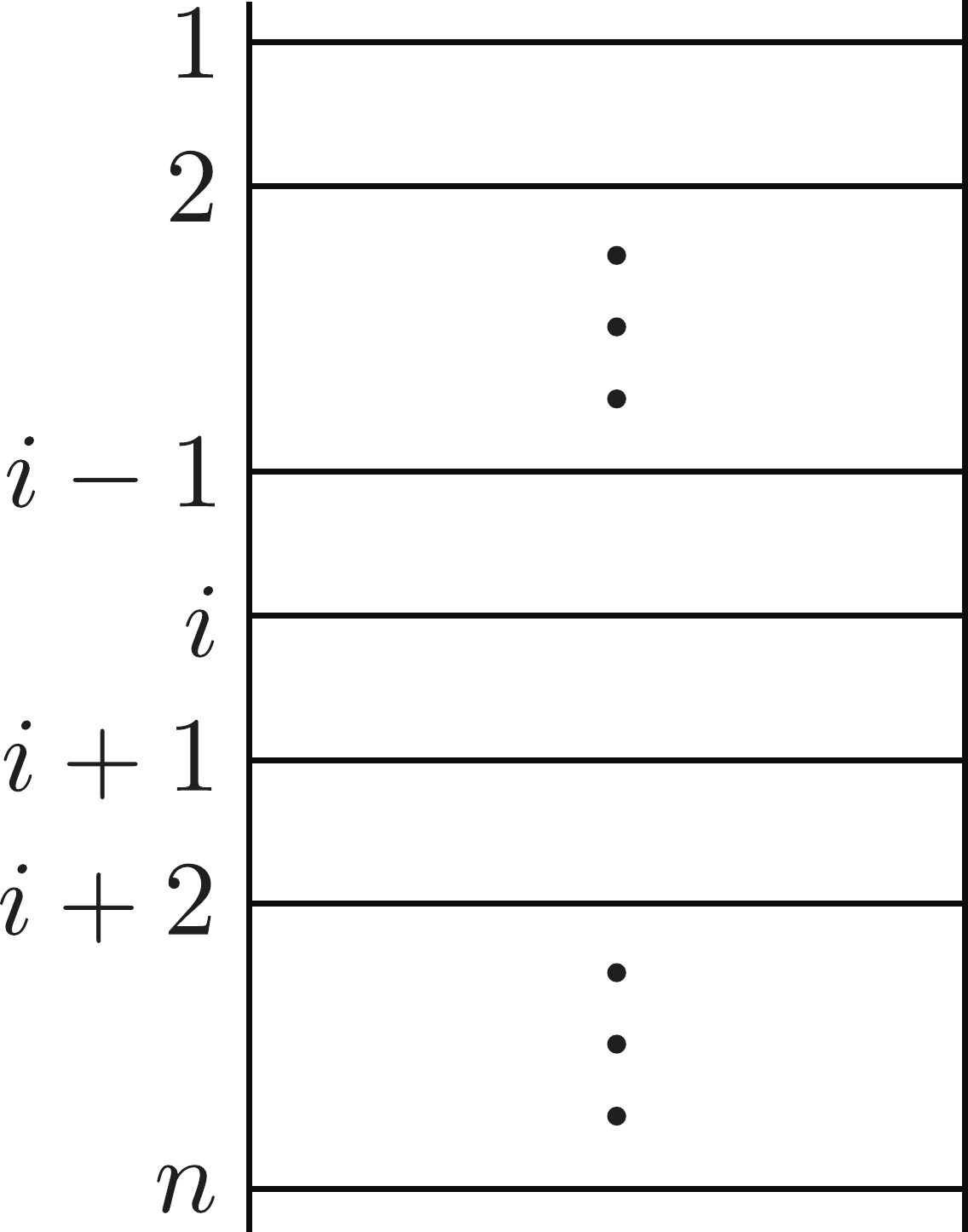}}}}
\mathbf{1}_{\TL_n} \quad & =
&& \vcenter{\hbox{\includegraphics[scale=0.275]{e-TLalgebra5.pdf}}} \quad \in \TL_n(\nu) \\[1em]
\label{ExtMe} 
\hphantom{\vcenter{\hbox{\includegraphics[scale=0.275]{e-TLalgebra5.pdf}}}}
\Gen_i^{\TL} = \Gen_i \quad & =
&& \vcenter{\hbox{\includegraphics[scale=0.275]{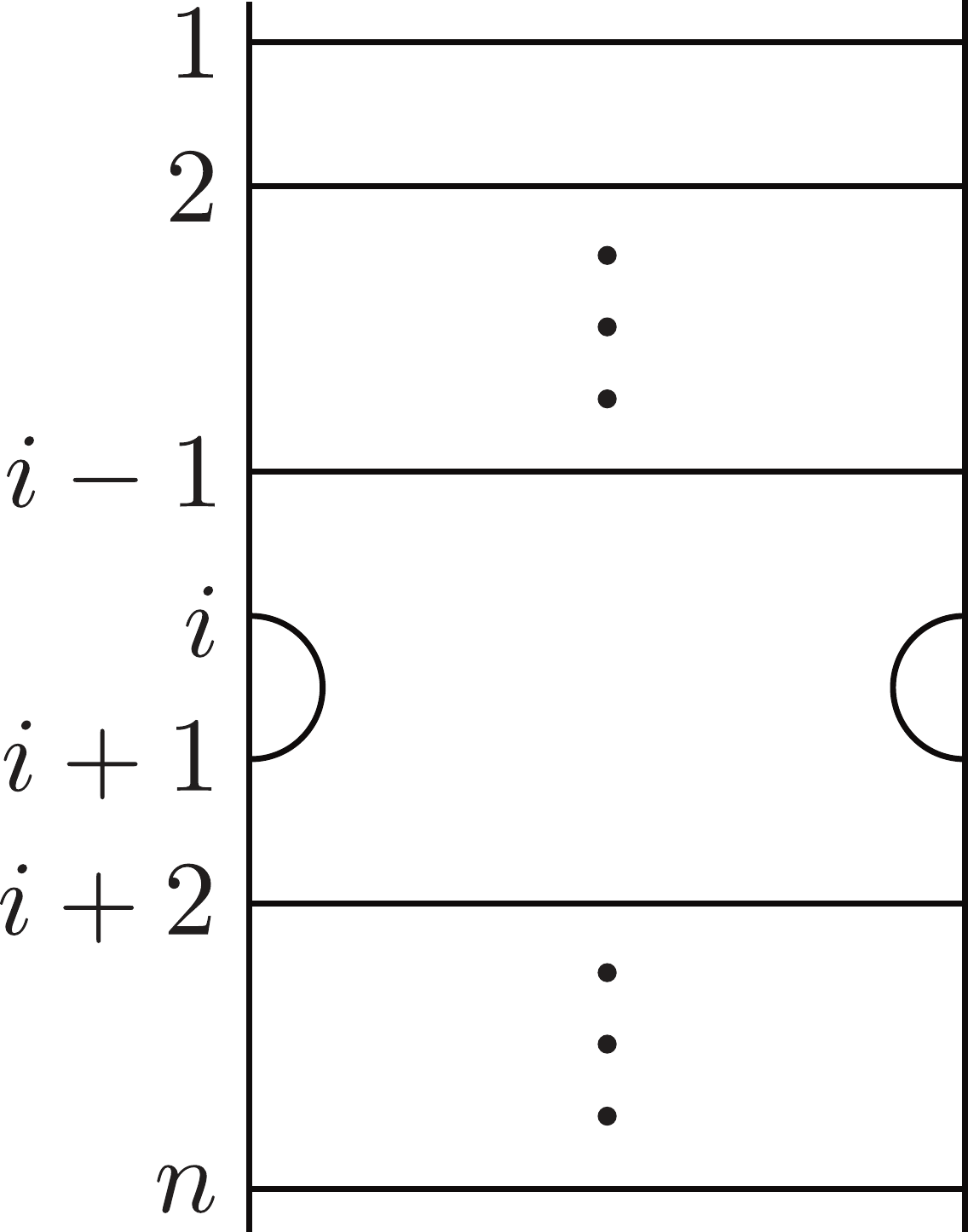}}} \quad \in \TL_n(\nu) ,
\quad \text{with $i \in \{1,2, \ldots, n-1\}$},
\end{align} 
and the two first relations~(\ref{WordRelations1}--\ref{WordRelations2}) (the third~\eqref{WordRelations3} being very intuitive, whence we omit its diagrams) read
\begin{align}
\label{TLrel1} 
\vcenter{\hbox{\includegraphics[scale=0.275]{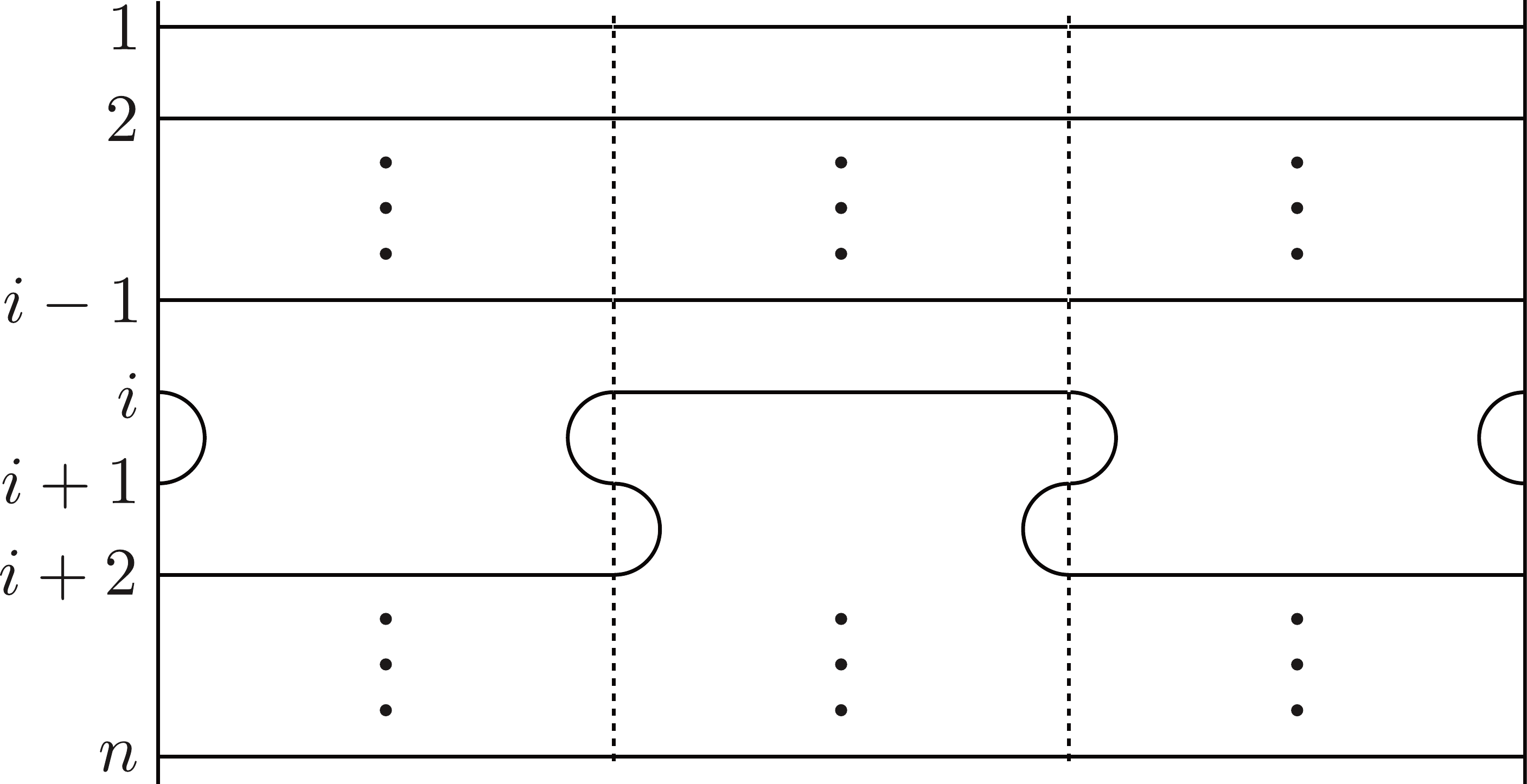}}}
\quad & \overset{\eqref{WordRelations1}}{=} \quad \hphantom{\nu \,\, \times \,\,} \; \vcenter{\hbox{\includegraphics[scale=0.275]{e-TLalgebra6.pdf} ,}} \\[1em] 
\label{TLrel2}
\vcenter{\hbox{\includegraphics[scale=0.275]{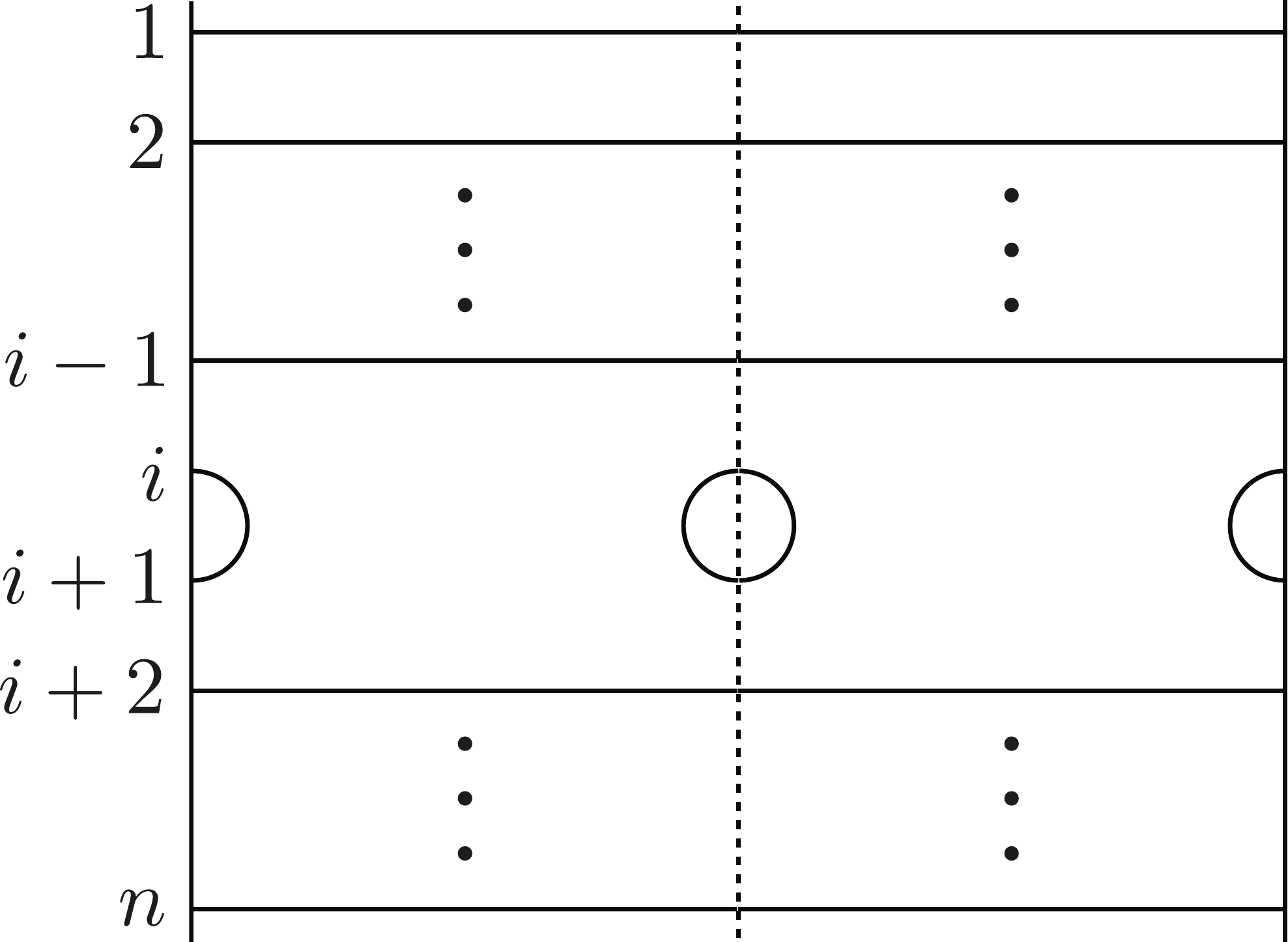}}}
\quad & \overset{\eqref{WordRelations2}}{=} \quad \nu \,\, \times \,\, \vcenter{\hbox{\includegraphics[scale=0.275]{e-TLalgebra6.pdf} .}} 
\end{align}

Throughout, we fix a nonzero complex number $q \in \bC^\times = \bC \setminus \{0\}$ and parameterize the loop fugacity 
$\nu \in \bC$ as
\begin{align}\label{fugacity} 
\nu = -q - q^{-1} = - [2] ,
\end{align} 
where $[2]$ is an example of a \emph{quantum integer}, defined for any $k \in \bZ$ as
\begin{align}\label{Qinteger} 
[k] = [k]_q := \frac{q^{k} - q^{-k}}{q - q^{-1}}. 
\end{align}
The representation theory of the Temperley-Lieb algebra is well-studied, e.g., in~\cite{pm, gwe, gl2, rsa, fp0}.
For instance, the algebra $\TL_n(\nu)$ is known to be semisimple if and only if either 
$n < \ppmin(q)$, or $n$ is odd and $q \in \{\pm \ii\}$, where
\begin{align} \label{MinPower} 
\pmin(q) : = 
 \begin{cases} \infty, & \text{$q$ is not a root of unity}, \\
p, & \text{$q=e^{\pi ip'/p}$ for coprime $p,p' \in \bZpos$}, 
 \end{cases} \qquad \qquad
\ppmin(q) := 
\begin{cases} \infty, & q \in \{\pm1\}, \\ \pmin(q), & q \not\in \{\pm1\} . 
\end{cases} 
\end{align}
We summarize further facts about the representation theory of $\TL_n(\nu)$ in section~\ref{RepThSec}.

\subsection{Jones-Wenzl projectors}


The purpose of this article is to consider a subalgebra of the Temperley-Lieb algebra that arises naturally 
in applications to conformal field theory and statistical 
physics~\cite{pms2, pms, bdmn, dgp, dge, nic, dgn, prt, mrr, fp0, fp3, fp1}.
We call this algebra the ``Jones-Wenzl algebra,''
and we define it in the following section~\ref{JWdefsec}.
In order to define this algebra, we first define the ``Jones-Wenzl projectors.''

For each $s \in \{0, 1, \ldots, \ppmin(q) - 1 \}$, the \emph{Jones-Wenzl projector of size} $s$
is the unique nonzero tangle $\WJProj\sub{s} \in \TL_s(\nu)$ satisfying the two properties~\cite{vj, hw, kl}
\begin{enumerate}[leftmargin=*, label = P\arabic*., ref = P\arabic*]
\itemcolor{red}
\item \label{wj1item} 
$\smash{\WJproj\sub{s}^2} = \smash{\WJproj\sub{s}}$, and

\item \label{wj2item}
$\smash{\Gen_i} \smash{\WJproj\sub{s}} = \smash{\WJproj\sub{s}} \smash{\Gen_i} = 0$, 
for all $i \in \{1, 2, \ldots, s - 1\}$.
\end{enumerate}
For example, we have
\begin{align}\label{WJ2} 
\WJProj\sub{0} = \text{the empty tangle}, \qquad \qquad
\WJProj\sub{1} = \mathbf{1}_{\TL_1}, \qquad \qquad \text{and} \qquad \qquad
\WJProj\sub{2} = \mathbf{1}_{\TL_2} - \nu^{-1} \Gen_1 .
\end{align} 
We represent the Jones-Wenzl projector $\WJProj\sub{s}$ diagrammatically as an empty \emph{projector box}:
\begin{align}\label{ProjBoxDiag}
s
\begin{cases}
\quad \vcenter{\hbox{\includegraphics[scale=0.275]{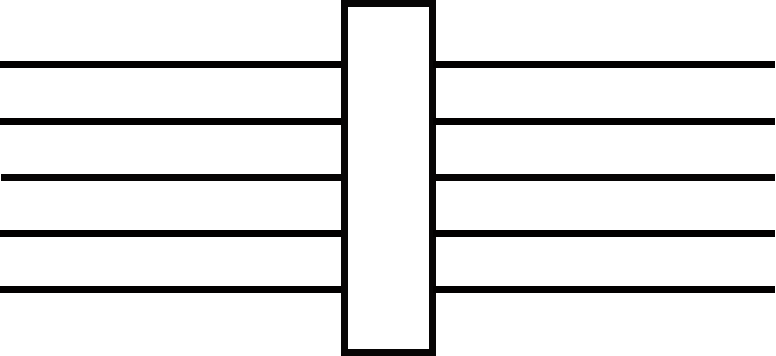}}}
\end{cases}
\; = \quad 
\vcenter{\hbox{\includegraphics[scale=0.275]{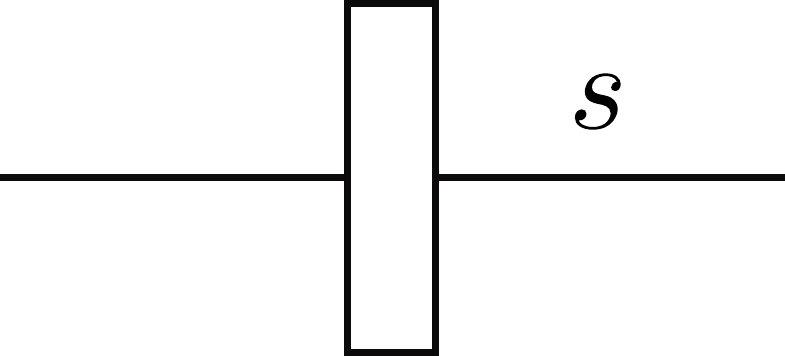}  ,}}
\end{align} 
where the label ``$s$" indicates a \emph{cable of size $s$}, that is, a collection of $s$ parallel links within a tangle:
\begin{align} \label{cable}
s
\begin{cases}
\quad \vcenter{\hbox{\includegraphics[scale=0.275]{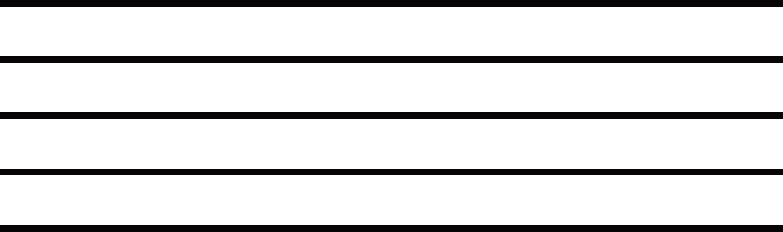}}}
\end{cases}
\; = \quad \raisebox{1pt}{\includegraphics[scale=0.275]{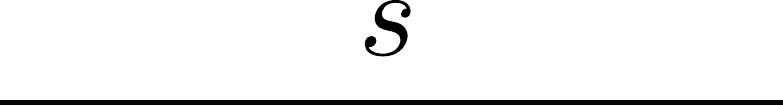}  .} 
\end{align} 
For example, the first, 
second, 
and third 
Jones-Wenzl projectors are respectively
\begin{alignat}{2} 
\label{ExampleProj1} 
\WJProj\sub{1}  \quad = \quad 
\vcenter{\hbox{\includegraphics[scale=0.275]{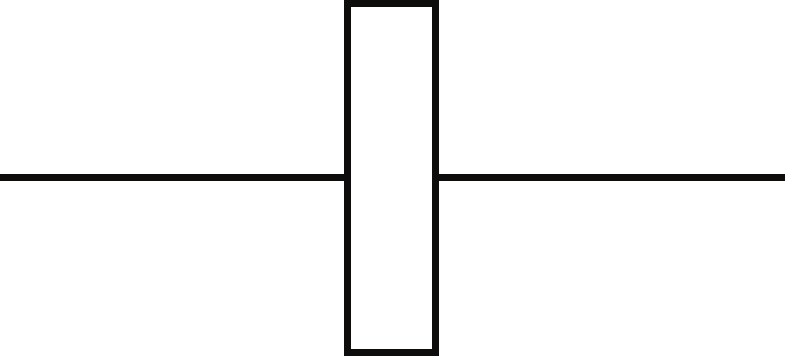}}} \quad & = \quad 
\vcenter{\hbox{\includegraphics[scale=0.275]{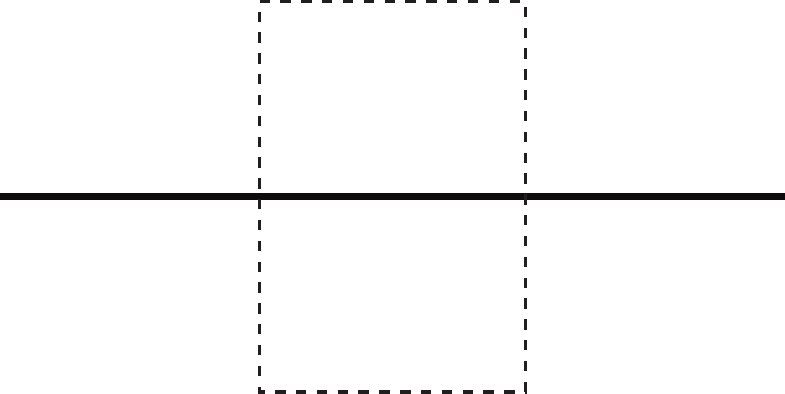} ,}} && \\[1em]
\label{ExampleProj2} 
\WJProj\sub{2}  \quad = \quad 
\vcenter{\hbox{\includegraphics[scale=0.275]{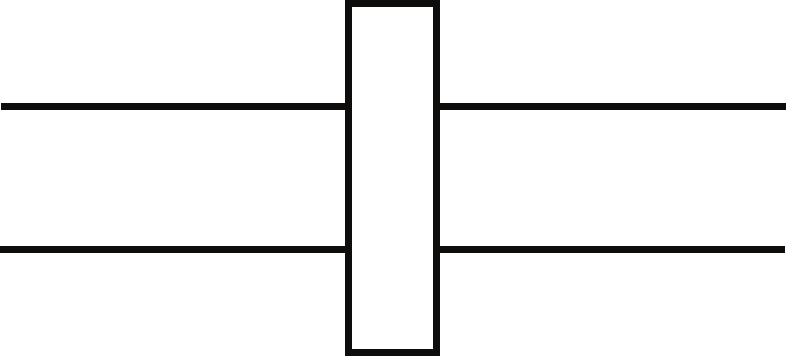}}} \quad & = \quad 
\vcenter{\hbox{\includegraphics[scale=0.275]{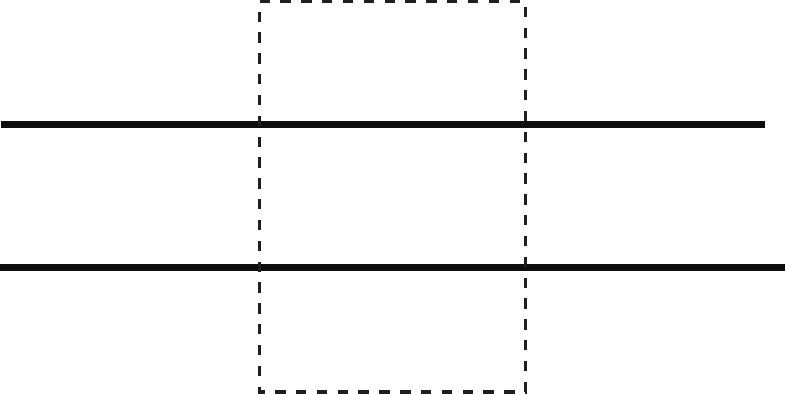}}} && \quad + \quad \frac{1}{[2]} \,\, \times \,\,
\vcenter{\hbox{\includegraphics[scale=0.275]{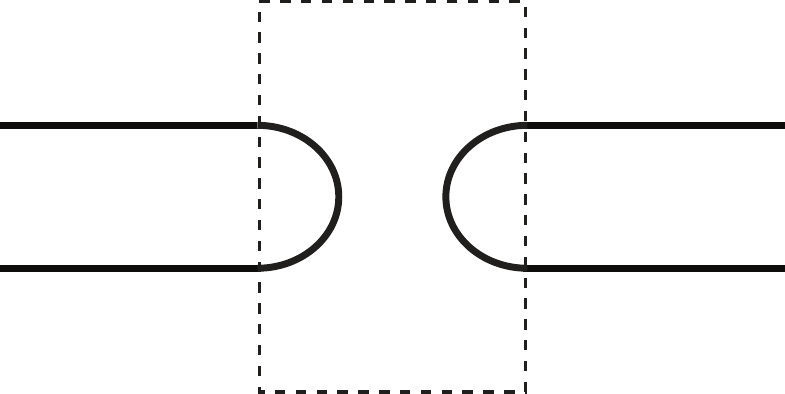} ,}} \\[1em] 
\nonumber
\WJProj\sub{3}  \quad = \quad 
\vcenter{\hbox{\includegraphics[scale=0.275]{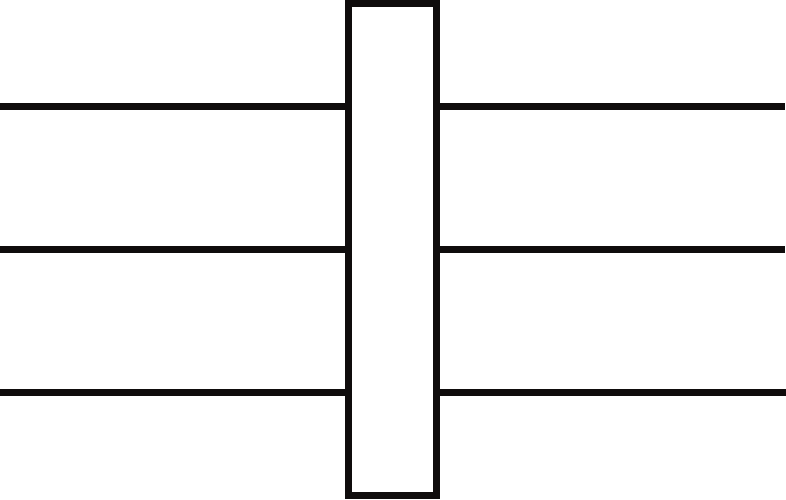}}} \quad & = \quad 
\vcenter{\hbox{\includegraphics[scale=0.275]{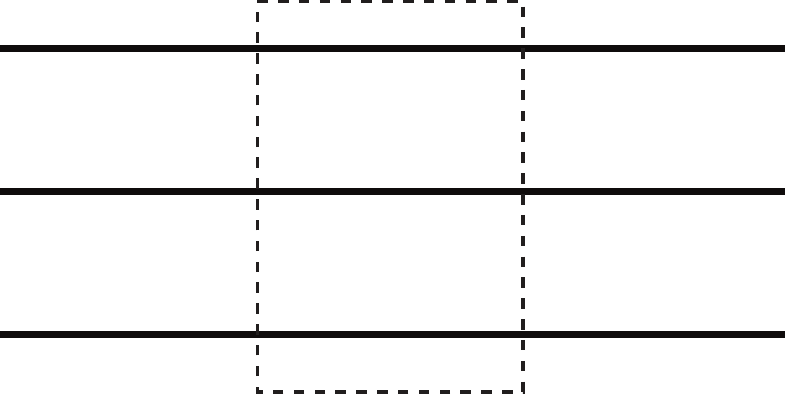}}} && \quad 
+ \quad \frac{[2]}{[3]}  \,\, \times \,\, \left( \; \vcenter{\hbox{\includegraphics[scale=0.275]{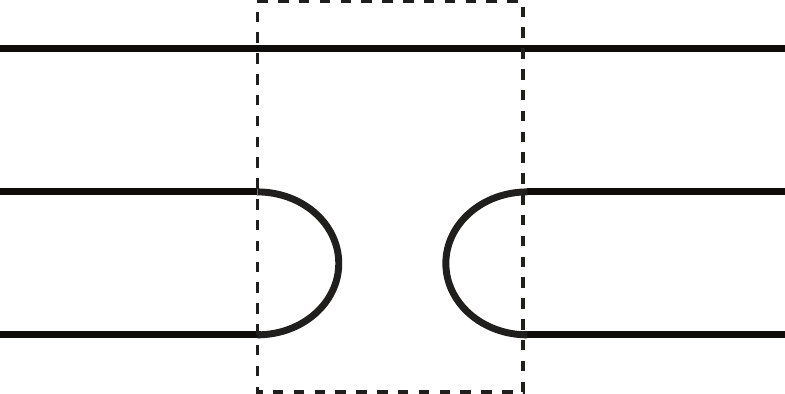}}} 
\quad + \quad \vcenter{\hbox{\includegraphics[scale=0.275]{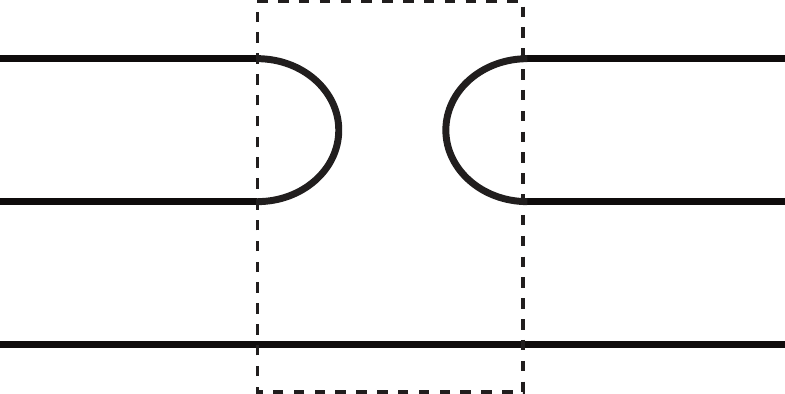}}} \; \right) \\
& && \quad + \quad \frac{1}{[3]}  \,\, \times \,\,
\left( \; \vcenter{\hbox{\includegraphics[scale=0.275]{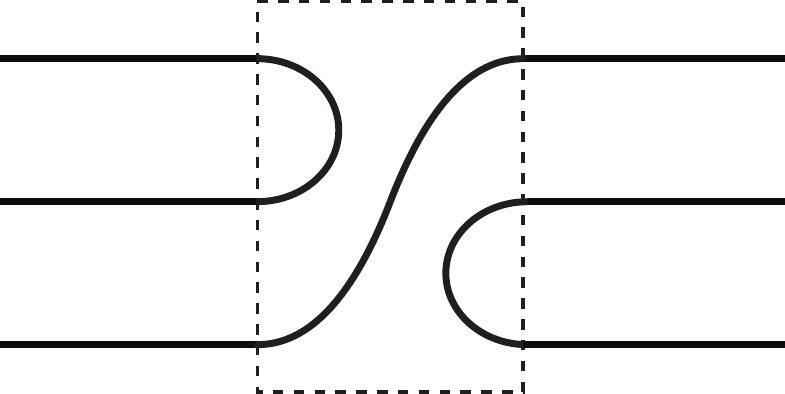}}} \quad + \quad 
\vcenter{\hbox{\includegraphics[scale=0.275]{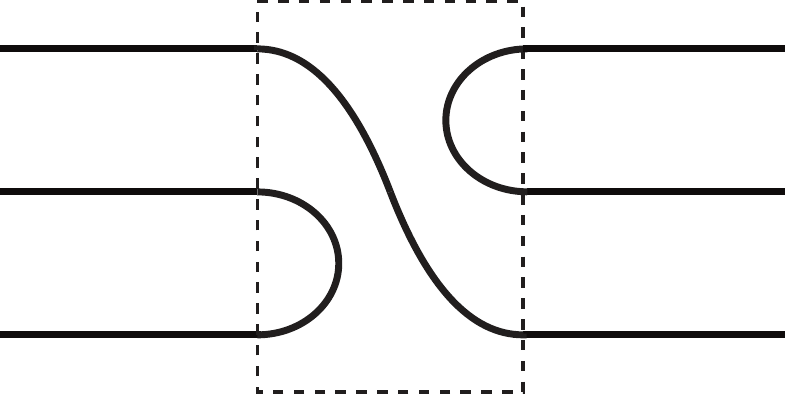}}} \; \right) .
\label{ExampleProj3} 
\end{alignat}

If $s \leq n$, then we may embed a projector box of size $s$ into a tangle in $\TL_n(\nu)$.  For example, the diagram
\begin{align}\label{CanonEmbed}
\vcenter{\hbox{\includegraphics[scale=0.275]{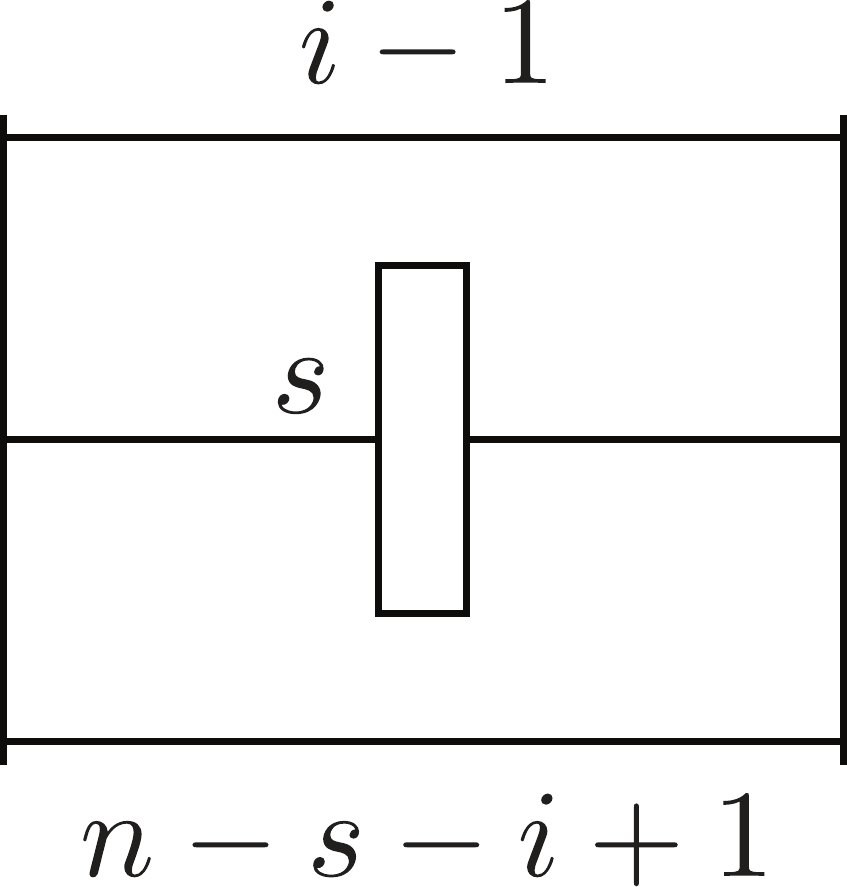}}} 
\end{align} 
represents the tangle in $\TL_n(\nu)$ obtained by replacing the box of size $s$ with the tangle $\WJProj\sub{s}$ within the larger diagram.
Abusing notation, we let the symbol $\WJProj\sub{s} \in \TL_n(\nu)$ also denote the tangle~\eqref{CanonEmbed} 
in $\TL_n(\nu)$ with $i = 1$.  Then the various projectors $\WJProj\sub{1} , \WJProj\sub{2}, \ldots,\WJProj\sub{n} \in \TL_n(\nu)$
satisfy the recursion relations~\cite{vj, hw, kl}
\begin{align}\label{wjrecursion} 
\smash{\WJproj\sub{1}} = \mathbf{1}_{\TL_n} 
\qquad \qquad \text{and} \qquad \qquad
\smash{\WJproj\sub{s+1}} = \smash{\WJproj\sub{s}} 
+ \frac{[s]}{[s+1]} \smash{\WJproj\sub{s}} \smash{\Gen_s} \smash{\WJproj\sub{s}} ,
\end{align}
for all $s \in \{1, 2, \ldots, n - 1\}$.
In terms of diagrams, recursion relation~\eqref{wjrecursion} reads
\begin{align}\label{RecursionDiagram}
\vcenter{\hbox{\includegraphics[scale=0.275]{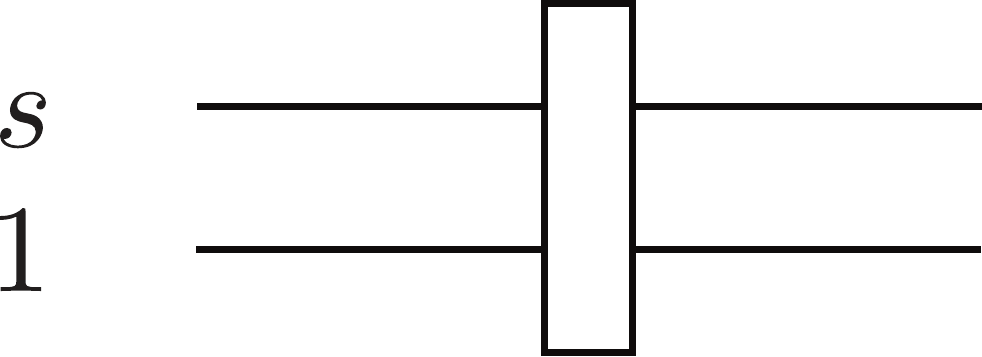}}} \quad = \quad 
\vcenter{\hbox{\includegraphics[scale=0.275]{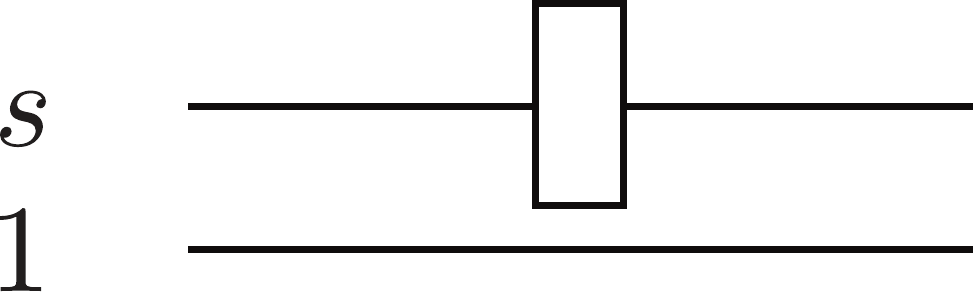}}} \quad + \quad 
\frac{[s]}{[s+1]} \,\, \times \,\, \vcenter{\hbox{\includegraphics[scale=0.275]{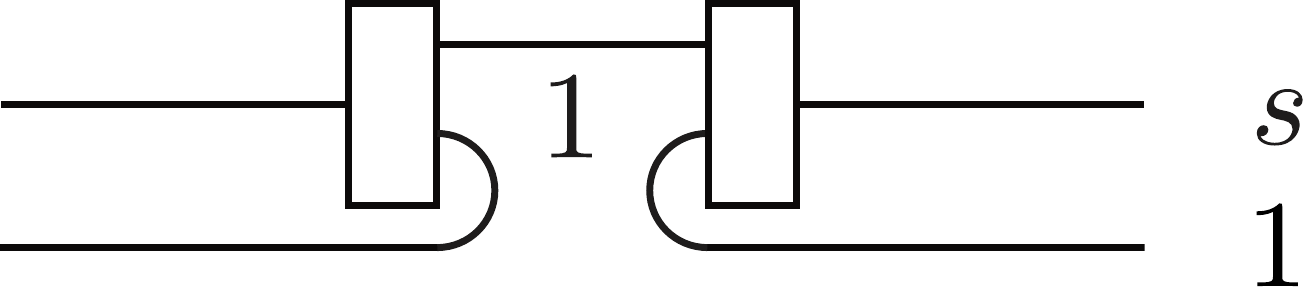} .}} 
\end{align}
With the graphical representation, properties~\ref{wj1item} and~\ref{wj2item} respectively translate to the diagram identities
\begin{align} 
\label{ProjectorID0} 
\tag{\ref{wj1item}} 
\vcenter{\hbox{\includegraphics[scale=0.275]{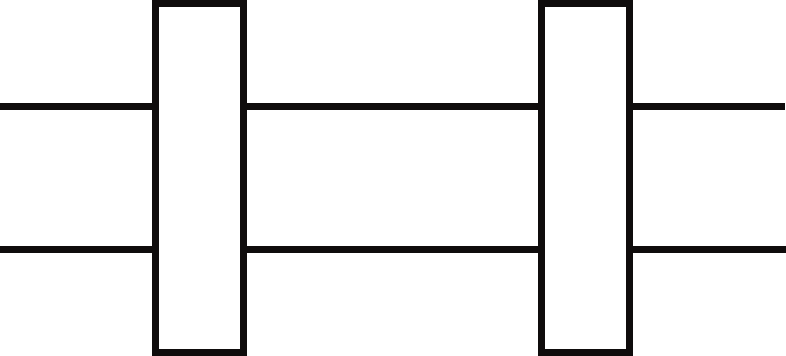}}} 
\quad & = \quad \vcenter{\hbox{\includegraphics[scale=0.275]{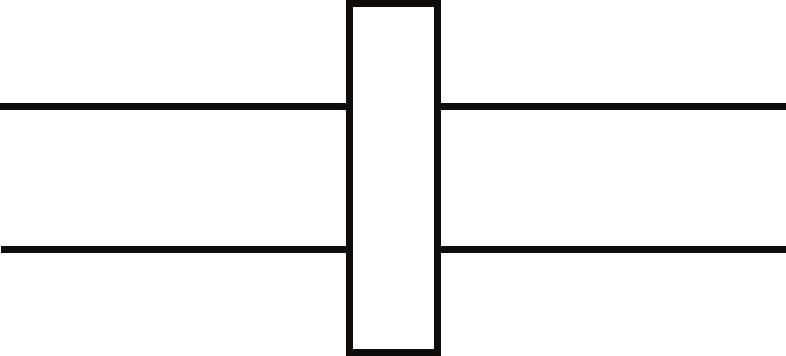}}} \\[1em]
\label{ProjectorID2} 
\tag{\ref{wj2item}} 
\vcenter{\hbox{\includegraphics[scale=0.275]{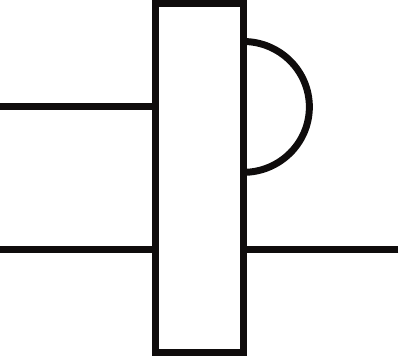}}} 
\quad = \quad \vcenter{\hbox{\includegraphics[scale=0.275]{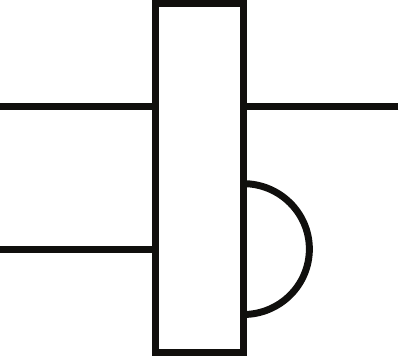}}} 
\quad & = \quad \vcenter{\hbox{\includegraphics[scale=0.275]{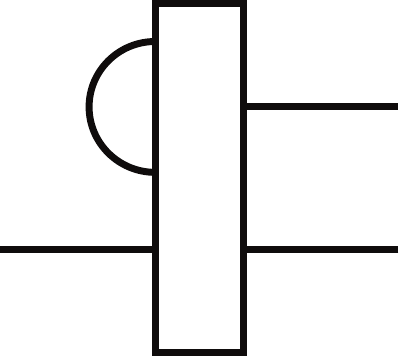}}} 
\quad = \quad \vcenter{\hbox{\includegraphics[scale=0.275]{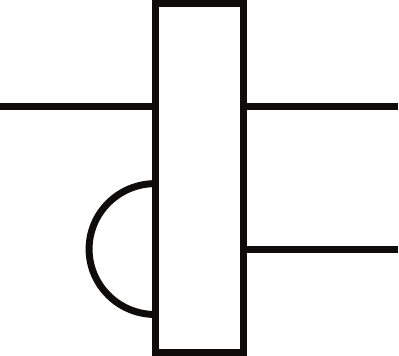}}}
\quad = \quad 0.
\end{align}
In fact, property~\eqref{ProjectorID0} can be strengthened to say that $\WJProj\sub{s}\WJProj\sub{t} = \WJProj\sub{t}\WJProj\sub{s} = \WJProj\sub{s}$ whenever $t \leq s$~\cite{kl}:
\begin{align} \label{ProjectorID1}
\tag{\ref{wj1item}$\red{'}$} 
\vcenter{\hbox{\includegraphics[scale=0.275]{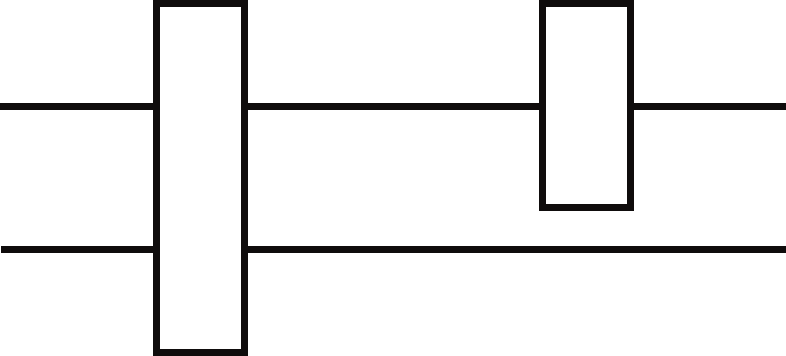}}} 
\quad = \quad \vcenter{\hbox{\includegraphics[scale=0.275]{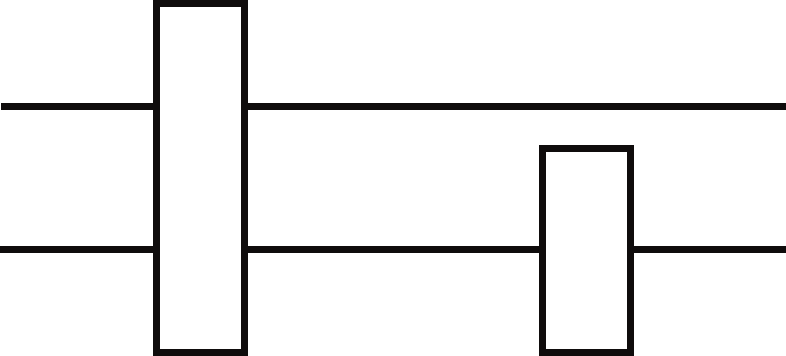}}} 
\quad = \quad \vcenter{\hbox{\includegraphics[scale=0.275]{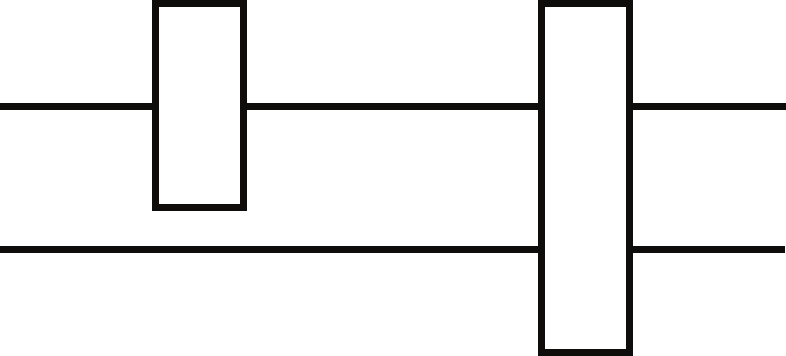}}} 
\quad = \quad \vcenter{\hbox{\includegraphics[scale=0.275]{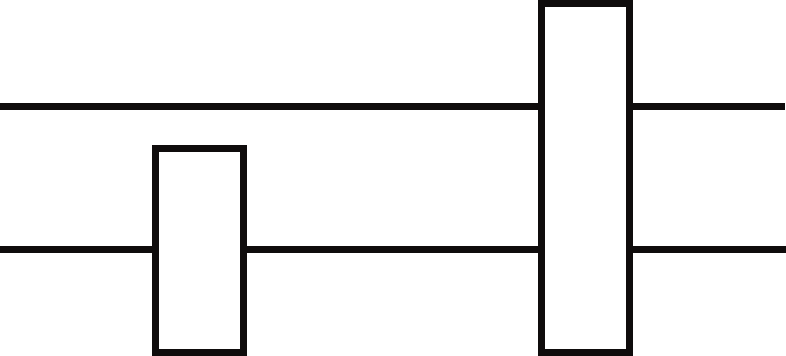}}} 
\quad = \quad \vcenter{\hbox{\includegraphics[scale=0.275]{e-ProjectorBox14.pdf} .}}
\end{align}
Finally, as a tangle in $\TL_s$, the Jones-Wenzl projector $\WJProj\sub{s}$ equals a linear combination of the link diagrams in $\LD_s$.  
Property~\ref{wj1item} implies that the coefficient of the unit $\mathbf{1}_{\TL_s}$ in this linear combination equals one.  
Hence, we have
\begin{align}\label{ProjDecomp} 
\WJproj\sub{s} = \mathbf{1}_{\TL_s} + \sum_{\substack{T \, \in \, \LD_s, \\  T \, \neq \, \mathbf{1}_{\TL_s}}} (\text{coef}_T) \,T,
\end{align}
for some coefficients $\text{coef}_T \in \bC$ (whose values depend on $q \in \bC^\times$). 
In fact, S.~Morrison derived an explicit formula for these coefficients in~\cite{sm}. 
In appendix~\ref{WJProjApp}, we give a new, alternative derivation of his formula.


\subsection{Jones-Wenzl algebra} \label{JWdefsec}


Now we define the Jones-Wenzl algebra, denoted by $\WJ_\multii(\nu)$.
Throughout, we fix a multiindex $\multii$ of nonnegative entries and use the following notation:
\begin{align} 
\label{PositiveIndexSetDef}
\bZpos^\# := \; & \bZpos \cup \bZpos^2 \cup \bZpos^3 \cup \dotsm , \\
\label{MultiindexNotation}
\multii := \; & (\sIndex_1, \sIndex_2,\ldots, \sIndex_{\np_\multii}) \in \{ (0) \} \cup \bZpos^\# , \\
\label{ndefn} 
\Summed_\multii := \; & \sIndex_1 + \sIndex_2 + \cdots + \sIndex_{\np_\multii} .
\end{align}
Using parameterization~(\ref{fugacity}--\ref{MinPower}), we assume throughout that $\max \multii < \ppmin(q)$.
Next, we define the \emph{Jones-Wenzl composite projector} to be the tangle
\begin{align}\label{WJCompProj} 
\WJProj_\multii \quad := \quad \vcenter{\hbox{\includegraphics[scale=0.275]{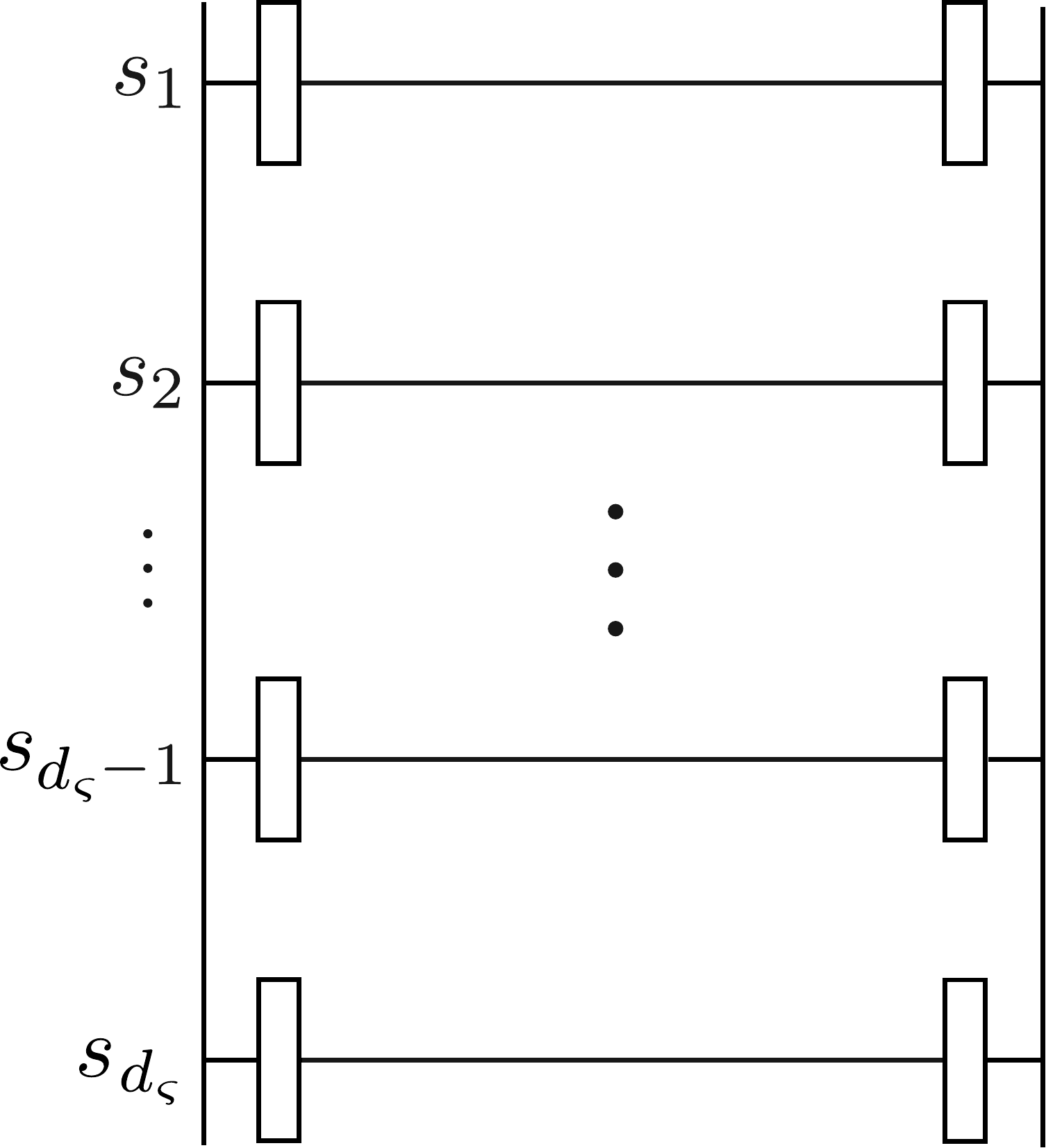} .}} 
\hphantom{\WJProj_\multii \quad := \quad}
\end{align} 
The restriction on the sizes of the indices in $\multii$ is needed in order for
all of the projector boxes in~\eqref{WJCompProj} to exist.
Then, the \emph{Jones-Wenzl algebra} $\WJ_\multii(\nu)$ is defined to be
\begin{align} \label{WJAdef}
\WJ_\multii(\nu) := \WJProj_\multii \TL_{\Summed_\multii}(\nu) \WJProj_\multii 
= \big\{ \WJProj_\multii T \WJProj_\multii \,|\, T \in \TL_{\Summed_\multii}(\nu) \big\} .
\end{align}
In other words, $\WJ_\multii(\nu)$ is the collection of all tangles in $\TL_{\Summed_\multii}(\nu)$ of the form
\begin{align} \label{JWform2-1} 
\vcenter{\hbox{\includegraphics[scale=0.275]{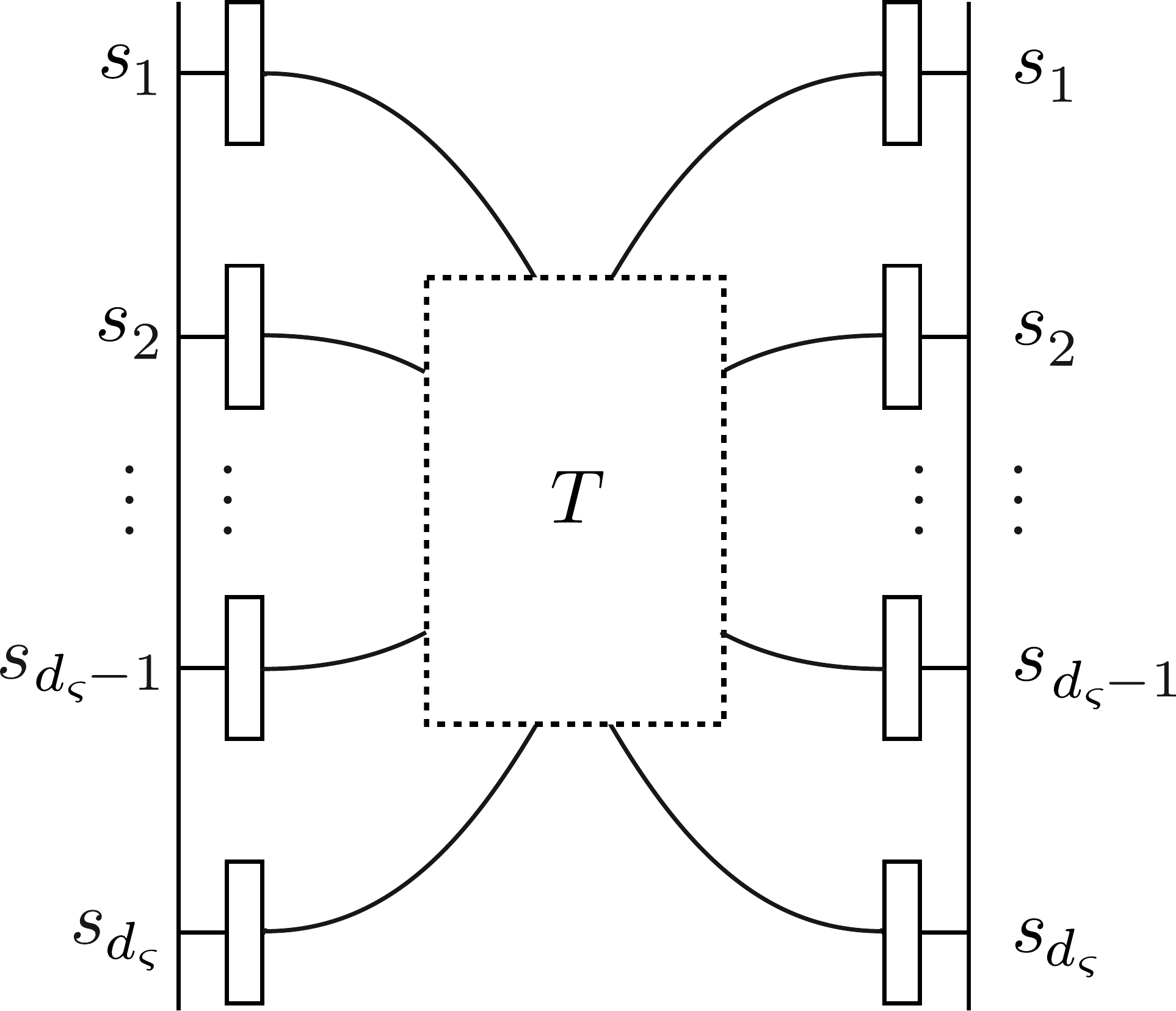} ,}}
\end{align} 
where $T \in \TL_{\Summed_\multii}(\nu)$. 
We note that 
by property~\eqref{ProjectorID2}, those tangles~\eqref{JWform2-1} that have a link with both endpoints at the same projector box are zero.
We call an element of $\WJ_\multii(\nu)$ a \emph{$\multii$-Jones-Wenzl tangle}. 
If $T$ is an $\Summed_\multii$-link diagram such that~\eqref{JWform2-1} does not vanish, then we call~\eqref{JWform2-1} 
a \emph{$\multii$-Jones-Wenzl link diagram.}
For example,
\begin{align}
\vcenter{\hbox{\includegraphics[scale=0.275]{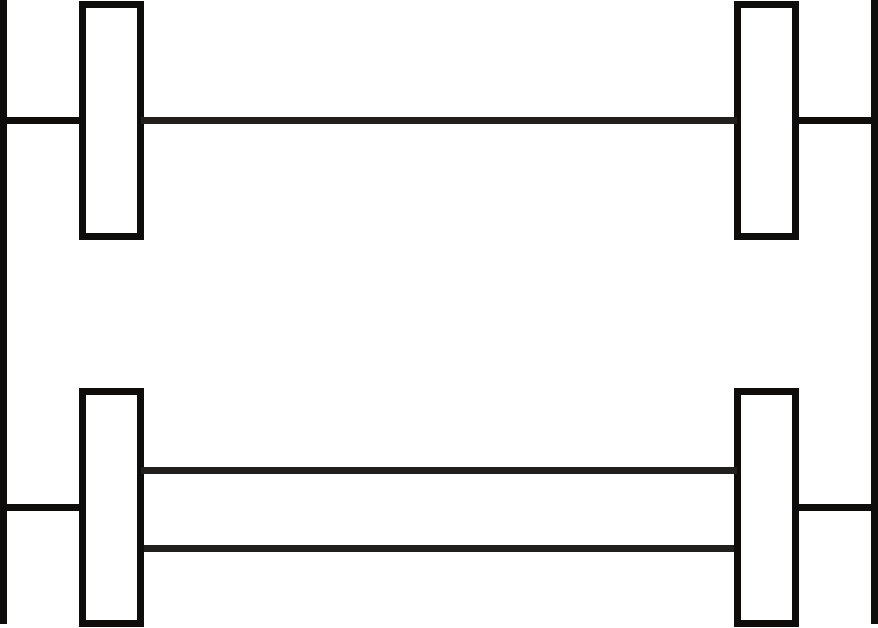}}} \qquad \qquad \text{and} \qquad \qquad
\vcenter{\hbox{\includegraphics[scale=0.275]{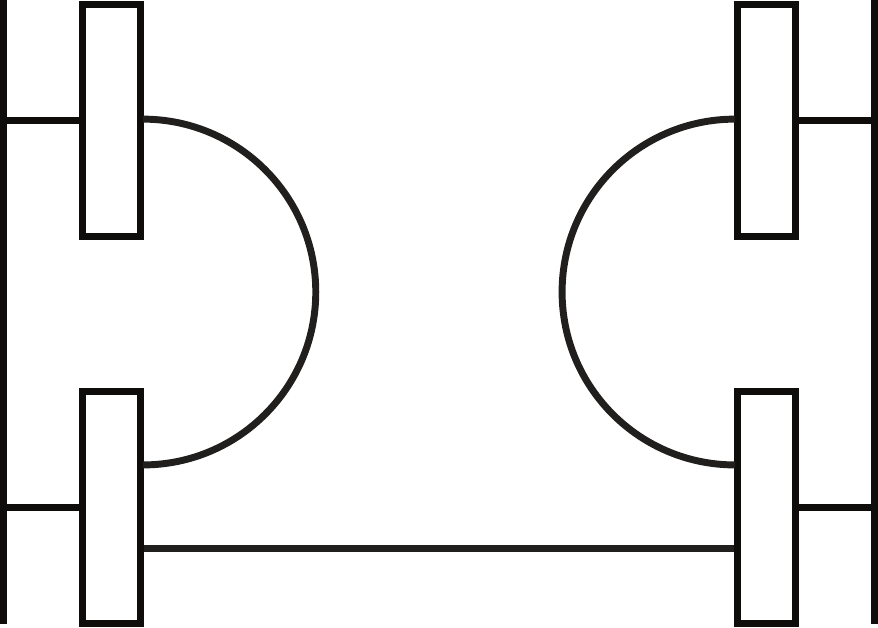}}} 
\end{align} 
are $(1,2)$-Jones-Wenzl link diagrams, and the following $(1,2)$-Jones-Wenzl tangle 
is not a Jones-Wenzl link diagram because it vanishes by property~\eqref{ProjectorID2}:
\begin{align}
\vcenter{\hbox{\includegraphics[scale=0.275]{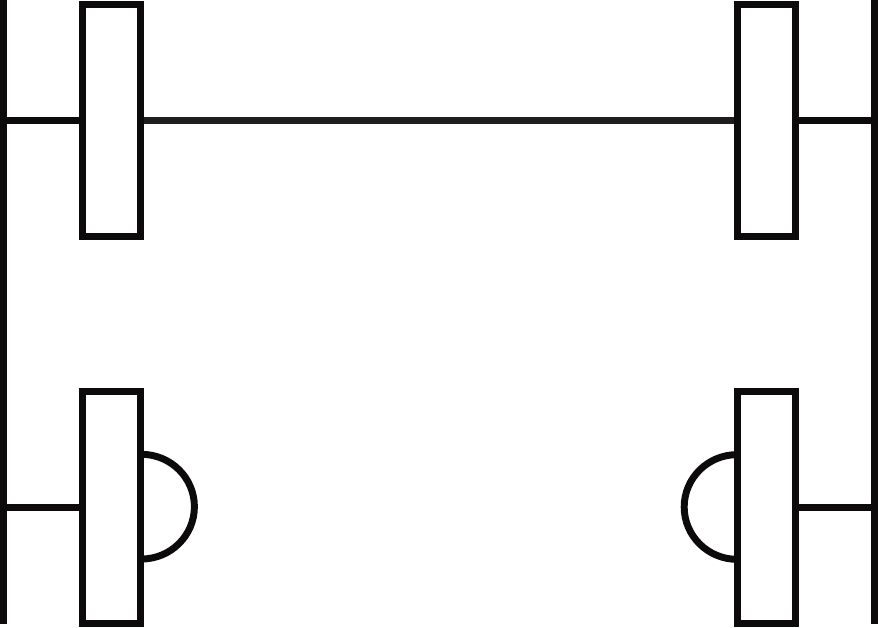}}} \quad \overset{\eqref{ProjectorID2}}{=} \quad 0 .
\end{align} 
We denote the set of all $\multii$-Jones-Wenzl link diagrams by $\PD_\multii$.
By definition, it forms a spanning set for the Jones-Wenzl algebra $\WJ_\multii(\nu)$,
and in fact, by~\cite[lemma~\red{B.1}]{fp0}, this spanning set is also a basis for $\WJ_\multii(\nu)$:
\begin{align}
\PD_\multii &:= \WJProj_\multii \LD_{\Summed_\multii} \WJProj_\multii \setminus \{0\} 
\qquad \qquad \overset{\eqref{WJAdef}}{\Longrightarrow} \qquad \qquad
\WJ_\multii(\nu) = \Span \PD_\multii .
\end{align}

The Jones-Wenzl algebra is a unital, associative algebra: indeed, 
it inherits the associative multiplication from the Temperley-Lieb algebra $\TL_{\Summed_\multii}(\nu)$, 
and property~\eqref{ProjectorID0} of the Jones-Wenzl projectors implies that $\WJProj_\multii$ is its unit: 
\begin{align} \label{WJunit}
\WJProj_\multii T = T \WJProj_\multii = T ,
\end{align}
for all tangles $T \in \TL_{\Summed_\multii}(\nu)$.  
Due to identity~\eqref{WJunit}, we sometimes write $\mathbf{1}_{\WJ_\multii} = \WJProj_\multii$.

\subsection{Main results: generators and relations for the Jones-Wenzl algebra}


The main purpose of this article is to find minimal generating sets and relations for the Jones-Wenzl algebra. 
In our forthcoming work~\cite{fp3}, we use these generating sets to obtain generalizations of the quantum Schur-Weyl duality
for the colored braid group and the Hopf algebra $U_q(\mathfrak{sl}_2)$.
Furthermore, in~\cite{fp1} we use those results to prove existence and uniqueness of certain monodromy-invariant 
correlation functions in conformal field theory.

To state our main theorem~\ref{GeneratorThm}, we first introduce needed notation.
For a multiindex $(r,t)$ with two entries, we set 
\begin{align} 
\DefectSet\sub{r,t} := \{ |r-t| , |r-t| + 2, \ldots, r+t\} .
\end{align}
We also define the following \emph{closed three-vertex} notation~\cite{kl, mv, cfs}:
for all $s \in \DefectSet\sub{r,t}$, we denote
\begin{align}\label{3vertex1}  
\vcenter{\hbox{\includegraphics[scale=0.275]{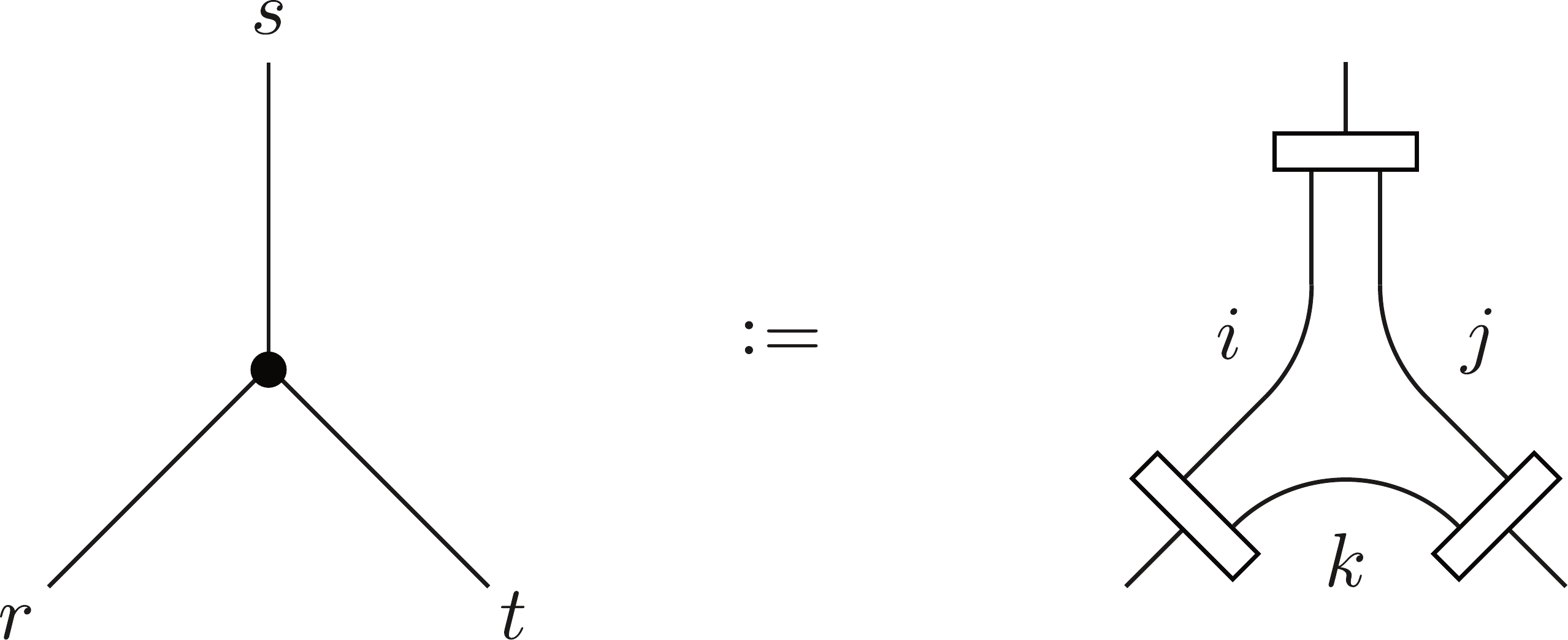} ,}} &
\qquad  \qquad \qquad
\begin{aligned} 
i & = \frac{r + s - t}{2} , \\[.7em] 
j & = \frac{s + t - r}{2} , \\[.7em] 
k & = \frac{t + r - s}{2} .
\end{aligned}
\end{align}

\begin{restatable}{theorem}{GeneratorThm} \label{GeneratorThm} 
Suppose $\Summed_\multii < \ppmin(q)$. Then the following hold:
\begin{enumerate} 
\itemcolor{red}

\item  \label{GeneratorThmItem1} 
The unit $\mathbf{1}_{\WJ_\multii} = \WJProj_\multii$~\eqref{WJCompProj} 
together with the $\multii$-Jones-Wenzl link diagrams 
\begin{align} \label{EGenerators-00} 
\ValGenWJ_i :=
\WJProj_\multii \Gen_{\sIndex_1 + \sIndex_2 + \dotsm + \sIndex_i}^{\TL} \WJProj_\multii 
\quad = \quad \vcenter{\hbox{\includegraphics[scale=0.275]{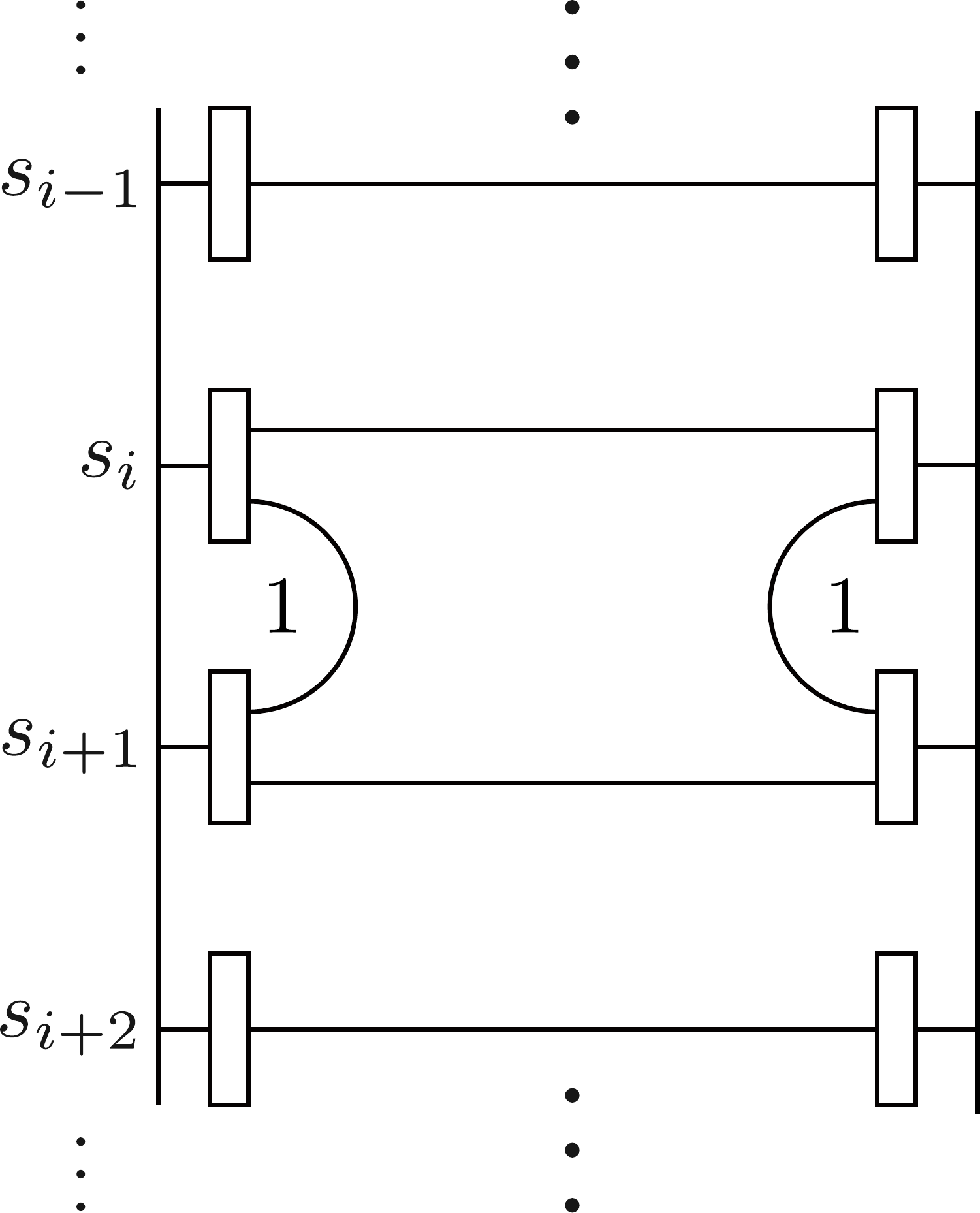} ,}}
\hphantom{\WJProj_\multii \Gen_{\sIndex_1 + \sIndex_2 + \dotsm + \sIndex_i}^{\TL} \WJProj_\multii 
\quad = \quad}
\end{align}
with $i \in \{1,2,\ldots,\np_\multii-1\}$, forms a minimal generating set for the Jones-Wenzl algebra $\WJ_\multii(\nu)$. 

\item \label{GeneratorThmItem2} 

Alternatively, 
the collection of all $\multii$-Jones-Wenzl tangles of the form
\begin{align} \label{MasterDiagramsWJ-00} 
\ValGenMWJ\super{s}_i \quad := \quad 
\vcenter{\hbox{\includegraphics[scale=0.275]{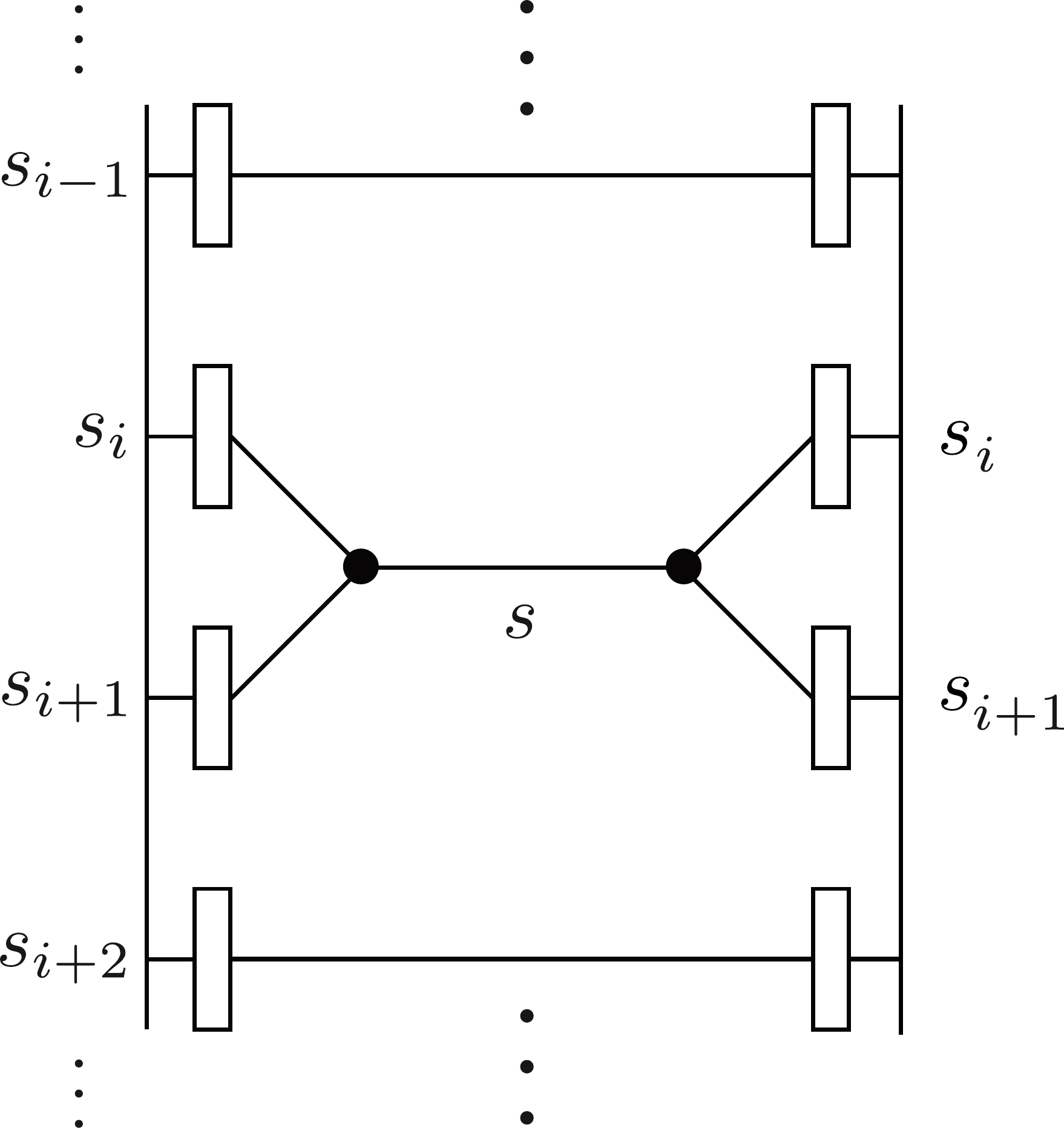} ,}}
\end{align}
with $s \in \DefectSet\sub{\sIndex_i,\sIndex_{i+1}}$ and $i \in \{ 1, 2, \ldots, \np_\multii - 1 \}$, 
forms a minimal generating set for $\WJ_\multii(\nu)$.

\end{enumerate}
\end{restatable}

We prove theorem~\ref{GeneratorThm} in section~\ref{GeneratorLemProofSect}.
In fact, we believe that item~\ref{GeneratorThmItem1} in theorem~\ref{GeneratorThm} 
also holds under the weaker assumption $\max \multii < \ppmin(q)$.
This assumption is necessary for the diagram algebra $\WJ_\multii(\nu)$ to be well-defined.

\begin{conj} \label{GeneratorConj} 
Item~\ref{GeneratorThmItem1} in theorem~\ref{GeneratorThm} holds whenever $\max \multii < \ppmin(q)$.
\end{conj}

Unfortunately, our explicit proof of theorem~\ref{GeneratorThm} in section~\ref{GeneratorLemProofSect}
relies heavily on the condition that $\Summed_\multii < \ppmin(q)$, which suggests that a proof of conjecture~\ref{GeneratorConj} 
may look quite different.
However, conjecture~\ref{GeneratorConj} is known to hold at least in the special case 
$\multii = (1,1,\ldots,1,k)$ with $\max \multii = k < \ppmin(q)$.
A proof for it is given in~\cite[appendix~\red{C}]{mrr}.
Also, in corollary~\ref{InitialCaseCor} in section~\ref{BaseCaseSec}, 
we prove conjecture~\ref{GeneratorConj} for the case of two projector boxes ($\np_\multii = 2$).

\bigskip

Another goal of this article is to find all independent relations for the generators in theorem~\ref{GeneratorThm}.
A presentation of the Jones-Wenzl algebra in terms of explicit generators and relations
gives means to extend the definition of this algebra beyond diagrams, 
a priori limited by the sizes of the Jones-Wenzl projectors.
In section~\ref{RelationLemProofSect}, we discuss special cases of $\multii$ for which all of the relations are known.
Below, we state some simple relations that we have found.
In general, however, the relations appear rather complicated, and therefore, we 
will not analyze the general case.

In the next proposition, we use the evaluation of the \emph{Theta network}
from lemma~\ref{ThetaLem} of appendix~\ref{TLRecouplingSect}:
\begin{align} 
\ThetaNet(r,s,t) 
= \frac{(-1)^{\frac{r + s + t}{2}} \left[ \frac{r + s + t}{2} + 1 \right]! \left[ \frac{ r + s - t }{2} \right]! \left[ \frac{ s + t - r}{2} \right]! \left[ \frac{t + r - s}{2} \right]! }{[ r ]! [s ]! [ t ]!} .
\end{align}

\begin{prop} \label{RelationProp} 
Suppose $\max \multii < \ppmin(q)$. 
\begin{enumerate} 
\itemcolor{red}
\item \label{WordRelationsWJItem1}
Generators~\eqref{EGenerators-00} satisfy the following relations, for all $i,j \in \{ 1, 2, \ldots, \np_\multii - 1 \}$\textnormal{:}
\begin{alignat}{4} 
\label{WordRelationsWJ01} 
\ValGenWJ_i \ValGenWJ_{i \pm 1} \ValGenWJ_i 
&= \ValGenWJ_i,  \qquad &&\textnormal{if $1 \leq i \pm1 \leq \np_\multii-1$} \quad &&&\textnormal{and $\sIndex_i = \sIndex_{i + 1} = 1$}, \\ 
\label{WordRelationsWJ02}
\ValGenWJ_i \ValGenWJ_{i + 1} \ValGenWJ_i - \ValGenWJ_{i + 1} \ValGenWJ_i \ValGenWJ_{i + 1}
&= \frac{1}{[\sIndex_{i + 1}]^2}  \, ( \ValGenWJ_i - \ValGenWJ_{i + 1} ) , \qquad
&&\textnormal{if $1 \leq i \leq \np_\multii-2$} \quad &&& \textnormal{and $\sIndex_i = \sIndex_{i+2} = 1$}, \\ 
\label{WordRelationsWJ03}
\ValGenWJ_i^2 &= - \frac{[\max (\sIndex_i, \sIndex_{i+1}) + 1]}{[\max (\sIndex_i, \sIndex_{i+1})]} \, \ValGenWJ_i, \qquad
&&\textnormal{if $\min (\sIndex_i, \sIndex_{i+1}) = 1$},  \\ 
\label{WordRelationsWJ04}
\ValGenWJ_i \ValGenWJ_j &= \ValGenWJ_j \ValGenWJ_i,  \qquad &&\textnormal{if $|i-j| > 1$}. 
\end{alignat}

\item \label{WordRelationsWJItem2} 
Generators~\eqref{MasterDiagramsWJ-00} satisfy the following relations, for all $i,j \in \{ 1, 2, \ldots, \np_\multii - 1 \}$\textnormal{:}
\begin{alignat}{2} 
\label{WordRelationsWJ05} 
\ValGenMWJ\super{s}_i \ValGenMWJ\super{s'}_i &= \delta_{s,s'} \frac{\ThetaNet(\sIndex_i,\sIndex_{i+1},s)}{(-1)^s[s+1]} \, \ValGenMWJ\super{s}_i, 
\qquad && \\ 
\label{WordRelationsWJ06}
\ValGenMWJ\super{s}_i \ValGenMWJ\super{s''}_j &= \ValGenMWJ\super{s''}_j \ValGenMWJ\super{s}_i,  
\qquad &&\textnormal{if $|i-j| > 1$},
\end{alignat}
for all $s,s' \in \DefectSet\sub{\sIndex_i,\sIndex_{i+1}}$ and $s'' \in \DefectSet\sub{\sIndex_j,\sIndex_{j+1}}$, and 
\begin{align} \label{WordRelationsWJ07} 
\mathbf{1}_{\WJ_\multii} 
& = \sum_{s \, \in \, \DefectSet\sub{\sIndex_i,\sIndex_{i+1}}} \frac{(-1)^s [s+1]}{\ThetaNet(\sIndex_i,\sIndex_{i+1},s)} \, \ValGenMWJ\super{s}_i .
\end{align} 
\end{enumerate}
\end{prop}
We prove proposition~\ref{RelationProp} in the end of section~\ref{ThreeNodeCase11}. 
Next, we make some remarks: 
\begin{itemize}[leftmargin=*]
\item Generators~\eqref{EGenerators-00} and relations~(\ref{WordRelationsWJ01}--\ref{WordRelationsWJ04})
reduce to the Temperley-Lieb generators $\{ \Gen_1, \Gen_2, \ldots, \Gen_{n-1} \}$
and relations~(\ref{WordRelations1}--\ref{WordRelations3}) when 
$\multii = (1,1,\ldots,1)$.
Thus, in this case we already know that~(\ref{WordRelationsWJ01}--\ref{WordRelationsWJ04}) 
are all of the independent relations satisfied by these generators. 
(This fact is proved, e.g., in~\cite{rsa}.)

\item When $\multii = (1,1,\ldots,1, 1, k)$ for some $k \geq 2$, (resp.~$\multii = (k_1,1,\ldots,1, 1, k_2)$ for some $k_1,k_2 \geq 2$),
the Jones-Wenzl algebra $\WJ_\multii(\nu)$ is closely related to the one-boundary Temperley-Lieb algebra,
i.e., blob algebra, and the boundary seam algebra
(resp.~two-boundary Temperley-Lieb algebra)~\cite{pms2, bdmn, mrr}.
We prove in section~\ref{ThreeNodeCase11} that in this case,~(\ref{WordRelationsWJ01}--\ref{WordRelationsWJ04})
are all of the independent relations satisfied by generators~\eqref{EGenerators-00}. 


\item However, for general multiindices $\multii$, further relations exist, and finding all of them appears rather complicated.
\end{itemize}

%
%

Finally, let us discuss the scope of the main results of the present article. In the companion article~\cite{fp0},
we investigate the representation theory of the ``valenced Temperley-Lieb algebra'' $\TL_\multii(\nu)$, which for $\max \multii < \ppmin(q)$
is isomorphic to the Jones-Wenzl algebra $\WJ_\multii(\nu)$. In~\cite[section~\red{2}]{fp0}, we define this algebra in terms of ``valenced'' tangles
analogous to the Jones-Wenzl tangles. However, the valenced Temperley-Lieb algebra should be well-defined by generators and relations
(but not necessarily as a diagram algebra) even if the condition $\max \multii < \ppmin(q)$ is not satisfied. 
To enlighten this, we present an example of such a definition:

Suppose $k_1, k_2, \ldots, k_{m+1} \geq 1$, and consider the multiindex
\begin{align} \label{specialmultii}
\multii = (\sIndex_1, \sIndex_2, \ldots, \sIndex_{\np_\multii}) = (k_1,1,1,k_2,1,1,\ldots,k_{m},1,1,k_{m+1}) ,
\end{align} 
with $\np_\multii = 3m + 1$. Fix also the parameters
\begin{align}
\lambda = (\lambda_s)_{s \in \bZpos} \in \bC^{\bZpos}
\qquad \qquad \text{and} \qquad \qquad 
\mu = (\mu_s)_{s \in \bZpos} \in \bC^{\bZpos} .
\end{align}
We define the \emph{(abstract) valenced Temperley-Lieb algebra} $\mathsf{A}_\multii(\lambda, \mu)$ to be the associative, 
unital $\bC$-algebra determined by 
its presentation with generators $\smash{\{ \ValGenWJ_i \}_{i=1}^{\np_\multii-1}}$ and relations 
similar to~(\ref{WordRelationsWJ01}--\ref{WordRelationsWJ04}):
\begin{alignat}{4} 
\label{WordRelationsVal1} 
\ValGenWJ_i \ValGenWJ_{i \pm 1} \ValGenWJ_i 
&= \ValGenWJ_i,  \qquad &&\textnormal{if $1 \leq i\pm1 \leq \np_\multii-1$} \quad  &&& \textnormal{and $\sIndex_i = \sIndex_{i + 1} = 1$}, \\ 
\label{WordRelationsVal2}
\ValGenWJ_i \ValGenWJ_{i + 1} \ValGenWJ_i - \ValGenWJ_{i + 1} \ValGenWJ_i \ValGenWJ_{i + 1}
&= \lambda_{\sIndex_{i + 1}}  \, ( \ValGenWJ_i - \ValGenWJ_{i + 1} ) , \qquad
&&\textnormal{if $1 \leq i \leq \np_\multii-2$} \quad &&& \textnormal{and $\sIndex_i = \sIndex_{i+2} = 1$}, \\ 
\label{WordRelationsVal3}
\ValGenWJ_i^2 &= \mu_{\max (\sIndex_i, \sIndex_{i+1})} \, \ValGenWJ_i, \qquad &&  \\ 
\label{WordRelationsVal4}
\ValGenWJ_i \ValGenWJ_j &= \ValGenWJ_j \ValGenWJ_i,  \qquad &&\textnormal{if $|i-j| > 1$}.
\end{alignat}
We expect that if $\max \multii = \max(k_1, k_2, \ldots, k_{m+1}) < \ppmin(q)$, 
then choosing 
\begin{align} \label{lammu}
\mu_s = - \frac{[s + 1]}{[s]} \qquad \qquad \text{and} \qquad \qquad \lambda_s = \frac{1}{[s]^2} ,
\end{align}
for all indices $s \in \{1,k_1, k_2, \ldots, k_{m+1}\}$,
the algebra $\mathsf{A}_\multii(\lambda, \mu)$ becomes isomorphic to the Jones-Wenzl diagram algebra $\WJ_\multii(\nu)$.
With a suitable dimension argument, for $\Summed_\multii < \ppmin(q)$
this should follow from theorem~\ref{GeneratorThm} and proposition~\ref{RelationProp}.


%
%
%
%
%

Along the lines of~(\ref{WordRelationsVal1}--\ref{WordRelationsVal4}), a general definition for the valenced Temperley-Lieb algebra of~\cite{fp0}
might be useful in applications, e.g., for logarithmic conformal field theories and critical planar statistical mechanics models, 
where the assumption that $\max \multii < \ppmin(q)$ can be violated. For instance,
the special case of $k_1 = k_2 = \cdots = k_m = 1$ and $k_{m+1} = k$, called the \emph{boundary seam algebra},
was introduced and investigated recently in the article~\cite{mrr} of A.~Morin-Duchesne, J.~Rasmussen, and D.~Ridout.
This algebra is related to transfer matrices of certain critical statistical physics models, which in the continuum limit are expected to 
be described by logarithmic minimal models.

\begin{center}
\bf Organization of this article
\end{center}

Section~\ref{RepThSec} consists of preliminaries and a summary of the representation theory of the Jones-Wenzl algebra.
These results immediately follow, e.g., from results obtained in~\cite{fp0} for the valenced Temperley-Lieb algebra
and the isomorphisms discussed in~\cite[appendix~\red{B}]{fp0}.
In section~\ref{CellSec}, we present explicit cellular bases for the Jones-Wenzl algebra, relating our work to the abstract
theory of J.~Graham and G.~Lehrer~\cite{gl, gl2}. 
Many representation-theoretic properties of the Jones-Wenzl algebra could also be obtained using this theory.

In section~\ref{GeneratorLemProofSect}, we prove the generator theorem~\ref{GeneratorThm}.  Our proof is constructive and no prerequisites are needed.
Then in section~\ref{RelationLemProofSect}, we discuss relations for the generators of the Jones-Wenzl algebra,
and prove proposition~\ref{RelationProp}.
In special cases, we find all of the independent relations.

In appendix~\ref{WJProjApp}, we 
show that all coefficients of the Jones-Wenzl projectors equal ratios of entries in the inverse of the meander matrix of~\cite{fgg}. 
We thus recover a result found in~\cite{sm}.
Our technique is to manipulate integration contours of Coulomb-gas-type 
integral formulas in the spirit of~\cite{fk, fkk, sfk3, sfk4}. 
In appendix~\ref{TLRecouplingSect}, we collect results related to Temperley-Lieb recoupling theory~\cite{kl, cfs}, 
for use throughout this article. 
The final appendix~\ref{LemmaApp} includes technical details.

\begin{center}
\bf Acknowledgements
\end{center}

We are very grateful to K.~Kyt\"ol\"a for numerous discussions and encouragement,
and cordially thank J.~Bellet\^ete, A.~Langlois-R\'emillard, D.~Ridout, 
and Y.~Saint-Aubin discussions and comments on this work.
E.P. is supported by the ERC AG COMPASP, the NCCR SwissMAP, and the Swiss NSF,
and she also acknowledges the earlier support from the Vilho, Yrj\"o and Kalle V\"ais\"al\"a Foundation.
During this work, S.F. was supported by the Academy of Finland grant number 288318, 
``Algebraic structures and random geometry of stochastic lattice models.''
\section{Representation theory of the Jones-Wenzl algebra} \label{RepThSec}

In this section, we briefly discuss representation theory of the Temperley-Lieb algebra
and the Jones-Wenzl algebra. We investigate the representations in more detail in the companion article~\cite{fp0}.

\subsection{Standard modules} 

Standard modules (termed ``cell modules'' in the theory of cellular algebras) 
are building blocks for representations of both the Temperley-Lieb algebra $\TL_n(\nu)$ 
and the Jones-Wenzl algebra $\WJ_\multii(\nu)$.
Indeed, certain quotients of these modules constitute the complete set of simple modules~\cite{gl2, rsa, fp0}
(see proposition~\ref{SimpleModuleProp}). The standard modules are spanned by ``link patterns,'' which we also define shortly.

The number $s \in \bZnn$ of crossing links in any $n$-link diagram 
\begin{align}
\vcenter{\hbox{\includegraphics[scale=0.275]{e-TLSplit6.pdf}}}
\end{align}
in $\LD_n$ necessarily belongs to the set
\begin{align}\label{DefectSet} 
\DefectSet_n := \{ n \; \text{mod} \; 2, \; (n \; \text{mod} \; 2) + 2, \; \ldots, \; n \} .
\end{align} 
Given an $n$-link diagram with $s$ crossing links, we create an \emph{$(n,s)$-link pattern}
by dividing the link diagram vertically in half, discarding the right half, and rotating the left half by $\pi/2$ radians:
\begin{align} 
\nonumber
\vcenter{\hbox{\includegraphics[scale=0.275]{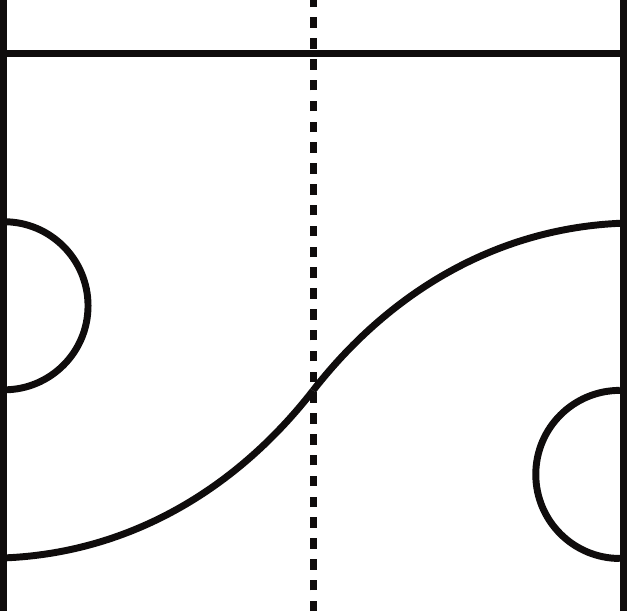}}} \qquad \qquad & \longmapsto \qquad \qquad
\vcenter{\hbox{\includegraphics[scale=0.275]{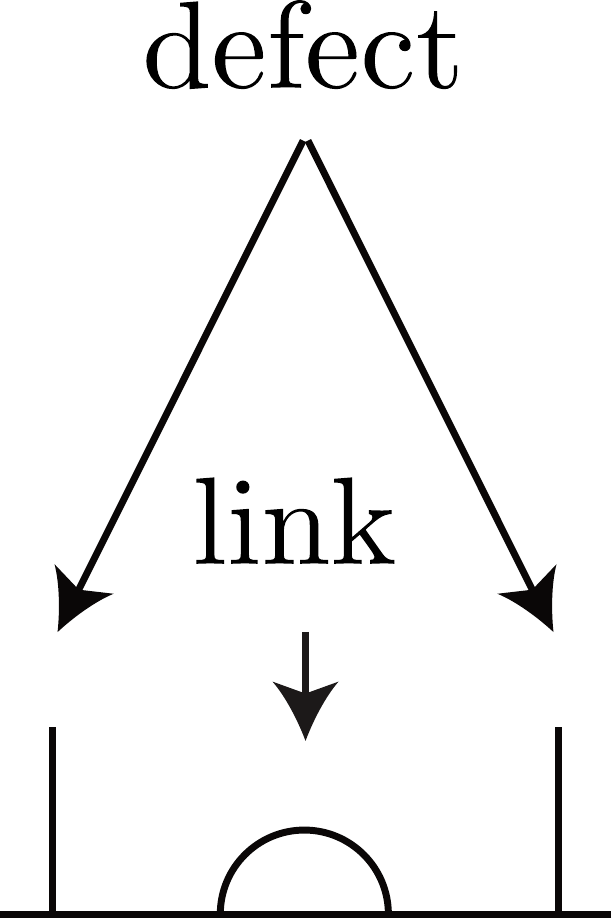}}} \\[1em]
\label{CutInHalf}
\text{link diagram} \qquad \qquad & \longmapsto \qquad \qquad \text{link pattern}.
\end{align}
We call the broken links in the $(n,s)$-link pattern \emph{defects}.  
We call a formal linear combination of $(n,s)$-link patterns with complex coefficients an \emph{$(n,s)$-link state},
and we denote by $\smash{\LS_n\super{s}}$ the complex vector space of $(n,s)$-link states.

We endow the space $\smash{\LS_n\super{s}}$ with a $\TL_n(\nu)$-action via the following diagram concatenation recipe. 
Given an $n$-link diagram $T \in \LD_n$
and an $(n,s)$-link pattern $\alpha \in \smash{\LP_n\super{s}}$, 
the latter rotated $-\pi/2$ radians, 
we concatenate $T$ to the left of $\alpha$,
remove the $k \geq 0$ loops formed by this concatenation, and multiply the result by $\nu^{k}$:
\begin{align} \label{loopex}
& \vcenter{\hbox{\includegraphics[scale=0.275]{e-TLalgebra1.pdf}}} \quad 
\vcenter{\hbox{\includegraphics[scale=0.275]{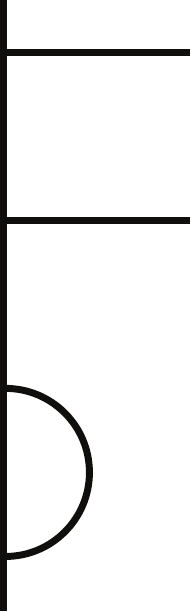}}} \quad
= \quad \vcenter{\hbox{\includegraphics[scale=0.275]{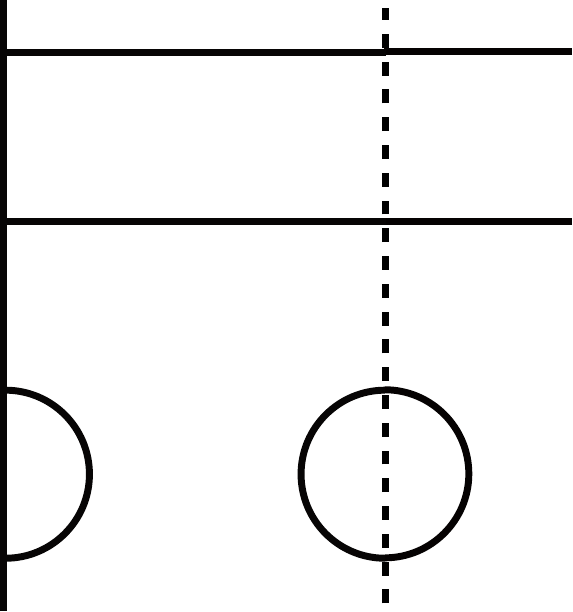}}}  \quad
 = \quad \nu \,\, \times \,\, 
\vcenter{\hbox{\includegraphics[scale=0.275]{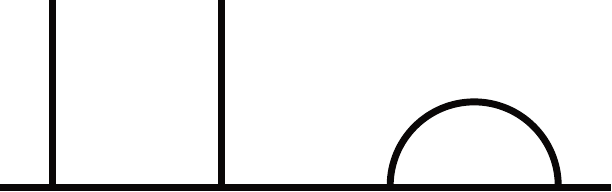} .}} 
\end{align}
Importantly, we set diagrams containing \emph{turn-back paths} to zero,
so $\TL_n(\nu)$ preserves the number $s$ of defects:
\begin{align} \label{turnbackex} 
& \vcenter{\hbox{\includegraphics[scale=0.275]{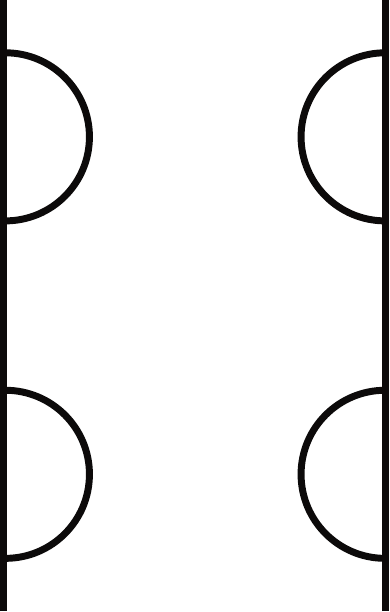}}} \quad 
\vcenter{\hbox{\includegraphics[scale=0.275]{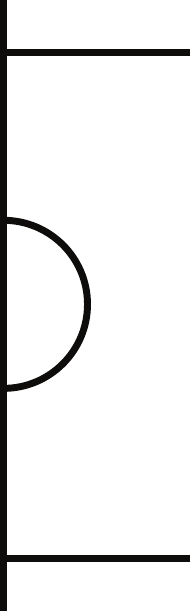}}} \quad
= \quad \vcenter{\hbox{\includegraphics[scale=0.275]{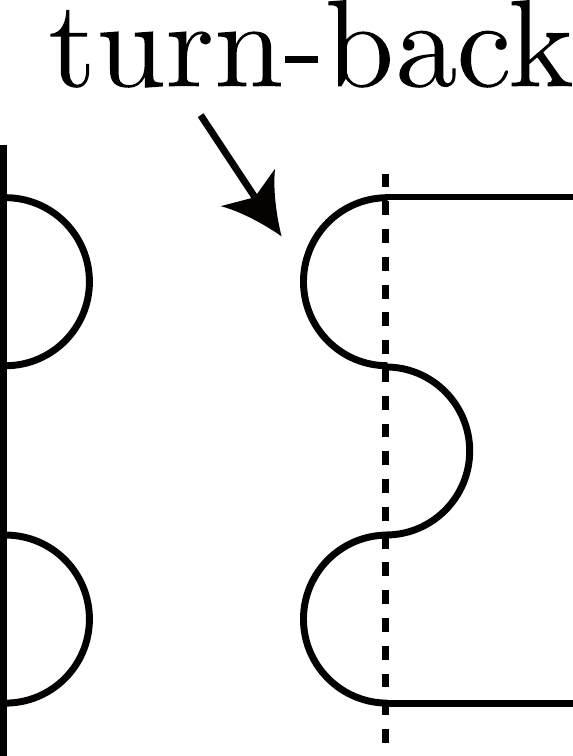}}} \quad
= \quad 0 \,\, \times \,\, 
\vcenter{\hbox{\includegraphics[scale=0.275]{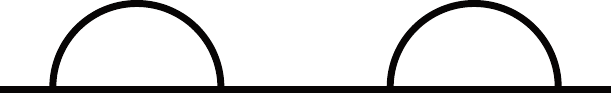}}} 
\quad = \quad 0 .
\end{align}
Bilinear extension of this recipe defines a $\TL_n(\nu)$-module structure on $\smash{\LS_n\super{s}}$,
and we thus call $\smash{\LS_n\super{s}}$ a $\TL_n(\nu)$-\emph{standard module}.
We also define the $\TL_n(\nu)$-\emph{link state module} to be the direct sum of all of the standard modules,
\begin{align} 
\LS_n :=  \bigoplus_{s \, \in \, \DefectSet_n} \LS_n\super{s} .
\end{align}

The representation theory of the Temperley-Lieb algebra 
$\TL_n(\nu)$ is completely understood~\cite{pm, hw, gwe, bw, gl, gl2, rsa, fp0}.
For example, $\TL_n(\nu)$ is semisimple if and only if the parameter $q \in \bC^\times$ 
that determines the fugacity 
$\nu$ via~\eqref{fugacity} satisfies
\begin{align}  
\text{either} \qquad n < \ppmin(q), \qquad \text{or} \qquad q = \pm \ii \text{ if $n$ is odd} ,
\end{align}
see, e.g.,~\cite[theorem~\red{8.1}]{rsa} and~\cite[Corollary~\red{6.10}]{fp0}.
In this case, the collection $\smash{\big\{ \LS_n\super{s} \,\big| \, s \in \DefectSet_n \big\}}$ 
is the complete set of non-isomorphic simple $\TL_n(\nu)$-modules. 
If $\TL_n(\nu)$ is not semisimple, some of its standard modules $\smash{\LS_n\super{s}}$ are not simple 
(but still indecomposable), and certain quotients of the standard modules are simple, see proposition~\ref{SimpleModuleProp}.

%

\bigskip

With notation~(\ref{MultiindexNotation}--\ref{ndefn}), we next introduce basic representations of 
the Jones-Wenzl algebra $\WJ_\multii(\nu) \subset \TL_{\Summed_\multii}(\nu)$.
Like for the Temperley-Lieb algebra, the simple modules of $\WJ_\multii(\nu)$ are quotients of its standard modules.
Elements in the latter are \emph{$(\multii,s)$-Jones-Wenzl link states}, that is, link states of the form
\begin{align}\label{JWLinkState3} 
\vcenter{\hbox{\includegraphics[scale=0.275]{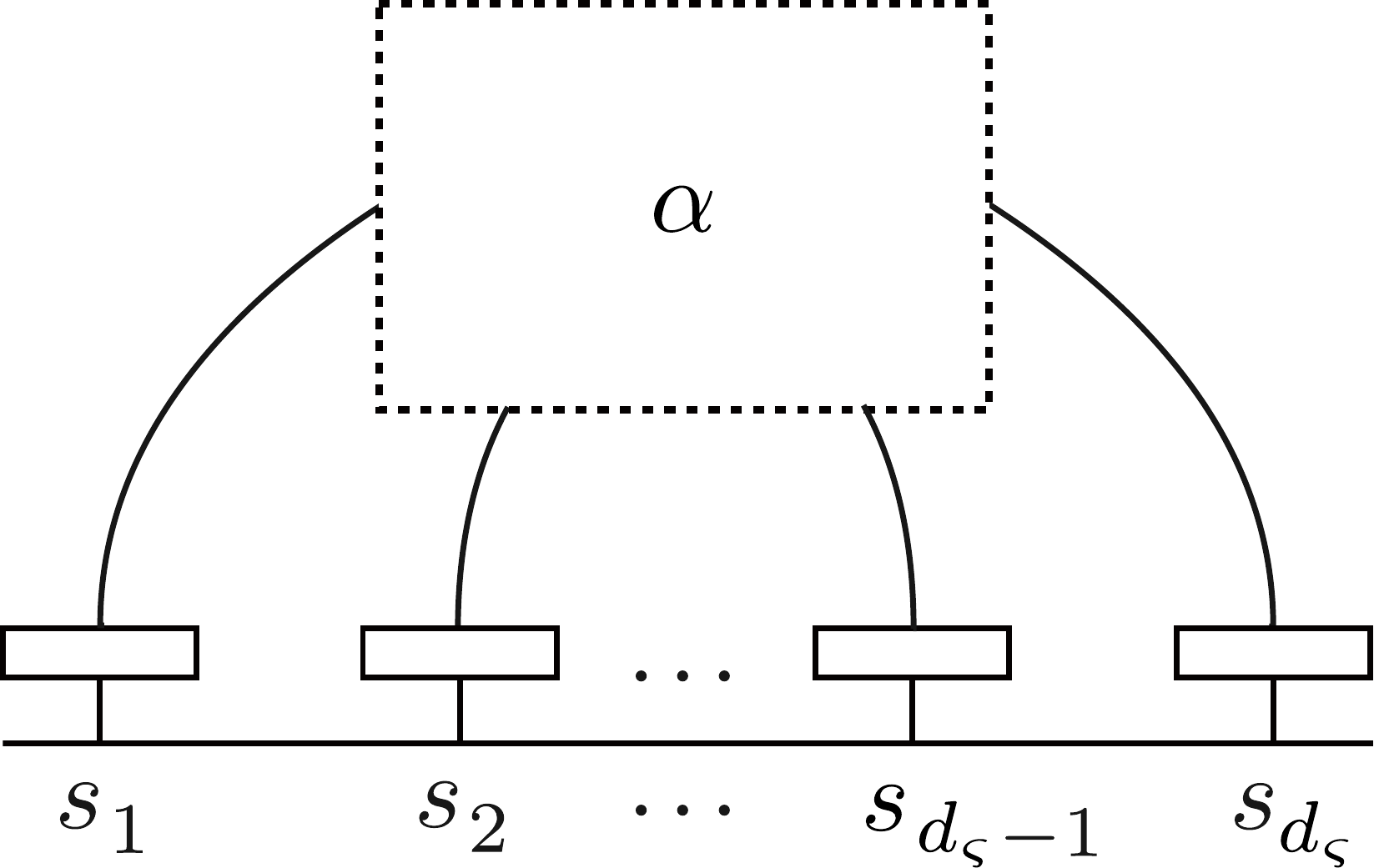} ,}}
\end{align}
for some 
$\alpha \in \smash{\LS_{\Summed_\multii}\super{s}}$.
By property~\eqref{ProjectorID2}, those link states~\eqref{JWLinkState3} that have a link with 
both endpoints at the same projector box are zero.
Also, if $\alpha$ is a $(\Summed_\multii, s)$-link pattern such that~\eqref{JWLinkState3} does not vanish, 
then we call~\eqref{JWLinkState3} a \emph{$(\multii,s)$-Jones-Wenzl link pattern}.
We define the sets of $(\multii,s)$-Jones-Wenzl link patterns and \emph{$\multii$-Jones-Wenzl link patterns} respectively by
\begin{align} \label{ProjsSpDefn2}  
\PP_\multii\super{s} := \WJProj_\multii \LP_{\Summed_\multii}\super{s} \setminus \{0\},
\qquad \qquad 
\PP_\multii := \WJProj_\multii \LP_{\Summed_\multii} \setminus \{0\}
= \bigcup_{s \, \in \, \DefectSet_\multii} \PP_\multii\super{s},
\end{align}
where $\DefectSet_\multii$ is the set of all integers $s \geq 0$ such that the set 
$\smash{\PP_\multii\super{s}}$ is not empty.
By breaking links into pairs of defects, it becomes evident that there are integers $\smin(\multii), \smax(\multii) \geq 0$ 
such that
\begin{align}\label{DefSet2} 
\DefectSet_\multii = \{ \smin(\multii), \smin(\multii) + 2, \; \ldots, \; \smax(\multii) \} . 
\end{align}
We refer to~\cite[section~\red{2B}]{fp0} 
for a complete determination of the set $\DefectSet_\multii$. 
Here we just note the obvious fact that the $(\multii,\Summed_\multii)$-Jones-Wenzl link pattern 
with only defects and no links has $\smax(\multii)$ defects, so
\begin{align} \label{smaxeq} 
\smax(\multii) = \Summed_\multii \overset{\eqref{ndefn}}{:=} \sIndex_1 + \sIndex_2 + \dotsm + \sIndex_{\np_\multii} .
\end{align}
Also, in the special cases of $\np_\multii = 1,2$, the set~\eqref{DefSet2} has a simple form~\cite[lemma~\red{2.1}]{fp0}:
\begin{align} \label{SpecialDefSet} 
\DefectSet\sub{s} = \{s\} \qquad \qquad \textnormal{and} \qquad\qquad \DefectSet\sub{r,t} = \{ |r-t| , |r-t| + 2, \ldots, r+t\} . 
\end{align}
The sets $\smash{\PP_\multii\super{s}}$ and $\PP_\multii$ 
span the complex vector spaces of $(\multii,s)$-Jones-Wenzl link states and \emph{$\multii$-Jones-Wenzl link states},
respectively defined as
\begin{align} \label{ProjsSpDefn}  
\PS_\multii\super{s} &:= \WJProj_\multii  \LS_{\Summed_\multii}\super{s} = \big\{ \WJProj_\multii \alpha \,|\, \alpha \in \LS_{\Summed_\multii}\super{s} \big\}, \\
\PS_\multii &:= \WJProj_\multii  \LS_{\Summed_\multii} = \big\{ \WJProj_\multii \alpha \,|\, \alpha \in \LS_{\Summed_\multii} \big\} 
= \bigoplus_{s \, \in \, \DefectSet_\multii} \PS_\multii\super{s} .
\end{align}
From~\eqref{ProjectorID2}, it is straightforward to see that they are $\WJ_\multii(\nu)$-modules. We call 
$\smash{\PS_\multii\super{s}}$ a $\WJ_\multii(\nu)$-\emph{standard module} and 
$\PS_\multii$ the $\WJ_\multii(\nu)$-\emph{link state module}.
The representation theory of the Jones-Wenzl algebra is analogous to that of the Temperley-Lieb algebra,
and we summarize salient facts about it in theorem~\ref{BigSSTHM} below.

We also note that by~\cite[corollaries~\red{2.7} and~\red{B.2}]{fp0}, the dimension of the Jones-Wenzl algebra is
\begin{align} \label{DimOfWJ}
\dim \WJ_\multii(\nu) = \sum_{s \, \in \, \DefectSet_\multii} \big(\Dim_\multii\super{s}\big)^2 ,
\end{align}
where for each $s \in \DefectSet_\multii$, using notation~\eqref{hats}, 
$\smash{\{\Dim_\multii\super{s}\}_{s \in \DefectSet_\multii}}$ is the unique solution to the recursion
\begin{align} 
\label{PreRecursion2} 
\Dim_\multii\super{s} \hspace{.1cm}
 = \hspace{.5cm} \sum_{\mathclap{r \, \in \, \DefectSet_{\lds} \, \cap \, \DefectSet\sub{s,t} }}  
\quad \Dim_{\lds}\super{r} 
\quad \qquad \text{and} \quad \qquad \Dim \sub{s}\super{s} = 1 .
\end{align}

\subsection{Bilinear form}

The simple modules of the Temperley-Lieb algebra are quotients of the standard modules 
by the radical of a natural invariant bilinear form. This bilinear form is given by evaluations of ``networks.''
A \emph{network} is a collection of nonintersecting, non-self-intersecting planar loops and paths within a rectangle. 
A path in a network can be a \emph{through-path}, which is a curve that respectively enters and 
exits the network at the bottom and top sides of the rectangle, or
a \emph{turn-back path}, which enters and exits 
at the same side of the rectangle, either top or bottom:
\begin{align}
\vcenter{\hbox{\includegraphics[scale=0.275]{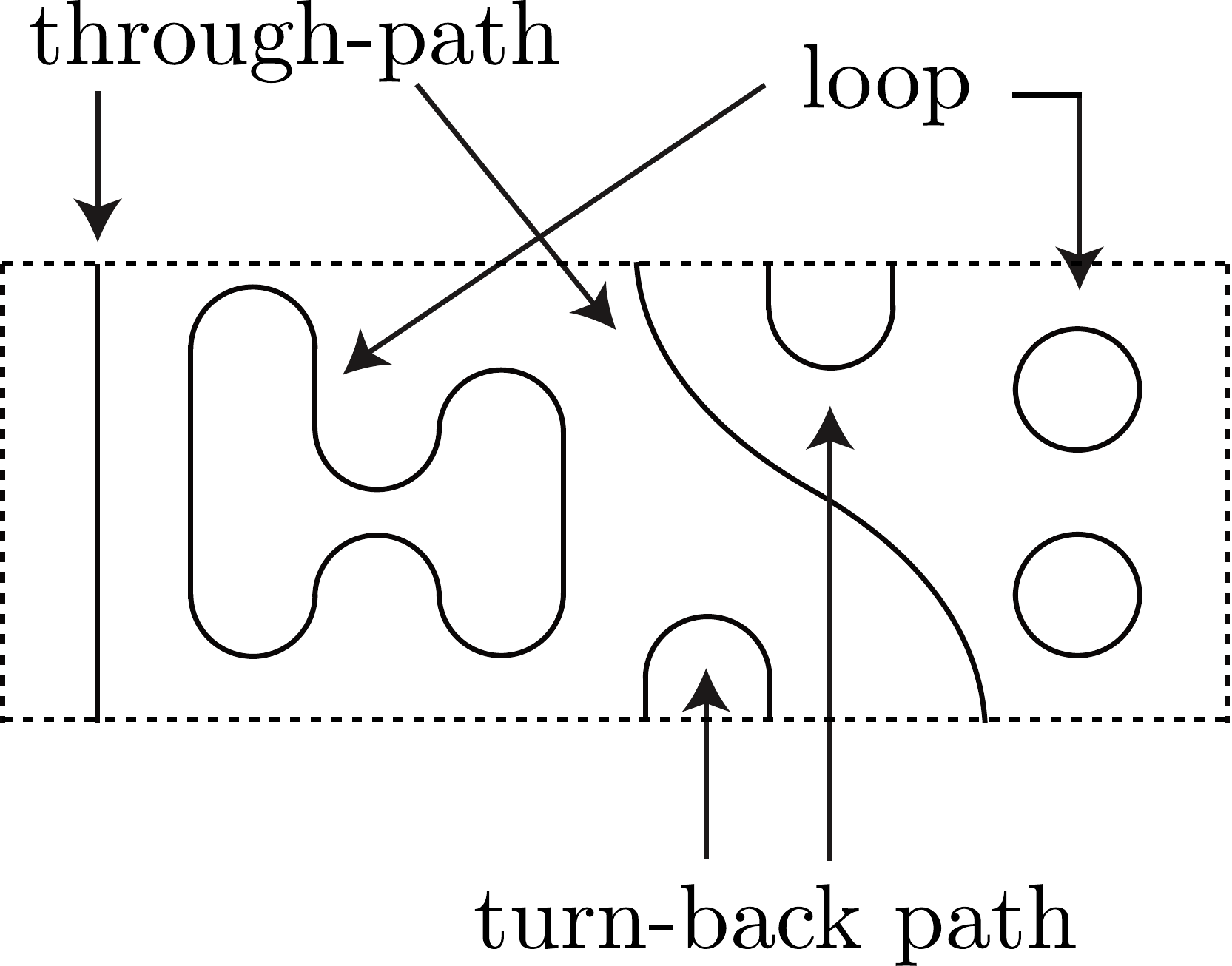} .}}
\end{align}
The \emph{evaluation of a network $T$} is a complex number obtained by
assigning all loops, through-paths, and turn-back paths in $T$ the following weights in $\bC$:
\begin{alignat}{7} 
\label{LoopWeight} 
& \text{loop weight (fugacity):} \qquad \qquad 
& \vcenter{\hbox{\includegraphics[scale=0.275]{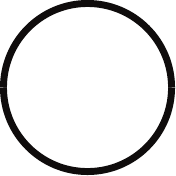}}}  \qquad & \text{and} \qquad  
& \raisebox{-18pt}{\includegraphics[scale=0.275]{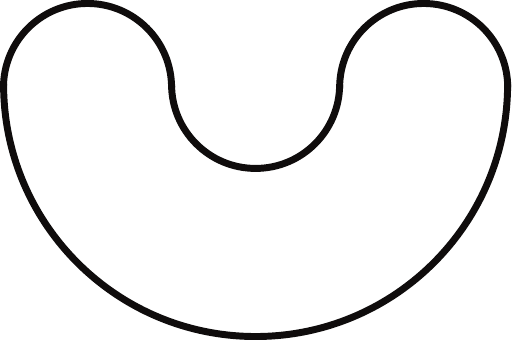}}  \qquad & \text{and} \qquad
& \raisebox{-11pt}{\includegraphics[scale=0.275]{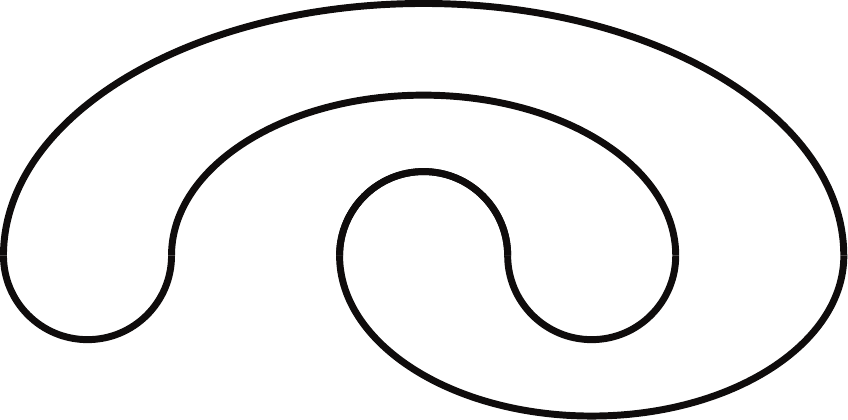}}  \qquad & \text{etc.}  
\quad = \quad \nu , \\[1em]
\label{ThroughPathWeight}
& \text{through-path weight:} \qquad \qquad
& \vcenter{\hbox{\hspace*{-4mm} \includegraphics[scale=0.275]{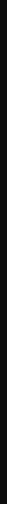}}}  \qquad & \text{and} \qquad
& \vcenter{\hbox{\includegraphics[scale=0.275]{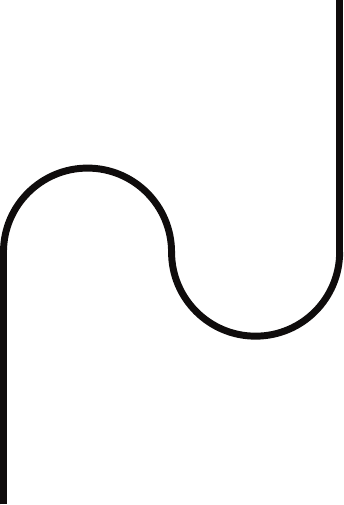} \hspace*{2mm}}}  \qquad & \text{and} \qquad
& \vcenter{\hbox{\includegraphics[scale=0.275]{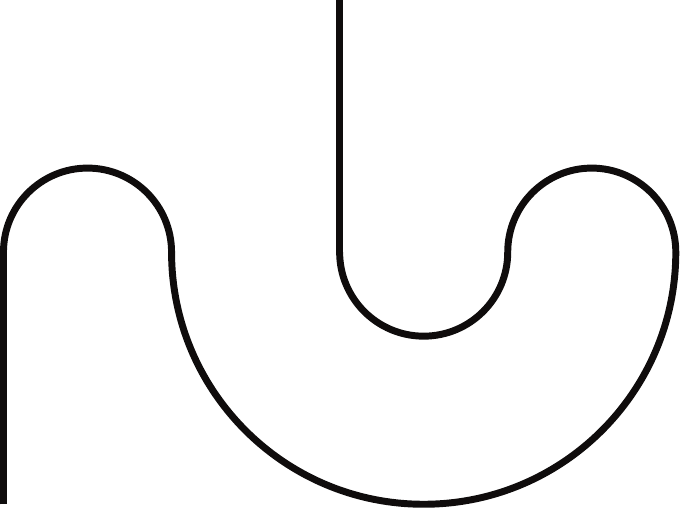}\hspace*{3mm}}}  \qquad & \text{etc.} 
\quad = \quad 1 , \\[1em]
\label{TurnBack0}
& \text{turn-back path weight:} \qquad \qquad
& \raisebox{-5pt}{\includegraphics[scale=0.275]{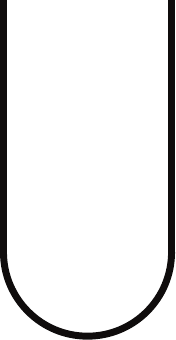}}  \qquad & \text{and} \qquad
& \raisebox{-5pt}{\includegraphics[scale=0.275]{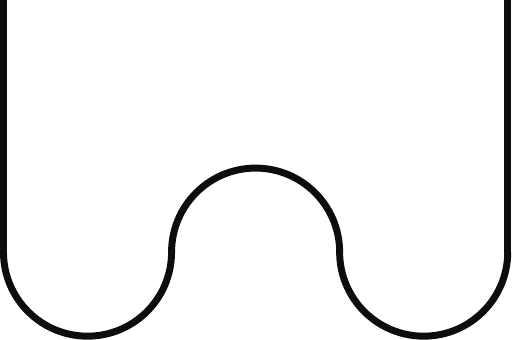}}  \qquad & \text{and} \qquad
& \raisebox{-19pt}{\includegraphics[scale=0.275]{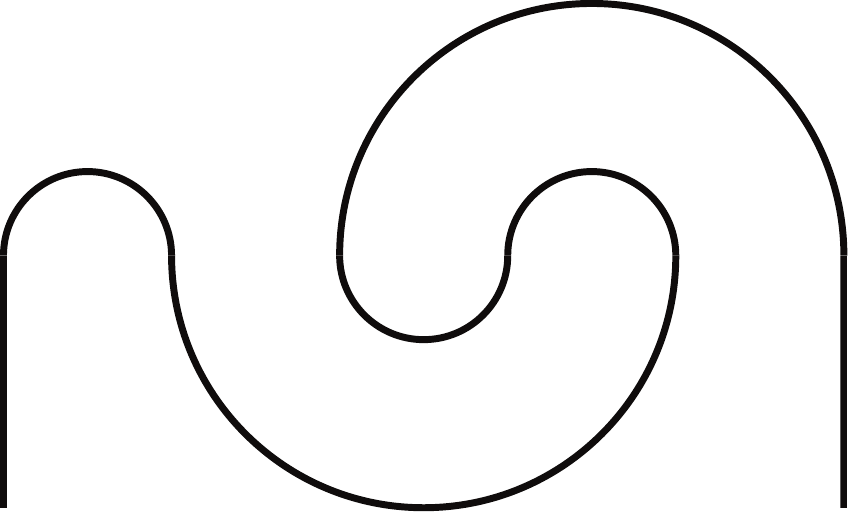}}  \qquad & \text{etc.} 
\quad = \quad 0 ,
\end{alignat}
and multiplying all of these factors together:
\begin{align} \label{evT} 
(\, T \,) &:= \prod \{ \text{the weights of all connected components in the network $T$} \} \\
\label{evT2} & = 
\begin{cases} 
\nu^{\textnormal{\# loops in $T$}} , 
& \textnormal{if the network $T$ has no turn-back path} , \\ 
0 , & \textnormal{if the network $T$ has a turn-back path,}
\end{cases}
\end{align}
where ``\# loops in $T$" stands for the number of loops in the network $T$.

We define the bilinear form on the link state module $\LS_n$ as follows. For two link patterns 
$\alpha,\beta \in \LP_n$, we horizontally reflect $\alpha$ so it is upside down, we concatenate it to $\beta$ from below, and delete the overlapping 
horizontal lines of $\alpha$ and $\beta$.  The resulting diagram is a network $\alpha \BarAction \beta$.  For instance, we have
\begin{align} 
\label{LinkStateRule} 
& \alpha \; = \; \raisebox{1pt}{\includegraphics[scale=0.275]{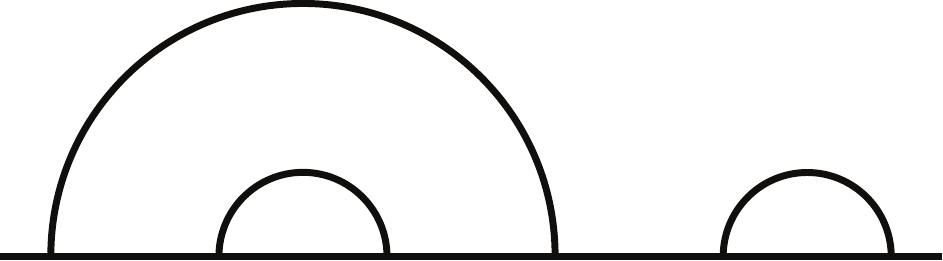}} \; , \qquad 
\beta \; = \; \raisebox{1pt}{\includegraphics[scale=0.275]{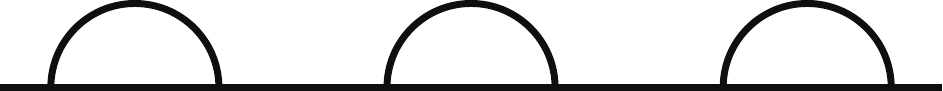}}
\qquad \qquad \Longrightarrow \qquad \qquad
\alpha \BarAction \beta \; = \;
\raisebox{-19pt}{\includegraphics[scale=0.275]{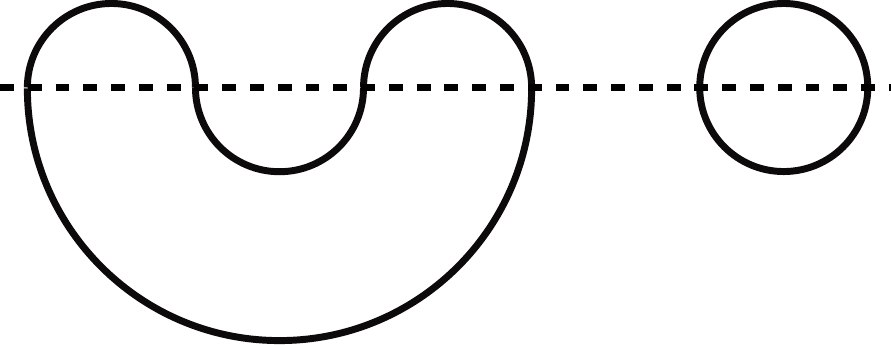}} \; , \\[.3em]
\label{LinkStateRuleDefect2}
& \alpha \; = \; \raisebox{1pt}{\includegraphics[scale=0.275]{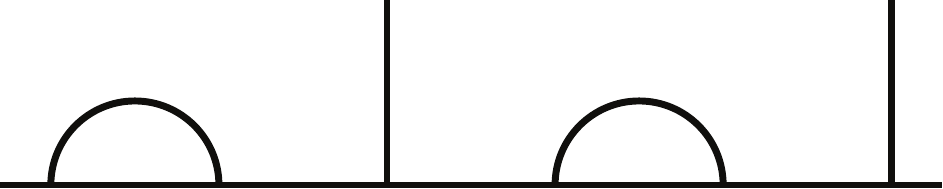}} \; , \qquad 
\beta \; = \; \raisebox{1pt}{\includegraphics[scale=0.275]{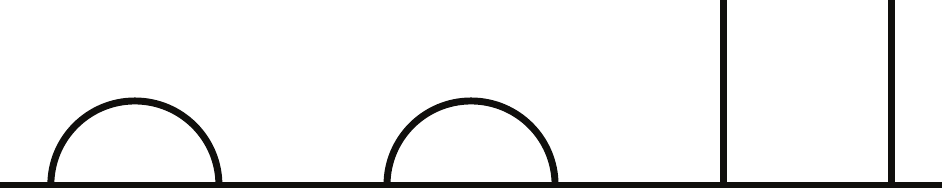}}
\qquad \qquad \Longrightarrow \qquad \qquad
\alpha \BarAction \beta \; = \;
\vcenter{\hbox{\includegraphics[scale=0.275]{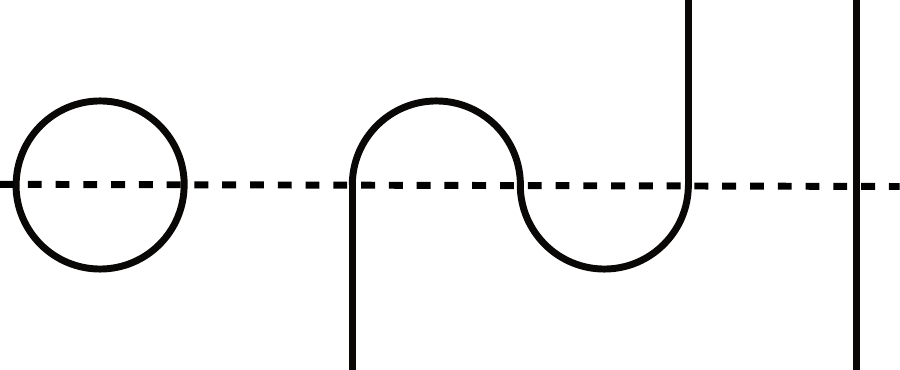}}} \; , \\[1em]
\label{LinkStateRuleDefect1}
& \alpha \; = \; \raisebox{1pt}{\includegraphics[scale=0.275]{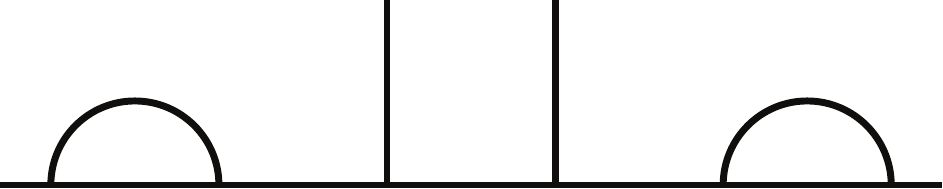}} \; , \qquad 
\beta \; = \; \raisebox{1pt}{\includegraphics[scale=0.275]{e-Connectivities6.pdf}} 
\qquad \qquad \Longrightarrow \qquad \qquad
\alpha \BarAction \beta \; = \;
\vcenter{\hbox{\includegraphics[scale=0.275]{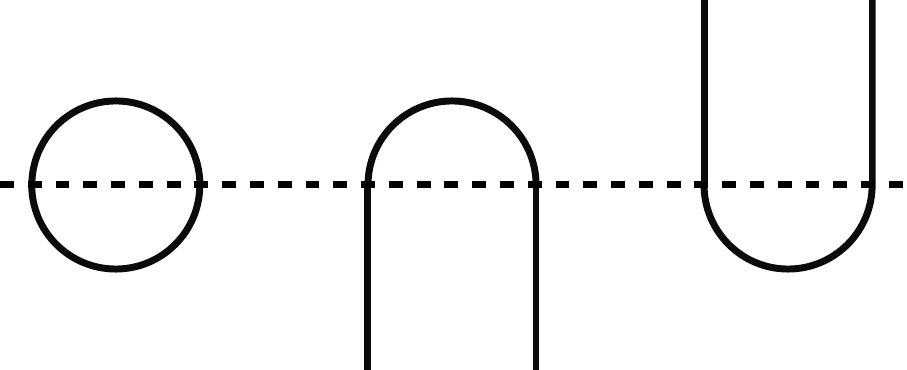}}} \; .
\end{align}
Then we define the \emph{link state bilinear form} 
$\smash{\BiForm{\cdot}{\cdot}} \colon \LS_n \times \LS_n \longrightarrow \bC$ by bilinear extension of the rule
\begin{align} \label{LSBiForm} 
(\alpha,\beta) \quad \longmapsto \quad \BiForm{\alpha}{\beta} ,
\end{align}
for 
all link patterns $\alpha,\beta \in \LP_n$.
If $\alpha,\beta \in \LS_0$, then the product in~\eqref{evT} is empty, so we take $\BiForm{\alpha}{\beta} = 1$.  
As examples, the bilinear forms $\BiForm{\alpha}{\beta}$ of the link patterns $\alpha$ and $\beta$ in 
(\ref{LinkStateRule},~\ref{LinkStateRuleDefect2},~\ref{LinkStateRuleDefect1}) respectively evaluate to 
\begin{align}
\bigg( \; \raisebox{-18pt}{\includegraphics[scale=0.275]{e-Connectivities4.pdf}} \; \bigg) \;
= \; \nu^2 , \qquad \qquad
\bigg( \; \vcenter{\hbox{\includegraphics[scale=0.275]{e-Connectivities10.pdf}}}  \; \bigg) \;
= \; \nu , \qquad \qquad
\bigg( \; \vcenter{\hbox{\includegraphics[scale=0.275]{e-Connectivities8.pdf}}}  \; \bigg) \;
= \; 0 . 
\end{align}

The link state bilinear form $\BiForm{\cdot}{\cdot}$ is \emph{symmetric}:
$\BiForm{\alpha}{\beta} = \BiForm{\beta}{\alpha}$, and \emph{invariant} in the sense that,
for all link patterns $\alpha, \beta \in \LS_n$ and for all tangles $T \in \TL_n(\nu)$, we have
\begin{align}
\label{InvarProp} \BiForm{\alpha}{T \beta} = \; & \BiForm{T^\dagger \alpha}{\beta} ,
\end{align}
where $T^\dagger \in \TL_n(\nu)$ denotes the tangle obtained by reflecting $T$ about a vertical axis:
\begin{align}\label{DaggerRefl} 
T \quad = \quad \vcenter{\hbox{\includegraphics[scale=0.275]{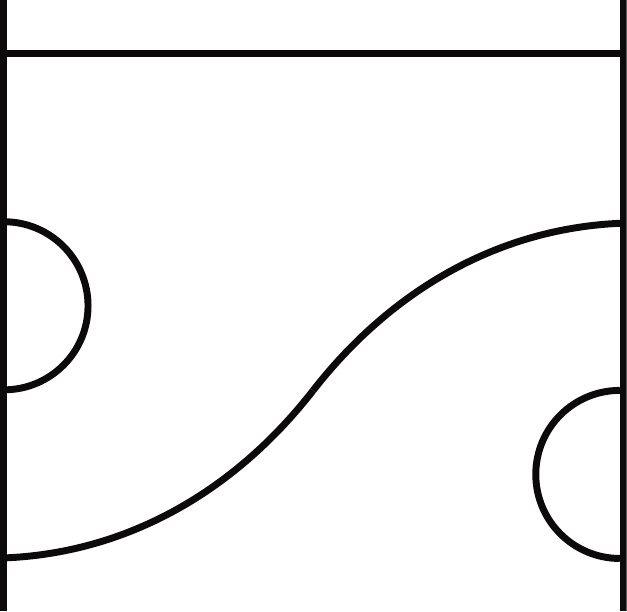}}} 
\qquad \qquad \Longrightarrow \qquad \qquad
T^\dagger \quad = \quad \vcenter{\hbox{\includegraphics[scale=0.275]{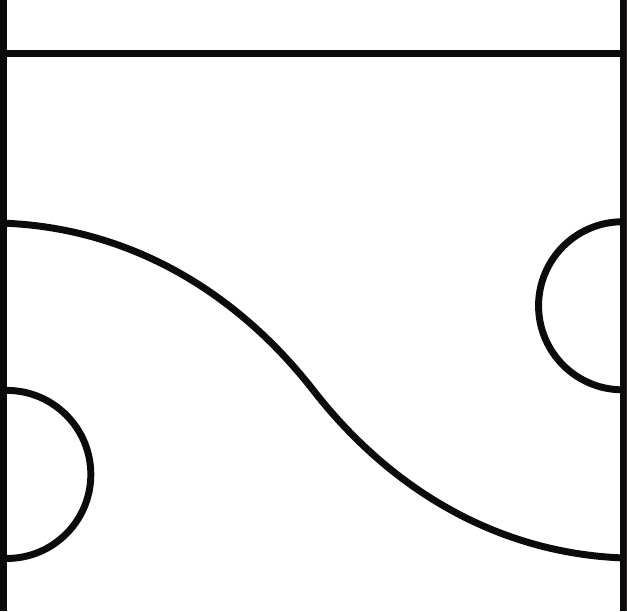} .}}
\end{align}
Invariance property~\eqref{InvarProp} guarantees that the radical of the bilinear form is a $\TL_n(\nu)$-submodule of $\LS_n$.

The restriction of the bilinear form $\BiForm{\cdot}{\cdot}$ to $\PS_\multii \subset \LS_{\Summed_\multii}$
defines a bilinear form on $\multii$-Jones-Wenzl link states. 
Property~\eqref{InvarProp} shows that the radical of this bilinear form,
\begin{align} 
\rad \PS_\multii := \; & \big\{\alpha \in \PS_\multii \, \big| \, \text{$\BiForm{\alpha}{\beta} = 0$ for all $\beta \in \PS_\multii$} \big\},
\end{align}
is a $\WJ_\multii(\nu)$-submodule of $\PS_\multii$. 
Furthermore, because $\multii$-Jones-Wenzl link patterns with different numbers of defects are orthogonal, $\rad \PS_\multii$
equals a direct sum of the radicals of the standard modules $\smash{\PS_\multii\super{s}}$:
\begin{align} 
\rad \PS_\multii = \bigoplus_{s \, \in \, \DefectSet_\multii} \rad \PS_\multii\super{s}, \qquad \text{where} \quad
\rad\smash{\PS_\multii\super{s}} := \; & \big\{\alpha \in\smash{\PS_\multii\super{s}} \, \big| \, \text{$\BiForm{\alpha}{\beta} = 0$ 
for all $\beta \in \smash{\PS_\multii\super{s}}$} \big\} ,
\end{align}
and, for each $s \in \DefectSet_\multii$, the radical
$\rad \smash{\PS_\multii\super{s}}$ is a $\WJ_\multii(\nu)$-submodule of $\smash{\PS_\multii\super{s}}$.

If $\rad \PS_\multii = \{0\}$, then 
for the $(\multii,s)$-Jones-Wenzl link pattern basis $\smash{\{ \alpha \in \smash{\PP_\multii\super{s}}\}}$ 
of the $\WJ_\multii(\nu)$-standard module $\smash{\PS_\multii\super{s}}$, we define 
the dual basis $\smash{\{ \alpha^\cheque \, | \, \in \smash{\PP_\multii\super{s}}\}}$ via the property
\begin{align} \label{DualLS} 
\BiForm{\alpha^\cheque}{\beta} = \delta_{\alpha,\beta} .
\end{align}

\subsection{Simple modules and semisimplicity}

Next, we summarize facts from the representation theory of the Jones-Wenzl algebra.
(We note that the Temperley-Lieb algebra is also included as a special case, with $\multii = (1,1,\ldots,1)$.)

\begin{theorem} \textnormal{\cite[theorem~\red{6.9}]{fp0}} \label{BigSSTHM}
Suppose $\max \multii < \ppmin(q)$. The following statements are equivalent:
\begin{enumerate}[label = \arabic*., ref = \arabic*]
\itemcolor{red}
\item \label{SSitem0}
The Jones-Wenzl algebra $\WJ_\multii(\nu)$ is semisimple, i.e., its Jacobson radical is trivial: $\rad \WJ_\multii(\nu) = \{0\}$.

\item \label{SSitem4}
We have $\rad \PS_\multii = \{0\}$.

\item \label{SSitem5}
The link state representation induced by the action of $\WJ_\multii (\nu)$ on $\PS_\multii$ is faithful.

\item \label{SSitem6}
The link state representation 
induces an isomorphism of algebras from $\WJ_\multii (\nu)$ to 
$\smash{\underset{s \, \in \, \DefectSet_\multii}{\bigoplus} \End \PS_\multii\super{s}}$.

\item \label{SSitem2}
The collection $\smash{\big\{ \PS_\multii\super{s} \,\big| \, s \in \DefectSet_\multii \big\}}$ 
is the complete set of non-isomorphic simple $\WJ_\multii(\nu)$-modules.

\end{enumerate}
If $\Summed_\multii < \ppmin(q)$, then all of the above properties hold.
\end{theorem}
\begin{proof}
The assertion follows from~\cite[theorem~\red{6.9}]{fp0} combined with the following facts:
by~\cite[corollary~\red{B.2}]{fp0}, the algebras $\TL_\multii (\nu)$ 
(the valenced Temperley-Lieb algebra defined in~\cite[section~\red{2D}]{fp0}) and $\WJ_\multii (\nu)$ are isomorphic, 
and their standard modules 
are isomorphic too, and furthermore, by~\cite[corollary~\red{B.3}]{fp0}, the radicals of the 
standard modules are isomorphic as well.
\end{proof}

In~\cite{fp0}, we also give an explicit characterization in terms of $q$ for
properties~\ref{SSitem0}--\ref{SSitem2} stated in theorem~\ref{BigSSTHM} to hold.
When none of the properties~\ref{SSitem0}--\ref{SSitem2} holds, $\WJ_\multii(\nu)$ fails to be semisimple.
However, also in this case, the complete set of simple $\WJ_\multii(\nu)$-modules is known to be given by 
quotients of the standard modules by their radicals:

\begin{prop} \textnormal{\cite[proposition~\red{6.7}]{fp0}}  \label{SimpleModuleProp}
Suppose $\max \multii < \ppmin(q)$.  The collection 
\begin{align} 
\big\{ \PS_\multii\super{s} / \rad \PS_\multii\super{s} \,\big| \, \rad \PS_\multii\super{s} \neq \PS_\multii\super{s} \big\}
\end{align}
is the complete set of non-isomorphic simple $\WJ_\multii(\nu)$-modules.
\end{prop}
\begin{proof}
This follows from~\cite[proposition~\red{6.7}]{fp0} by the isomorphisms~\cite[corollaries~\red{B.2} and~\red{B.3}]{fp0}.
\end{proof}

\subsection{Cellularity of the Jones-Wenzl algebra} \label{CellSec}

The purpose of this section is to specify the connection of the Jones-Wenzl algebra $\WJ_\multii(\nu)$ to 
the theory of ``cellular algebras'' \`a la J.~Graham and G.~Lehrer~\cite{gl}.
Many diagram algebras, such as the Temperley-Lieb algebra, are cellular.
Graham and Lehrer have developed a category-theoretic approach for investigating the representation theory 
of such algebras, and indeed, quite general results can be obtained using this abstract theory.

Cellular algebras are characterized by the existence of a ``cellular basis'' with certain useful properties.

\begin{defn} \textnormal{\cite[Definition~\red{1.1}]{gl}} \label{CellDef}
A \emph{cellular $\bC$-algebra} is an associative algebra $(\mathsf{A}, \cdot)$ over the field $\bC$
equipped with the additional structure of a \emph{cell datum} $(\Poset, \Cell, \CellMap, \star)$, where
\begin{enumerate} 
\itemcolor{red}

\item $\Poset$ is a finite partially ordered set,

\item $\smash{\Cell = \{ \Cell(\poset) \, | \, \poset \in \Poset \}}$ is a collection of finite sets $\smash{\Cell(\poset)}$,

\item $\CellMap$ is a function from 
$\smash{\underset{\poset \, \in \, \Poset}{\bigsqcup} \Cell(\poset) \times \Cell(\poset)}$ to $\mathsf{A}$, and

\item $\star$ is a $\bC$-linear map from $\mathsf{A}$ to $\mathsf{A}$,
\end{enumerate}
such that the following properties hold:
\begin{enumerate}[label = $\mathrm{c}$\arabic*., ref = $\mathrm{c}$\arabic*]
\itemcolor{red}
\item \label{Cellitem1}
The function 
$\CellMap \colon \smash{\underset{\poset \, \in \, \Poset}{\bigsqcup} \Cell(\poset) \times \Cell(\poset)} \longrightarrow \mathsf{A}$ 
is injective and its image is a basis for $\mathsf{A}$.

\item \label{Cellitem2}
The map $\star \colon \mathsf{A} \longrightarrow \mathsf{A}$ is an anti-involution such that 
$\CellMap(\alpha,\beta)^\star  = \CellMap(\beta,\alpha)$, for all $\alpha,\beta \in \smash{\Cell(\poset)}$ and $\poset \in \Poset$.

\item \label{Cellitem3}
For any two elements $\alpha,\beta \in \smash{\Cell(\poset)}$ and for all elements $a \in \mathsf{A}$, we have
\begin{align}
a \cdot \CellMap(\alpha,\beta) = \sum_{\gamma \, \in \, \Cell(\poset)} c_a(\gamma,\alpha) \, \CellMap(\gamma,\beta)
\qquad \big( \textnormal{mod } \mathsf{A}\super{< \poset} \big) ,
\end{align}
where the coefficients $c_a(\gamma,\alpha) \in \bC$ are independent of $\beta$, and where
\begin{align}
\mathsf{A}\super{< \poset} = \Span \big\{\CellMap(\delta,\epsilon) \, \big| \, \delta,\epsilon \in \Cell(\posetprime) \text{ and } \posetprime < \poset \big\} .
\end{align}
\end{enumerate}
\end{defn}

The image of the function $\CellMap$ is called a \emph{cellular basis} for $\mathsf{A}$.
For each $\poset \in \Poset$, the set $\smash{\Cell(\poset)}$ spans an $\mathsf{A}$-module
$\smash{\CellMod\super{\poset} = \Span \Cell(\poset)}$, called a \emph{cell module}, with $\mathsf{A}$-action
defined by linear extension of~\cite[Definition~\red{2.1}]{gl} 
\begin{align} \label{cellmodule}
a . \beta := \sum_{\gamma \, \in \, \Cell(\poset)} c_a(\gamma,\beta) \gamma ,
\end{align}
for all basis elements $\beta \in \Cell(\poset)$ of the module $\smash{\CellMod\super{\poset}}$ and for all elements $a \in \mathsf{A}$ of the algebra.

It was proved in~\cite{gl} that any cellular algebra admits nice structural properties, including the following:
\begin{itemize}[leftmargin=*]
\item On each $\smash{\CellMod\super{\poset}}$, there exists a symmetric invariant bilinear form $\Lpar \cdot, \cdot \Rpar$, which is determined by the formula
\begin{align} \label{GLBiform}
\CellMap(\alpha,\beta) \CellMap(\gamma,\delta) = \Lpar \beta,\gamma \Rpar \, \CellMap(\alpha, \delta) \qquad \big( \textnormal{mod } \mathsf{A}\super{< \poset} \big) ,
\end{align}
for all $\alpha, \beta, \gamma, \delta \in \smash{\CellMod\super{\poset}}$.

\item 
The radical $\rad \smash{\CellMod\super{\poset}}$
of the bilinear form on $\smash{\CellMod\super{\poset}}$ is an $\mathsf{A}$-submodule of $\smash{\CellMod\super{\poset}}$.
When $\rad \smash{\CellMod\super{\poset}} \neq \smash{\CellMod\super{\poset}}$, the quotient module 
$\smash{\CellMod\super{\poset} / \rad \CellMod\super{\poset}}$ is simple, and $\rad \smash{\CellMod\super{\poset}}$
is equal to the intersection of all maximal proper submodules of $\smash{\CellMod\super{\poset}}$.

\item The collection $\smash{\big\{ \CellMod\super{\poset} / \rad \CellMod\super{\poset} \,\big| \, \rad \CellMod\super{\poset} \neq \CellMod\super{\poset} \big\}}$
is the complete set of non-isomorphic simple $\mathsf{A}$-modules. 

\item The following statements are equivalent:
\begin{itemize}

\item[(a):] The algebra $\mathsf{A}$ is semisimple.

\item[(b):] All of the nonzero $\mathsf{A}$-modules $\smash{\CellMod\super{\poset} / \rad \CellMod\super{\poset}}$
are simple and non-isomorphic.

\item[(c):] The bilinear form on the direct sum module $\smash{\underset{\poset \, \in \, \Poset}{\bigoplus} \CellMod\super{\poset}}$
is nondegenerate, that is, $\rad \smash{\CellMod\super{\poset}} = \{0\}$, for all $\poset \in \Poset$.
\end{itemize}
\end{itemize}

The above properties are very reminiscent of some of the results stated in theorem~\ref{BigSSTHM} and proposition~\ref{SimpleModuleProp}.
Indeed, the Jones-Wenzl algebra is cellular, so one could apply the theory of~\cite{gl} to it.
However, this would not give concrete knowledge about, e.g., the structure of the simple modules and their radicals,
such as explicit bases and dimensions for them. In the case of the Temperley-Lieb and Jones-Wenzl algebras,
very concrete investigations were carried out by D.~Ridout and Y.~Saint-Aubin for the case of $\TL_n(\nu)$~\cite{rsa},
and by the authors of the present article in~\cite{fp0}.

\begin{prop} \label{CellPropo}
Suppose $\max \multii < \ppmin(q)$. The Jones-Wenzl algebra $\WJ_\multii(\nu)$ is cellular.
\end{prop}
\begin{proof}
We verify that $\WJ_\multii(\nu)$ has the following cell datum $\big(\DefectSet_\multii, \PP_\multii, \BarAction \cdot \quad \cdot \BarAction, \star\big)$:
\begin{enumerate}[leftmargin=*]
\itemcolor{red}

\item $\Poset = \DefectSet_\multii$, defined in~\eqref{DefSet2}, is a finite partially ordered set,

\item $\smash{\Cell(s) = \PP_\multii\super{s}}$, defined in~\eqref{ProjsSpDefn2}, for all $s \in \DefectSet_\multii$, are finite sets,

\item  \label{Basistem} 
$\CellMap = \BarAction \cdot \quad \cdot \BarAction \colon 
\smash{\underset{s \, \in \, \DefectSet_\multii}{\bigcup} \PP_\multii\super{s} \times \PP_\multii\super{s}} \longrightarrow \PD_\multii$,
defined in lemma~\ref{WJSandwichLem}
by the assignment
\begin{align} \label{Sandwichmap}
\BarAction \alpha \quad \beta \BarAction
\quad := \quad \vcenter{\hbox{\includegraphics[scale=0.275]{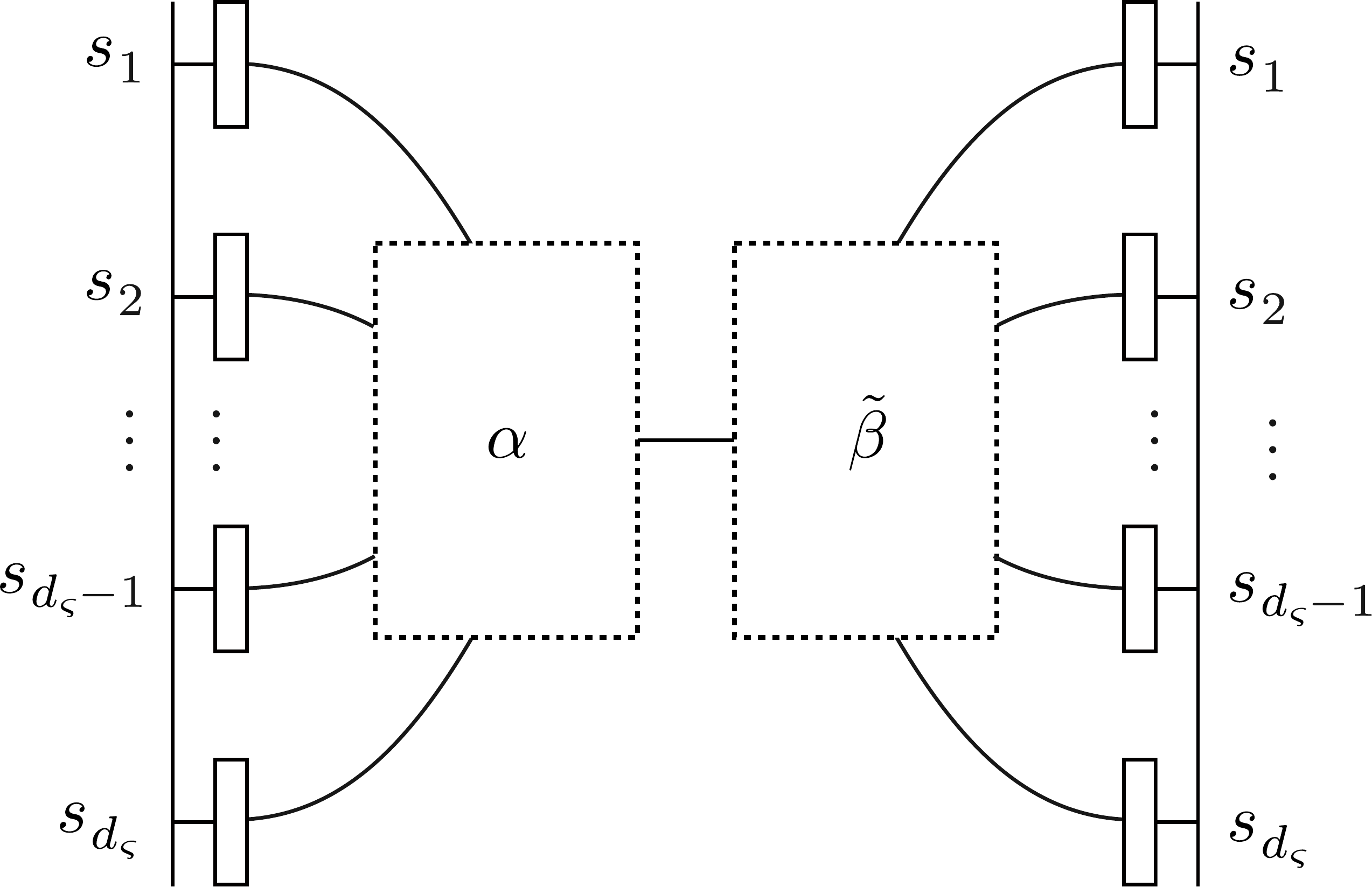} ,}}
\end{align}
where $\tilde{\beta}$ is obtained by reflecting $\beta$ about a vertical axis, 
is a bijection onto the basis $\PD_\multii$ of $\WJ_\multii(\nu)$, and

\item $\star = \dagger$, defined via~\eqref{DaggerRefl}, is a $\bC$-linear anti-involution from $\WJ_\multii(\nu)$ to $\WJ_\multii(\nu)$, 
and by definition, for all $\multii$-Jones-Wenzl link patterns $\alpha, \beta \in \PP_\multii$, we have 
\begin{align}
\BarAction \alpha \quad \beta \BarAction^\dagger = \BarAction \beta \quad \alpha \BarAction .
\end{align}
\end{enumerate}

Properties~\ref{Cellitem1} and~\ref{Cellitem2} for $\WJ_\multii(\nu)$ are immediate from the definitions. 
To check property~\ref{Cellitem3}, we note that
\begin{align}
\WJ_\multii(\nu)\super{< s} = \Span \{\BarAction \delta \quad \epsilon \BarAction \, | \, \delta, \epsilon \in \PP_{\multii}\super{r} \text{ and } r < s \} .
\end{align}
We first consider a $\multii$-Jones-Wenzl tangle $T \in \WJ_\multii(\nu)$.
We expand it in the basis 
$\smash{\big\{ \BarAction \gamma \quad \delta \BarAction \, \big| \, s \in \DefectSet_\multii, \; \gamma, \delta \in \PS_\multii\super{s}\big\}}$
as
\begin{align} \label{Ttangel}
T = \sum_{r \, \in \, \DefectSet_\multii} \sum_{\gamma , \delta \, \in \, \PP_{\multii}\super{r}} 
c_{\gamma,\delta}\super{r} \BarAction \gamma \quad \delta \BarAction ,
\end{align}
for some coefficients $\smash{c_{\gamma,\delta}\super{r}} \in \bC$. 
Then we let $\alpha, \beta \in \smash{\PP_\multii\super{s}}$ for some $s \in \DefectSet_\multii$.
We have
\begin{align}
\label{Tcellaction}
T \BarAction \alpha \quad \beta \BarAction 
& \underset{\eqref{Ttangel}}{\overset{\eqref{wjrecursion}}{=}} 
\sum_{\substack{r \, \in \, \DefectSet_\multii \\ r \geq s}} \sum_{\gamma , \delta \, \in \, \PP_{\multii}\super{r}} 
c_{\gamma,\delta}\super{r} \BarAction \gamma \quad \delta \BarAction \BarAction \alpha \quad \beta \BarAction 
\qquad \big( \textnormal{mod } \WJ_\multii(\nu)\super{< s} \big) ,
%
\end{align}
because by recursion property~\eqref{wjrecursion} of the Jones-Wenzl projectors, 
each tangle $\BarAction \gamma \quad \delta \BarAction \BarAction \alpha \quad \beta \BarAction$
with $\gamma, \delta  \in \smash{\PP_{\multii}\super{r}}$
is a linear combination of $\multii$-Jones-Wenzl link diagrams with at most $\min(r,s)$ crossing links.
Furthermore, the diagrams in~\eqref{Tcellaction} with exactly $s$ crossing links
have necessarily the form $\BarAction \eta \quad \beta \BarAction$ for some $\eta \in \smash{\PP_{\multii}\super{s}}$.
In conclusion, we have 
\begin{align}
T \BarAction \alpha \quad \beta \BarAction 
& \overset{\eqref{Tcellaction}}{=}
\sum_{\eta \, \in \, \PP_{\multii}\super{s}} c_T(\eta,\alpha) \BarAction \eta \quad \beta \BarAction 
\qquad \big( \textnormal{mod } \WJ_\multii(\nu)\super{< s} \big) ,
\end{align}
for some coefficients $c_T(\eta,\alpha) \in \bC$ that depend on $\alpha$ and $T$, but not on $\beta$.
This shows property~\ref{Cellitem3} for $\WJ_\multii(\nu)$.
\end{proof}

We remark that, admitting the fact from~\cite{gl} that the Temperley-Lieb algebra $\TL_{\Summed_\multii}(\nu)$ is cellular,
proposition~\ref{CellPropo} also follows immediately from~\cite[proposition~\red{4.3}]{kxi} and definition~\eqref{WJAdef} 
of the Jones-Wenzl algebra $\WJ_\multii(\nu) := \WJProj_\multii \TL_{\Summed_\multii}(\nu) \WJProj_\multii$.
(For this, one has to note that 
$\WJProj_\multii \in \TL_{\Summed_\multii}(\nu)$ is an idempotent fixed by the involution $T \mapsto T^\dagger$ in $\TL_{\Summed_\multii}(\nu)$.)
However, it is instructive to prove proposition~\ref{CellPropo} by providing with explicit cell datum for $\WJ_\multii(\nu)$.
Moreover, cellularity of the Temperley-Lieb algebra also follows a special case of the above calculation.

\bigskip

With cellular basis 
$\smash{\big\{ \BarAction \alpha \quad \beta \BarAction \, \big| \, s \in \DefectSet_\multii, \; \alpha, \beta \in \PS_\multii\super{s}\big\}}$,
the $\WJ_\multii(\nu)$-action~\eqref{cellmodule} on the cell module $\smash{\PS_{\multii}\super{s}}$ reads
\begin{align}
T . \alpha := \sum_{\eta \, \in \, \PP_{\multii}\super{s}} c_T(\eta,\alpha) \eta ,
\end{align}
for all basis elements $\alpha \in \PP_{\multii}\super{s}$ of the module $\smash{\PS_{\multii}\super{s}}$ 
and for all Jones-Wenzl tangles $T \in \WJ_\multii(\nu)$. One can convince oneself that this action 
coincides with the diagrammatic $\WJ_\multii(\nu)$-action on the standard module $\smash{\PS_{\multii}\super{s}}$.
Also, the bilinear form~\eqref{GLBiform} induced by the cellular structure coincides with the 
diagrammatic bilinear form~\eqref{LSBiForm} restricted to $\smash{\PS_{\multii}\super{s}}$.

\bigskip

As a final remark, we note that there is an alternative cellular basis for $\WJ_\multii(\nu)$ when $\Summed_\multii < \ppmin(q)$
(which holds, e.g., when $\WJ_\multii(\nu)$ is semisimple). This basis is obtained as in the proof of proposition~\ref{CellPropo} 
but replacing $\BarAction \alpha \quad \beta \BarAction$ by $\BarAction \alpha \;\; \ProjBox \;\; \beta \BarAction$, 
where ``$\; \ProjBox \;$'' is a Jones-Wenzl projector box across all of the crossing links in $\BarAction \alpha \quad \beta \BarAction$.
The main virtue of the cellular basis
$\smash{\{\BarAction \alpha \;\; \ProjBox \;\; \beta \BarAction \, | \, s \in \DefectSet_\multii , \; \alpha, \beta \in \PP_{\multii}\super{s} \}}$
is that property~\ref{Cellitem3} holds without any correction terms in $\WJ_\multii(\nu)\super{< s}$: 
indeed, item~\ref{ExtractLemItem} of lemma~\ref{CollectionLem} from appendix~\ref{TLRecouplingSect} gives
for all $\alpha, \beta, \gamma, \delta \in \smash{\PP_{\multii}\super{s}}$ the identity
\begin{align} \label{GLBiformobvious}
\BarAction \alpha \;\; \ProjBox \;\; \beta \BarAction \BarAction \gamma \;\; \ProjBox \;\; \delta \BarAction
\overset{\eqref{ExtractID}}{=}
\BiForm{\beta}{\gamma} \BarAction \alpha \;\; \ProjBox \;\; \delta \BarAction .
\end{align}
It is also obvious from~\eqref{GLBiformobvious} that the bilinear form~\eqref{GLBiform} induced by this cellular structure 
coincides with the diagrammatic bilinear form~\eqref{LSBiForm}.
A downside is that this basis is only well-defined when all of the projector boxes ``$\; \ProjBox \;$''  in all of the basis
elements are well-defined, which happens exactly when $\smax(\multii) = \Summed_\multii < \ppmin(q)$.

\section{Minimal generating sets for the Jones-Wenzl algebra} \label{GeneratorLemProofSect}

The purpose of this section is to prove theorem~\ref{GeneratorThm}:

\GeneratorThm*

We remark that the unit~\eqref{WJCompProj} of $\WJ_\multii(\nu)$
is obtained from generators~\eqref{MasterDiagramsWJ-00} via the relation
\begin{align} \label{AllProjesSumToOne}
\vcenter{\hbox{\includegraphics[scale=0.275]{e-CompositeProjector.pdf}}} 
\quad = \quad 
\sum_{s \, \in \, \DefectSet\sub{\sIndex_i,\sIndex_{i+1}}} \frac{(-1)^s [s+1]}{\ThetaNet(\sIndex_i,\sIndex_{i+1},s)} \,\, \times \,\,
\vcenter{\hbox{\includegraphics[scale=0.275]{e-Generators_3Vertex.pdf} ,}}
\end{align}
for any $i \in \{ 1, 2, \ldots, \np_\multii - 1 \}$.
In~\cite{fp3}, we prove that each diagram~\eqref{MasterDiagramsWJ-00} 
equals a nonzero multiple of a submodule projector 
in a tensor product representation of the Hopf algebra $U_q(\mathfrak{sl}_2)$ (with multiplicative constant as in the above sum).
Relation~\eqref{AllProjesSumToOne} says that summing over all of these projectors gives the identity operator.
Another way to view this relation is a decomposition of the unit~\eqref{WJCompProj} of $\WJ_\multii(\nu)$
into a sum of orthogonal (but not central) idempotents: indeed by identity~\eqref{LoopErasure1} 
from appendix~\ref{TLRecouplingSect}, we have
\begin{align} \label{orthogonalidem}
\hspace*{-3mm}
\vcenter{\hbox{\includegraphics[scale=0.275]{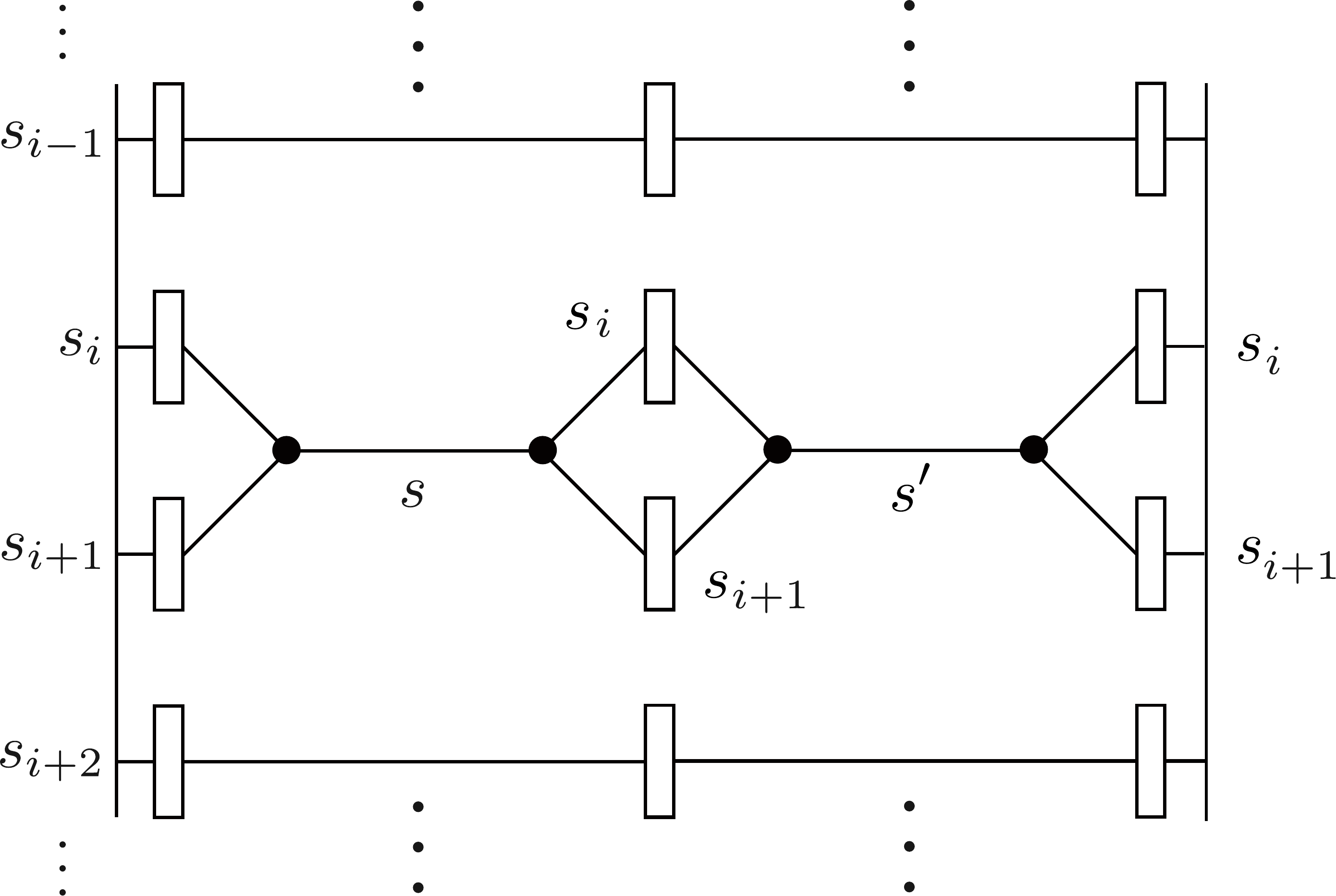}}}
\overset{\eqref{LoopErasure1}}{=} 
\; \delta_{s,s'} \frac{\ThetaNet(\sIndex_i,\sIndex_{i+1},s)}{(-1)^s [s+1]} \,\, \times \,\,
\vcenter{\hbox{\includegraphics[scale=0.275]{e-Generators_3Vertex.pdf} .}} 
\end{align}
Relation~\eqref{AllProjesSumToOne} is a special case of the ``quantum 6j-symbols''~\cite[equation~\red{9.15}, page~\red{99}]{kl}.

\bigskip

We have divided the proof of theorem~\ref{GeneratorThm} 
into several parts in order to clarify its structure and emphasize the main ideas. 
The proof constitutes rather involved diagram calculations presented in this section, but no prerequisites are needed.
Section~\ref{PreliminariesSec} contains some simple observations needed in the proof.
In section~\ref{BaseCaseSec}, we prove theorem~\ref{GeneratorThm} for the case of two projectors, $\np_\multii = 2$. 
Then, in section~\ref{IndStepSec} we construct certain basis tangles of $\WJ_\multii(\nu)$
from the claimed generators 
by induction in the number $\np_\multii$ of projectors.
In section~\ref{TheProofSec}, we finish the proof 
using the work in sections~\ref{BaseCaseSec} and~\ref{IndStepSec}.
Throughout, we fix 
$\multii = (\sIndex_1, \sIndex_2, \ldots,\sIndex_{\np_\multii})$ as in~(\ref{MultiindexNotation},~\ref{ndefn}), with $\Summed_\multii < \ppmin(q)$.

\subsection{Preliminary results} \label{PreliminariesSec}

To begin, we collect useful change of basis operations in lemma~\ref{ManyBasesLem}. For this, we 
denote by $\Defect_\alpha$ the number of defects in a $\multii$-Jones-Wenzl link state $\alpha$, that is,
\begin{align} 
\alpha \in \PS_\multii\super{s} \qquad \Longrightarrow \qquad \Defect_\alpha :=s ,
\end{align} 
we denote $\DefectSet\sub{\alpha,t} := \DefectSet\sub{\Defect_\alpha,t}$,
and we define the index set
\begin{align} 
\label{tRange} 
\mathsf{R}_{\alpha,\beta} 
&:= \left\{ 0, 1, \ldots, \min \Big(\sIndex_{\np_\multii} - \frac{|\Defect_\alpha - \Defect_\beta|}{2} , 
\Defect_\alpha , \Defect_\beta \Big) \right\}, \qquad \text{for all $\alpha, \beta \in \PP_\multii$.}
\end{align} 
We also frequently use the notations 
\begin{align} \label{hats} 
\multii = (\sIndex_1, \sIndex_2, \ldots, \sIndex_{\np_\multii}) 
\qquad \qquad \Longrightarrow \qquad \qquad & 
\begin{cases}
\lds := (\sIndex_1, \sIndex_2, \ldots, \sIndex_{\np_\multii-1}) \quad & \textnormal{and} \quad t\hphantom{'} := \sIndex_{\np_\multii} , \\
\fds := (\sIndex_2, \sIndex_3, \ldots, \sIndex_{\np_\multii})  \quad & \textnormal{and} \quad t' := \sIndex_1 ,
\end{cases} \\
\label{hatsMax} 
\Summed_{\lds} \overset{\eqref{ndefn}}{:=} 
\sIndex_1 + \sIndex_2 + \cdots + \sIndex_{\np_\multii - 1} \overset{\eqref{smaxeq}}{=} \smax(\lds) 
\qquad \qquad & \text{and} \qquad \qquad
\Summed_{\fds} \overset{\eqref{ndefn}}{:=}  
\sIndex_2 + \sIndex_3 + \cdots + \sIndex_{\np_\multii} \overset{\eqref{smaxeq}}{=} \smax(\fds) .
\end{align} 
Now, the following set is a basis for the Jones-Wenzl algebra $\WJ_\multii(\nu)$:
\begin{align} \label{MGenerators} 
\PD0_\multii
 & \; := \quad \left\{ 
\vcenter{\hbox{\includegraphics[scale=0.275]{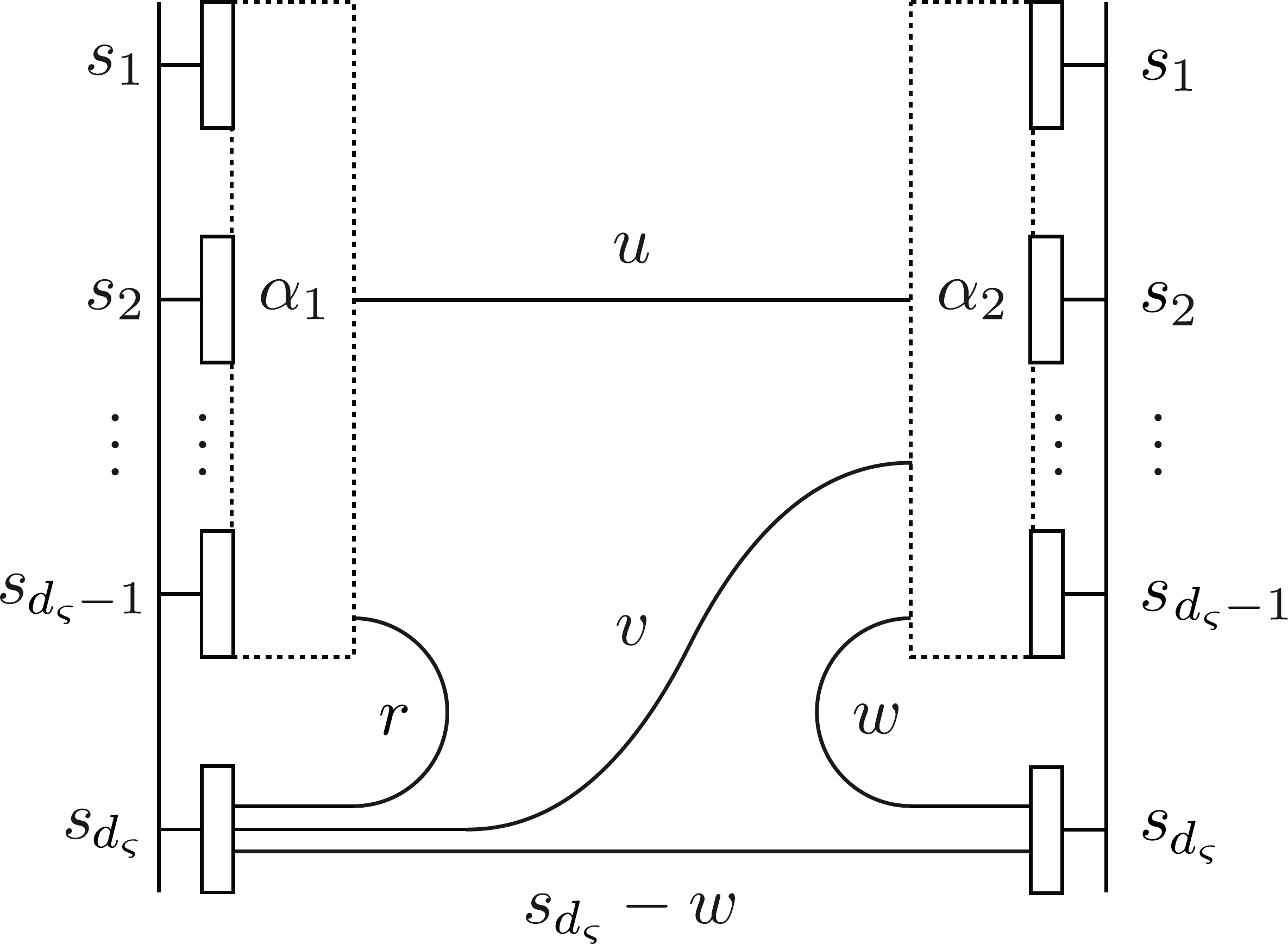}}} \qquad \right|\left .
\quad \vcenter{\hbox{\includegraphics[scale=0.275]{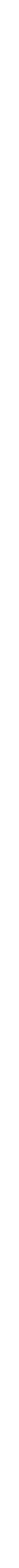}}} 
\begin{array}{l} 
\alpha_1, \alpha_2 \in \PP_{\lds} , \\[10pt]
r \in \mathsf{R}_{\alpha_1,\alpha_2} , \\[10pt] 
u := \min(\Defect_{\alpha_1}, \Defect_{\alpha_2}) - r , \\[10pt] 
v := \dfrac{|\Defect_{\alpha_1} - \Defect_{\alpha_2}|}{2} , \\[10pt] 
w := r + \dfrac{|\Defect_{\alpha_1} - \Defect_{\alpha_2}|}{2} 
\end{array} 
\,\,\, \right\} .
\end{align}

Using lemma~\ref{ChangeOfBasisLem} from appendix~\ref{LemmaApp}
(i.e.,~\cite[lemma~\red{4.4}]{fp0}), we construct other bases for $\WJ_\multii(\nu)$.

\begin{lem} \label{ManyBasesLem} 
Suppose $\Summed_\multii < \ppmin(q)$.
Then each of the following sets is a basis for the Jones-Wenzl algebra $\WJ_\multii(\nu)$:
\begin{align}
\label{InsertTwoBoxes}
\PD1_\multii 
& \; := \quad \left\{ \vcenter{\hbox{\includegraphics[scale=0.275]{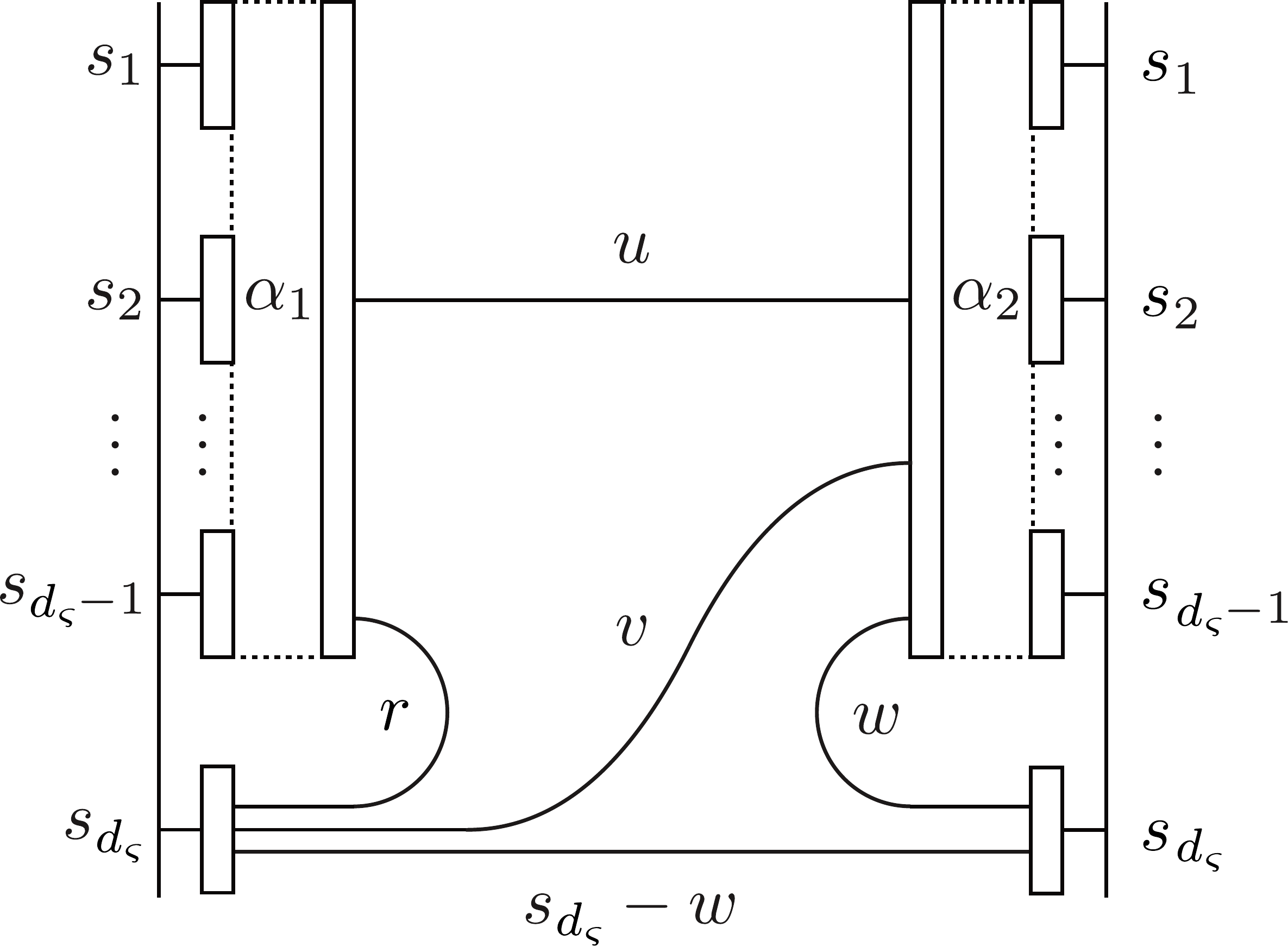}}} \qquad \right|\left.
\quad \vcenter{\hbox{\includegraphics[scale=0.275]{e-Generators38.pdf}}} 
\begin{array}{l} 
\alpha_1, \alpha_2 \in \PP_{\lds} , \\[10pt]
r \in \mathsf{R}_{\alpha_1,\alpha_2} , \\[10pt] 
u := \min(\Defect_{\alpha_1}, \Defect_{\alpha_2})-r , \\[10pt] 
v := \dfrac{|\Defect_{\alpha_1} - \Defect_{\alpha_2}|}{2} , \\[10pt] 
w := r + \dfrac{|\Defect_{\alpha_1} - \Defect_{\alpha_2}|}{2} 
\end{array} 
\,\,\, \right\} , \\[1em]
\label{InsertOneHorizBox}
\PD2_\multii 
& \; := \quad \left\{ 
\vcenter{\hbox{\includegraphics[scale=0.275]{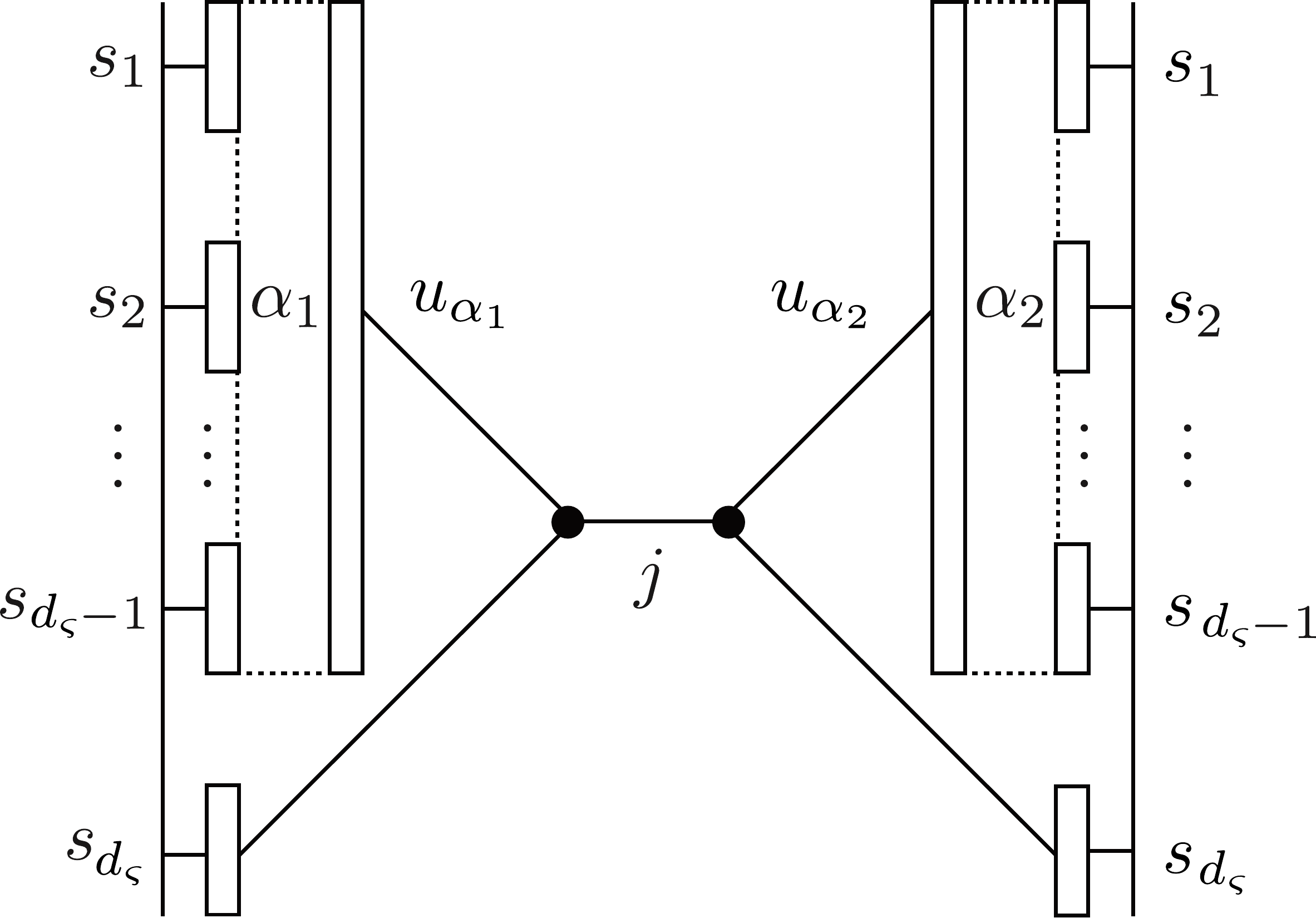}}} \qquad \right|\left.
\quad \vcenter{\hbox{\includegraphics[scale=0.275]{e-Generators38.pdf}}} 
\begin{array}{l} 
\hphantom{u:= \min(\Defect_{\alpha_1}, \Defect_{\alpha_2}) - t,} \\[10pt] 
\alpha_1, \alpha_2 \in \PP_{\lds} , \\[10pt]
j \in \DefectSet\sub{\alpha_1,t} \cap \DefectSet\sub{\alpha_2,t},
\\[10pt] 
\textnormal{with $t = \sIndex_{\np_\multii}$}
\end{array} 
\,\,\, \right\} \\[1em]
\label{InsertOneBox}
\PD3_\multii 
& \; := \quad \left\{\vcenter{\hbox{\includegraphics[scale=0.275]{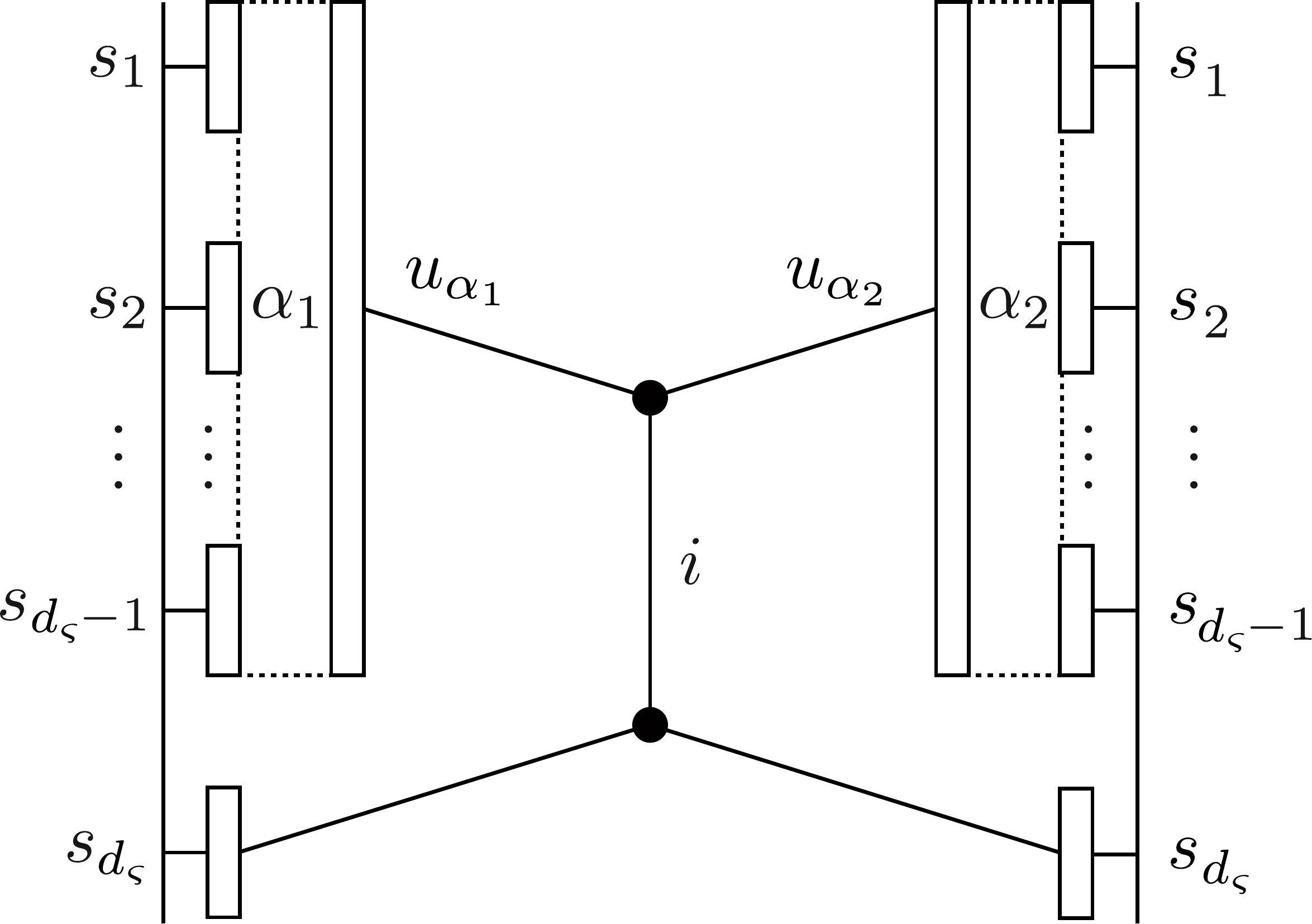}}} \qquad \right|\left.
\quad \vcenter{\hbox{\includegraphics[scale=0.275]{e-Generators38.pdf}}}
\begin{array}{l} 
\hphantom{u := \min(\Defect_{\alpha_1}, \Defect_{\alpha_2}) - t,} \\[10pt] 
\alpha_1, \alpha_2 \in \PP_{\lds} , \\[10pt] 
i \in \DefectSet\sub{\alpha_1,\alpha_2} \cap \DefectSet\sub{t,t},
\\[10pt] 
\textnormal{with $t = \sIndex_{\np_\multii}$}
\end{array} 
\,\,\, \right\} . 
\end{align}
\end{lem}

\begin{proof} 
First, we label the left and right link states of~\eqref{MGenerators} by $a$ and $c$ 
respectively, and we label the left-bottom and right-bottom projector boxes $b$ and $d$ respectively.  
Then we obtain $\PD1_\multii$~\eqref{InsertTwoBoxes} from~\eqref{MGenerators} by sending
\begin{align} \label{1stmap0} 
\vcenter{\hbox{\includegraphics[scale=0.275]{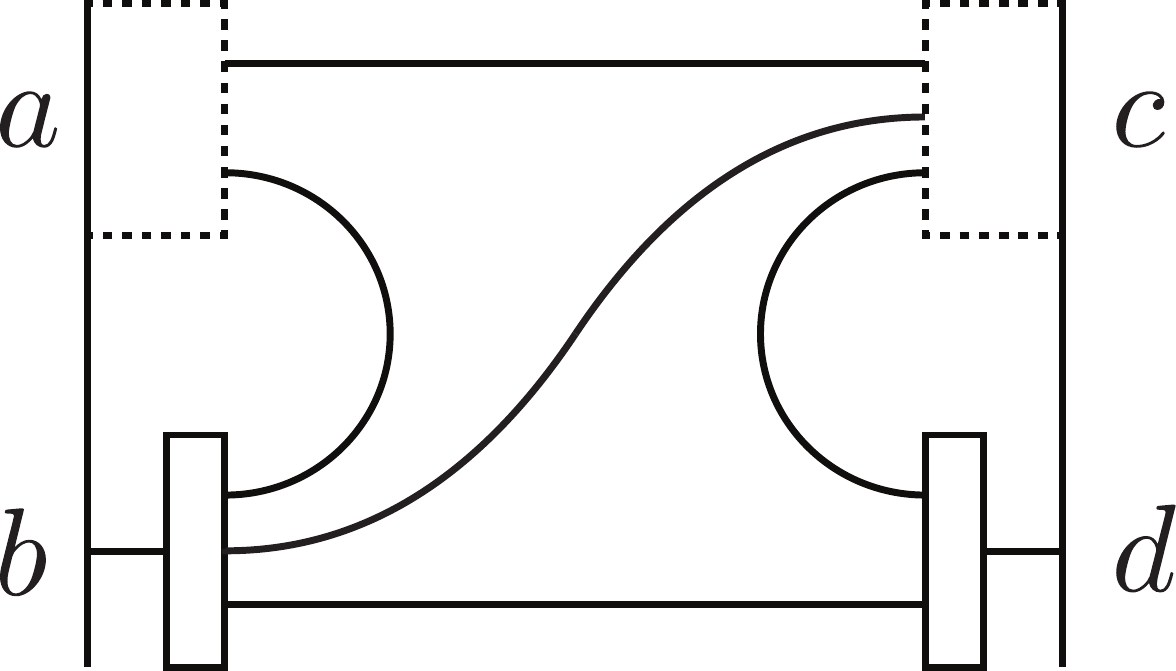}}} \quad
& \qquad \longmapsto \qquad
 \quad \vcenter{\hbox{\includegraphics[scale=0.275]{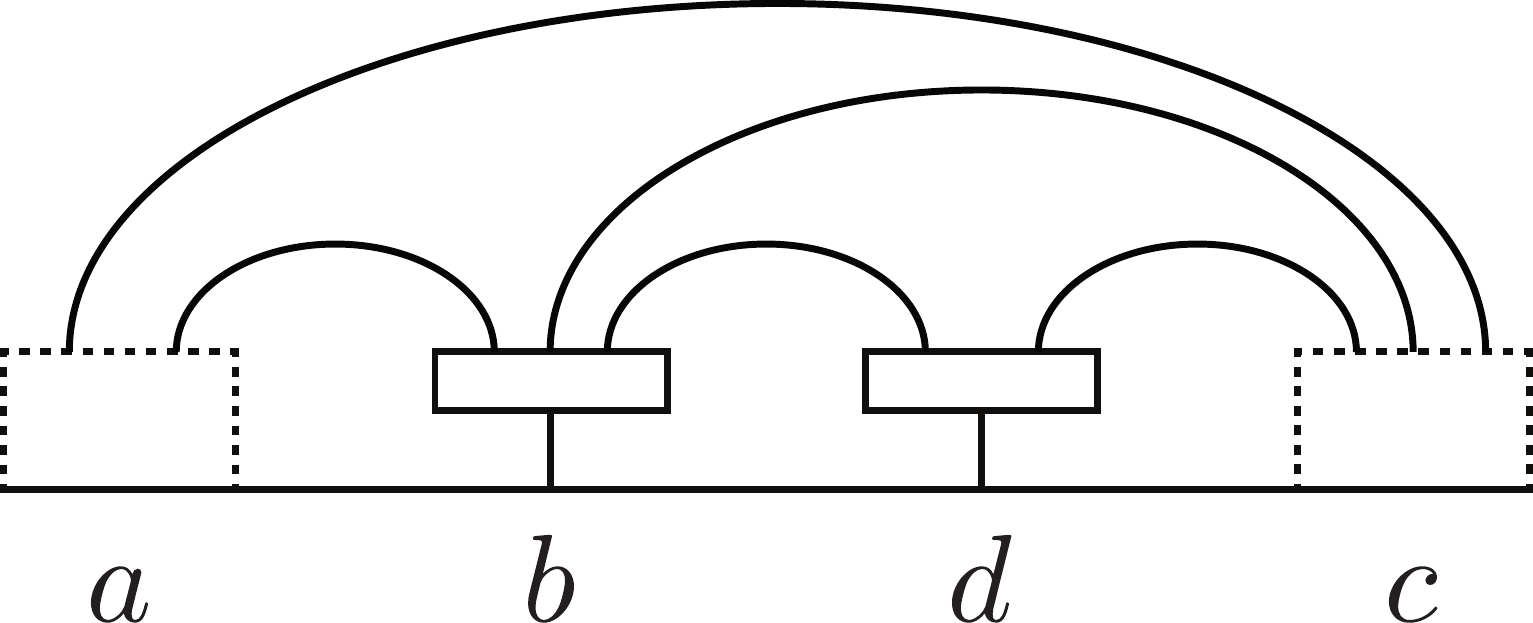}}}  \\[1em]
\label{2ndmap0} 
& \qquad \longmapsto \qquad \quad \vcenter{\hbox{\includegraphics[scale=0.275]{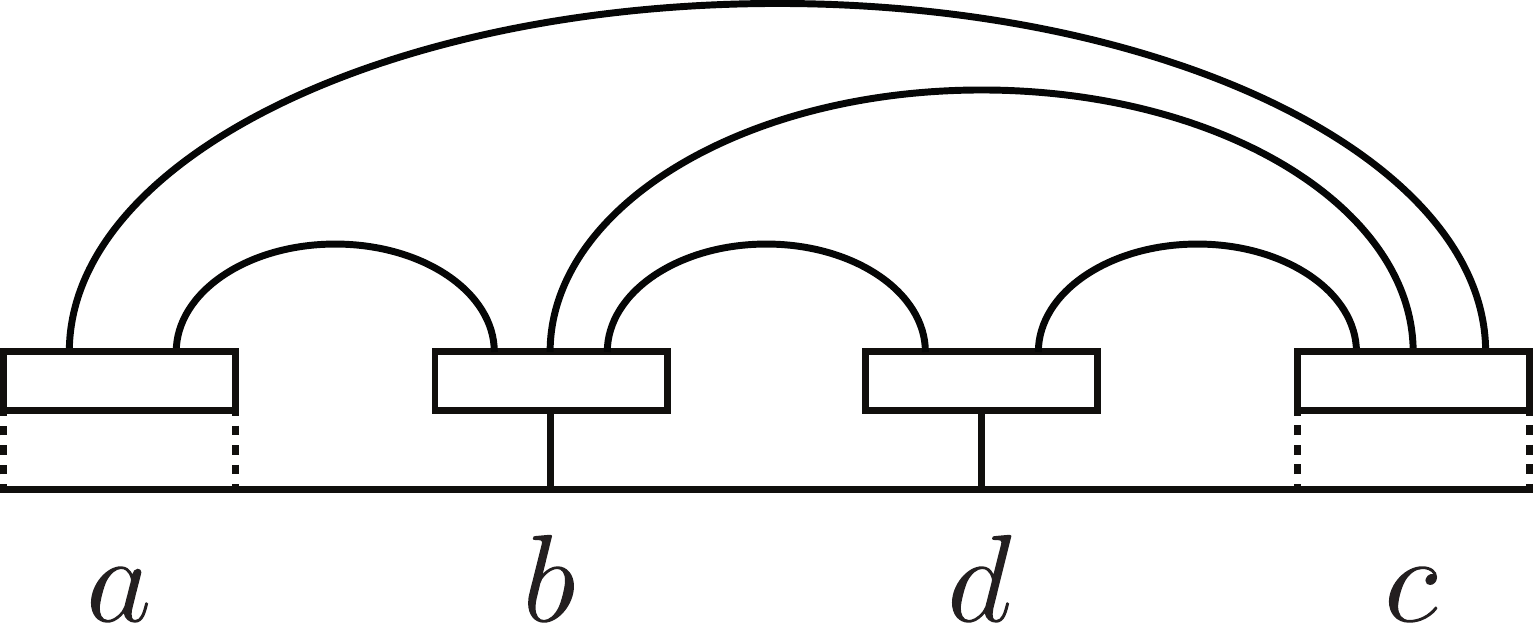}}}  \\[1em]
\label{3rdmap0} 
& \qquad \longmapsto \qquad \quad \vcenter{\hbox{\includegraphics[scale=0.275]{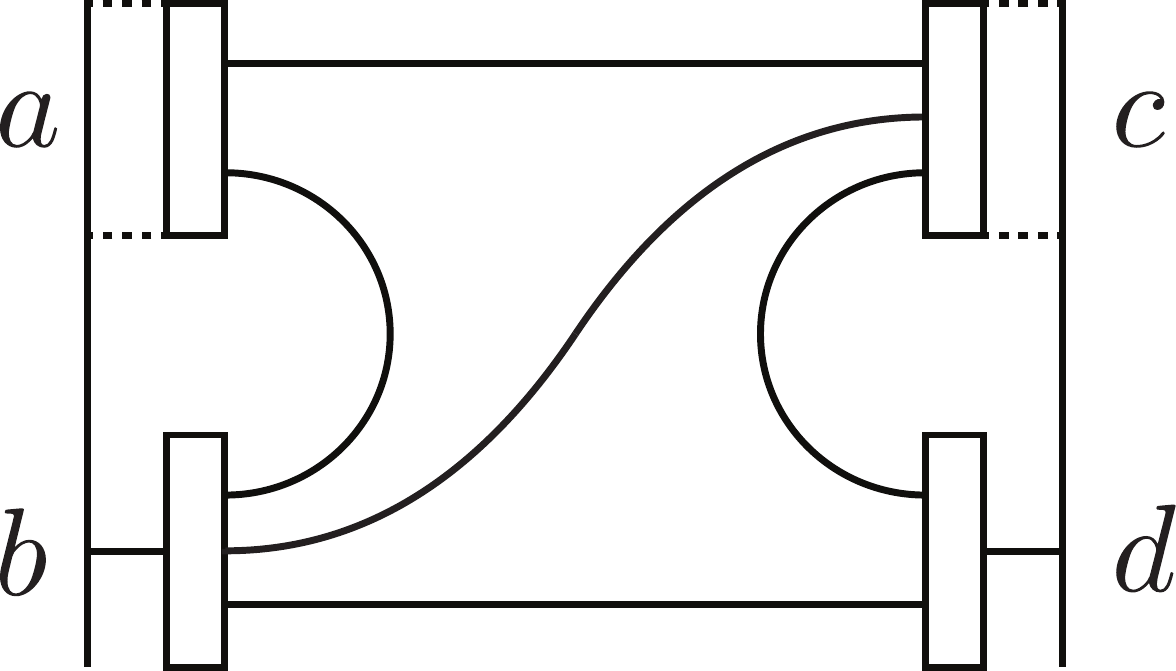} .}} 
\end{align}
The first and third maps are obvious isomorphisms of vector spaces, and lemma~\ref{ChangeOfBasisLem} 
implies that the second map is an isomorphism of vector spaces too.  Because the composition 
sends~\eqref{MGenerators} to $\PD1_\multii$, $\PD1_\multii$ is a basis for $\WJ_\multii(\nu)$.

Next, we obtain $\PD2_\multii$~\eqref{InsertOneHorizBox} by applying this composition of three maps, 
similar to (\ref{1stmap0}--\ref{3rdmap0}), to $\PD1_\multii$:
\begin{align} \label{1stmap} 
\vcenter{\hbox{\includegraphics[scale=0.275]{e-Basis7.pdf}}} \quad
& \qquad \longmapsto \qquad \quad \vcenter{\hbox{\includegraphics[scale=0.275]{e-Trans5.pdf}}}  \\[1em]
\label{2ndmap} 
& \qquad \longmapsto \qquad \quad \vcenter{\hbox{\includegraphics[scale=0.275]{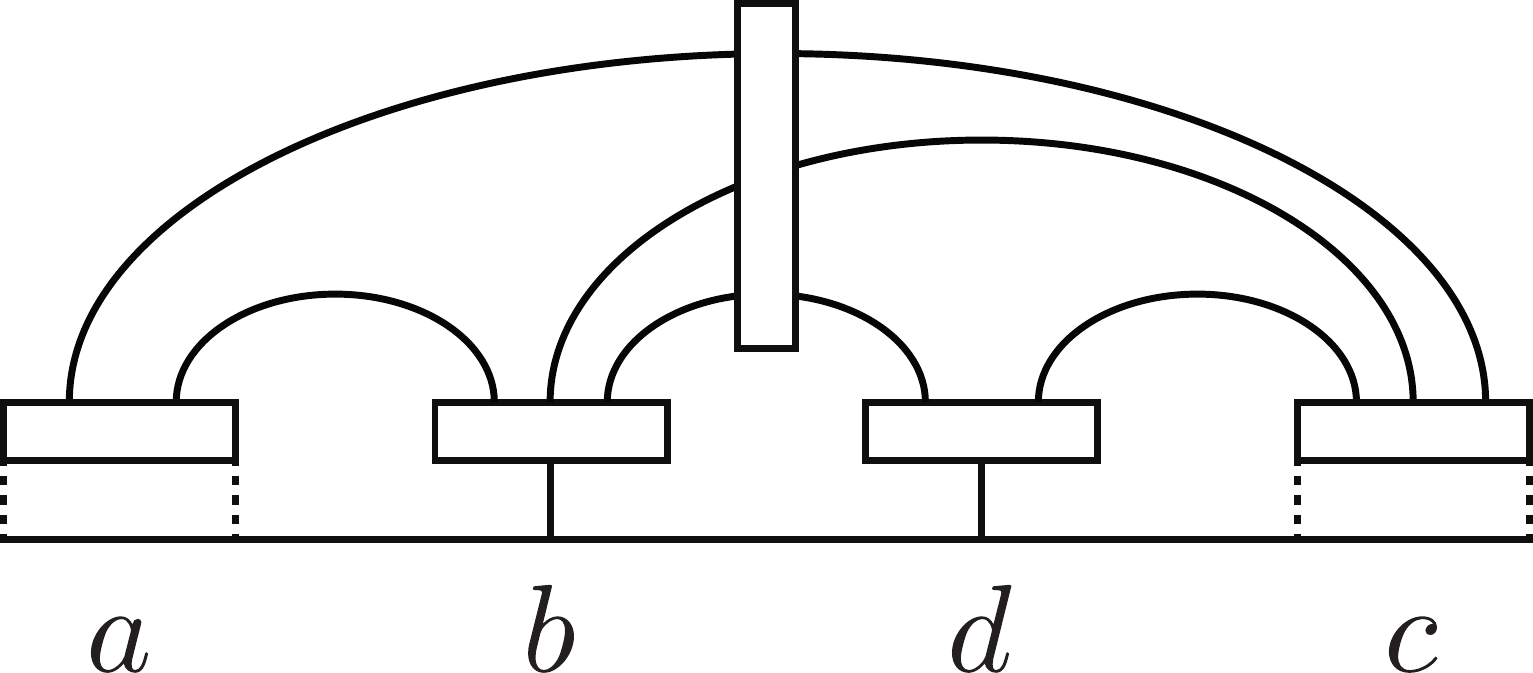}}}  \\[1em]
\label{3rdmap} 
& \qquad \longmapsto \qquad \quad \vcenter{\hbox{\includegraphics[scale=0.275]{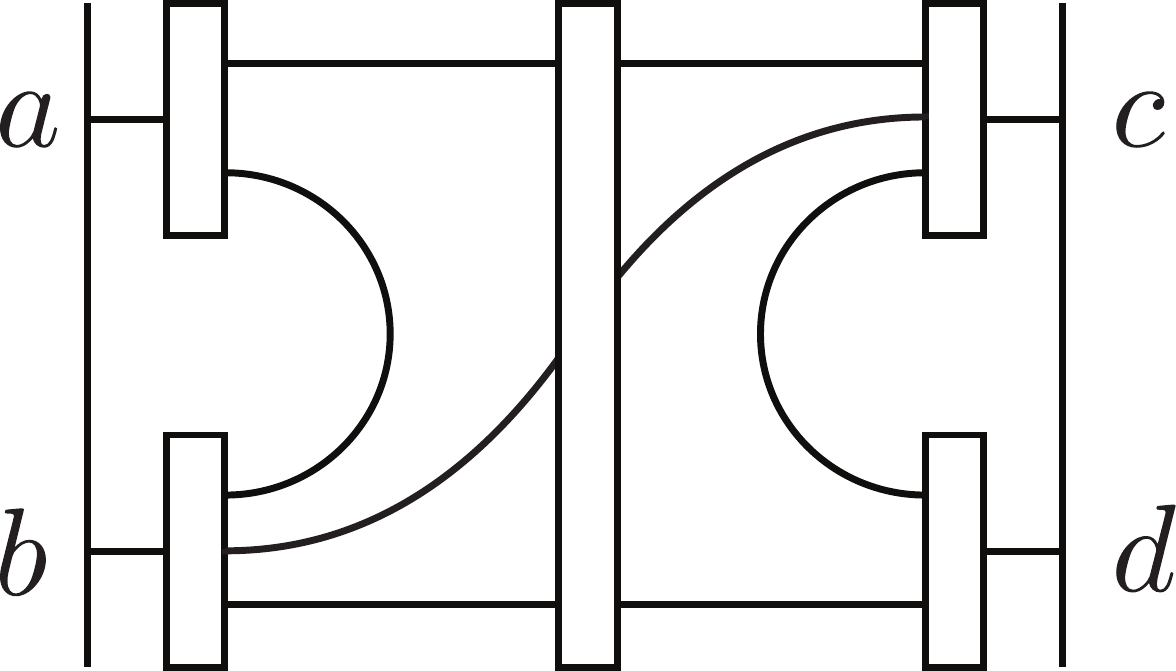}}} 
\quad =  \quad \vcenter{\hbox{\includegraphics[scale=0.275]{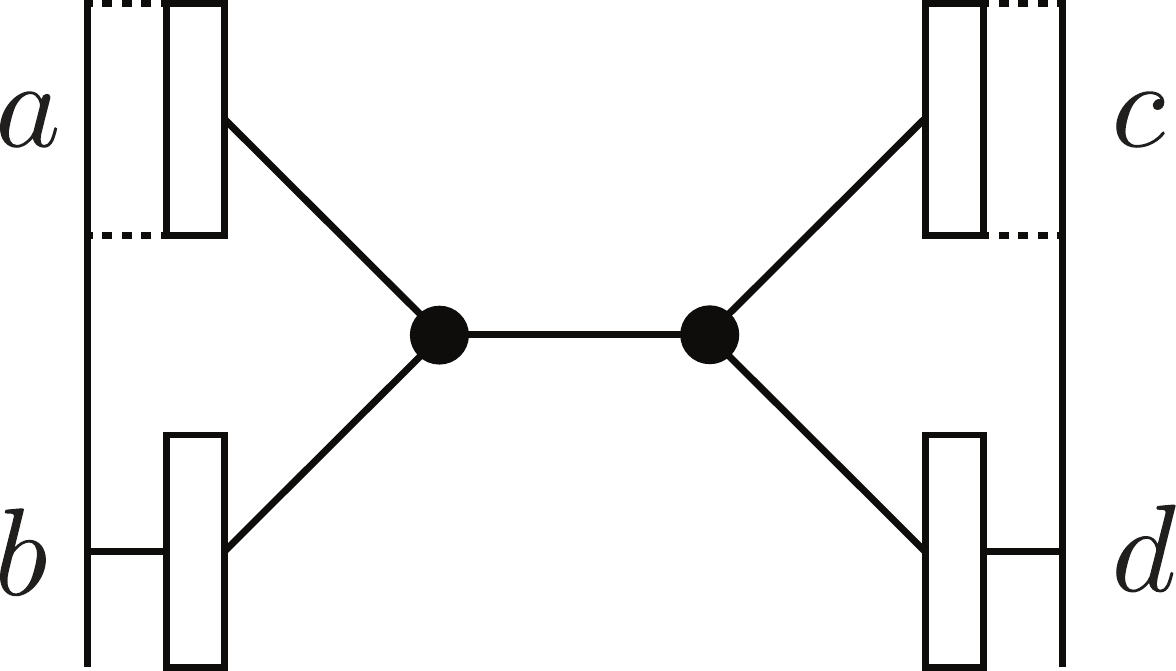} .}} 
\end{align}
We must check that the size $j$ of the new projector box, set vertically across all horizontal crossing links 
in~\eqref{InsertTwoBoxes}, is less than $\ppmin(q)$.  Indeed, with $\lds = (\sIndex_1, \sIndex_2, \ldots, \sIndex_{\np_\multii - 1})$
and $t = \sIndex_{\np_\multii}$, we have $\Summed_\multii = \smax(\smash{\lds}) + t$ and
\begin{align} 
j \leq \max (\DefectSet\sub{\alpha_1,t} \cap \DefectSet\sub{\alpha_2,t}) 
= \min(\Defect_{\alpha_1} + t, \Defect_{\alpha_2} + t) \leq  \smax(\smash{\lds}) + t
= \Summed_\multii < \ppmin(q). 
\end{align} 
Reasoning as in the case of $\PD1_\multii$, we conclude that $\PD2_\multii$ is a basis for $\WJ_\multii(\nu)$.

Finally, we obtain $\PD3_\multii$~\eqref{InsertOneBox} from $\PD1_\multii$ via the procedure of 
the previous paragraph with $b$ and $c$ exchanged.  The size $i$ of the new projector box
satisfies
\begin{align} 
i \leq \max (\DefectSet\sub{\alpha_1,\alpha_2} \cap \DefectSet\sub{t,t}) 
= \min(2 t, \Defect_{\alpha_1} + \Defect_{\alpha_2}) \leq 
\min(2 t, 2 \smax(\smash{\lds})) \leq \smax(\smash{\lds}) + t = \Summed_\multii < \ppmin(q).
\end{align} 
Hence, reasoning as in the case of $\PD2_\multii$, we conclude that $\PD3_\multii$ is a basis for $\WJ_\multii(\nu)$.
\end{proof}

Later, in the proof of corollary~\ref{FinalCor}, we will show that every tangle in the first basis $\PD1_\multii$~\eqref{InsertTwoBoxes}
is a polynomial in the claimed generators~\eqref{EGenerators-00} 
of $\WJ_\multii(\nu)$. 
This is a key step to prove theorem~\ref{GeneratorThm}.
We use the third basis $\PD3_\multii$~\eqref{InsertOneBox} in the proofs of lemmas~\ref{TProductLem} and~\ref{SameSideLem2}.
We include the second basis $\PD2_\multii$~\eqref{InsertOneHorizBox} for completeness. 

In fact, using proposition~\ref{SpecialTProp}, one may explicitly work out the coefficients of the change of basis from $\PD1_\multii$ to $\PD2_\multii$ and 
$\PD3_\multii$.  We leave this to the reader, as we do not use these explicit coefficients in the present article.

\bigskip

In the next lemmas~\ref{ParityLem} and~\ref{MinMaxLem}, we collect further technical results of combinatorial nature. 
Lemma~\ref{ParityLem} is needed in the proofs of lemmas~\ref{MinMaxLem},~\ref{TProductLem} and \ref{SameSideLem2},
and lemma~\ref{MinMaxLem} is needed in the proof of lemma~\ref{2stepLem}. 



\begin{lem} \label{ParityLem} 
For any Jones-Wenzl link pattern $\alpha \in \PP_\multii$, 
the following hold:
\begin{enumerate} 
\itemcolor{red}
\item\label{PaRittt1} All elements of the union $\DefectSet\sub{\alpha,\sIndex_1} \cup \smash{\DefectSet_\fds}$ have the same parity.

\item \label{PaRittt3} 
Upon writing $\alpha$ in the form
\begin{align} \label{WJsubForm} 
\alpha \quad = \quad \vcenter{\hbox{\includegraphics[scale=0.275]{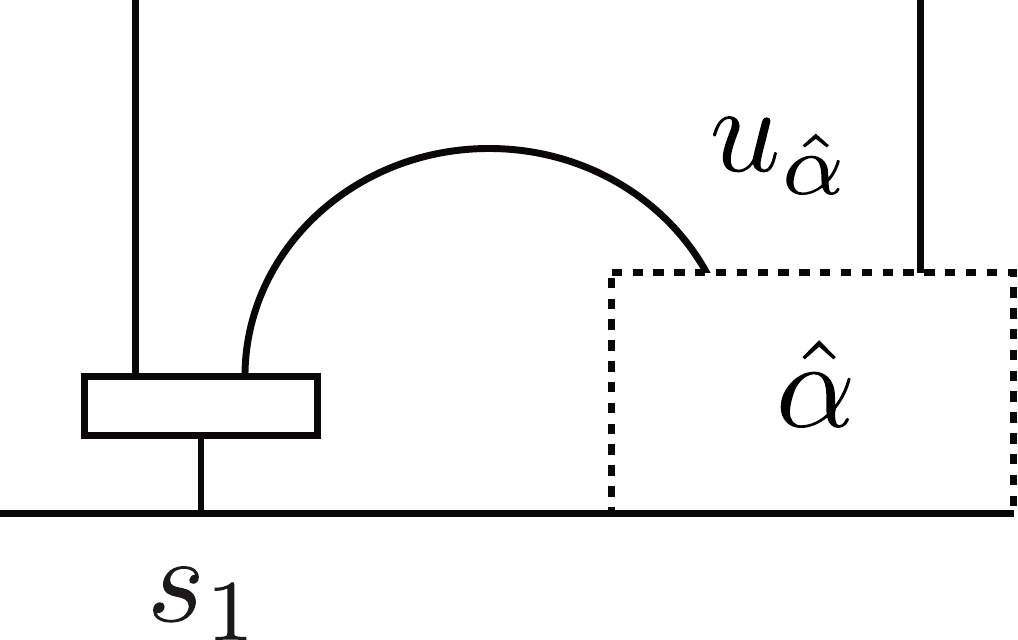} ,}} 
\end{align} 
we have $\Defect_{\hat{\alpha}} \in \DefectSet\sub{\alpha,\sIndex_1} \cap \smash{\DefectSet_\fds}$.

\item\label{PaRittt2} The intersection $\DefectSet\sub{\alpha,\sIndex_1} \cap \smash{\DefectSet_\fds}$ is not empty, and it is given by
\begin{align} \label{EalphSet} 
\DefectSet\sub{\alpha,\sIndex_1} \cap \DefectSet_{\fds} 
= \big\{ \max(\smin(\fds),|\Defect_\alpha - \sIndex_1|), \,
\max(\smin(\fds),|\Defect_\alpha - \sIndex_1|) + 2 , \,
\ldots, \, \min(  \smax(\smash{\fds}), \Defect_\alpha + \sIndex_1) \big\} .
\end{align} 
\end{enumerate}
\end{lem}

\begin{proof} 
We prove items~\ref{PaRittt1}--\ref{PaRittt2} as follows:
\begin{enumerate}[leftmargin=*]
\itemcolor{red}
\item The sets $\DefectSet\sub{\alpha,\sIndex_1}$ and $\smash{\DefectSet_\fds}$ each comprise a sequence of integers that increases in increments of two.  
Explicitly, 
\begin{align} \label{1Sets} 
\DefectSet\sub{\alpha,\sIndex_1} \overset{\eqref{SpecialDefSet}}{=} \big\{ |\Defect_\alpha - \sIndex_1|, |\Defect_\alpha - \sIndex_1| + 2, 
\ldots, \Defect_\alpha + \sIndex_1 \big\}, \qquad \smash{\DefectSet_\fds} \overset{\eqref{DefSet2}}{=} \{ 
\smin(\fds), \smin(\fds) + 2, 
\ldots, \smax(\smash{\fds}) \} ,
\end{align}
so if any two elements from each set~\eqref{1Sets} have the same parity, 
then all elements of the union $\DefectSet\sub{\alpha,\sIndex_1} \cup \smash{\DefectSet_\fds}$ do. 
We write this as $a \sim b$ for $a \in \DefectSet\sub{\alpha,\sIndex_1}$ and $b \in \smash{\DefectSet_\fds}$.
Now, to prove item~\ref{PaRittt1}, we only have to note that $\Defect_\alpha \in \DefectSet_\multii$ implies
\begin{align} 
\Defect_\alpha \sim \Summed_\multii 
\qquad  \qquad \overset{\eqref{ndefn}}{\Longrightarrow} \qquad \qquad 
\Defect_\alpha + \sIndex_1 \sim  2 \sIndex_1 + \sIndex_2 + \dotsm +\sIndex_{\np_\multii} \sim \sIndex_2 + \sIndex_3 + \dotsm +\sIndex_{\np_\multii} 
\overset{\eqref{hatsMax}}{=} \smax(\fds) . 
\end{align} 

\item 
The claim in item~\ref{PaRittt3}
amounts to showing that $\Defect_{\hat{\alpha}} \in \DefectSet\sub{\alpha,\sIndex_1}$, 
because $\Defect_{\hat{\alpha}} \in \smash{\DefectSet_\fds}$ by definition.  
Now, the defects of $\hat{\alpha}$ either are defects of $\alpha$, or attach 
to the nodes of the projector box of size $\sIndex_1$.  Hence, we have
\begin{align} \label{Extreme1} 
\Defect_{\hat{\alpha}} \leq \Defect_\alpha + \sIndex_1, \qquad\qquad \text{with equality achieved for} \qquad
\quad \vcenter{\hbox{\includegraphics[scale=0.275]{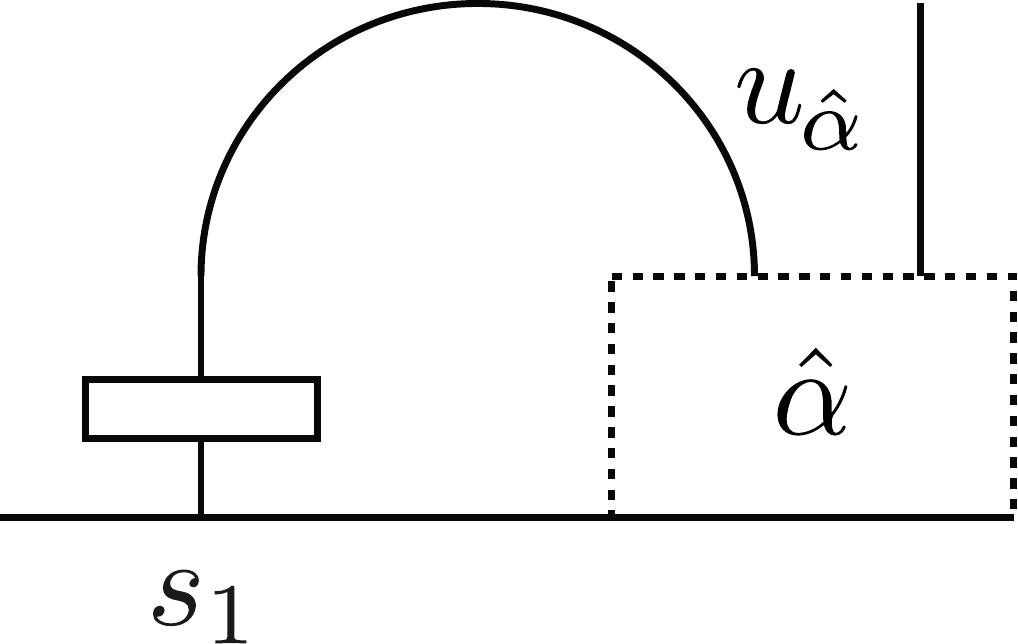} .}} 
\end{align} 
Also, the defects of $\alpha$ either attach to the nodes of the projector box of size $\sIndex_1$, or are themselves defects of $\hat{\alpha}$, so 
\begin{align} \label{Extreme2} 
\Defect_{\alpha} \leq \sIndex_1+ \Defect_{\hat{\alpha}}, \qquad\qquad \text{with equality achieved for} \qquad
\quad \vcenter{\hbox{\includegraphics[scale=0.275]{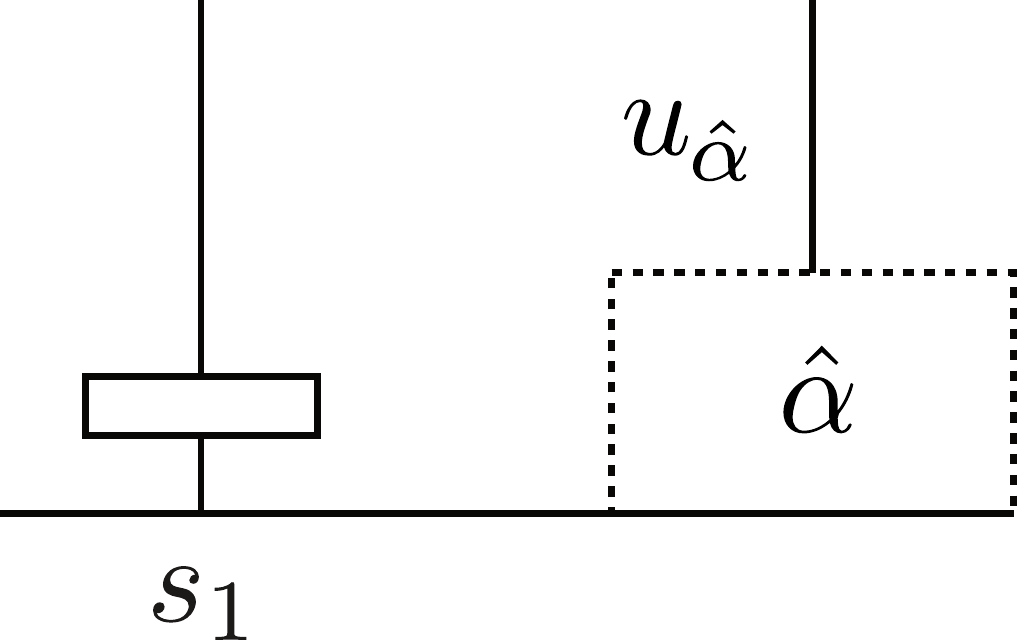} .}} 
\end{align} 
Moreover, both links and defects of $\alpha$ may attach to the projector box of size $\sIndex_1$, and all links that attach to this projector box must 
connect to defects of $\hat{\alpha}$.  Hence, we have
\begin{align} \label{Extreme3} 
\sIndex_1 \leq \Defect_{\hat{\alpha}} + \Defect_{\alpha}, \qquad\qquad \text{with equality achieved for} \qquad
\quad \vcenter{\hbox{\includegraphics[scale=0.275]{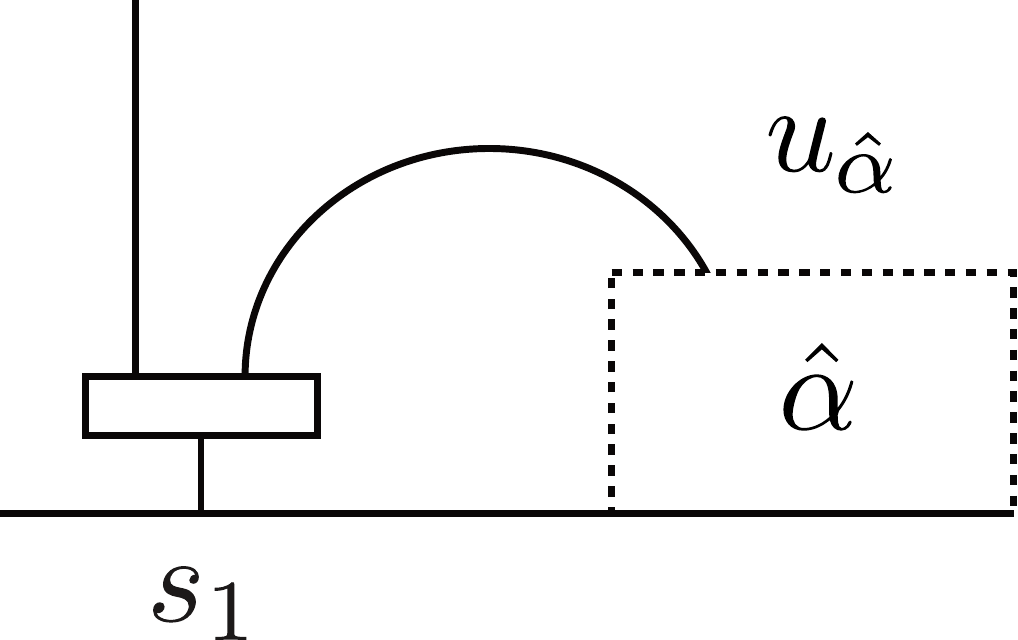} .}}
\end{align} 
Altogether, (\ref{Extreme1}--\ref{Extreme3}) imply that $| \Defect_\alpha - \sIndex_1 | \leq \Defect_{\hat{\alpha}} \leq \Defect_\alpha + \sIndex_1$.  
Moreover, with $\Defect_{\hat{\alpha}} \in \smash{\DefectSet_\fds}$, item~\ref{PaRittt1} says that $\Defect_{\hat{\alpha}}$ has the same parity 
as the elements of $\DefectSet\sub{\alpha,\sIndex_1}$.  
Hence, we conclude from~\eqref{1Sets}  that $\Defect_{\hat{\alpha}} \in \DefectSet\sub{\alpha,\sIndex_1}$, proving item~\ref{PaRittt3}.

\item Expressions~\eqref{1Sets} for the sets $\DefectSet\sub{\alpha,\sIndex_1}$ and $\smash{\DefectSet_{\fds}}$ 
combined with item~\ref{PaRittt3}
imply expression~\eqref{EalphSet} for the intersection $\DefectSet\sub{\alpha,\sIndex_1} \cap \smash{\DefectSet_{\fds}}$.
This proves item~\ref{PaRittt2}.
\end{enumerate}
\end{proof} 


\begin{lem} \label{MinMaxLem}
For all Jones-Wenzl link patterns $\alpha_1,\alpha_2 \in \PP_\multii$ 
such that $\Defect_{\alpha_2} = \Defect_{\alpha_1} + 2$, we have 
\begin{align}
\DefectSet\sub{\alpha_1,\sIndex_1} \cap \DefectSet\sub{\alpha_2,\sIndex_1} \cap \DefectSet_\fds \; \neq \; \emptyset .
\end{align}
\end{lem}
\begin{proof}
We recall from item~\ref{PaRittt2} of lemma~\ref{ParityLem} that for any Jones-Wenzl link pattern $\alpha \in \PP_\multii$, we have
\begin{align}
\label{MinMax1b}
 \max (\DefectSet\sub{\alpha,\sIndex_1} \cap \DefectSet_\fds) &\overset{\eqref{EalphSet}}{=} \min( \smax(\smash{\fds}), \Defect_{\alpha} + \sIndex_1), \\
\label{MinMax1a}
 \min (\DefectSet\sub{\alpha,\sIndex_1} \cap \DefectSet_\fds) &\overset{\eqref{EalphSet}}{=} \max \big(\smin(\fds), |\Defect_\alpha - \sIndex_1| \big). 
\end{align}
By~\eqref{MinMax1b}, we clearly have $\max (\DefectSet\sub{\alpha_1,\sIndex_1} \cap \smash{\DefectSet_\fds}) \leq \smax(\smash{\fds})$, 
which leads us to consider two cases:

\begin{enumerate}[leftmargin=*]
\itemcolor{red}
\item \label{ContainIt1SF} $\max (\DefectSet\sub{\alpha_1,\sIndex_1} \cap \smash{\DefectSet_\fds}) = \smax(\smash{\fds})$:  
In this case, we trivially have $\smax(\smash{\fds}) \in \DefectSet\sub{\alpha_1,\sIndex_1} \cap \smash{\DefectSet_\fds}$.  Furthermore,~\eqref{MinMax1b} with $\alpha = \alpha_1$ gives $\smax(\smash{\fds}) \leq \Defect_{\alpha_1} + \sIndex_1$,
and with $\Defect_{\alpha_2} = \Defect_{\alpha_1} + 2$, this also implies that 
\begin{align}
\smax(\smash{\fds}) < \Defect_{\alpha_1} + \sIndex_1 + 2 = \Defect_{\alpha_2} + \sIndex_1 
\qquad \Longrightarrow \qquad 
\smax(\smash{\fds}) \overset{\eqref{MinMax1b}}{=} \max (\DefectSet\sub{\alpha_2,\sIndex_1} \cap \smash{\DefectSet_\fds}) .
\end{align}
We conclude that $\smax(\smash{\fds}) \in \DefectSet\sub{\alpha_2,\sIndex_1} \cap \smash{\DefectSet_\fds}$ too. 
Thus, we have $\smax(\smash{\fds}) \in \DefectSet\sub{\alpha_1,\sIndex_1} \cap \DefectSet\sub{\alpha_2,\sIndex_1} \cap \smash{\DefectSet_\fds}$, so this set is not empty.

\item \label{ContainIt2SF} $\max (\DefectSet\sub{\alpha_1,\sIndex_1} \cap \smash{\DefectSet_\fds}) < \smax(\smash{\fds})$: 
In this case,~\eqref{MinMax1b} with $\alpha = \alpha_1$ gives $\Defect_{\alpha_1} + \sIndex_1 <  \smax(\smash{\fds})$. Thus, we have 
\begin{align} \label{FirstInclus} 
\Defect_{\alpha_1} + \sIndex_1 = \min( \smax(\smash{\fds}), \Defect_{\alpha_1} + \sIndex_1) \overset{\eqref{MinMax1b}}{=} \max (\DefectSet\sub{\alpha_1,\sIndex_1} \cap \DefectSet_\fds) 
\qquad \Longrightarrow \qquad \Defect_{\alpha_1} + \sIndex_1 \in \DefectSet\sub{\alpha_1,\sIndex_1} \cap \DefectSet_\fds . 
\end{align} 
To finish, we prove that $\Defect_{\alpha_1} + \sIndex_1 \in \DefectSet\sub{\alpha_2,\sIndex_1}$. For this purpose, 
by lemma~\ref{ParityLem}, we only need to show that
\begin{align} 
\label{ConseqIneq00} 
\min \DefectSet\sub{\alpha_2,\sIndex_1} \leq \Defect_{\alpha_1} + \sIndex_1 \leq \max \DefectSet\sub{\alpha_2,\sIndex_1} .
\end{align}
First, with $\Defect_{\alpha_2} = \Defect_{\alpha_1} + 2$, we have
\begin{align} 
\Defect_{\alpha_1} + \sIndex_1 < \Defect_{\alpha_1} + \sIndex_1 + 2 = \Defect_{\alpha_2} + \sIndex_1 \overset{\eqref{1Sets}}{=} \max \DefectSet\sub{\alpha_2,\sIndex_1}.
\end{align} 
Second, we observe that
\begin{align} 
\begin{cases} 
\sIndex_1 - \Defect_{\alpha_2} \leq \sIndex_1 \leq \Defect_{\alpha_1} + \sIndex_1 \\
\Defect_{\alpha_2} - \sIndex_1 \leq \Defect_{\alpha_2} - \sIndex_1 + 2 (\sIndex_1 - 1) = \Defect_{\alpha_1} + \sIndex_1 
\end{cases} 
\qquad \Longrightarrow \qquad 
\min \DefectSet\sub{\alpha_2,\sIndex_1} \overset{\eqref{1Sets}}{=} |\Defect_{\alpha_2} - \sIndex_1| \leq \Defect_{\alpha_1} + \sIndex_1.
\end{align} 
Hence, we have 
\begin{align} \label{Conseq} 
\Defect_{\alpha_1} + \sIndex_1 \in \DefectSet\sub{\alpha_2,\sIndex_1} \qquad \overset{\eqref{FirstInclus}}{\Longrightarrow} \qquad 
\Defect_{\alpha_1} + \sIndex_1 \in \DefectSet\sub{\alpha_1,\sIndex_1} \cap \DefectSet\sub{\alpha_2,\sIndex_1} \cap \smash{\DefectSet_\fds} . 
\end{align} 
\end{enumerate}
Together, items~\ref{ContainIt1SF}  and~\ref{ContainIt2SF} show that the intersection $\DefectSet\sub{\alpha_1,\sIndex_1} \cap \DefectSet\sub{\alpha_2,\sIndex_1} \cap \smash{\DefectSet_\fds}$ is never empty.
\end{proof}

We remark that by linearity, lemmas~\ref{ParityLem} and~\ref{MinMaxLem} may be applied to all Jones-Wenzl link states 
$\alpha, \alpha_1, \alpha_2 \in \smash{\PS_\multii\super{s}}$.

As the last preliminary result, we construct new tangles by concatenating a basis tangle of $\WJ_\multii(\nu)$ from
$\PD3_\multii$~\eqref{InsertOneBox} from the left and right with tangles of type appearing later 
in hypothesis~\ref{IndAss1}.
This enables induction in the number $\np_\multii$ of projectors later on in the proofs of lemmas~\ref{SameSideLem2} and~\ref{2stepLem}.
To state this result, we denote by $\one{\{A\}}$ the indicator function on the statement $A$, equaling one if $A$ is true and zero if $A$ is false.
We also use the notation $\flds := (\sIndex_2, \sIndex_3, \ldots, \sIndex_{\np_\multii - 1})$ and $t = \sIndex_{\np_\multii}$.

\begin{lem} \label{TProductLem} 
Suppose $\Summed_\multii < \ppmin(q)$. Let
\begin{align} \label{DefectCond} 
\alpha_1,\alpha_2 \in 
\PP_\lds, \qquad 
\Defect_1 \in \DefectSet\sub{\alpha_1,\sIndex_1} \cap \DefectSet_\flds, \quad \Defect_2 \in \DefectSet\sub{\alpha_2,\sIndex_1} \cap \DefectSet_\flds, 
\qquad  \beta_1,\gamma_1 \in \smash{\PP_\flds\super{\Defect_1}}, \quad \beta_2,\gamma_2 \in \smash{\PP_\flds\super{\Defect_2}} ,
\end{align} 
and, for each index
\begin{align} \label{EDefSet} 
i \in \DefectSet\sub{\Defect_1,\Defect_2} \cap \DefectSet\sub{t,t} 
= \big\{ |\Defect_1-\Defect_2|, |\Defect_1 - \Defect_2| + 2 | \, , \ldots, 
\min(\Defect_1 + \Defect_2, 2 t) \big\},
\end{align} 
let $\smash{T_i \left( 
\begin{smallmatrix}
\alpha_1 & \beta_1 & \gamma_1 \\ 
\alpha_2 & \beta_2 & \gamma_2 
\end{smallmatrix} 
\right)}$ 
denote the following product of three tangles: 
\begin{align} \label{TProduct}  
\vcenter{\hbox{\includegraphics[scale=0.275]{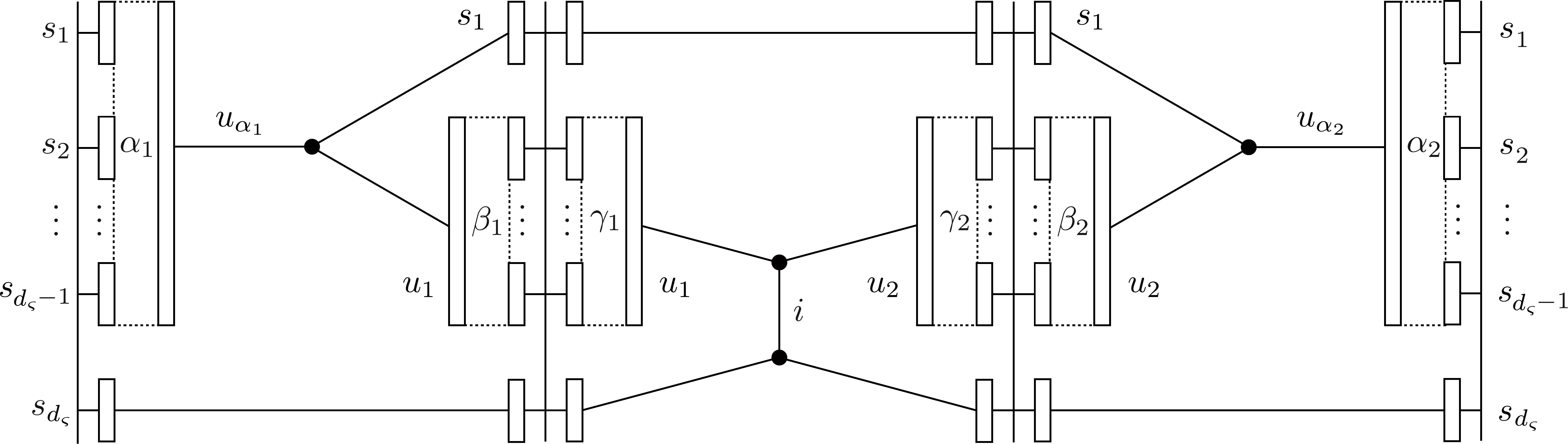} .}} 
\end{align} 
Then we have
\begin{align} \label{TProduct2} T_i 
\begin{pmatrix}
\alpha_1 & \beta_1 & \gamma_1 \\ 
\alpha_2 & \beta_2 & \gamma_2 
\end{pmatrix} 
\; = \; c_i 
\begin{pmatrix}
\alpha_1 & \beta_1 & \gamma_1 \\ 
\alpha_2 & \beta_2 & \gamma_2 
\end{pmatrix}
\,\, \times \,\, \vcenter{\hbox{\includegraphics[scale=0.275]{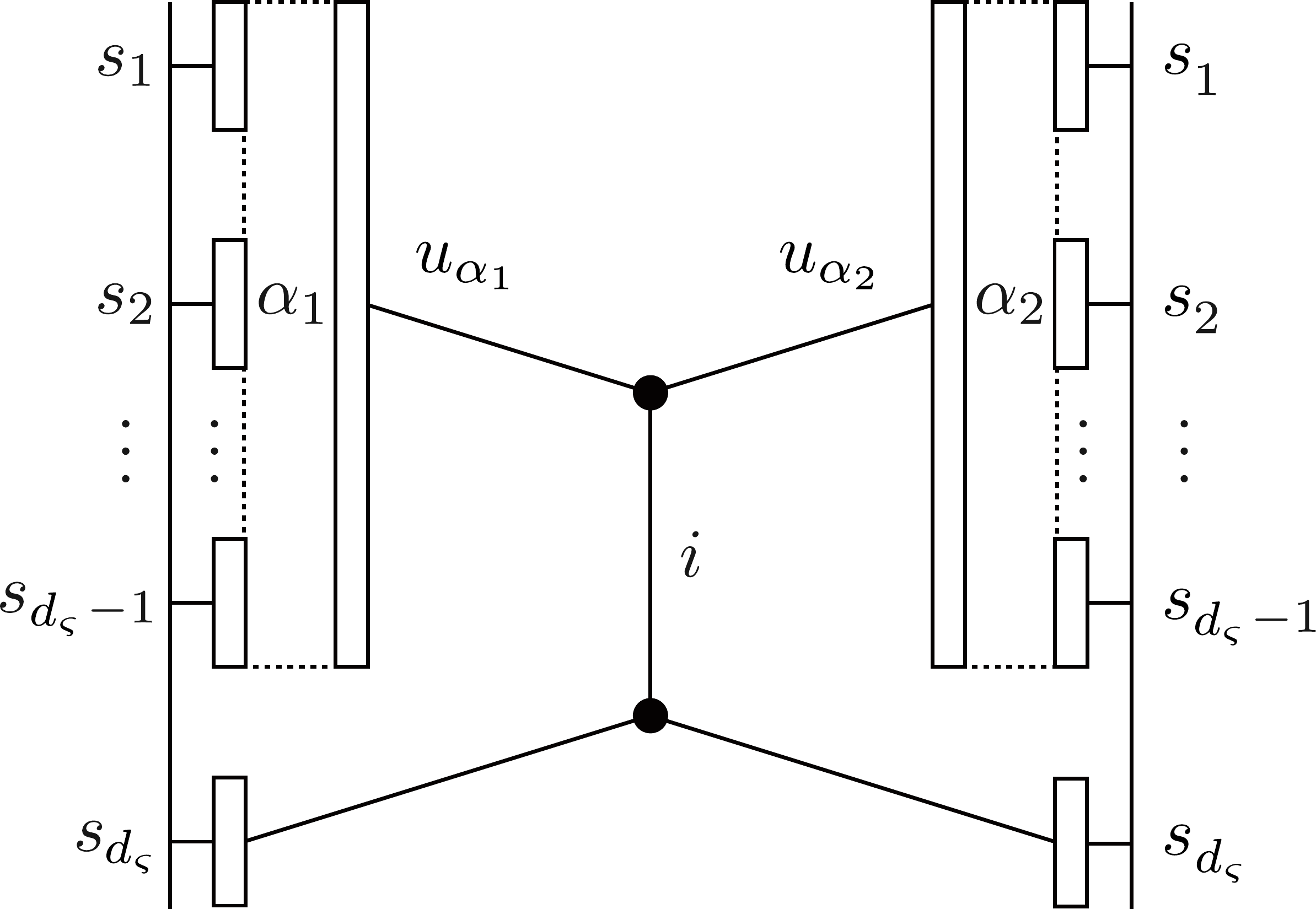} ,}} 
\end{align} 
where the coefficient equals 
\begin{align} \label{TProductCoeff} 
c_i 
\begin{pmatrix}
\alpha_1 & \beta_1 & \gamma_1 \\ 
\alpha_2 & \beta_2 & \gamma_2 
\end{pmatrix} = 
\TetraNet \left[
\begin{array}{lll}
\Defect_1 & \Defect_{\alpha_1} & i \\ 
\Defect_{\alpha_2} & \Defect_2 & \sIndex_1 
\end{array} 
\right] 
\dfrac{\BiForm{\beta_1}{\gamma_1} \BiForm{\beta_2}{\gamma_2}}{\ThetaNet(i,\Defect_{\alpha_1},\Defect_{\alpha_2})}
\; \one{ \left\{i \in \DefectSet\sub{\alpha_1,\alpha_2} \cap \DefectSet\sub{\Defect_1,\Defect_2} \cap \DefectSet\sub{t,t} \right\}}.
\end{align} 
\end{lem}

\begin{proof} 
To begin, we verify that the lemma statement makes sense.  First, we observe that by lemma~\ref{ParityLem} with $\multii \mapsto \lds$, the sets 
$\smash{\DefectSet\sub{\alpha_1,\sIndex_1} \cap \DefectSet_\flds}$ and $\smash{\DefectSet\sub{\alpha_2,\sIndex_1} \cap \DefectSet_\flds}$ are not empty and all elements 
of their union have the same parity.  Therefore, the elements of $\DefectSet\sub{\Defect_1,\Defect_2}$ are even, so the intersection 
$\DefectSet\sub{\Defect_1,\Defect_2} \cap \DefectSet\sub{t,t}$ is not empty. 
Finally, conditions~\eqref{DefectCond} on $\Defect_1$ and $\Defect_2$ guarantee that the three-vertices of the left and right tangles in 
the product~\eqref{TProduct} exist, and the condition 
$i \in  \DefectSet\sub{\Defect_1,\Defect_2} \cap \DefectSet\sub{t,t}$
guarantees that the three-vertices of the middle tangle in this product exist.

Now we compute the product tangle $T_i$ of~\eqref{TProduct}, for all $i$ as in~\eqref{EDefSet}. On the one hand, 
item~\ref{ExtractLemItem} of lemma~\ref{CollectionLem} from appendix~\ref{TLRecouplingSect}
shows that tangle~\eqref{TProduct} simplifies to 
\begin{align} \label{FirstT} 
T_i \left( 
\begin{array}{lll} 
\alpha_1 & \beta_1 & \gamma_1 \\ 
\alpha_2 & \beta_2 & \gamma_2 
\end{array} 
\right)
\; \overset{\eqref{ExtractID}}{=} \;  \, \BiForm{\beta_1}{\gamma_1} \BiForm{\beta_2}{\gamma_2} 
\,\, \times \,\, \vcenter{\hbox{\includegraphics[scale=0.275]{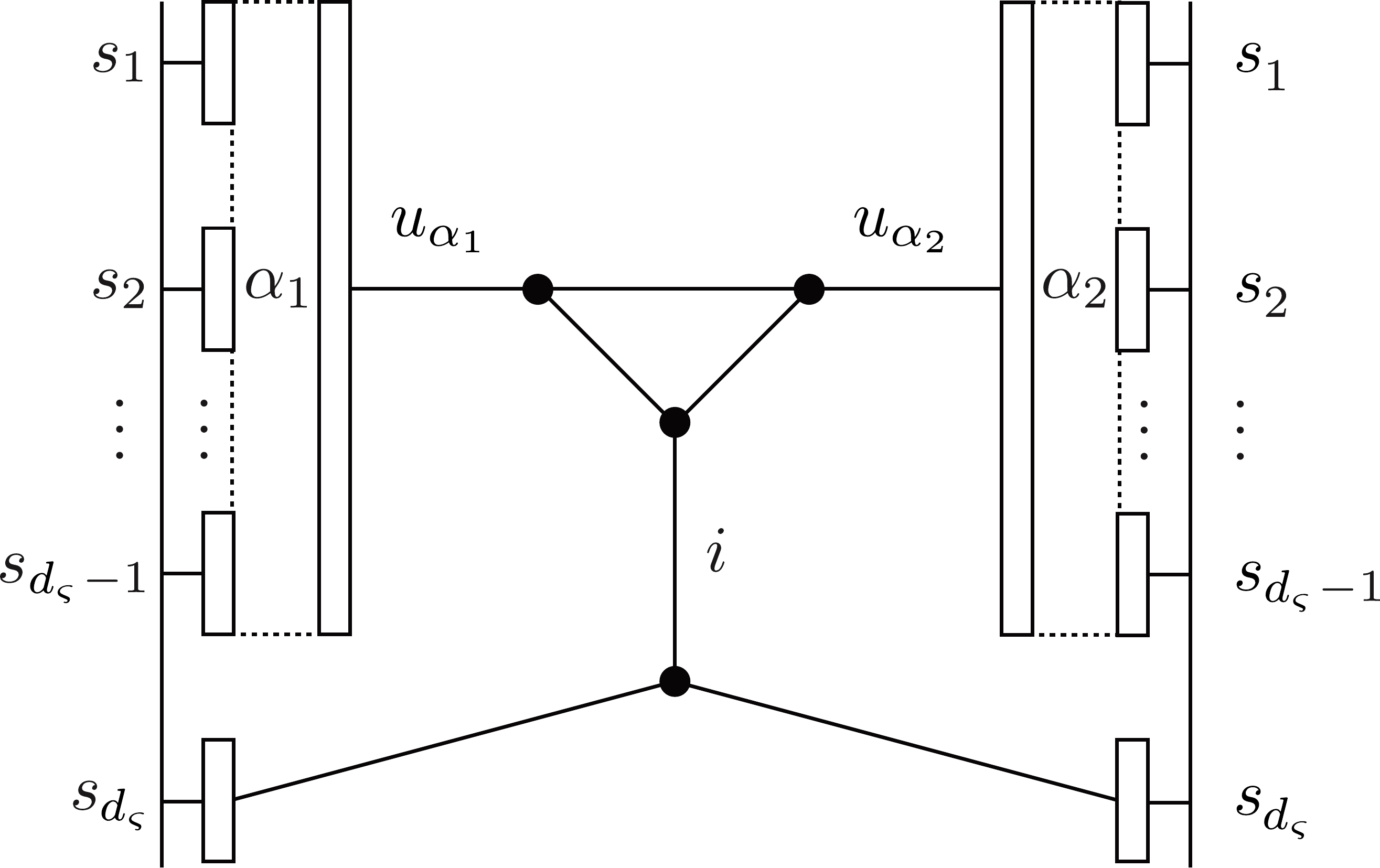} .}} 
\end{align}

On the other hand, we may express the tangle in~\eqref{FirstT} as a linear combination of tangles of 
the basis $\PD3_\multii$~\eqref{InsertOneBox} of lemma~\ref{ManyBasesLem}, with the same left/right link patterns 
$\alpha_1 , \alpha_2 \in \smash{\PP_{\lds}}$ as in this tangle. We then have
\begin{align} \label{SecondT}
T_i \left( 
\begin{array}{lll} 
\alpha_1 & \beta_1 & \gamma_1 \\ 
\alpha_2 & \beta_2 & \gamma_2 
\end{array} \right) \,\,\, 
= \; \sum_{j \, \in \, \DefectSet\sub{\alpha_1,\alpha_2} \cap \, \DefectSet\sub{t,t}}
c_{ij} \left( 
\begin{array}{lll} 
\alpha_1 & \beta_1 & \gamma_1 \\ 
\alpha_2 & \beta_2 & \gamma_2 
\end{array} \right) 
\,\, \times \,\, \vcenter{\hbox{\includegraphics[scale=0.275]{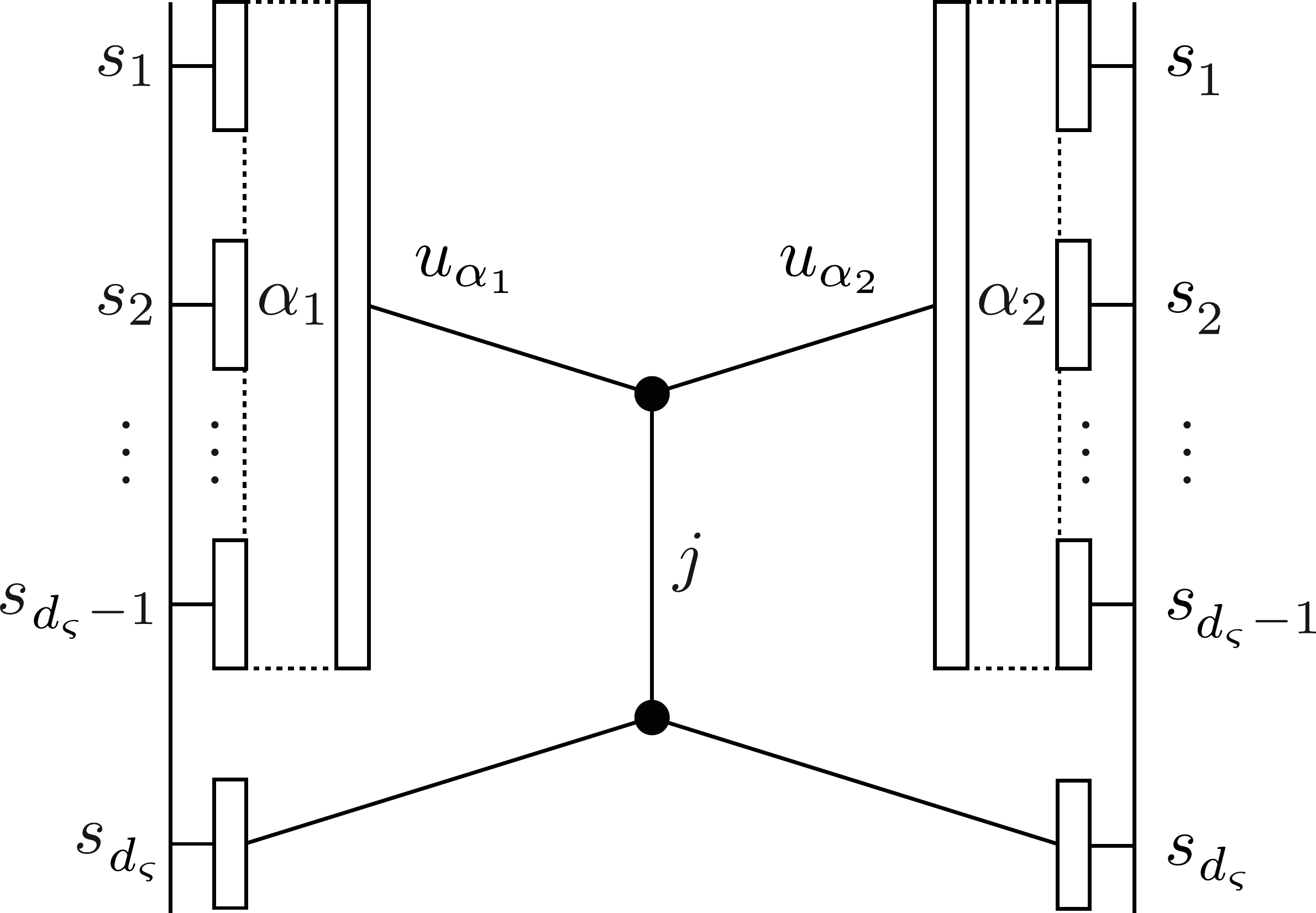} ,}} 
\end{align} 
for some coefficients 
$c_{ij} = \smash{c_{ij} \left( 
\begin{smallmatrix}
\alpha_1 & \beta_1 & \gamma_1 \\ 
\alpha_2 & \beta_2 & \gamma_2 
\end{smallmatrix} 
\right)} \in \bC$, which we find by equating~\eqref{FirstT} with~\eqref{SecondT} and inserting the tangles of either side into the ``dual" tangle 
\begin{align} \label{DualDiagram} 
\vcenter{\hbox{\includegraphics[scale=0.275]{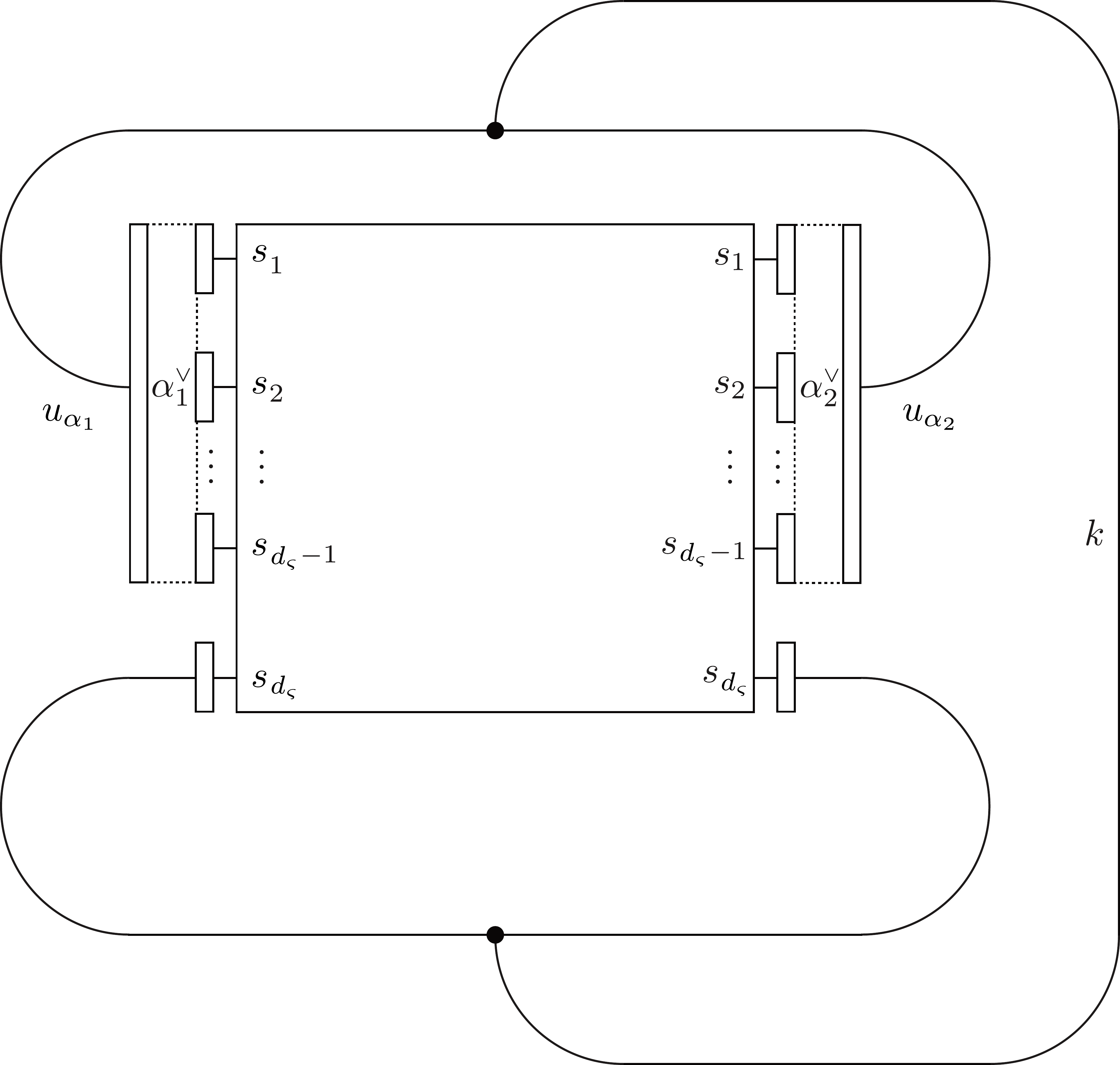} ,}} 
\end{align} 
for each 
\begin{align} \label{EalphSet2} 
k \in \DefectSet\sub{\alpha_1,\alpha_2} \cap \DefectSet\sub{t,t} 
= \big\{ |\Defect_{\alpha_1} - \Defect_{\alpha_2}|, |\Defect_{\alpha_1} - \Defect_{\alpha_2}| + 2, 
\ldots, \min(\Defect_{\alpha_1} + \Defect_{\alpha_2}, 2 t) \big\},
\end{align} 
thereby closing all links into loops.  
(We note that with our assumption $\smax(\smash{\lds}) < \Summed_\multii < \ppmin(q)$, 
theorem~\ref{BigSSTHM} shows that $\rad \smash{\PS_{\lds}} = \{0\}$, so
the dual link states $\alpha_1^\cheque, \alpha_2^\cheque \in \smash{\PS_{\lds}}$~\eqref{DualLS} do exist.) Thus, we arrive with
\begin{align} \label{InsertionResult}
\BiForm{\beta_1}{\gamma_1} 
\BiForm{\beta_2}{\gamma_2} \quad \vcenter{\hbox{\includegraphics[scale=0.275]{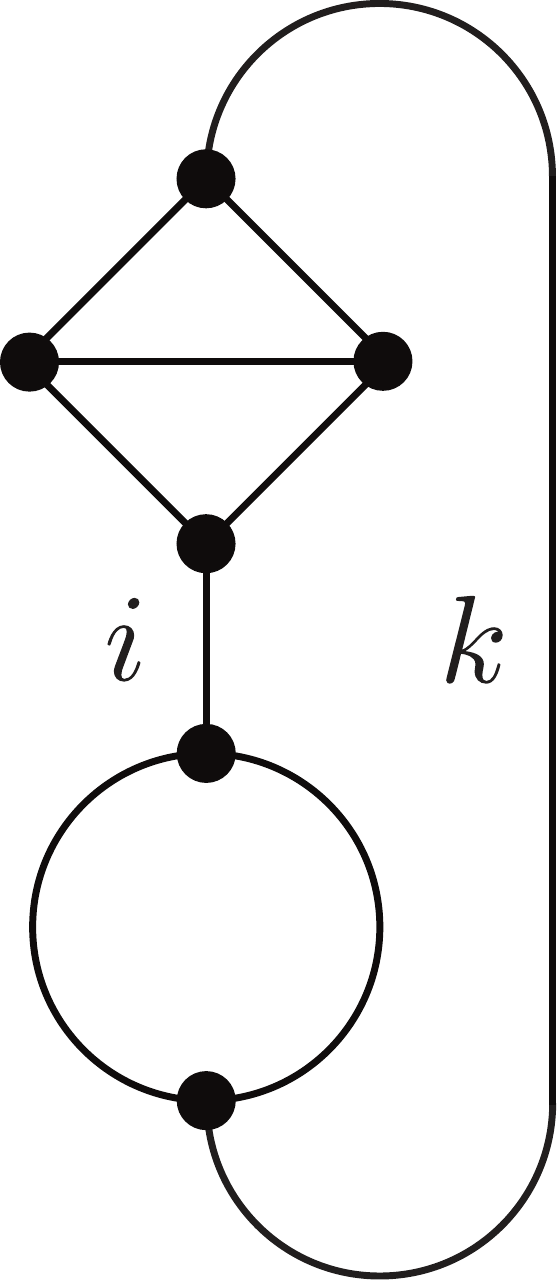}}} 
\quad = \quad  
\sum_{j \, \in \, \DefectSet\sub{\alpha_1,\alpha_2} \cap \, \DefectSet\sub{t,t}}
c_{ij} \left( 
\begin{array}{lll} 
\alpha_1 & \beta_1 & \gamma_1 \\ 
\alpha_2 & \beta_2 & \gamma_2 
\end{array} 
\right) \,\, \times \,\, \vcenter{\hbox{\includegraphics[scale=0.275]{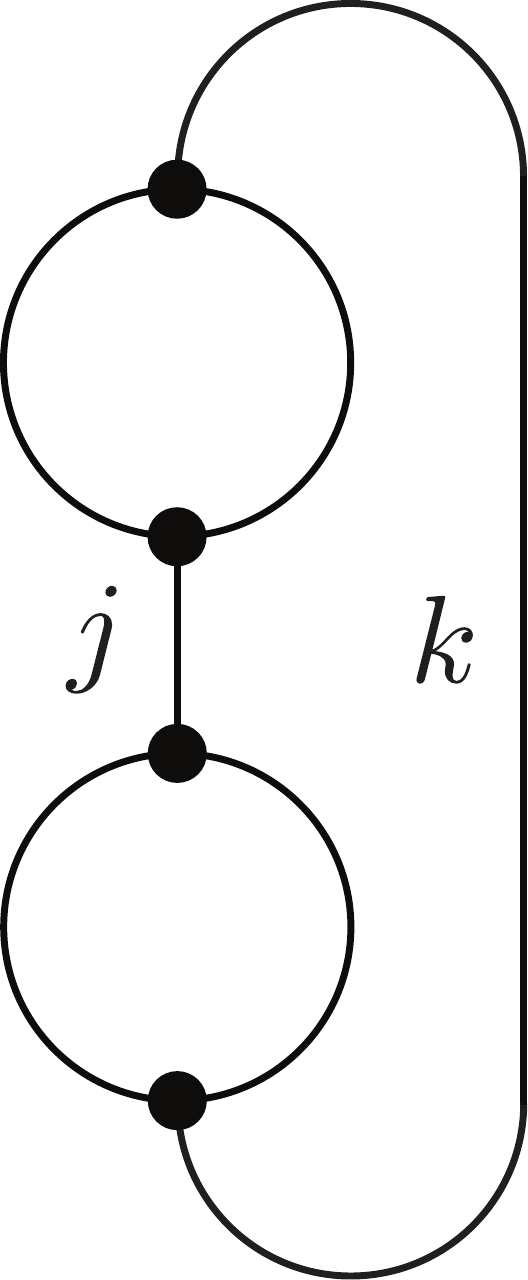} .}}
\end{align} 
(We do not label the sizes of all cables in these networks.  One may infer those sizes from~\eqref{TProduct}.)  
Then, using the simplification rules from
items~\ref{ExtractLemItem} and~\ref{LoopErasureLemItem} of lemma~\ref{CollectionLem},
we delete the lower loop of either network, finding
\begin{align}
& \delta_{ik} \BiForm{\beta_1}{\gamma_1} \BiForm{\beta_2}{\gamma_2} 
\frac{\ThetaNet(i,t,t)}{(-1)^i[i+1]} 
\,\, \times \,\, \vcenter{\hbox{\includegraphics[scale=0.275]{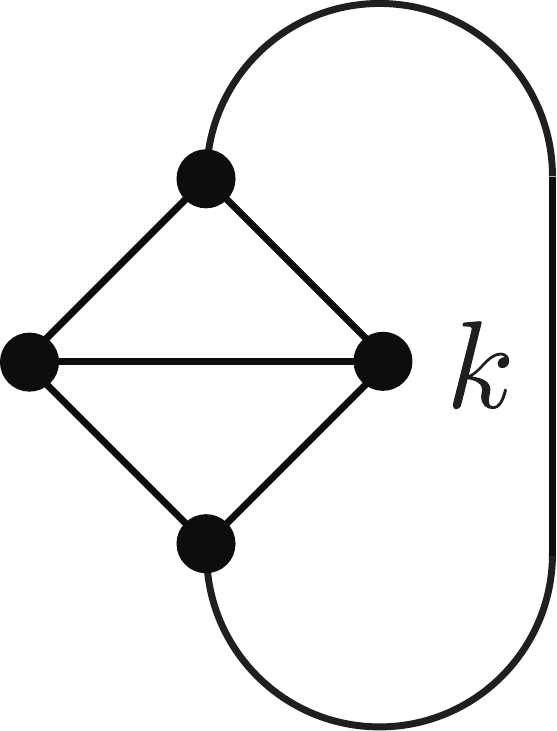}}} 
\\
= &  \; \sum_{j \, \in \, \DefectSet\sub{\alpha_1,\alpha_2} \cap \,\DefectSet\sub{t,t}}
c_{ij} \left( 
\begin{array}{lll}
\alpha_1 & \beta_1 & \gamma_1 \\ 
\alpha_2 & \beta_2 & \gamma_2 
\end{array} 
\right)
\delta_{jk} \frac{\ThetaNet(j,t,t)}{(-1)^j[j+1]}
\,\, \times \,\, \vcenter{\hbox{\includegraphics[scale=0.275]{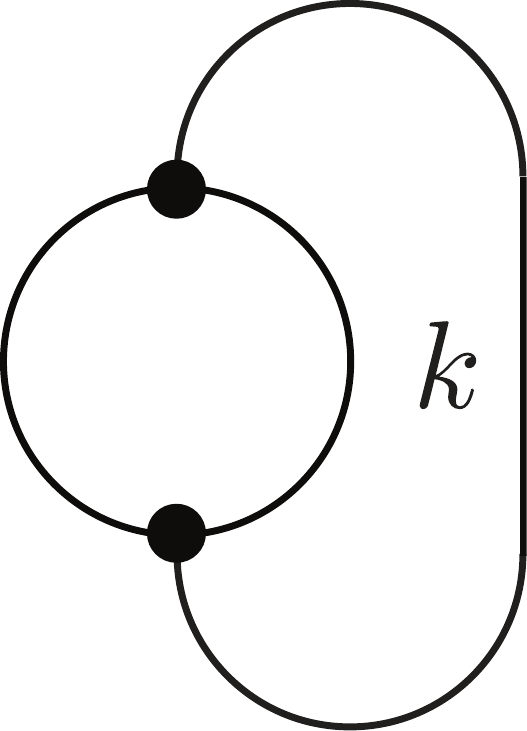} .}}
\end{align}
The left (resp.~right) side has a Tetrahedral~\eqref{TetraDefinition} (resp.~Theta~\eqref{ThetaDefinition}) network. Hence, we have
\begin{align} \label{ciks} 
c_{ik} \left( 
\begin{array}{lll} 
\alpha_1 & \beta_1 & \gamma_1 \\ 
\alpha_2 & \beta_2 & \gamma_2 
\end{array} 
\right) 
= \delta_{ik} \, \TetraNet \left[ 
\begin{array}{lll}
\Defect_1 & \Defect_{\alpha_1} & i \\ 
\Defect_{\alpha_2} & \Defect_2 & \sIndex_1 
\end{array} 
\right] 
\frac{\BiForm{\beta_1}{\gamma_1} \BiForm{\beta_2}{\gamma_2}}{\ThetaNet(i,\Defect_{\alpha_1},\Defect_{\alpha_2})} . 
\end{align} 
After inserting this result into~\eqref{SecondT} and recalling (\ref{EDefSet},~\ref{EalphSet2}), we arrive with 
asserted identities~(\ref{TProduct2},~\ref{TProductCoeff}), where the coefficients are given by
$c_i := c_{ii}$.  This concludes the proof.
\end{proof}  
\subsection{Case of two projectors} \label{BaseCaseSec}

Now we prove item~\ref{GeneratorThmItem1} theorem~\ref{GeneratorThm} for $\multii = (\sIndex_1, \sIndex_2) \in \bZpos^2$.
Our task is to construct all tangles in $\WJ\sub{\sIndex_1, \sIndex_2}(\nu)$ from 
\begin{align} 
\label{EGeneratorsTwonode}
\WJProj\sub{\sIndex_1,\sIndex_2} \quad = \quad \vcenter{\hbox{\includegraphics[scale=0.275]{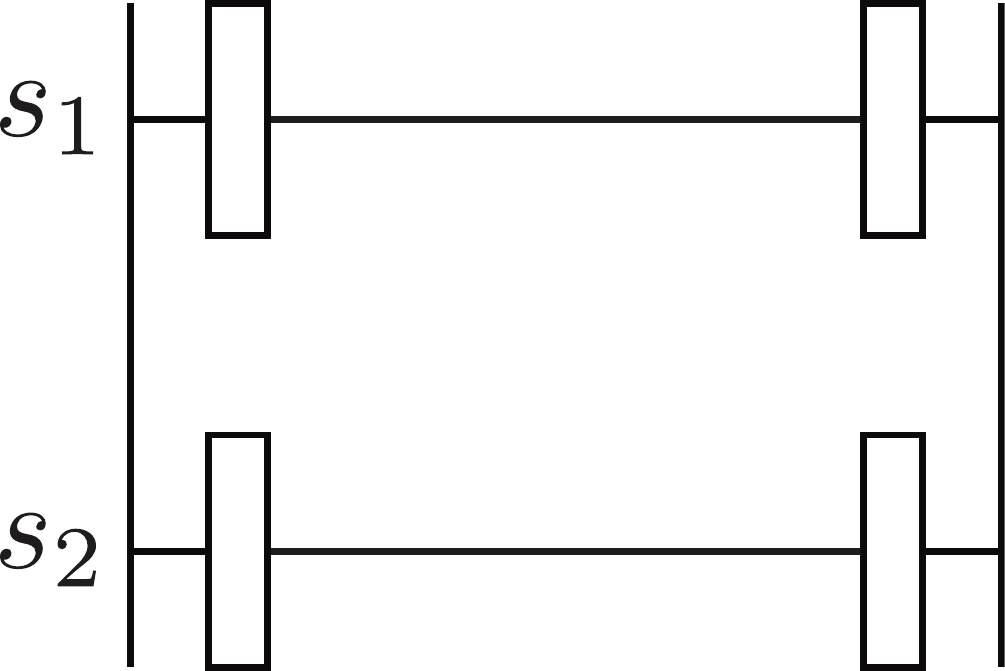}}} 
\qquad \qquad \textnormal{and} \qquad \qquad
\ValGenWJ_1 \quad = \quad \vcenter{\hbox{\includegraphics[scale=0.275]{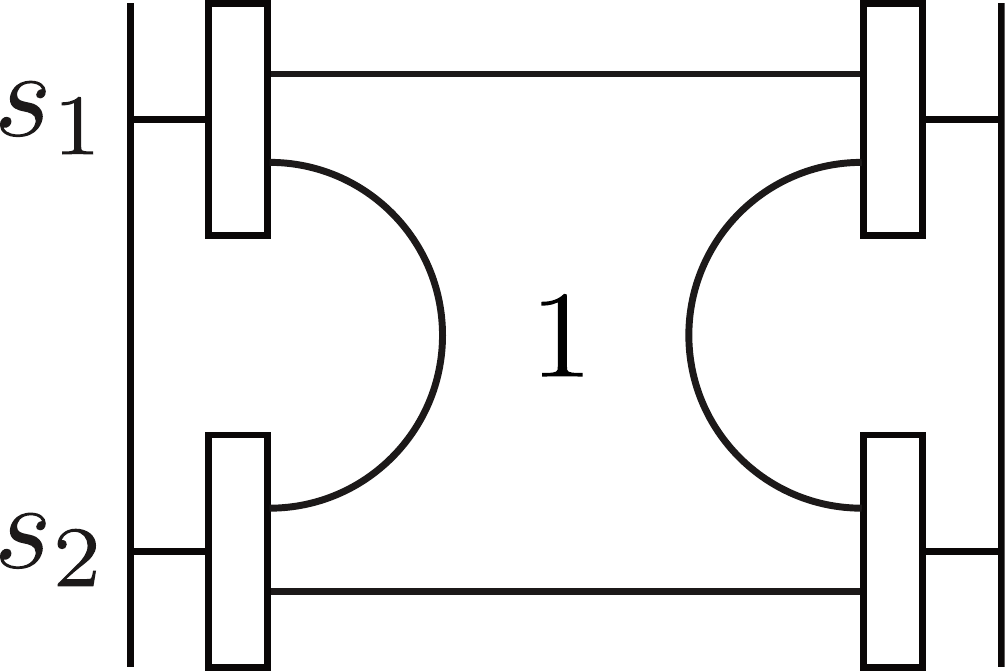} ,}} 
\end{align}  
which constitute the generating set
\begin{align}\label{GensetTwo}
\mathsf{G}\sub{\sIndex_1, \sIndex_2} := \WJProj\sub{\sIndex_1, \sIndex_2} 
\{ \mathbf{1}_{\TL_{\sIndex_1+\sIndex_2}}, \Gen^\TL_{\sIndex_1} \} \WJProj\sub{\sIndex_1, \sIndex_2} .
\end{align}
In fact, in this case we obtain a stronger result: the statement of corollary~\ref{InitialCaseCor} 
implies item~\ref{GeneratorThmItem1} of theorem~\ref{GeneratorThm} with $\multii = (\sIndex_1, \sIndex_2)$ under the weaker assumption 
$\max (\sIndex_1, \sIndex_2) < \ppmin(q)$ instead of $\sIndex_1 + \sIndex_2 < \ppmin(q)$.

\begin{lem} \label{InitialCaseLem} 
Suppose $\max (\sIndex_1, \sIndex_2) < \ppmin(q)$.  
Every tangle of the form
\begin{align} \label{PthDiagram} 
\vcenter{\hbox{\includegraphics[scale=0.275]{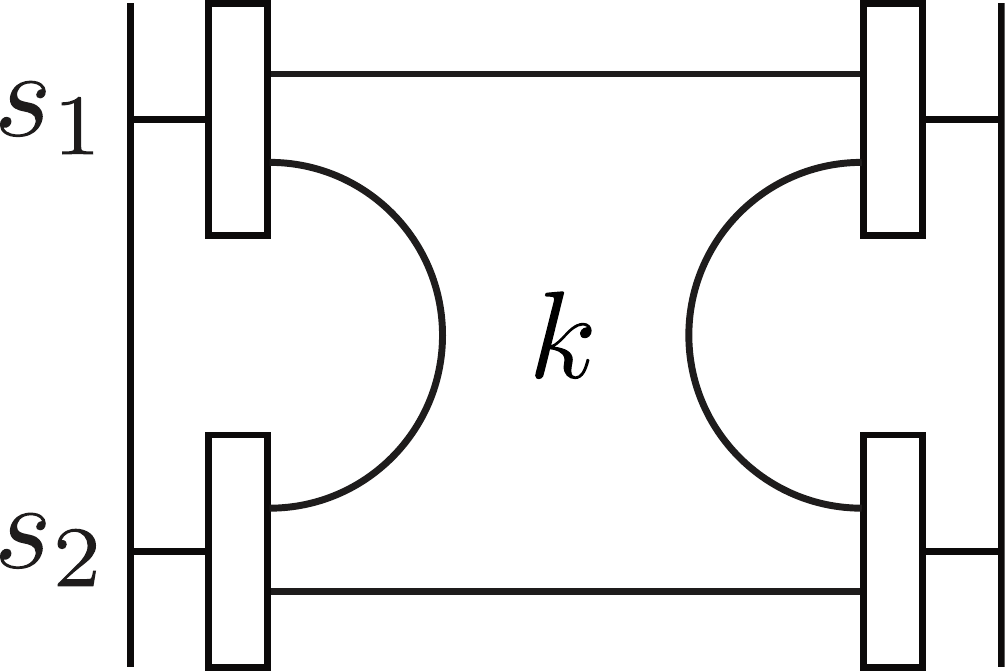} ,}} 
\end{align} 
where $k \in \{0,1,\ldots,\min (\sIndex_1, \sIndex_2)\}$,
is a polynomial in the elements of $\mathsf{G}\sub{\sIndex_1, \sIndex_2}$.
\end{lem}

\begin{proof}
We prove the claim by induction on $k \geq 0$. 
It is clearly true if $k \in \{0,1\}$.  
Assuming that the claim holds for tangles in~\eqref{PthDiagram} with cables of size $k$ joining the two boxes, 
we form the product
\begin{align} \label{MultiplyDiagrams} 
\vcenter{\hbox{\includegraphics[scale=0.275]{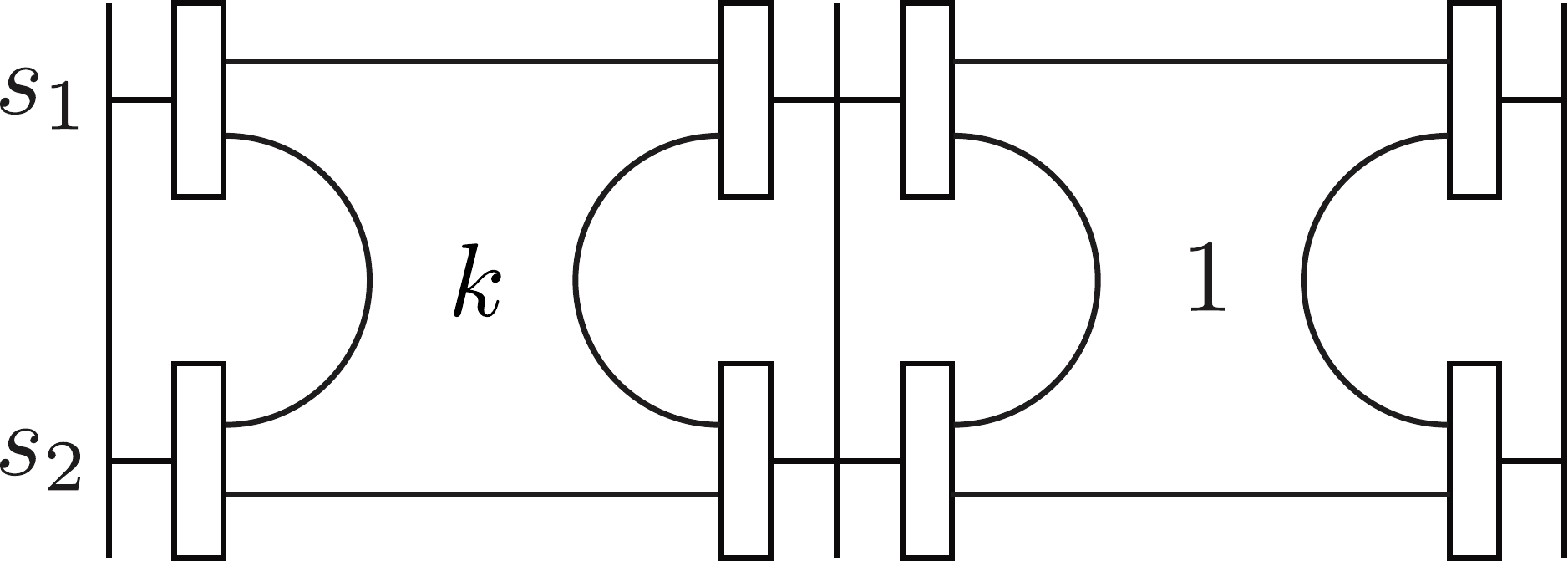}}} \quad 
\overset{\eqref{ProjectorID2}}{=}
\quad \vcenter{\hbox{\includegraphics[scale=0.275]{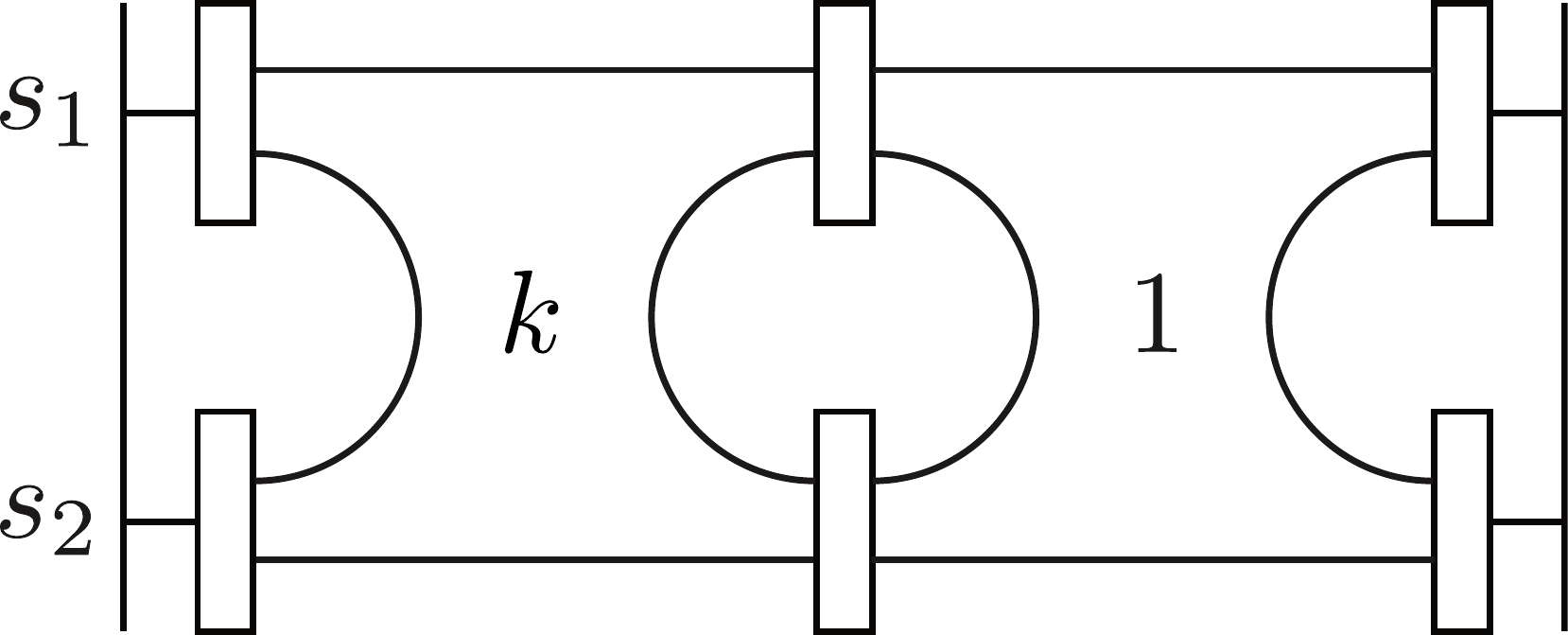} .}} 
\end{align} 
Next, we decompose the upper-middle projector box of~\eqref{MultiplyDiagrams} over all internal link diagrams. 
By rule~\eqref{ProjectorID2}, the only nonvanishing terms of~\eqref{MultiplyDiagrams} with the upper-middle box decomposed are
\begin{align} \label{TwoTerms} 
\vcenter{\hbox{\includegraphics[scale=0.275]{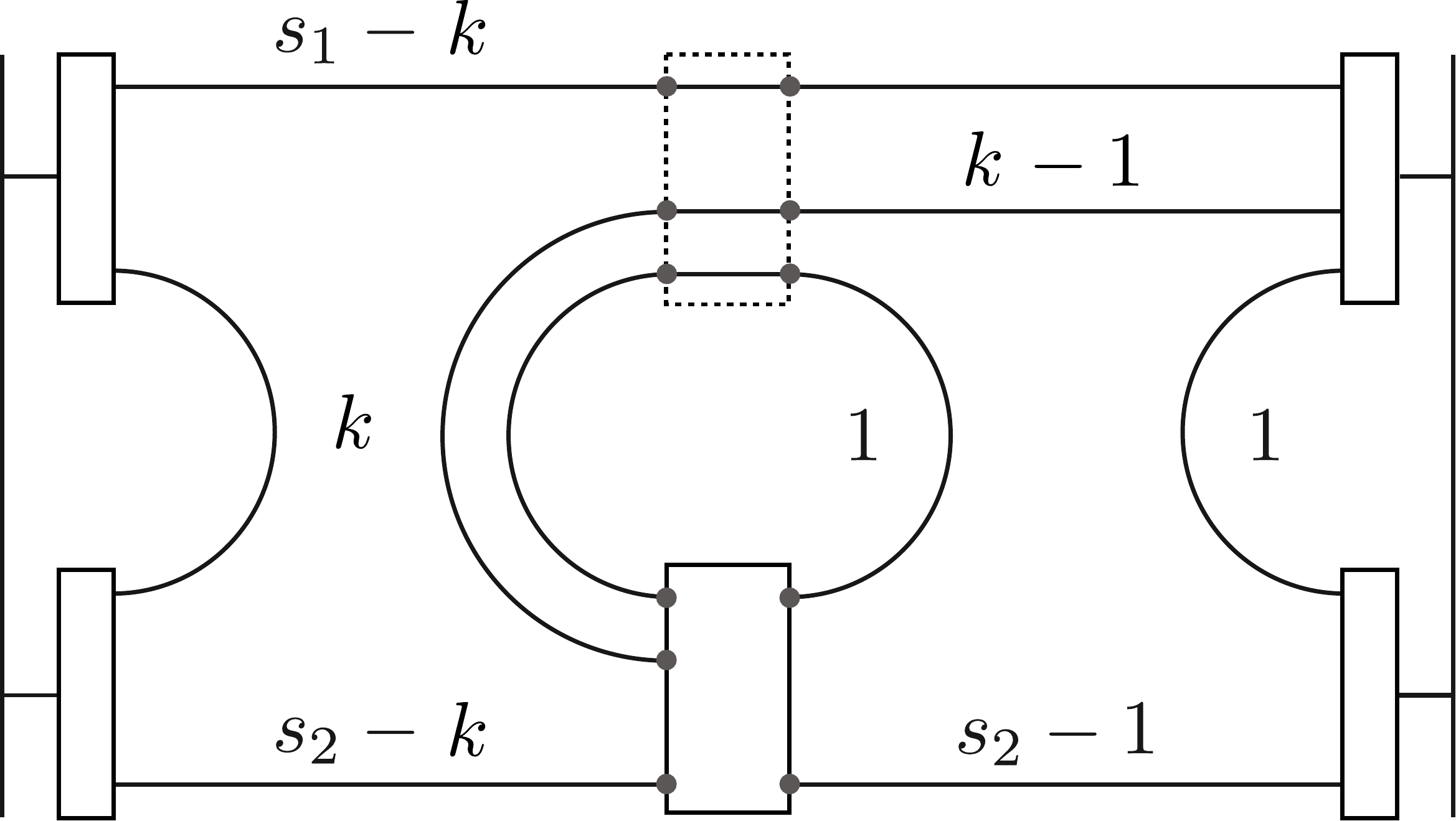}}} 
\quad + \; \frac{[\sIndex_1 - k]}{[\sIndex_1]} \,\, \times \,\, \vcenter{\hbox{\includegraphics[scale=0.275]{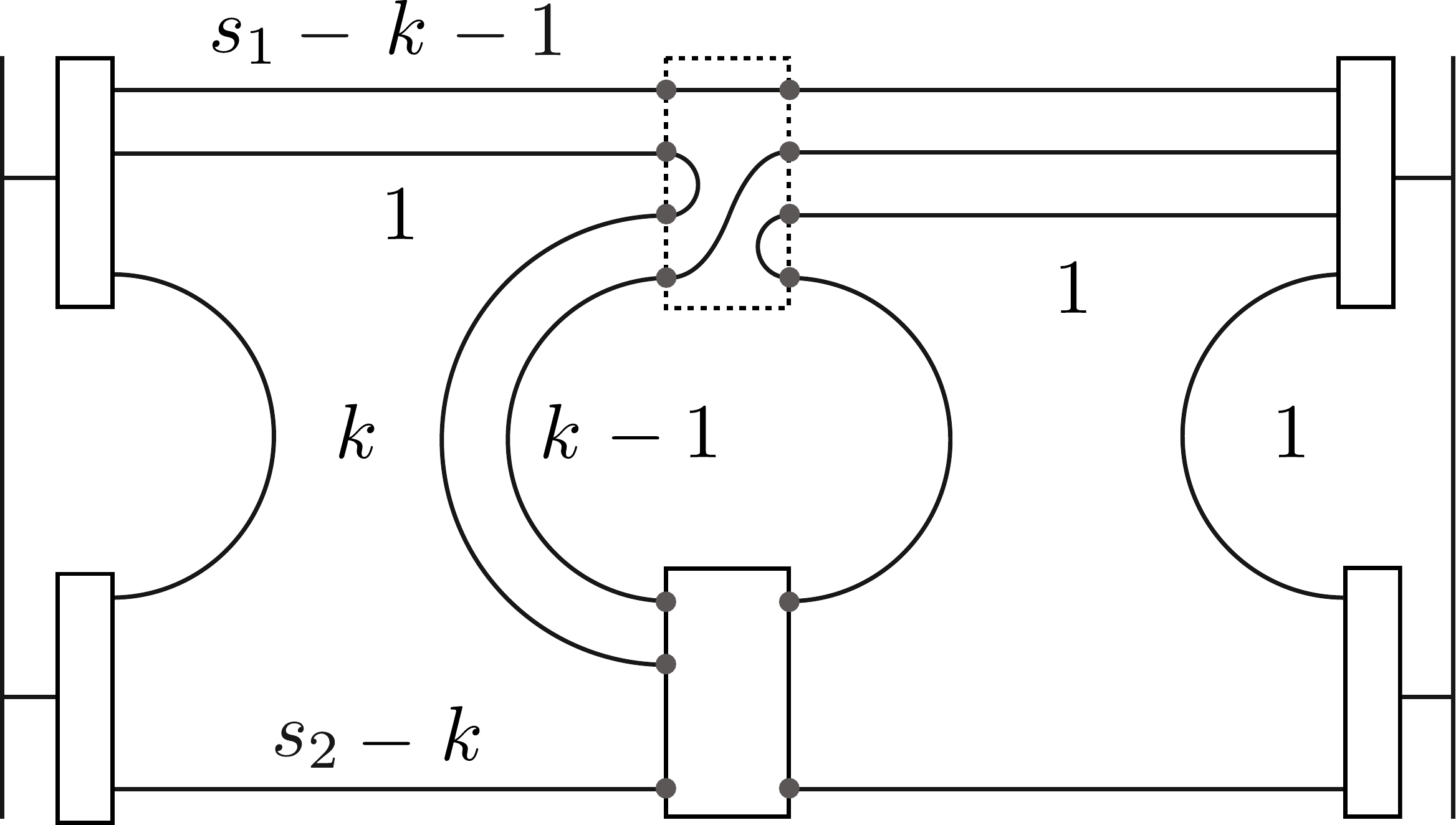} ,}}
\end{align} 
where the coefficient on the second term follows from formula~\eqref{SpecialT} of proposition~\ref{SpecialTProp}.
In the first tangle of~\eqref{TwoTerms}, one loop both enters and exits the bottom-middle box from opposite sides. 
Using identities (\ref{ProjectorID1},~\ref{DeltaTangleGen}), we get
\begin{align} \label{TwoTerms1} 
\hspace*{-3mm}
\vcenter{\hbox{\includegraphics[scale=0.275]{e-Generators115_withs1s2.pdf}}} \quad
= \; - \frac{[\sIndex_2 + 1]}{[\sIndex_2]} \,\, \times \,\, \vcenter{\hbox{\includegraphics[scale=0.275]{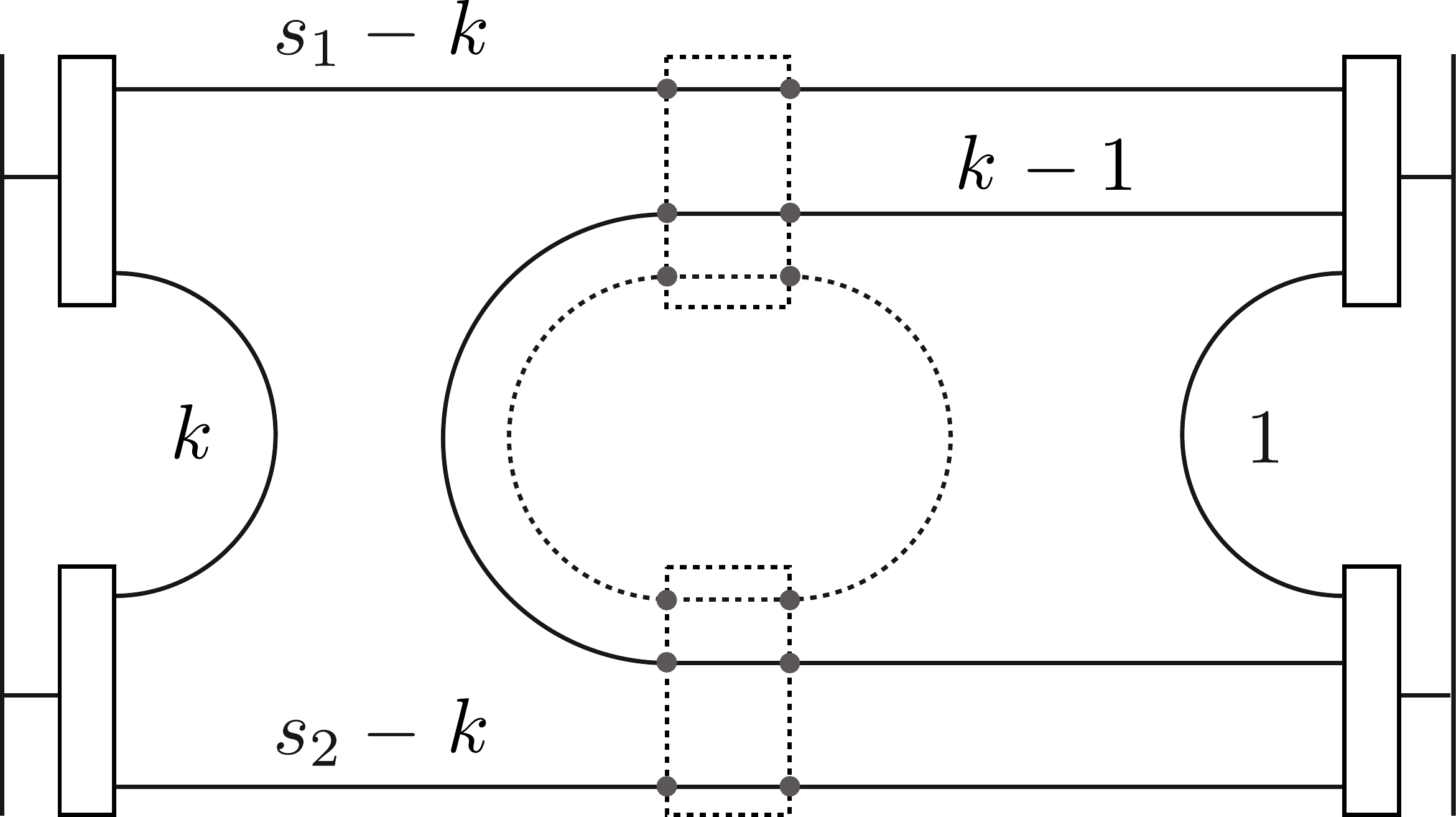} .}} 
\end{align} 
In the second tangle of~\eqref{TwoTerms}, all but two terms vanish after we decompose the lower-middle box.  
Using again~\eqref{SpecialT} from proposition~\ref{SpecialTProp} to find the coefficients of the tangles in the box decomposition, 
these terms are
\begin{align} \label{TwoTerms2} 
\vcenter{\hbox{\includegraphics[scale=0.275]{e-Generators116_withs1s2.pdf}}} \quad
= \; & \; \frac{[\sIndex_2 - k + 1]}{[\sIndex_2]} \,\, \times \,\,
\vcenter{\hbox{\includegraphics[scale=0.275]{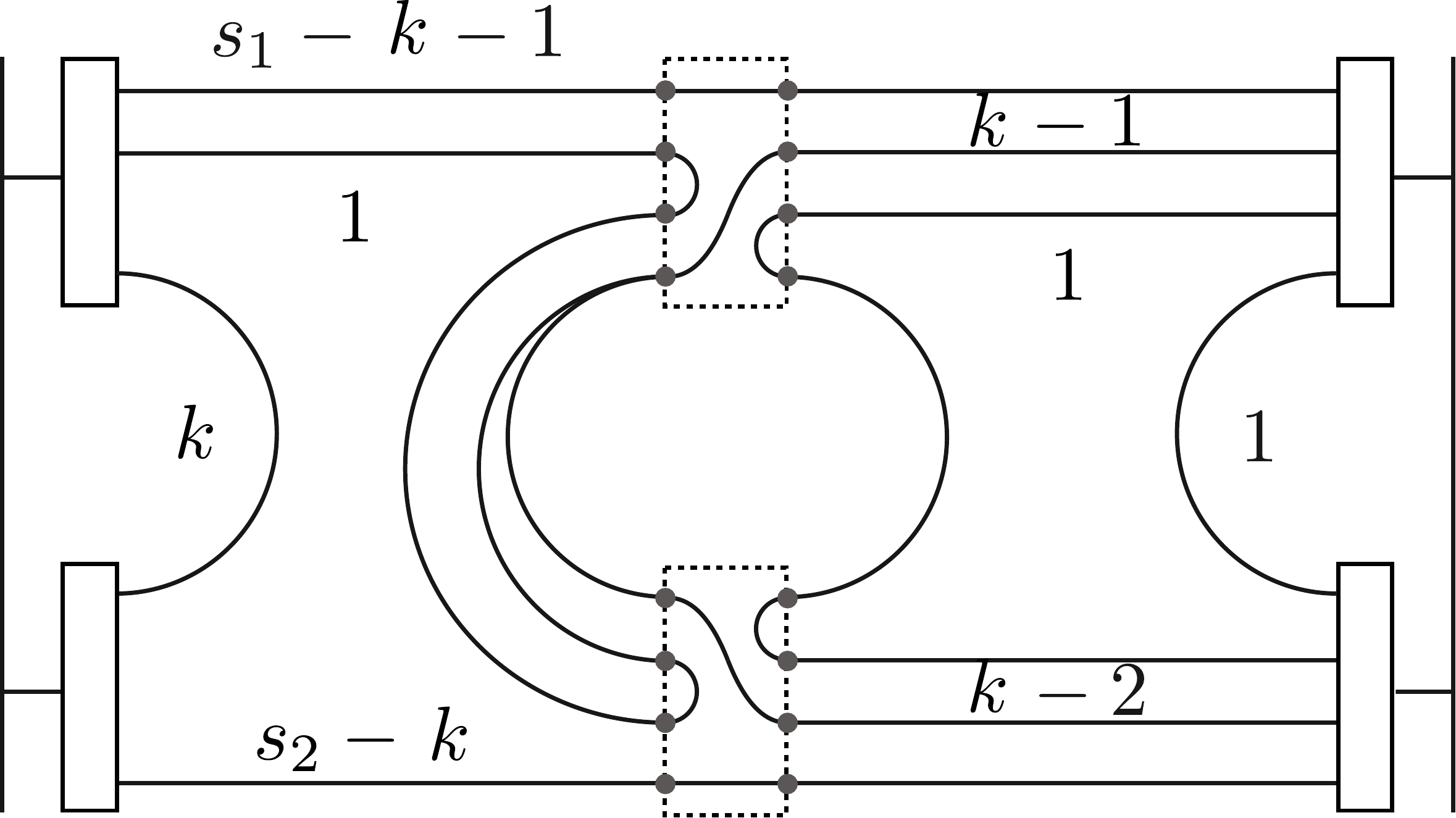}}} \hspace*{3mm} \\ 
\nonumber
\; & \quad + \frac{[\sIndex_2 - k]}{[\sIndex_2]} \,\, \times \,\,
\vcenter{\hbox{\includegraphics[scale=0.275]{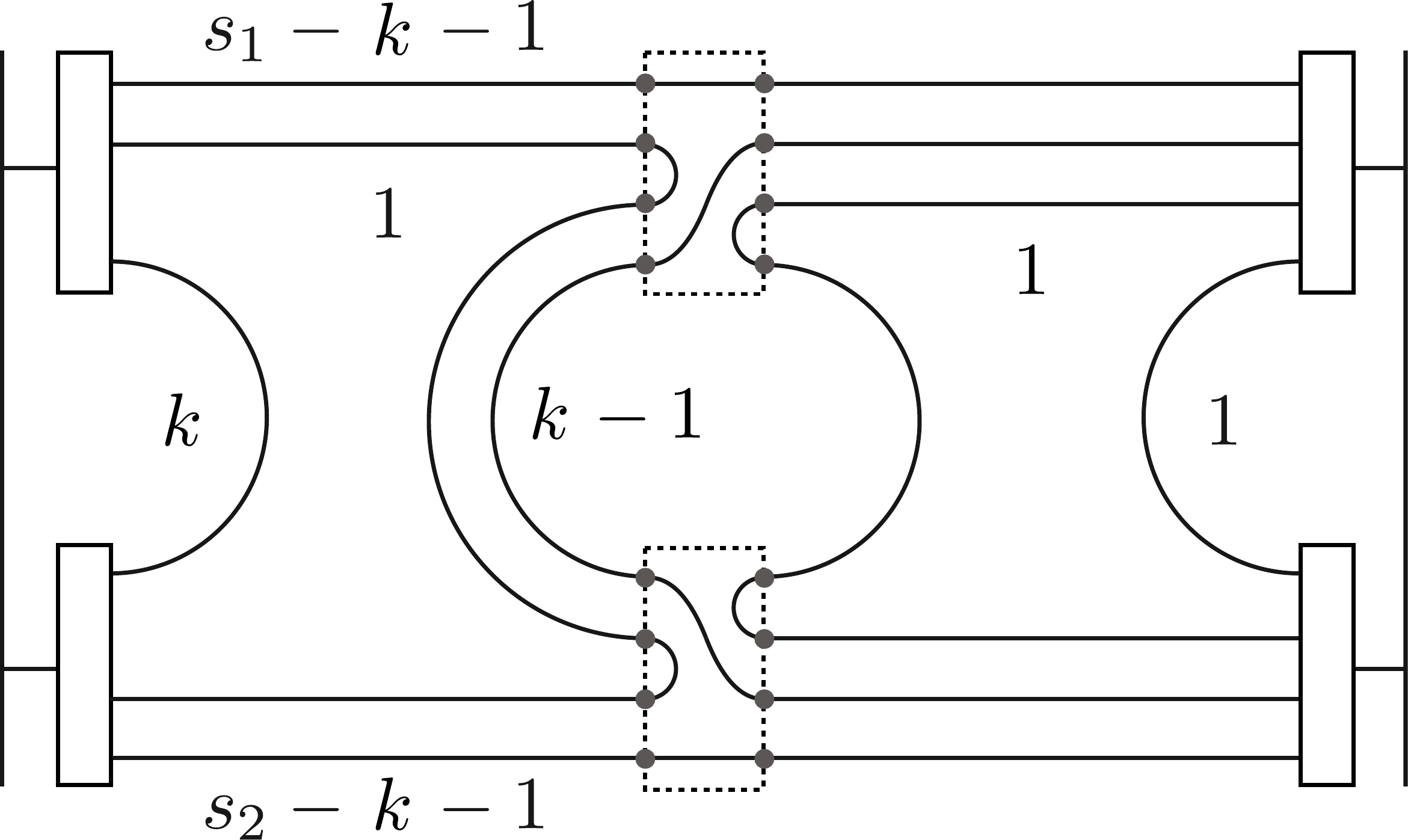} .}}
\end{align}
After combining (\ref{TwoTerms1},~\ref{TwoTerms2}) and using 
the simple identity~\eqref{QintID} from lemma~\ref{CollectionLem}, we find that~\eqref{MultiplyDiagrams} equals
\begin{alignat}{2} 
\label{TopLineTangles} 
\vcenter{\hbox{\includegraphics[scale=0.275]{e-Generators2.pdf}}} \quad
= \; & \; - \frac{[\sIndex_1 + \sIndex_2 - k + 1]}{[\sIndex_1][\sIndex_2]} \,\, \times \,\,
&&  \vcenter{\hbox{\includegraphics[scale=0.275]{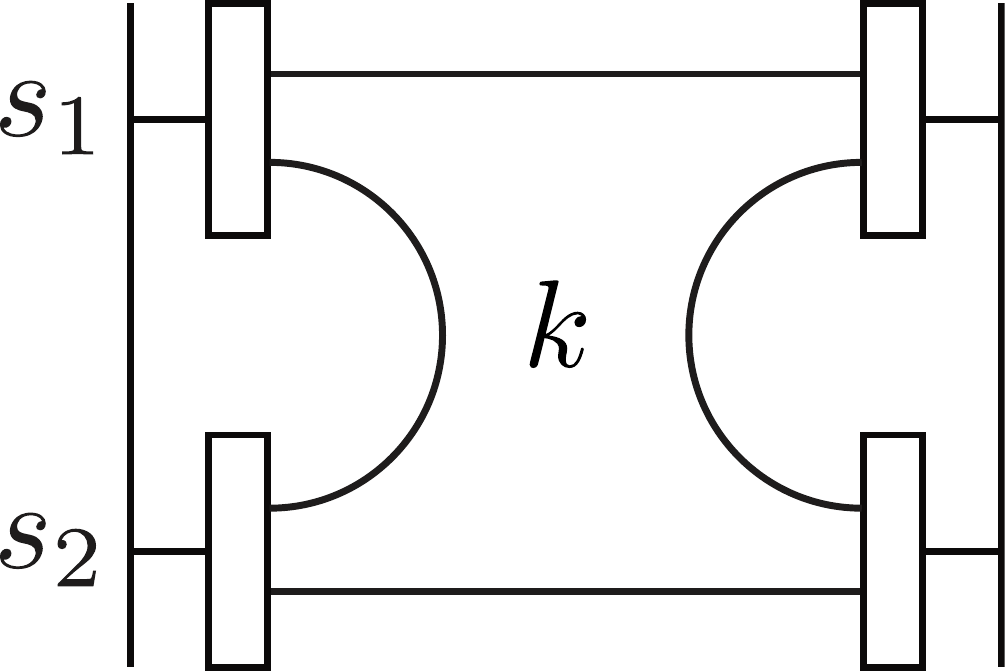}}}  \\[1em]
\label{BotTangle} 
\; & \quad + \frac{[\sIndex_1 - k][\sIndex_2 - k]}{[\sIndex_1][\sIndex_2]} \,\, \times \,\,
&&  \vcenter{\hbox{\includegraphics[scale=0.275]{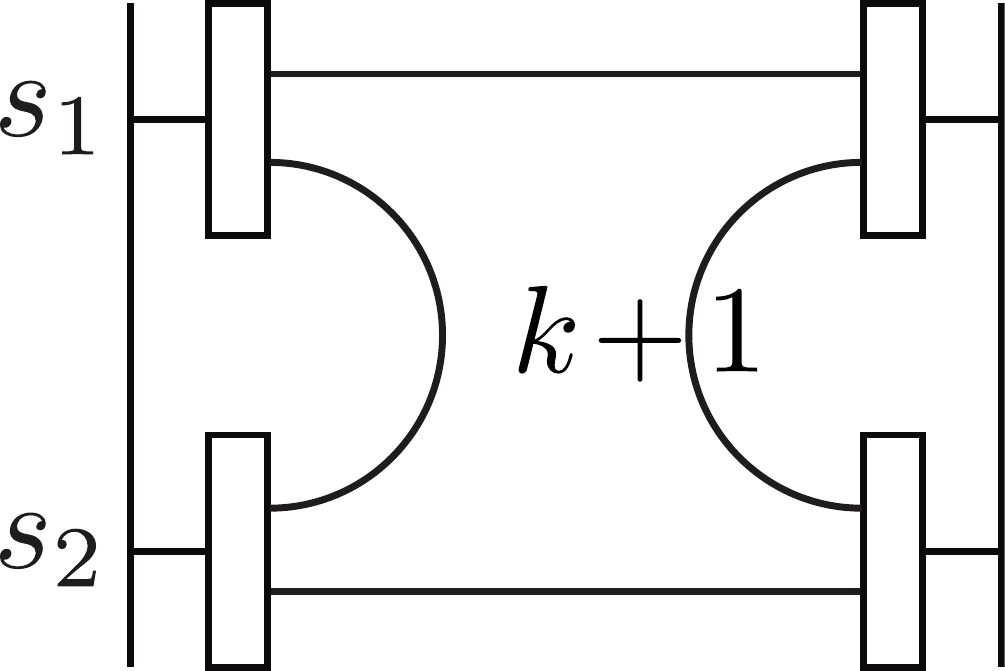} .}}
\end{alignat}
Now, by the induction hypothesis, the two tangles in~\eqref{TopLineTangles} are polynomials in 
the tangle $\WJProj\sub{\sIndex_1,\sIndex_2}\Gen_{\sIndex_1}^{\TL}\WJProj\sub{\sIndex_1,\sIndex_2}$. 
Therefore, so is the tangle in~\eqref{BotTangle}. 
This finishes the induction step and concludes the proof.
\end{proof}

\begin{cor} \label{InitialCaseCor} 
Suppose $\max (\sIndex_1, \sIndex_2) < \ppmin(q)$.  Then the Jones-Wenzl algebra $\WJ\sub{\sIndex_1, \sIndex_2}(\nu)$ is generated by 
the collection $\mathsf{G}\sub{\sIndex_1, \sIndex_2}$~\eqref{GensetTwo} of Jones-Wenzl tangles.  
\end{cor}
\begin{proof}
The collection of tangles in lemma~\ref{InitialCaseLem} is a basis for $\WJ\sub{\sIndex_1, \sIndex_2}(\nu)$, so it follows 
that $\mathsf{G}\sub{\sIndex_1, \sIndex_2}$ generates $\WJ\sub{\sIndex_1, \sIndex_2}(\nu)$.  
\end{proof}
\subsection{Induction step and supporting lemmas} \label{IndStepSec}


In this section, we perform the core of the induction step to prove theorem~\ref{GeneratorThm} for a general multiindex~$\multii$. 
We use induction on the number $\np_\multii$ of projectors, the initial case being the content of the previous section~\ref{BaseCaseSec}. 
The present section contains a series of lemmas, and in the next section~\ref{TheProofSec}, we summarize the proof of theorem~\ref{GeneratorThm}.

\begin{InductAssump} \label{IndAss1}
Let $\np_\multii \geq 3$, and assume that item~\ref{GeneratorThmItem1} of theorem~\ref{GeneratorThm} holds
for any multiindex in the set $\{ (0) \} \cup \bZpos \cup \bZpos^2 \cup \dotsm \cup \bZpos^{\np_\multii - 1}$.
Denoting
\begin{align}
\label{GenSet2} 
\mathsf{G}_\multii &:= 
\WJProj_\multii \{ \mathbf{1}_{\TL_{\Summed_\multii}}, \Gen_{\sIndex_1}^{\TL}, \Gen_{\sIndex_1 + \sIndex_2}^{\TL}, \ldots, \Gen_{\sIndex_1 + \sIndex_2 + \dotsm + \sIndex_{\np_\multii-1}}^{\TL} \}  \WJProj_\multii , \\
\mathsf{G}_\lds &:= 
\WJProj_\multii \{ \mathbf{1}_{\TL_{\Summed_\multii}}, \Gen_{\sIndex_1}^{\TL}, \Gen_{\sIndex_1 + \sIndex_2}^{\TL}, \ldots, \Gen_{\sIndex_1 + \sIndex_2 + \dotsm + \sIndex_{\np_\multii-2}}^{\TL} \}  \WJProj_\multii , \\
\mathsf{G}_\fds &:= 
\WJProj_\multii \{ \mathbf{1}_{\TL_{\Summed_\multii}}, \Gen_{\sIndex_2}^{\TL}, \Gen_{\sIndex_2 + \sIndex_3}^{\TL}, \ldots, \Gen_{\sIndex_1 + \sIndex_2 + \dotsm + \sIndex_{\np_\multii-1}}^{\TL} \}  \WJProj_\multii ,
\end{align}
this assumption implies that $\smash{\WJ_\lds(\nu)} \subset \WJ_\multii(\nu)$ and $\smash{\WJ_\fds(\nu)} \subset \WJ_\multii(\nu) $ are generated by 
the respective collections $\smash{\mathsf{G}_\lds}$ and $\smash{\mathsf{G}_\fds}$ of Jones-Wenzl tangles.  
Specifically, this is equivalent to assuming that all of the tangles
\begin{align} \label{InductDiagrams} 
\vcenter{\hbox{\includegraphics[scale=0.275]{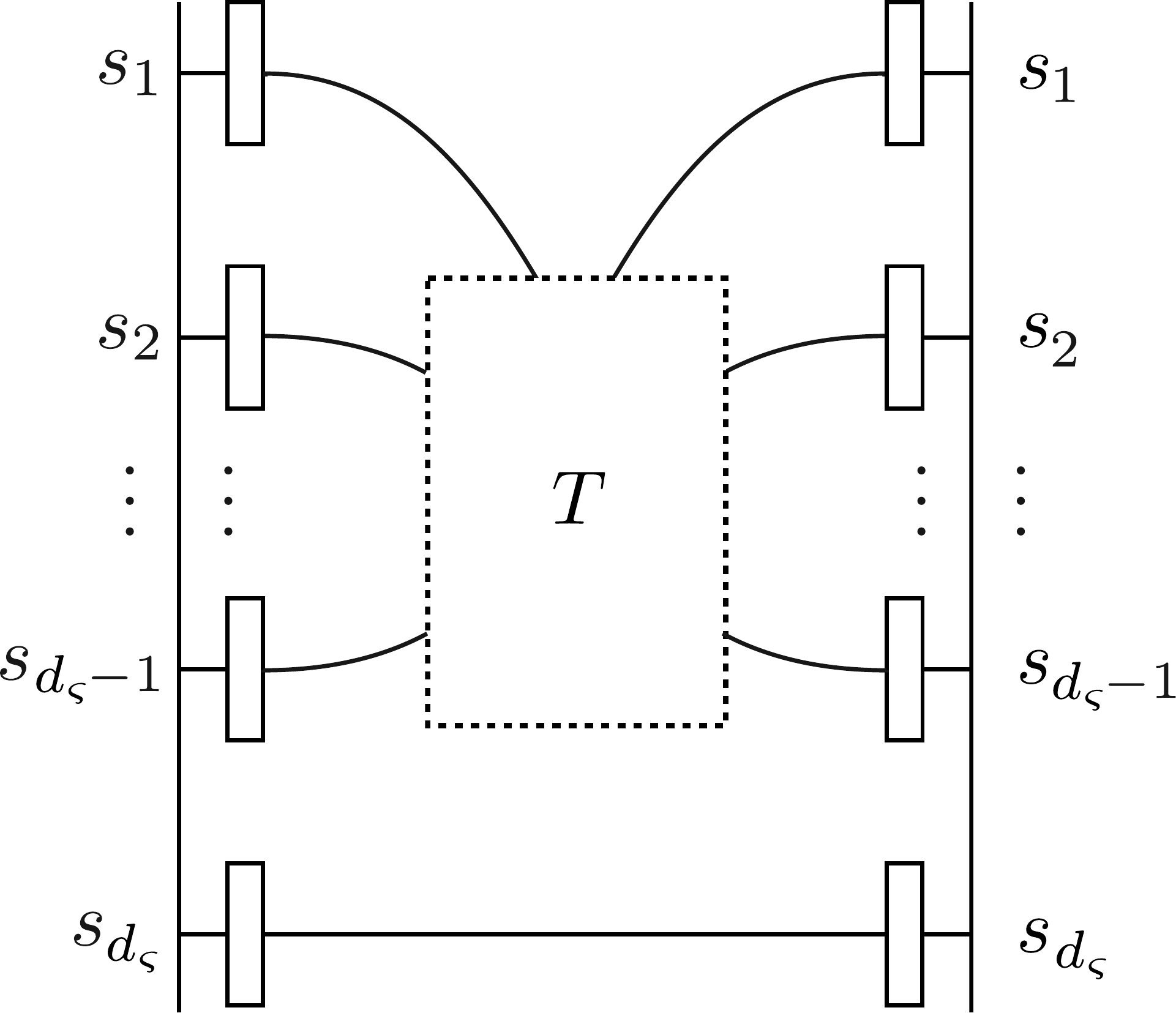}}} \quad
\qquad \qquad \textnormal{and} \qquad \qquad
\quad \vcenter{\hbox{\includegraphics[scale=0.275]{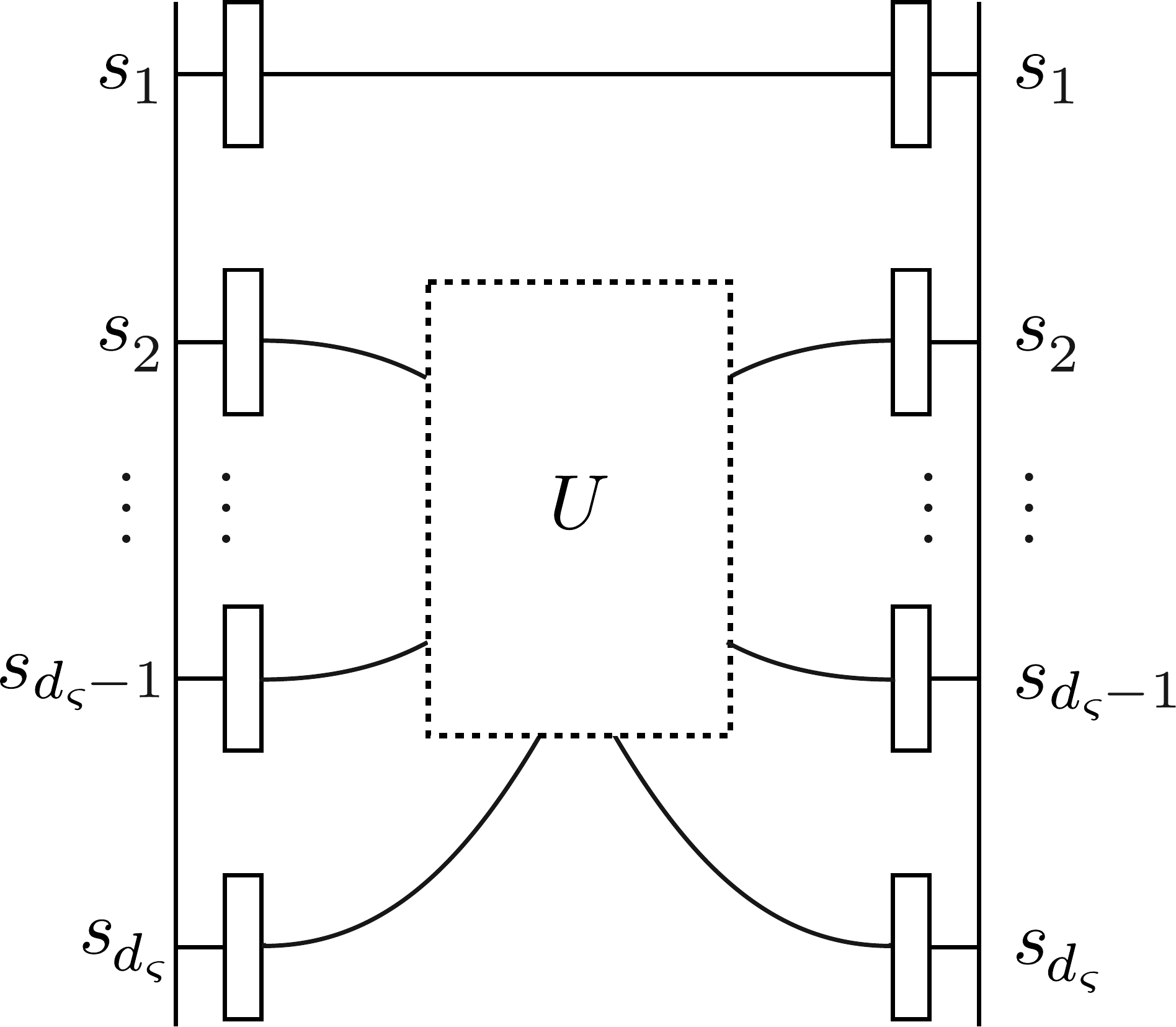} ,}} 
\end{align} 
where $T \in \smash{\TL_{\smax(\lds)}(\nu)}$ and $U \in \smash{\TL_{\smax(\fds)}(\nu)}$, 
are polynomials in the elements of the collections $\smash{\mathsf{G}_{\lds}}$ and $\smash{\mathsf{G}_{\fds}}$, respectively.
\end{InductAssump}

\begin{claim} \label{IndClaim}
Suppose 
$\Summed_\multii < \ppmin(q)$. If induction hypothesis~\ref{IndAss1} holds, 
then the Jones-Wenzl algebra $\WJ_\multii(\nu)$ is generated by the collection $\mathsf{G}_\multii$~\eqref{GenSet2} 
of $\multii$-Jones-Wenzl tangles. 
Equivalently, each tangle of the form
\begin{align} \label{PreMGenerators} 
\vcenter{\hbox{\includegraphics[scale=0.275]{e-GenericTangle_WJ.pdf} ,}} 
\end{align} 
where $T \in \TL_{\Summed_\multii}(\nu)$, is a polynomial in the elements of the collection $\mathsf{G}_\multii$. 
\end{claim}

Assuming that induction hypothesis~\ref{IndAss1} holds,
in the next lemmas and corollaries~\ref{SameSideLem1}--\ref{DifferentSideCor}, 
we apply induction on $\np_\multii$ to construct certain simple basis tangles in $\PD1_\multii$~\eqref{InsertTwoBoxes} from the claimed generator
set $\mathsf{G}_\multii$ of $\WJ_\multii(\nu)$.
After this, we prove claim~\ref{IndClaim} in corollary~\ref{FinalCor}, by showing that every tangle in $\PD1_\multii$
is indeed a polynomial in the claimed generator set $\mathsf{G}_\multii$.  
Then, it only remains to conclude in section~\ref{TheProofSec} with the proof of theorem~\ref{GeneratorThm}.

\begin{center}
\bf Constructing simple basis tangles
\end{center}

We begin by constructing the basis tangles of type $\PD1_\multii$~\eqref{InsertTwoBoxes} with $\Defect_{\alpha_1} = \Defect_{\alpha_2}$, 
and $v = 0$, and $r = w  = 1$. 
First, in lemma~\ref{SameSideLem1} (resp.~lemma~\ref{SameSideLem2}) we construct such tangles with maximal (resp.~minimal)
number of crossing links. Later, in lemma~\ref{SameSideLem} and corollary~\ref{SameSideCor}, we construct such tangles with any number of crossing links. 
To prove the latter results, we use certain simple tangles in the basis $\PD3_\multii$~\eqref{InsertOneBox}, treated in lemma~\ref{2stepLem} and corollary~\ref{2stepCor}.

\begin{lem} \label{SameSideLem1} 
Suppose $\Summed_\multii < \ppmin(q)$.  If induction hypothesis~\ref{IndAss1} holds, then
the following tangle is a polynomial in the elements of $\mathsf{G}_\multii$\textnormal{:}
\begin{align} \label{SameSideTangle1} 
\vcenter{\hbox{\includegraphics[scale=0.275]{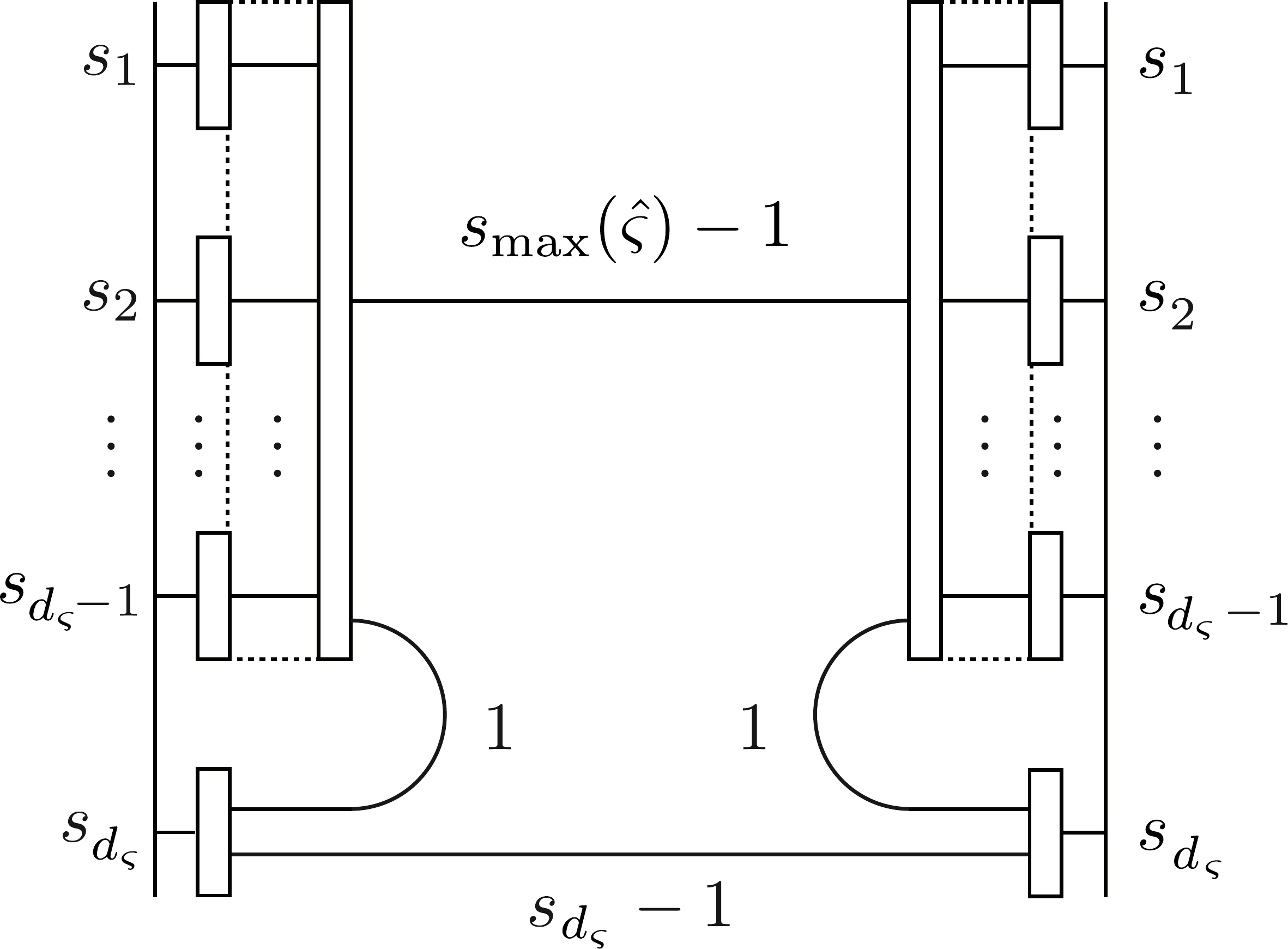} .}} 
\end{align} 
\end{lem} 

\begin{proof} 
We generate the sought tangle~\eqref{SameSideTangle1} from the following product:
\begin{align}  
\hspace*{-7mm}
\vcenter{\hbox{\includegraphics[scale=0.275]{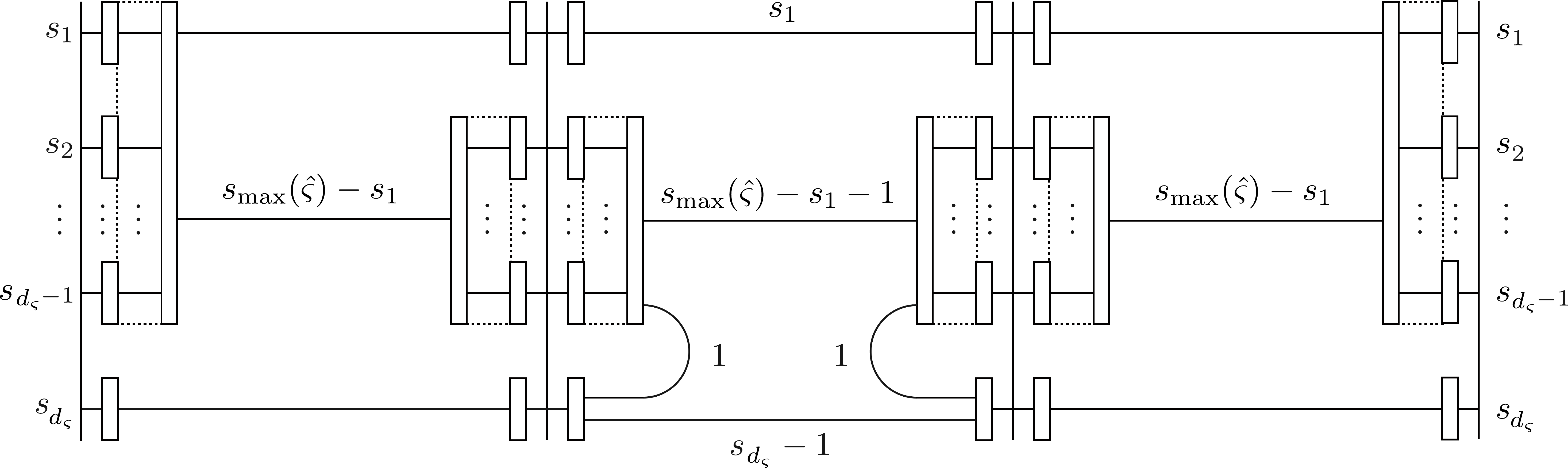} .}}
\end{align}  
By induction hypothesis~\ref{IndAss1}, the left (resp.~middle, resp.~right) tangle of this product is a polynomial in 
the elements of $\smash{\mathsf{G}_\lds}$ (resp.~$\smash{\mathsf{G}_\fds}$, resp.~$\smash{\mathsf{G}_\lds}$).  
With $\smash{\mathsf{G}_{\lds} \cup \mathsf{G}_\mathsf{\fds}} = \mathsf{G}_\multii$, the claim follows.
\end{proof}

\begin{lem} \label{SameSideLem2} 
Suppose $\Summed_\multii < \ppmin(q)$.  If induction hypothesis~\ref{IndAss1} holds, then for all 
Jones-Wenzl link states $\alpha_1, \alpha_2 \in \PS_\lds$, 
the following tangle is a polynomial in the elements of $\mathsf{G}_\multii$\textnormal{:}
\begin{align} \label{SameSideTangle2} 
\vcenter{\hbox{\includegraphics[scale=0.275]{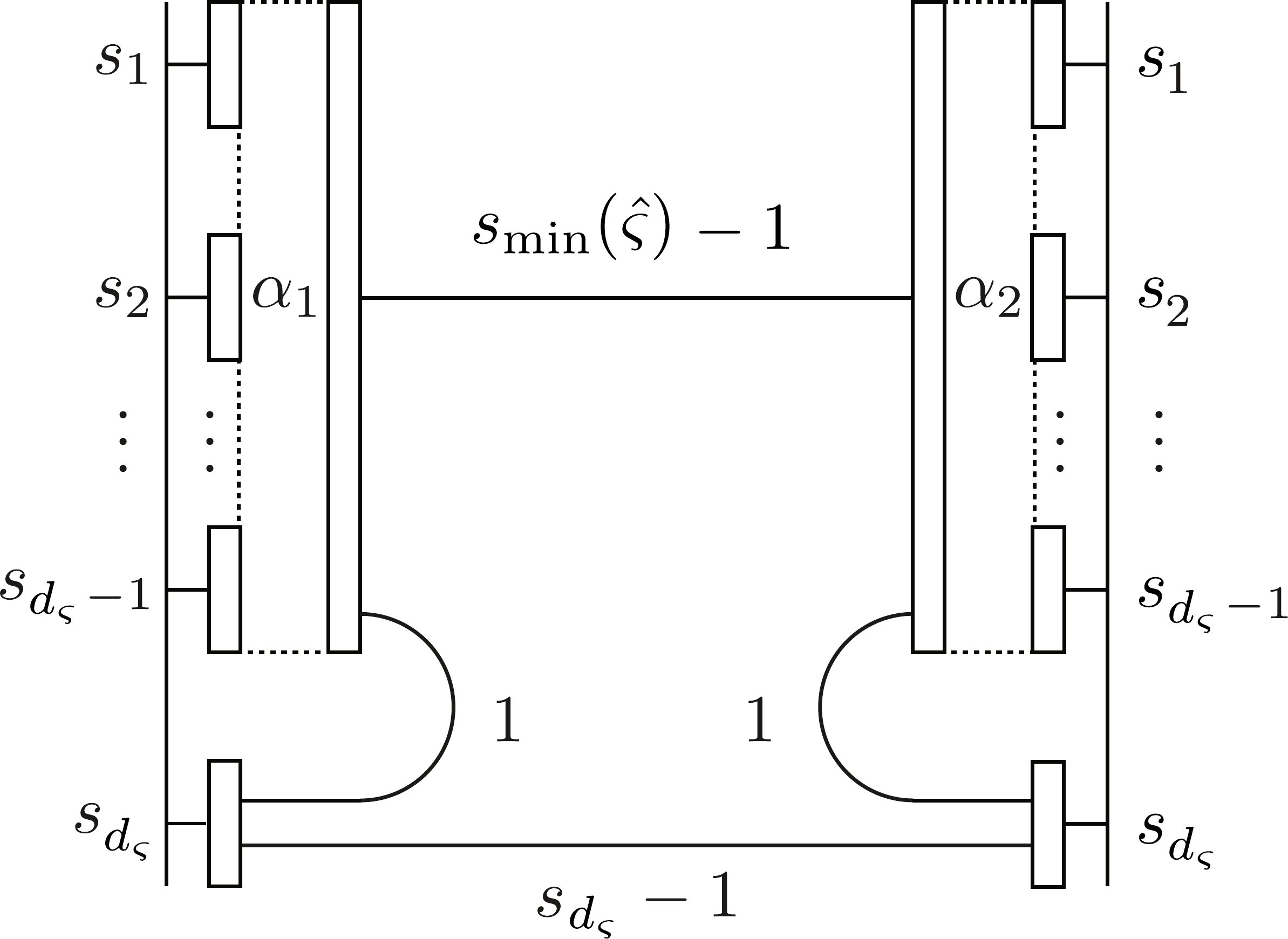} .}} 
\end{align} 
\end{lem} 

\begin{proof}  
We note that our assumptions give
\begin{align} \label{AssumptionsGivePbig}
\np_\multii \geq 3 \quad \text{and} \quad \Summed_\multii < \ppmin(q) \qquad \Longrightarrow \qquad
\sIndex_1, \smin(\lds) < \ppmin(q) - 2 \quad \text{and} \quad \ppmin(q) > 3 . 
\end{align}
By linearity, we may also assume that $\alpha_1$ and $\alpha_2$ are link patterns.
The idea of the proof is to generate the sought tangle~\eqref{SameSideTangle2} by forming the product~\eqref{TProduct} of lemma~\ref{TProductLem}
with judiciously chosen link states $\beta_1,\beta_2, \gamma_1, \gamma_2$.

Writing $\alpha_1$ and $\alpha_2$ in the form of~\eqref{WJsubForm}, 
$\alpha_1'$ and $\alpha_2'$ being the corresponding sub-link patterns,
item~\ref{PaRittt3} of lemma~\ref{ParityLem} says that 
$\Defect_{\alpha_1'} \in \DefectSet\sub{\alpha_1,\sIndex_1}$ and $\Defect_{\alpha_2'} \in \DefectSet\sub{\alpha_2,\sIndex_1}$. 
According to lemma~\ref{DefectLem} in appendix~\ref{LemmaApp}, because $\Defect_{\alpha_1} = \Defect_{\alpha_2} = \smin(\lds)$, 
all defects of $\alpha_1$ and $\alpha_2$ attach to a single projector box. 
There are three cases to consider:
\begin{enumerate}[leftmargin=*]
\itemcolor{red}

\item \label{SameSideIt2a} 
\emph{All defects of $\alpha_1$ and $\alpha_2$ attach to the first box}: 
In this case, we have $\Defect_{\alpha_1'} = \sIndex_1 - \smin(\lds) = \Defect_{\alpha_2'}$.
Hence, we may form the product~\eqref{TProduct} of lemma~\ref{TProductLem} with
\begin{align} \label{SubsLinkSt} 
\beta_1 = \alpha_1', \qquad \gamma_1 = (\alpha_1')^\cheque, \qquad 
\beta_2 = (\alpha_2')^\cheque, \qquad \gamma_2 = \alpha_2' ,
\qquad \text{and} \qquad 
\Defect_1 = \Defect_{\alpha_1'}, \qquad 
\Defect_2 = \Defect_{\alpha_2'} 
\end{align} 
(where the dual elements~\eqref{DualLS} exist by theorem~\ref{BigSSTHM}, 
because we assume $\smax(\smash{\flds}) < \Summed_\multii < \ppmin(q)$).
Also, we have
\begin{align} 
\begin{cases} 
\Defect_1 \overset{\eqref{SubsLinkSt}}{=} \Defect_{\alpha_1'} = \sIndex_1 - \smin(\lds) = \Defect_{\alpha_2'} 
\overset{\eqref{SubsLinkSt}}{=} \Defect_2, \\ 
\text{$\sIndex_1 - \smin(\lds) \geq 1$ by item~\ref{DefIt2} of lemma~\ref{DefectLem}}, 
\end{cases} 
\quad &  \Longrightarrow \quad \Defect_1 = \Defect_2, \quad \Defect_1 + \Defect_2 = 2\Defect_1 \geq 2 \\
& \Longrightarrow \quad 2 \in \DefectSet\sub{\Defect_1,\Defect_2} \overset{\eqref{SpecialDefSet}}{=} \{0,2,\ldots,2\Defect_1\} ,
\end{align}
and $2 \in \DefectSet\sub{t,t} = \{0,2,\ldots,2 t\}$ because $t := \sIndex_{\np_\multii} \geq 1$.  
Altogether, we may therefore set
\begin{align} \label{SubsLinkSt2} 
i = 2 \in \DefectSet\sub{\Defect_1,\Defect_2} \cap \DefectSet\sub{t,t} ,
\end{align} 
in addition to~\eqref{SubsLinkSt}, in the product tangle~\eqref{TProduct}.
Then we arrive with
\begin{align} \label{PreUpTangle} 
\hspace*{-5mm}
\vcenter{\hbox{\includegraphics[scale=0.275]{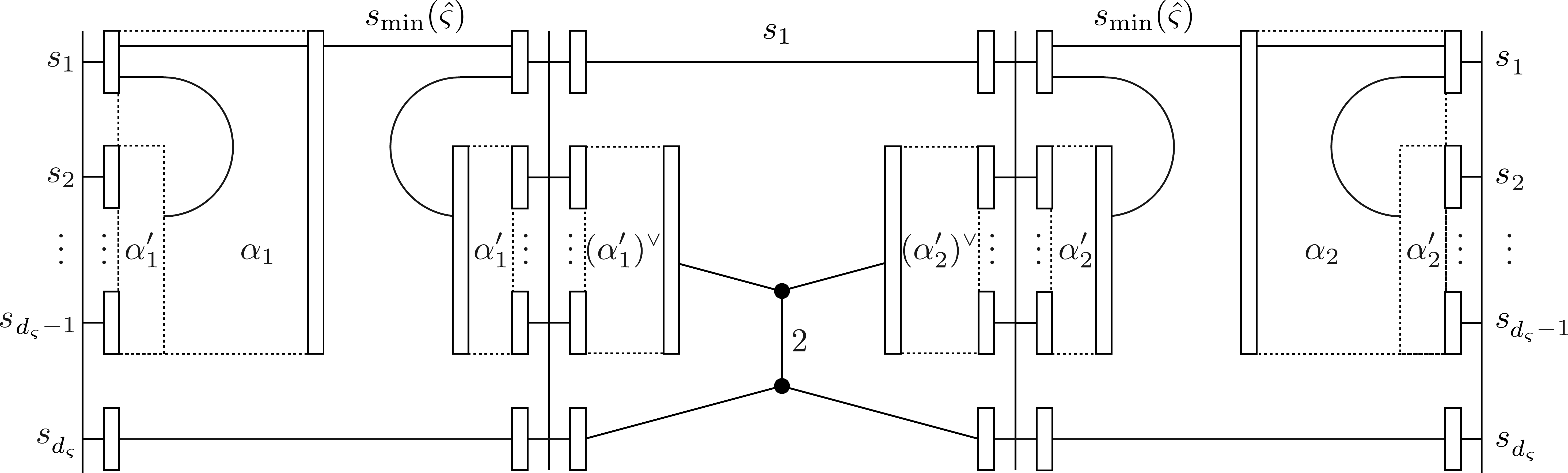} ,}}
\end{align} 
which by 
item~\ref{ExtractLemItem} of lemma~\ref{CollectionLem}
and the identity 
$\big( (\alpha_1')^\cheque \, \big| \, \alpha_1' \big) = 1 = \big( (\alpha_2')^\cheque  \big| \, \alpha_2' \big)$
immediately simplifies to
\begin{align} \label{UpTangle} 
\vcenter{\hbox{\includegraphics[scale=0.275]{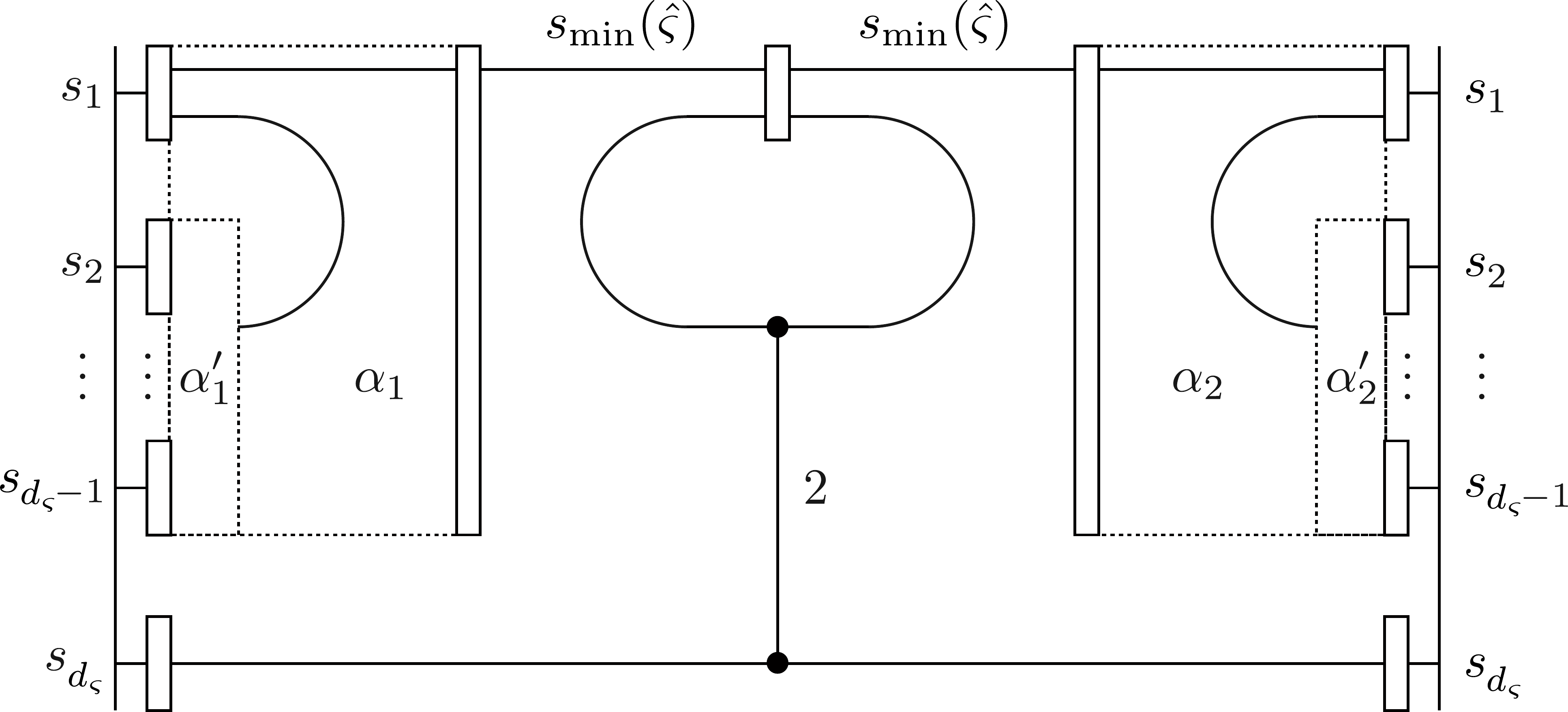} .}} 
\end{align} 
Now with $\Defect = \Defect_{\alpha_1} = \Defect_{\alpha_2} = \smin(\lds)$ and $t := \sIndex_{\np_\multii}$, we expand this tangle over the basis $\PD3_\multii$~\eqref{InsertOneBox}:
\begin{align} \label{UpExpand} 
\eqref{UpTangle} \quad = \quad \sum_{j \, \in \, \DefectSet\sub{\Defect,\Defect} 
\cap \, \DefectSet\sub{t,t}}
c_j \quad
\vcenter{\hbox{\includegraphics[scale=0.275]{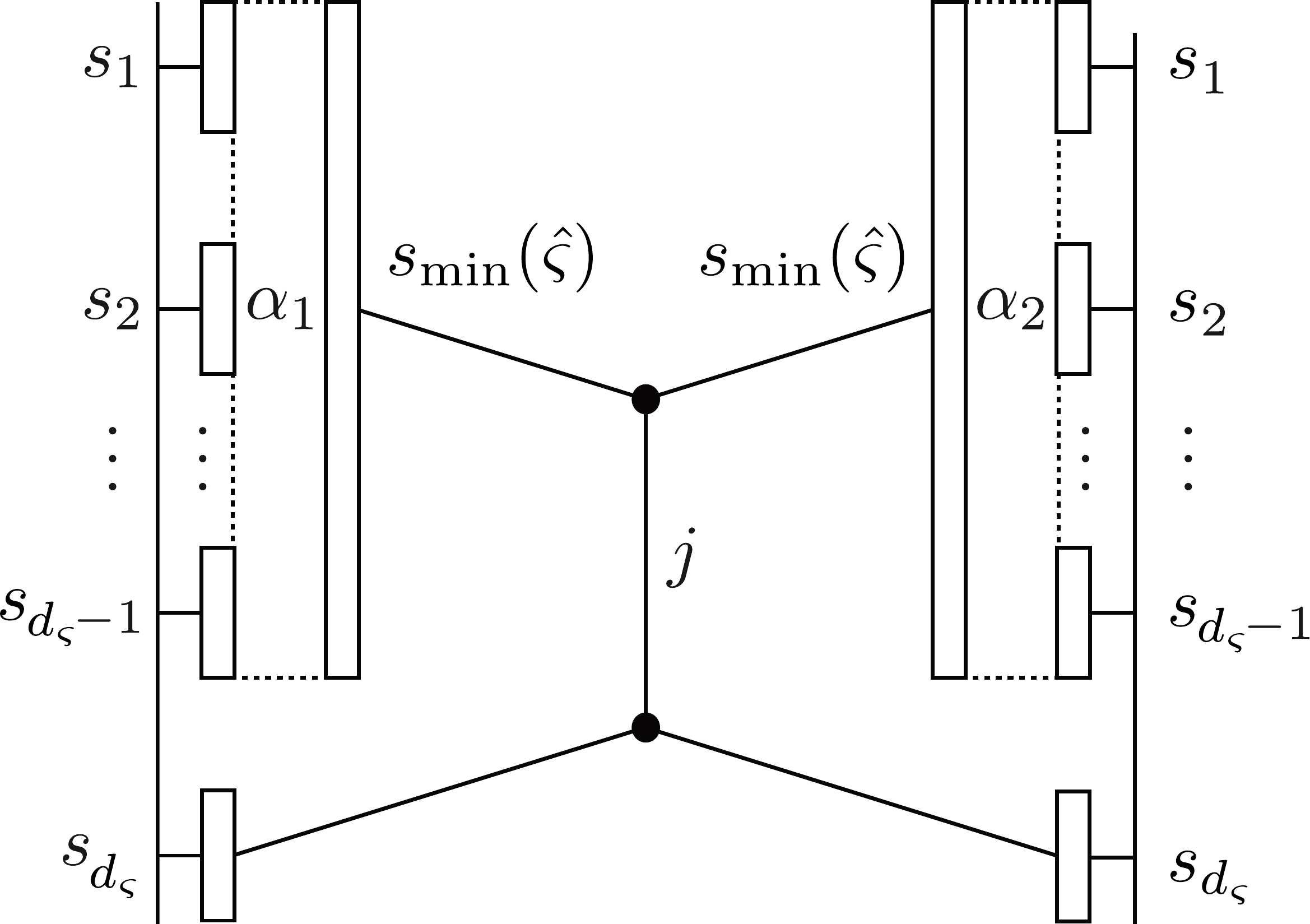} .}} 
\end{align}
To find the coefficients $c_j \in \bC$, we proceed as in the proof of lemma~\ref{TProductLem}.
We insert both sides of~\eqref{UpExpand} into 
the ``dual" tangle~\eqref{DualDiagram}, thereby closing all links into loops.  After simplifying the result,~\eqref{UpExpand} becomes
\begin{align} \label{InsertionResult2}
\vcenter{\hbox{\includegraphics[scale=0.275]{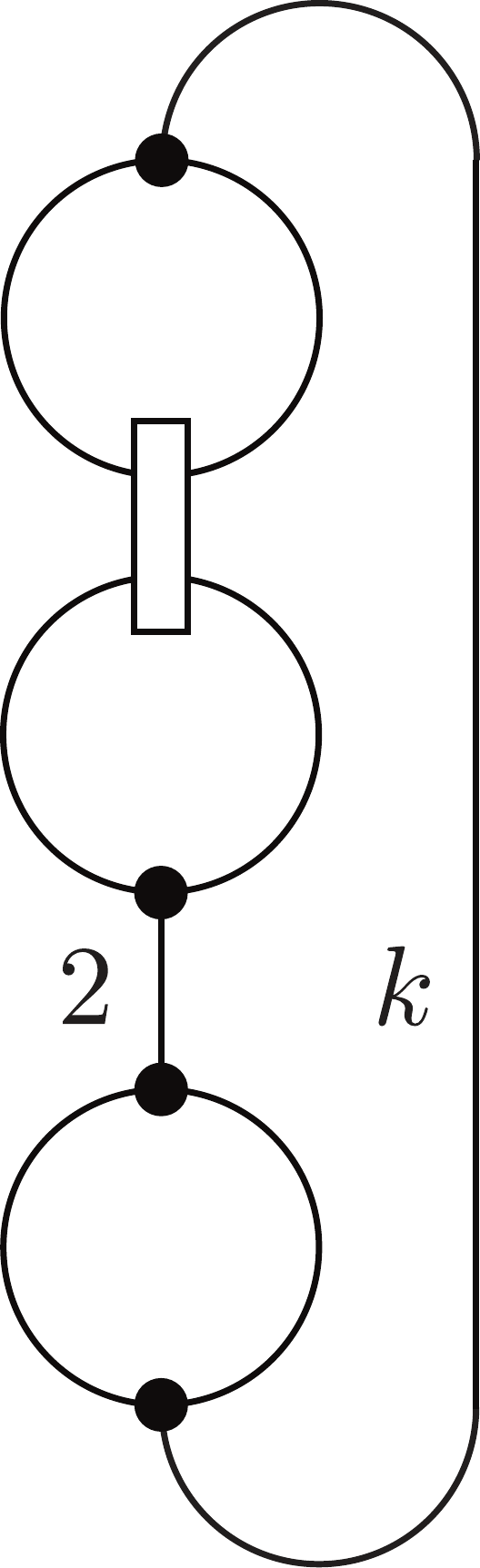}}} 
\quad = \quad \sum_{j \, \in \, \DefectSet\sub{\Defect,\Defect} 
\cap \, \DefectSet\sub{t,t}} 
c_j \,\, \times \,\,
\vcenter{\hbox{\includegraphics[scale=0.275]{e-Generators20.pdf}}} 
\end{align} 
(where we do not label the sizes of all cables; one may infer those sizes from~\eqref{UpExpand}).  
Then, using 
items~\ref{ExtractLemItem} and~\ref{LoopErasureLemItem} of lemma~\ref{CollectionLem}
we delete the lower loop of either network to obtain
\begin{align} \label{resultTanlgess}
\delta_{2,k} 
\frac{\ThetaNet(2,t,t)}{[3]}  \,\, \times \,\,
\vcenter{\hbox{\includegraphics[scale=0.275]{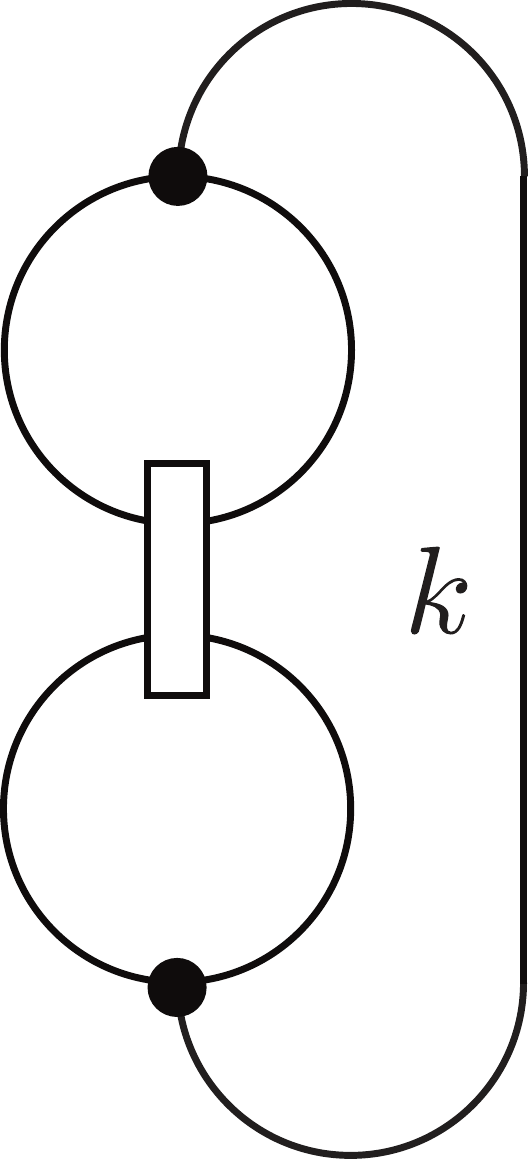}}} \quad = \quad
\sum_{j \, \in \, \DefectSet\sub{\Defect,\Defect} \cap \,\DefectSet\sub{t,t}} 
c_j \,
\delta_{jk} \frac{\ThetaNet(j,t,t)}{(-1)^j[j+1]} 
 \,\, \times \,\, \vcenter{\hbox{\includegraphics[scale=0.275]{e-Generators22.pdf} .}} 
\end{align} 
The network on the left side of~\eqref{resultTanlgess} evaluates to~\eqref{LoopBox} by lemma~\ref{NetEvalLem}, and the network 
on the right side of~\eqref{resultTanlgess} is the Theta network~\eqref{ThetaDefinition}
with $(r,s,t) = (k,\smin(\lds), \smin(\lds))$.
Thus, using lemma~\ref{ThetaLem} we arrive with
\begin{align} \label{coefck} 
c_k = \delta_{2,k} \frac{[\sIndex_1+1]^2}{[3] [\sIndex_1] \, \ThetaNet(2, \smin(\lds), \smin(\lds))}
\overset{\eqref{ThetaFormula1}}{=}
\delta_{2,k} \frac{[\sIndex_1+1]^2 [\smin(\lds)] [ \smin(\lds) - 2]! }{[3] [\sIndex_1] [2] [ \smin(\lds) + 1] [ \smin(\lds) + 2]} .
\end{align} 
By~\eqref{AssumptionsGivePbig}, these coefficients are finite, and only $c_2$ is nonzero.  
After inserting~\eqref{coefck} into~\eqref{UpExpand}, 
using (\ref{ExampleProj2},~\ref{3vertex1}) to decompose the three-vertices on the right side of~\eqref{UpExpand}, and rearranging,~\eqref{UpExpand} 
becomes
\begin{align} 
\label{InduD}
\hspace*{-5mm}
& \vcenter{\hbox{\includegraphics[scale=0.275]{e-Generators134min.pdf}}} 
\quad = \quad 
\nu  \,\, \times \,\, \vcenter{\hbox{\includegraphics[scale=0.275]{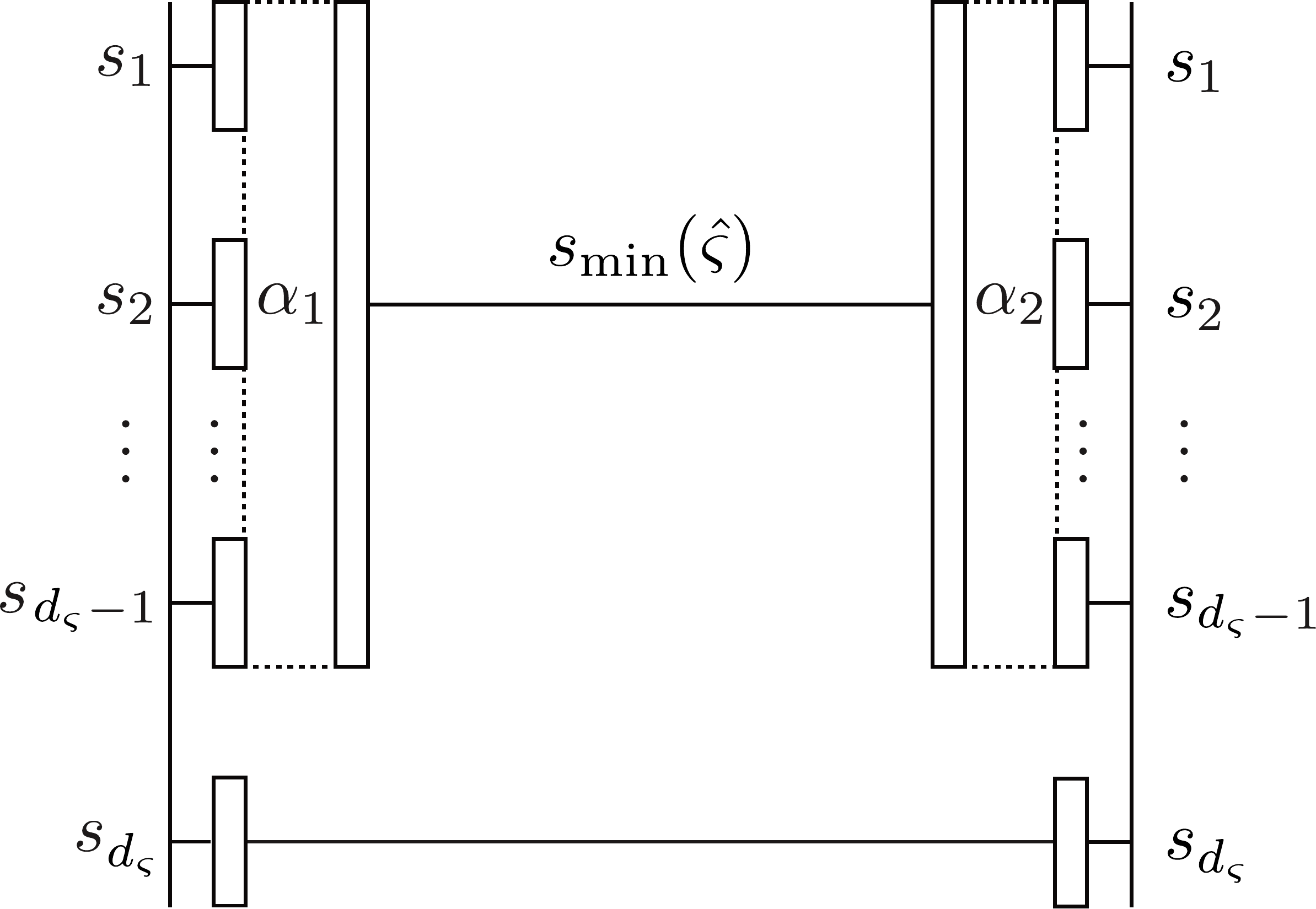}}} \\ 
\nonumber
- \frac{\nu}{c_2}  \,\, \times \,\,  & \vcenter{\hbox{\includegraphics[scale=0.275]{e-Generators135.pdf} .}} 
\end{align}
Now, by induction hypothesis~\ref{IndAss1}, the first tangle on the right side of~\eqref{InduD} is a polynomial in the elements of the collection 
$\smash{\mathsf{G}_\lds}$, and the left (resp.~middle, resp.~right) tangle in the product on the right side is a polynomial in the elements 
of $\smash{\mathsf{G}_\lds}$ (resp.~$\smash{\mathsf{G}_\fds}$, resp.~$\smash{\mathsf{G}_\lds}$).  
With $\smash{\mathsf{G}_{\lds} \cup \mathsf{G}_\mathsf{\fds}} = \mathsf{G}_\multii$, this proves the lemma for this special case.

\item \label{SameSideIt2b} 
\emph{All defects of $\alpha_1$ and $\alpha_2$ attach to the $j$:th and $k$:th box respectively, with $j,k \in \{2,3,\ldots,\np_\multii-1\}$}:  
As in item~\ref{SameSideIt2a}, we form the product~\eqref{TProduct} of lemma~\ref{TProductLem} with substitutions~(\ref{SubsLinkSt},~\ref{SubsLinkSt2}).
To justify this choice, we note that by our assumption that all defects of $\alpha_1$ and $\alpha_2$ attach to the $j$:th and $k$:th box 
we have $\Defect_{\alpha_1'} = \sIndex_1 + \smin(\lds) = \Defect_{\alpha_2'}$, so
\begin{align} 
\begin{cases} 
\Defect_1 \overset{\eqref{SubsLinkSt}}{=} \Defect_{\alpha_1'} = \sIndex_1 + \smin(\lds) 
= \Defect_{\alpha_2'} \overset{\eqref{SubsLinkSt}}{=} \Defect_2, \\ 
\sIndex_1 + \smin(\lds) \geq 2 , 
\end{cases} 
\quad &  \Longrightarrow \quad 
\Defect_1 = \Defect_2, \quad \Defect_1 + \Defect_2 = 2\Defect_1 \geq 2 \\
& \Longrightarrow \quad 2 \in \DefectSet\sub{\Defect_1,\Defect_2} \overset{\eqref{SpecialDefSet}}{=} \{0,2,\ldots,2\Defect_1\} ,
\end{align}
and $2 \in \DefectSet\sub{t,t} = \{0,2,\ldots,2 t\}$ because $t := \sIndex_{\np_\multii} \geq 1$. 
With substitutions~(\ref{SubsLinkSt},~\ref{SubsLinkSt2}) into~\eqref{TProduct}, we arrive with
\begin{align} \label{PreDownTangle}
\hspace*{-5mm}
\vcenter{\hbox{\includegraphics[scale=0.275]{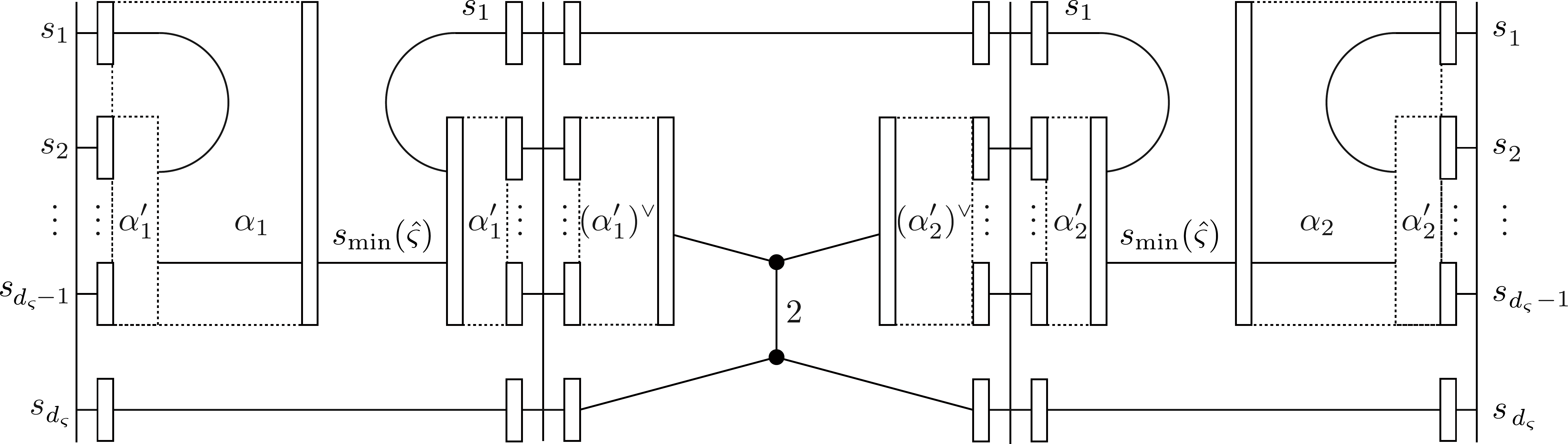} ,}} 
\end{align}  
which by 
item~\ref{ExtractLemItem} of lemma~\ref{CollectionLem}
with $\big( (\alpha_1')^\cheque \, \big| \, \alpha_1' \big) = 1 = \big( (\alpha_2')^\cheque  \big| \, \alpha_2' \big)$
immediately simplifies to
\begin{align} \label{DownTangle} 
\vcenter{\hbox{\includegraphics[scale=0.275]{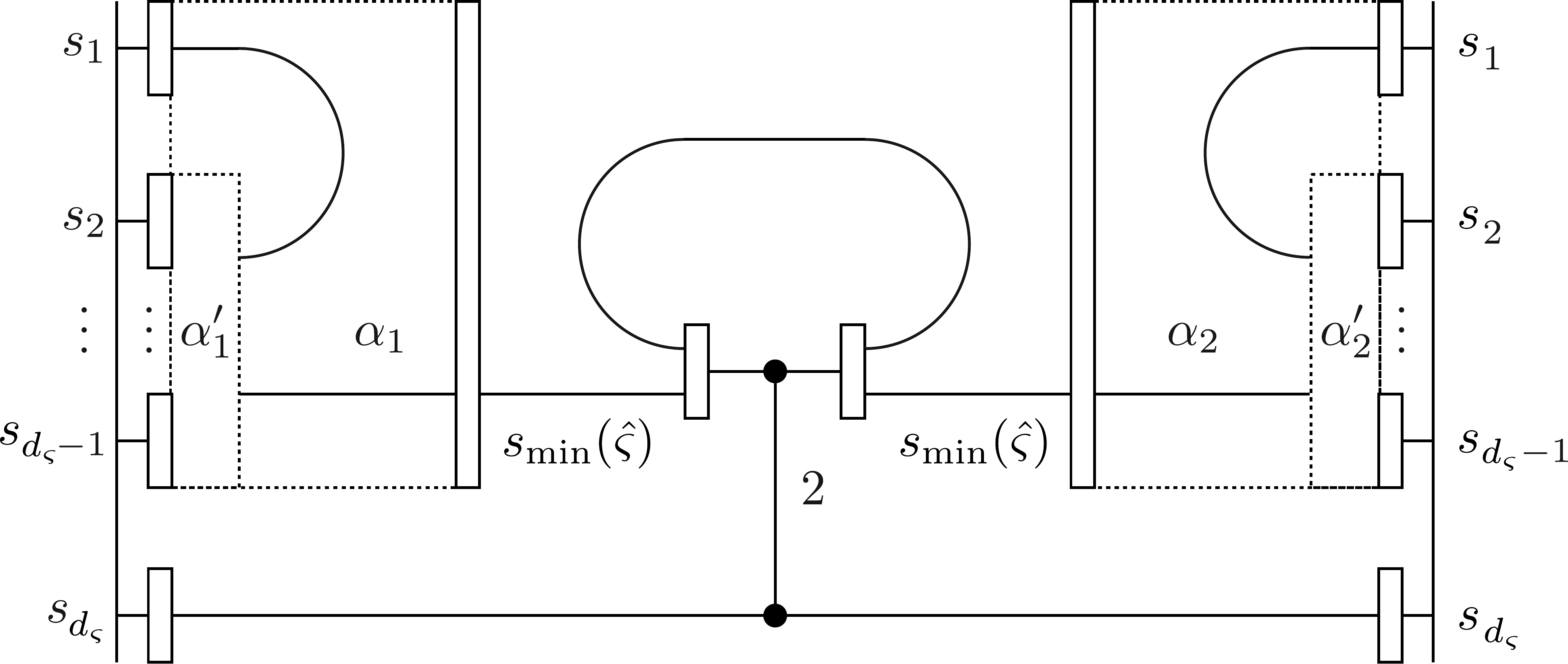} .}} 
\end{align} 
With $\Defect = \Defect_{\alpha_1} = \Defect_{\alpha_2} = \smin(\lds)$ and $t := \sIndex_{\np_\multii}$, 
we expand this tangle over the basis $\PD3_\multii$~\eqref{InsertOneBox}:
\begin{align} \label{DownExpand} 
\eqref{DownTangle} \quad = \quad 
\sum_{j \, \in \, \DefectSet\sub{\Defect,\Defect} \cap \, \DefectSet\sub{t,t}} 
c_j  \,\, \times \,\,
\vcenter{\hbox{\includegraphics[scale=0.275]{e-Generators147.pdf} ,}} 
\end{align}
and to find the coefficients $c_j \in \bC$, we proceed as in item~\ref{SameSideIt2a}. In the present situation, inserting
both sides of~\eqref{DownExpand} into the ``dual" tangle~\eqref{DualDiagram}, 
and simplifying the result gives
\begin{align} \label{InsertionResult3}
\vcenter{\hbox{\includegraphics[scale=0.275]{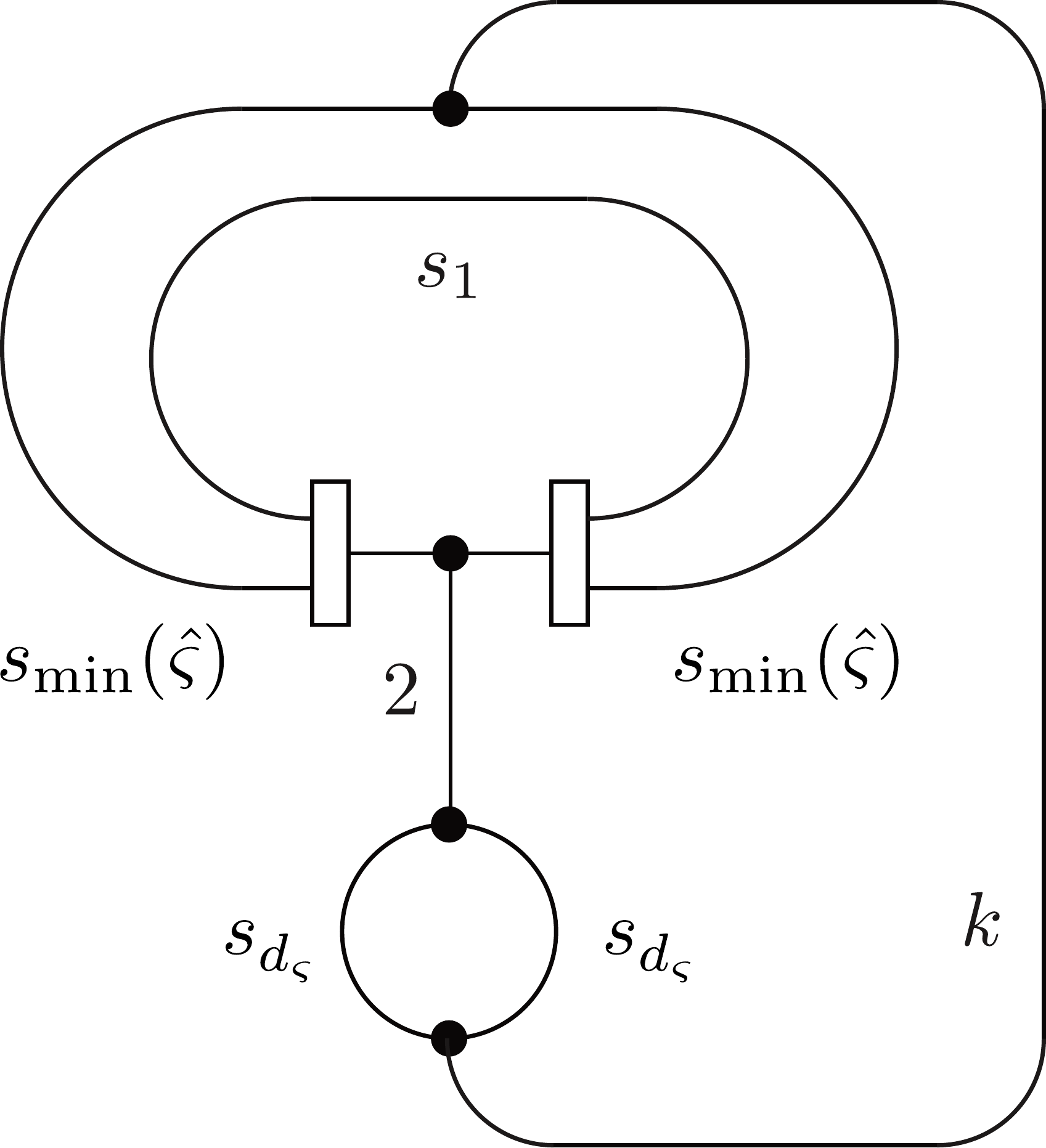}}} \quad = \quad 
\sum_{j \, \in \, \DefectSet\sub{\Defect,\Defect} \cap \, \DefectSet\sub{t,t}} 
c_j 
\,\, \times \,\, \vcenter{\hbox{\includegraphics[scale=0.275]{e-Generators20.pdf} .}} 
\end{align} 
Again, using 
items~\ref{ExtractLemItem} and~\ref{LoopErasureLemItem} of lemma~\ref{CollectionLem},
we delete the lower loop of either network to find
\begin{align} \label{BotSi2}
\delta_{2,k} 
\frac{\ThetaNet(2,t,t)}{[3]}  \,\, \times \,\,
\vcenter{\hbox{\includegraphics[scale=0.275]{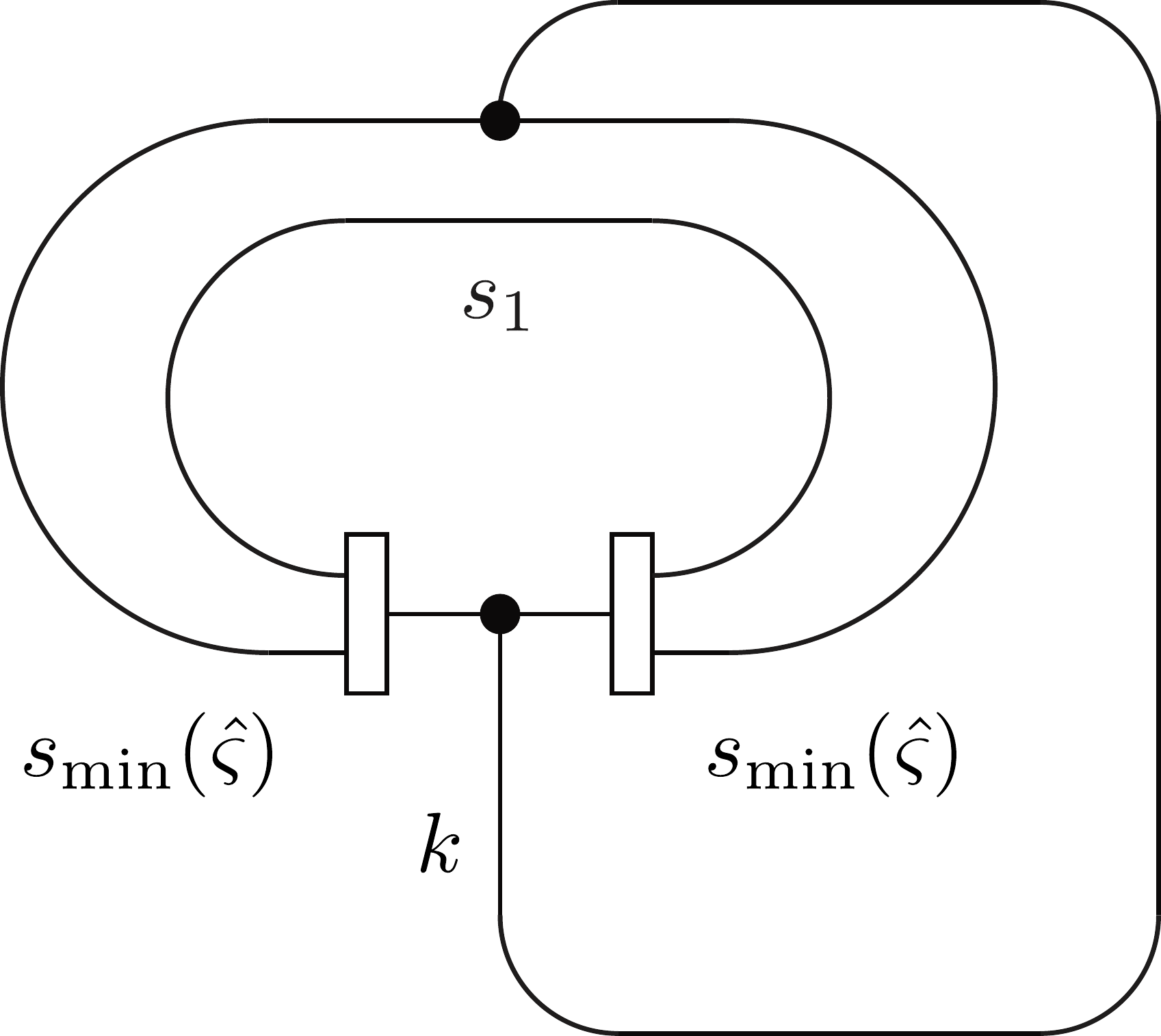}}} \quad = \quad 
\sum_{j \, \in \, \DefectSet\sub{\Defect,\Defect} \cap \, \DefectSet\sub{t,t}} 
c_j \, \delta_{jk} \frac{\ThetaNet(j,t,t)}{(-1)^j[j+1]} 
\,\, \times \,\, \vcenter{\hbox{\includegraphics[scale=0.275]{e-Generators22.pdf} .}}
\end{align} 
The networks on both sides are Theta networks defined in~\eqref{ThetaDefinition}, which respectively evaluate to
$\smash{\ThetaNet(2,\smin(\lds) + \sIndex_1, \smin(\lds) + \sIndex_1)}$ and  $\smash{\ThetaNet(k,\smin(\lds), \smin(\lds))}$.
Thus, we arrive with
\begin{align} \label{coefck2} 
& c_k = \delta_{2,k}\frac{\ThetaNet(2, \smin(\lds) + \sIndex_1, \smin(\lds) + \sIndex_1)}{\ThetaNet(2, \smin(\lds) , \smin(\lds) )} \\
& \; \overset{\eqref{ThetaFormula1}}{=}
\delta_{2,k} (-1)^{\sIndex_1} \frac{ [\smin(\lds) + \sIndex_1 + 1] [\smin(\lds) + \sIndex_1 + 2] [\smin(\lds)] [ \smin(\lds) - 2]! }{ [ \smin(\lds) + 1] [ \smin(\lds) + 2] [ \smin(\lds) + \sIndex_1] [ \smin(\lds) + \sIndex_1 - 2]!} . 
\end{align} 
Using~\eqref{AssumptionsGivePbig} and the fact from~\eqref{PreDownTangle} that 
\begin{align}
\smin(\lds) + \sIndex_1 \leq \smax(\flds) = \Summed_\multii - (\sIndex_{\np_\multii} + \sIndex_1) < \ppmin(q) - 2 ,
\end{align}
we see that these coefficients are finite, and only $c_2$ is nonzero.  
Now, proceeding with (\ref{DownExpand},~\ref{coefck2}) as we do with (\ref{UpExpand},~\ref{coefck}) just beneath~\eqref{coefck} of item~\ref{SameSideIt2a}, 
we finish the proof of the lemma for this special case.

\item \label{SameSideIt2c} 
\emph{The defects of $\alpha_1$ attach to the first box, and the defects of $\alpha_2$ attach to the $i$:th box for some $i \in \{2,3,\ldots,\np_\multii-1\}$ 
(or vice versa)}:  
In this case, we form the product~\eqref{TProduct} with substitutions $i = 2$ and
\begin{align} \label{SubsLinkSt3} 
\beta_1 = \beta_2 = \alpha_1', \qquad \gamma_1 = \gamma_2 = (\alpha_1')^\cheque 
\qquad \Longrightarrow \qquad \Defect_1 = \Defect_{\alpha_1'} = \Defect_2 . 
\end{align} 
These substitutions into~\eqref{TProduct} give
\begin{align} \label{PreUpDownTangle} 
\hspace*{-5mm}
\vcenter{\hbox{\includegraphics[scale=0.275]{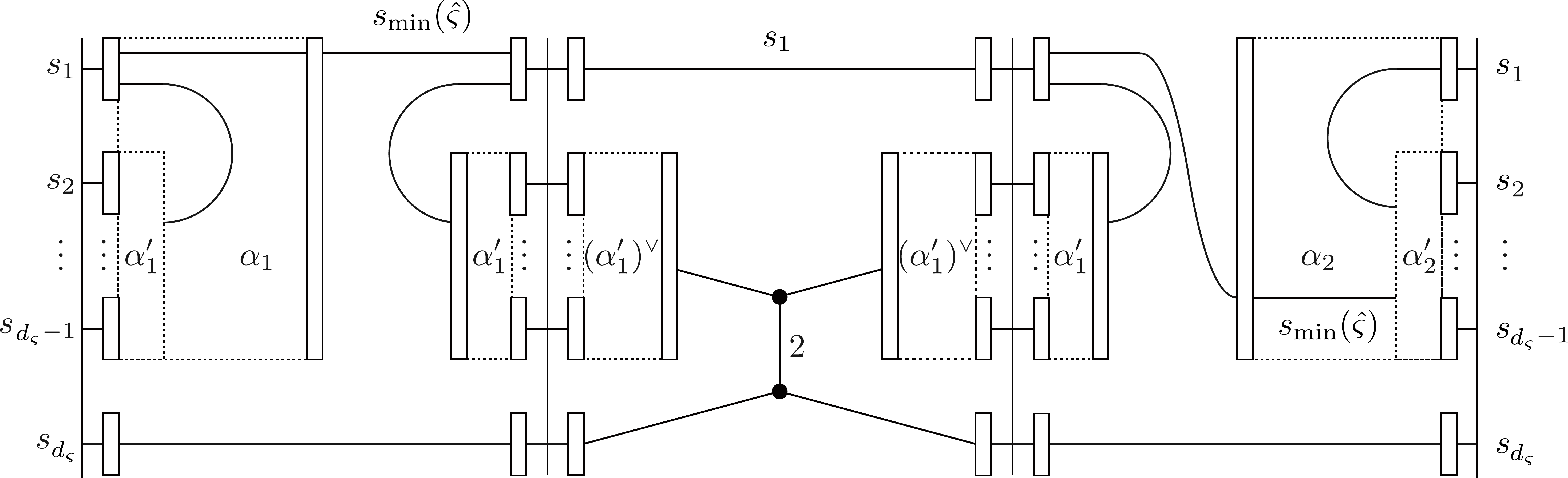} .}}
\end{align} 
Repeating the analysis of item~\ref{SameSideIt2a} above proves the lemma for this special case.
\end{enumerate}
This concludes the proof. 
\end{proof}

Now we construct certain simple basis tangles in $\PD3_\multii$~\eqref{InsertOneBox}
from the claimed generator set $\mathsf{G}_\multii$ of $\WJ_\multii(\nu)$.
We make use of them below in the proof of lemma~\ref{SameSideLem} and corollary~\ref{SameSideCor}.

\begin{lem} \label{2stepLem} 
Suppose $\Summed_\multii < \ppmin(q)$.  If induction hypothesis~\ref{IndAss1} holds, then for all 
Jones-Wenzl link states $\alpha_1, \alpha_2 \in \PS_\lds$
such that $\Defect_{\alpha_2} = \Defect_{\alpha_1} + 2$, the following tangle is a polynomial in the elements of $\mathsf{G}_\multii$\textnormal{:}
\begin{align} \label{2stepTangle} 
\vcenter{\hbox{\includegraphics[scale=0.275]{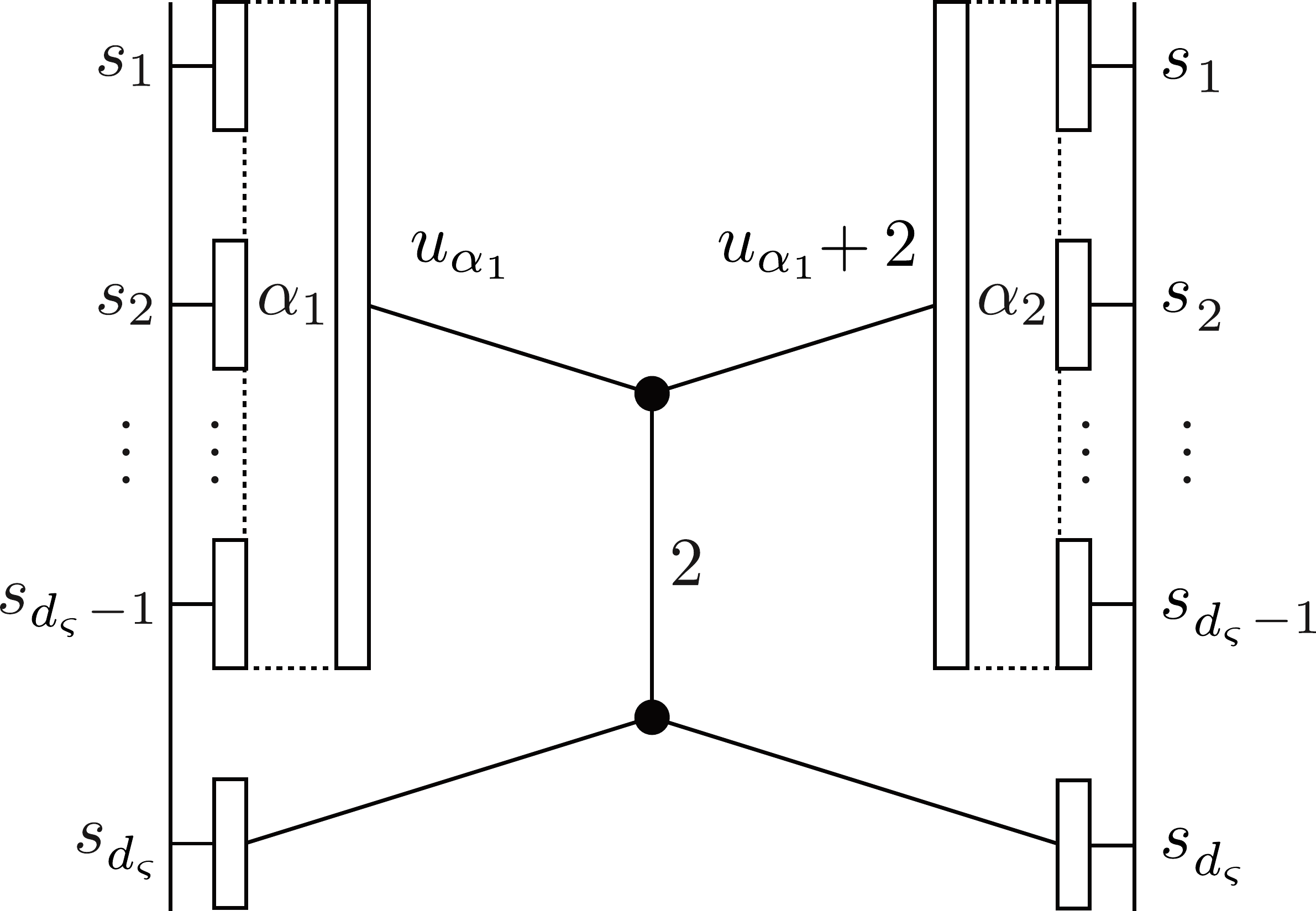}}} \quad 
= \quad \vcenter{\hbox{\includegraphics[scale=0.275]{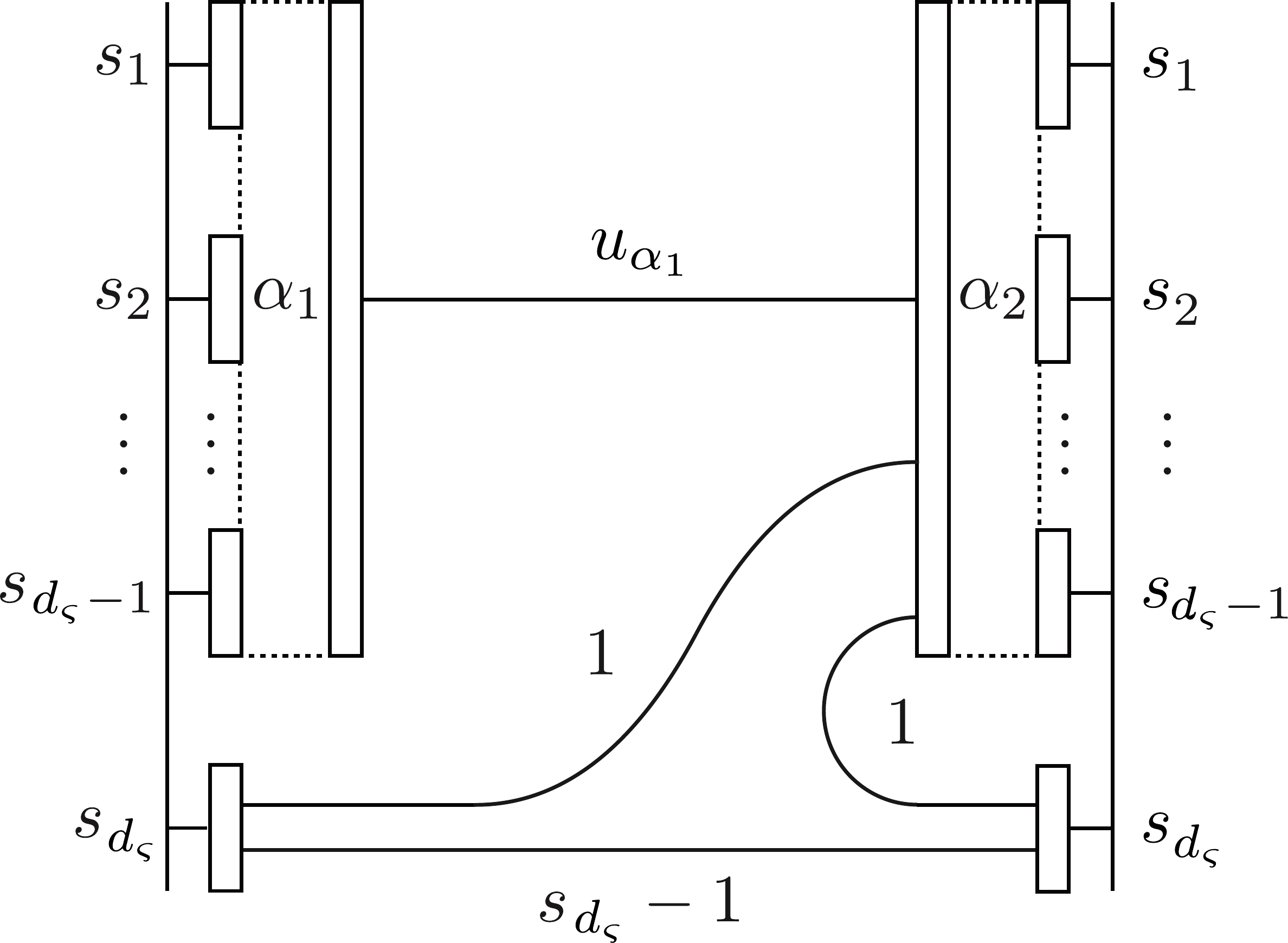} .}} 
\end{align} 
\end{lem} 

\begin{proof}  
It is apparent from~\eqref{3vertex1} that equality in~\eqref{2stepTangle} holds for $\Defect_{\alpha_2} = \Defect_{\alpha_1} + 2$.  
To generate the left tangle of~\eqref{2stepTangle}, we form the product~\eqref{TProduct} of lemma~\ref{TProductLem} as follows: we choose   
\begin{align} \label{Subs} 
\gamma_1 = \beta_1\superscr{\cheque}, \qquad \gamma_2 = \beta_2\superscr{\cheque} 
\end{align} 
(where the dual elements~\eqref{DualLS} exist by theorem~\ref{BigSSTHM}, 
because we assume $\smax(\smash{\flds}) < \Summed_\multii < \ppmin(q)$),
and we set
\begin{align} \label{MysterySet} 
\Defect := \Defect_1 = \Defect_2 \in \DefectSet\sub{\alpha_1,\sIndex_1} \cap \DefectSet\sub{\alpha_2,\sIndex_1} \cap \DefectSet_\flds .
\end{align} 
(Lemma~\ref{MinMaxLem} implies that we can make this choice.)
Then, in~\eqref{TProduct2} we have (with $t := \sIndex_{\np_\multii}$)
\begin{align} \label{MysterySeti} 
i = 2 \in \DefectSet\sub{\Defect_1,\Defect_2}\cap\DefectSet\sub{t,t} = \{0,2,\ldots,2\min(\Defect,t)\} .
\end{align}
Now, by lemma~\ref{TProductLem}, the product~\eqref{TProduct} with these substitutions evaluates to 
\begin{align} \label{Resol} 
\TetraNet \left[
\begin{array}{ccc}
\Defect & \Defect_{\alpha_1} & 2 \\ 
\Defect_{\alpha_2} & \Defect & \sIndex_1 
\end{array} 
\right]\frac{\one{\{2 \in \DefectSet\sub{\alpha_1,\alpha_2} \cap \DefectSet\sub{\Defect,\Defect} \cap 
\DefectSet\sub{t,t}\}}}{\ThetaNet(2,\Defect_{\alpha_1},\Defect_{\alpha_2})}  \,\, \times \,\,
\vcenter{\hbox{\includegraphics[scale=0.275]{e-Generators132.pdf} .}} 
\end{align}  
By (\ref{EDefSet},~\ref{EalphSet2}) with $\Defect_{\alpha_2} = \Defect_{\alpha_1} + 2$, the indicator $\one$ in~\eqref{Resol} 
equals one.  Combining this observation with item~\ref{NetworkIt3} of lemma~\ref{NetEvalLem}
from appendix~\ref{TLRecouplingSect}, we find that the coefficient in~\eqref{Resol} equals
\begin{align} 
\label{coefficient_for_resol}
\TetraNet \left[
\begin{array}{ccc}
\Defect & \Defect_{\alpha_1} & 2 \\ 
\Defect_{\alpha_2} & \Defect & \sIndex_1 
\end{array} 
\right] 
\frac{1}{\ThetaNet(2,\Defect_{\alpha_1},\Defect_{\alpha_2})} 
\overset{\eqref{SpecialTetNetwork}}{=}  \; &
\frac{1}{[\Defect]} \left[\frac{\Defect - \Defect_{\alpha_1} + \sIndex_1}{2}\right]
\frac{\ThetaNet(\Defect,\sIndex_1,\Defect_{\alpha_1}+2)}{\ThetaNet(2,\Defect_{\alpha_1},\Defect_{\alpha_1}+2)} \\
\overset{\eqref{ThetaFormula1}}{=} \; &
\frac{1}{[\Defect]} 
\frac{\left[ \frac{\Defect + \sIndex_1 + \Defect_{\alpha_1}}{2} + 2 \right]! \left[ \frac{\Defect + \sIndex_1 - \Defect_{\alpha_1}}{2} \right]! \left[ \frac{\sIndex_1 + \Defect_{\alpha_1} - \Defect}{2} + 1 \right]! \left[ \frac{\Defect + \Defect_{\alpha_1} - \sIndex_1}{2} + 1 \right]!}{[\Defect]! [\sIndex_1]! [\Defect_{\alpha_1} + 3]!} .
\label{coefficient_for_resol2}
\end{align}
Using the facts that $\sIndex_1, \Defect < \ppmin(q)$ and
\begin{align}
\Defect_{\alpha_1} + 2 \leq \smax(\smash{\lds}) 
\overset{\eqref{hatsMax}}{=} \Summed_\multii - \sIndex_{\np_\multii} < \ppmin(q) - 1
\qquad \Longrightarrow \qquad 
\Defect_{\alpha_1} < \ppmin(q) - 3 ,
\end{align}
we see that the denominator of~\eqref{coefficient_for_resol2} is finite and nonzero, and using the fact from~\eqref{TProduct} that
\begin{align}
\Defect + \sIndex_1 \leq \smax(\smash{\lds}) 
\overset{\eqref{hatsMax}}{=} \Summed_\multii - \sIndex_{\np_\multii} < \ppmin(q) - 1 ,
\end{align}
we see that the numerator of~\eqref{coefficient_for_resol2} is finite and nonzero.
Therefore, the product in~\eqref{TProduct} with substitutions~\eqref{Subs}--\eqref{MysterySeti}
gives the left tangle $T$ of~\eqref{2stepTangle}, up to a nonzero constant.

Now to see that $T$ is a polynomial in the elements of $\mathsf{G}_\multii$, we observe that the left (resp.~middle, resp.~right) tangle 
of the product~\eqref{TProduct} giving $T$ is a polynomial in the elements of the collection $\smash{\mathsf{G}_\lds}$ (resp.~$\smash{\mathsf{G}_\fds}$, 
resp.~$\smash{\mathsf{G}_\lds}$), by induction hypothesis~\ref{IndAss1}.  
With $\smash{\mathsf{G}_{\lds} \cup \mathsf{G}_\mathsf{\fds}} = \mathsf{G}_\multii$, 
it then follows that $T$ is a polynomial in the elements of $\mathsf{G}_\multii$. 
\end{proof}

\begin{cor} \label{2stepCor} 
Suppose $\Summed_\multii < \ppmin(q)$.  If induction hypothesis~\ref{IndAss1} holds, then for all 
Jones-Wenzl link states $\alpha_1, \alpha_2 \in \PS_\lds$
such that $\Defect_{\alpha_1} = \Defect_{\alpha_2} + 2$, 
the following tangle is a polynomial in the elements of $\mathsf{G}_\multii$\textnormal{:}
\begin{align} \label{2stepTangle2} 
\vcenter{\hbox{\includegraphics[scale=0.275]{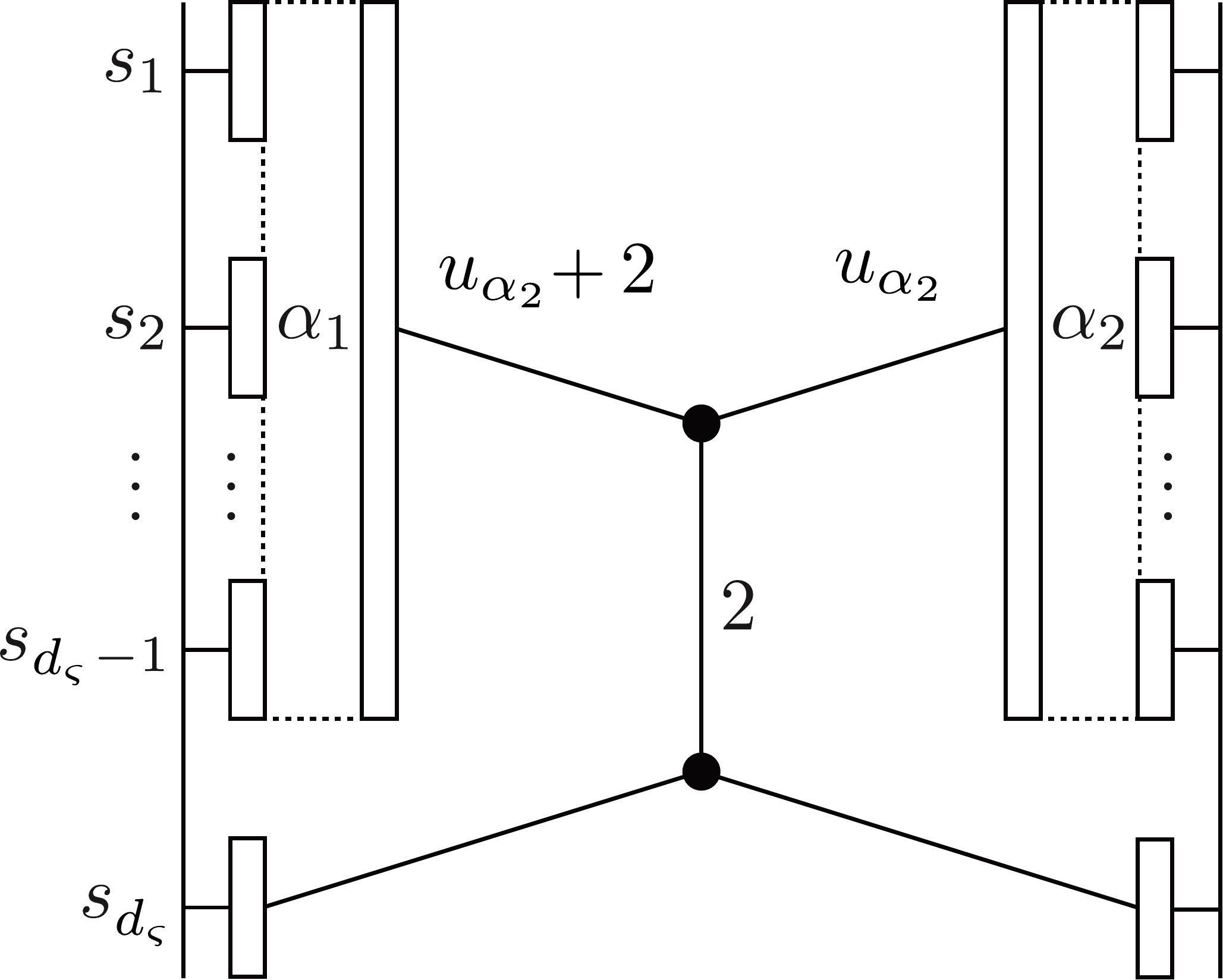}}} \quad
= \quad \vcenter{\hbox{\includegraphics[scale=0.275]{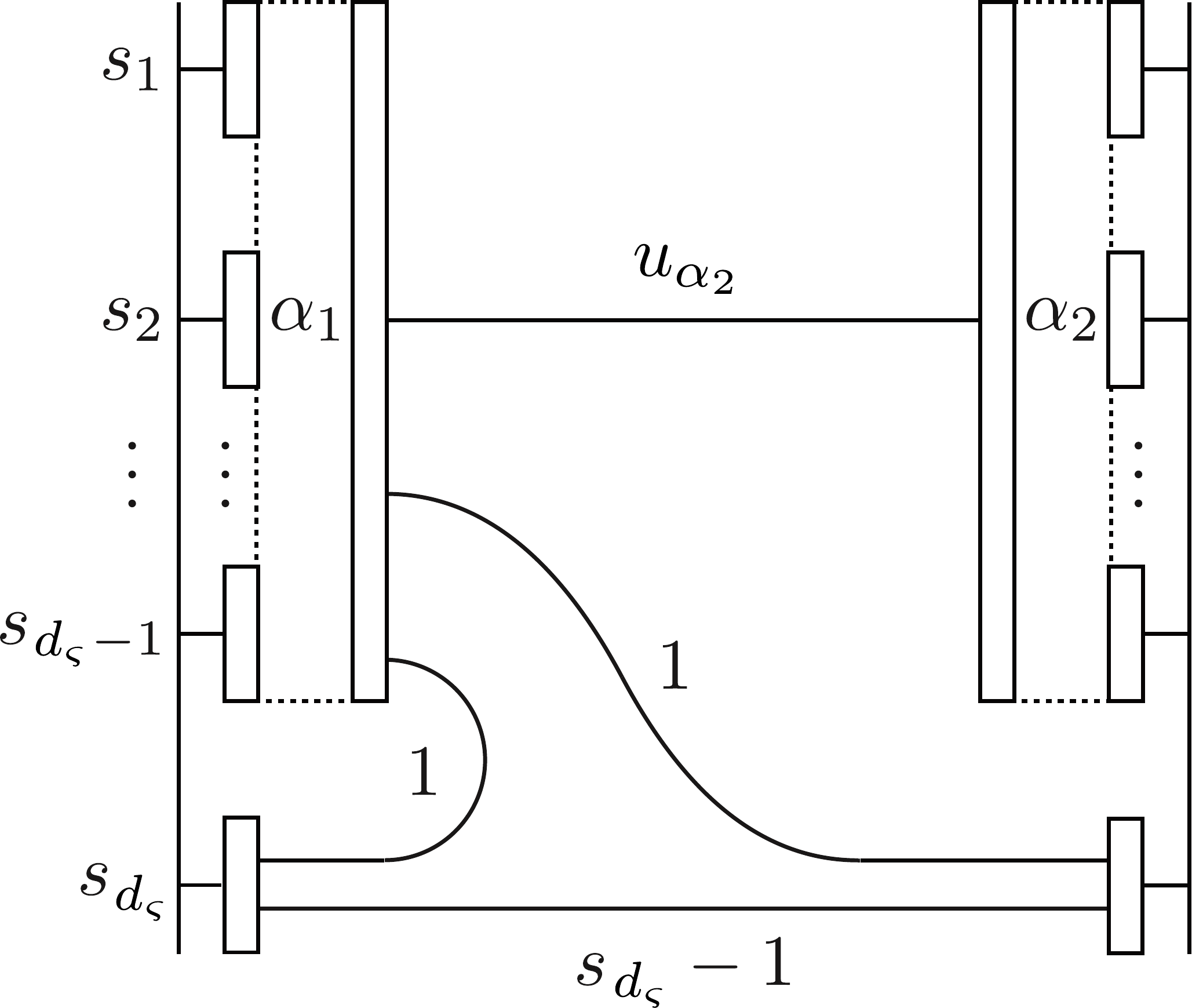} .}} 
\end{align} 
\end{cor} 

\begin{proof} 
We obtain this result by vertically reflecting the tangles of lemma~\ref{2stepLem} and exchanging $\alpha_1$ and $ \alpha_2$.
\end{proof}

Now we are ready to construct basis tangles of type $\PD1_\multii$~\eqref{InsertTwoBoxes} with $\Defect_{\alpha_1} = \Defect_{\alpha_2}$, 
and $v = 0$, and $r = w  = 1$, having any number of crossing links. 
The next lemma~\ref{SameSideLem} and corollary~\ref{SameSideCor} thus generalize lemmas~\ref{SameSideLem1} and~\ref{SameSideLem2}.

\begin{lem} \label{SameSideLem} 
Suppose $\Summed_\multii < \ppmin(q)$ and $\sIndex_{\np_\multii} \leq \sIndex_1$. 
If induction hypothesis~\ref{IndAss1} holds, then for all 
Jones-Wenzl link states $\alpha_1, \alpha_2 \in \PS_\lds$ such that
\begin{align} \label{SameSideLemDefectNumber}
\smin(\lds) + 2 \leq \, \Defect := \Defect_{\alpha_1} = \Defect_{\alpha_2} \, \leq \smax(\lds) - 2 ,
\end{align}
the following tangle is a polynomial in the elements of $\mathsf{G}_\multii$\textnormal{:}
\begin{align} \label{SameSideTangle} 
\vcenter{\hbox{\includegraphics[scale=0.275]{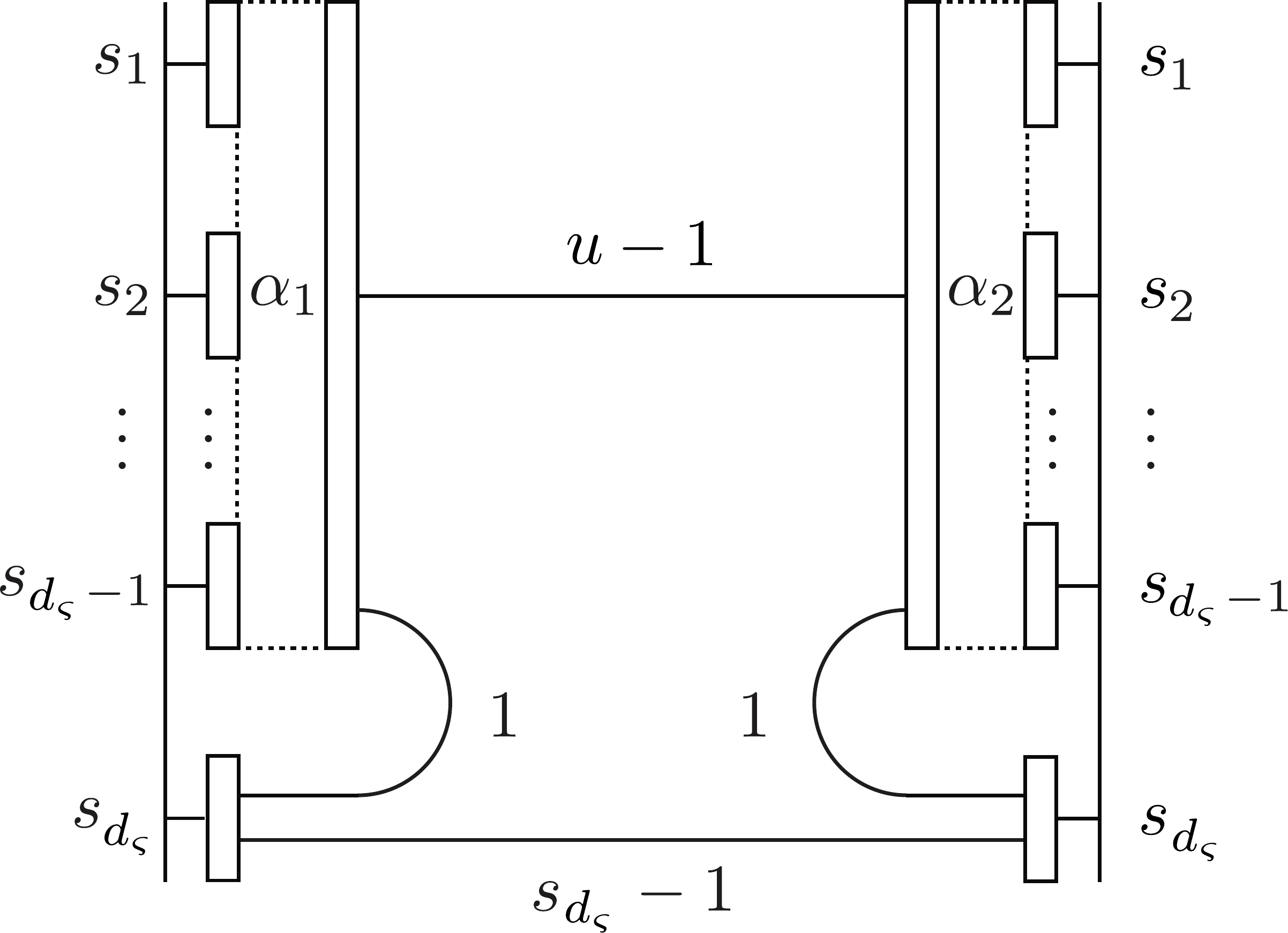} .}} 
\end{align} 
\end{lem} 

\begin{proof}
For later use, we note that if $\np_\multii = 3$, then we must have $\sIndex_2 > 1$;
otherwise, we have $\smash{\lds} = (\sIndex_1, \sIndex_2) = (\sIndex_1,1)$ and
\begin{align} 
\Defect:= \Defect_{\alpha_1} = \Defect_{\alpha_2} \in \DefectSet_\lds \qquad 
\text{where $\DefectSet_\lds = \DefectSet\sub{\sIndex_1,1} \overset{\eqref{SpecialDefSet}}{=} \{\sIndex_1\pm1\},$} 
\end{align} 
so either $\Defect = \sIndex_1 - 1 = \smin(\smash{\lds})$, or $\Defect = \sIndex_1 + 1 = \smax(\smash{\lds})$, 
contradicting condition~\eqref{SameSideLemDefectNumber} on $\Defect$. 
We also observe that this condition guarantees that
\begin{align} 
\Defect \pm 2 \in \DefectSet_\lds \qquad \Longrightarrow \qquad \smash{\PP_\lds\super{\Defect\pm2}} \neq \emptyset, 
\end{align} 
which we frequently use in this proof.  As usual, by linearity, we assume that $\alpha_1$ and $\alpha_2$ are Jones-Wenzl link patterns.

First, we take the tangle~\eqref{2stepTangle} of lemma~\ref{2stepLem}, $\alpha_1$ in~\eqref{2stepTangle} being 
the link pattern consisting of $\Defect - 2$ defects, and its vertical reflection with the replacement $\alpha_2 \mapsto \alpha_1$.
Multiplying them, with the latter on the left, we obtain 
\begin{align} \label{Mirror1.5} 
\vcenter{\hbox{\includegraphics[scale=0.275]{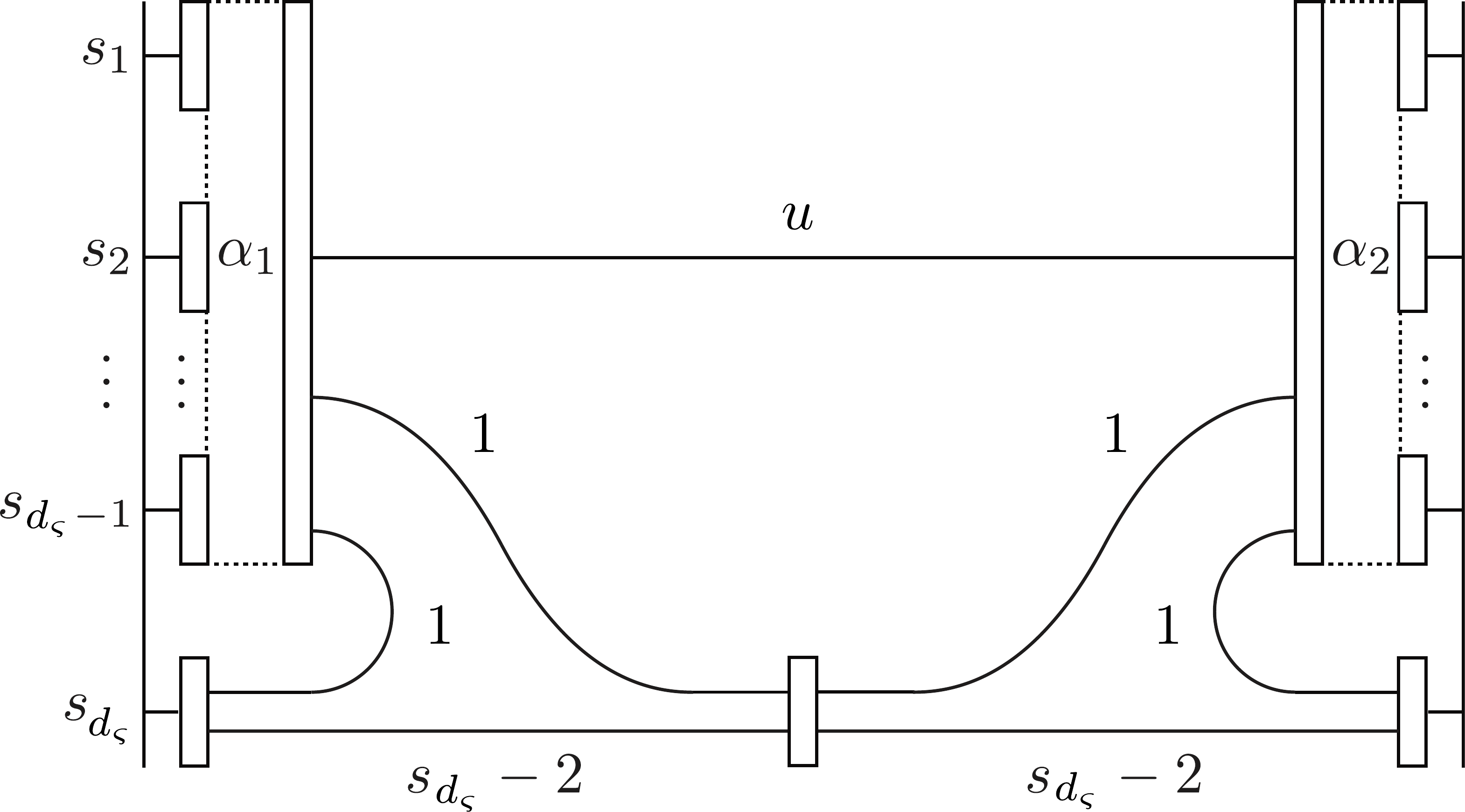} .}} 
\end{align} 
Next, we decompose the middle projector box of this result.  
By property~\eqref{ProjectorID2}, only two tangles in this decomposition, each of the form~\eqref{SpecialTDiag} 
with $j = k = 0$ and $i \in \{0,1\}$, give nonzero terms.  Using~\eqref{SpecialT}, we see that~\eqref{Mirror1.5} equals
\begin{align} \label{Sys1}
\vcenter{\hbox{\includegraphics[scale=0.275]{e-Generators134General.pdf}}} 
\; + \quad \frac{[\sIndex_{\np_\multii}-1]}{[\sIndex_{\np_\multii}]} \,\, \times \,\,
\vcenter{\hbox{\includegraphics[scale=0.275]{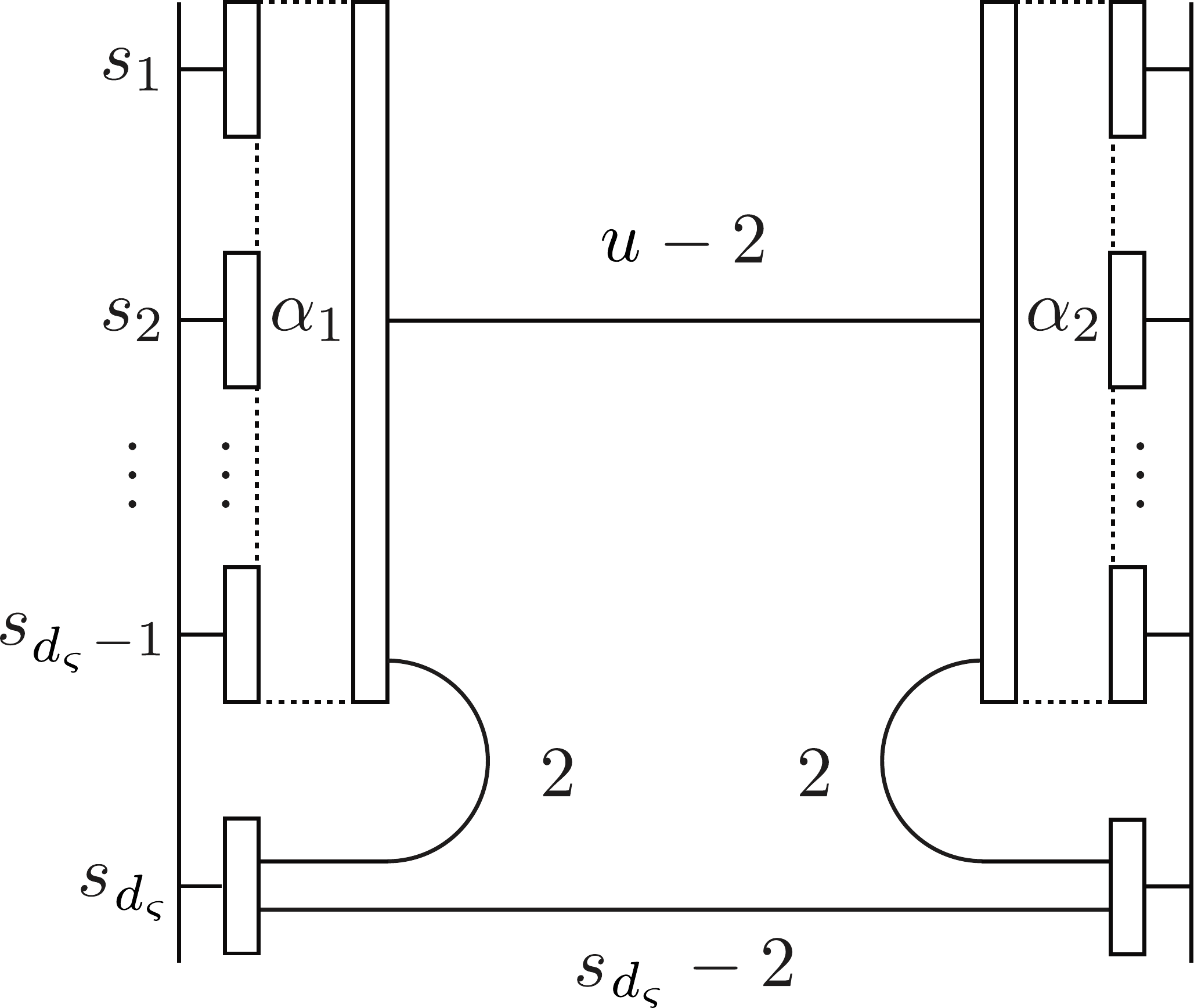} .}} 
\end{align}

Second, we take the tangle~\eqref{2stepTangle} of lemma~\ref{2stepLem}, $\alpha_2$ in~\eqref{2stepTangle} being 
the link pattern consisting of $\Defect + 2$ defects, and its vertical reflection with the replacement $\alpha_1 \mapsto \alpha_2$.
Multiplying them, now with the latter on the right, we obtain
\begin{align} 
\label{Mirror2NoAlpha} 
\vcenter{\hbox{\includegraphics[scale=0.275]{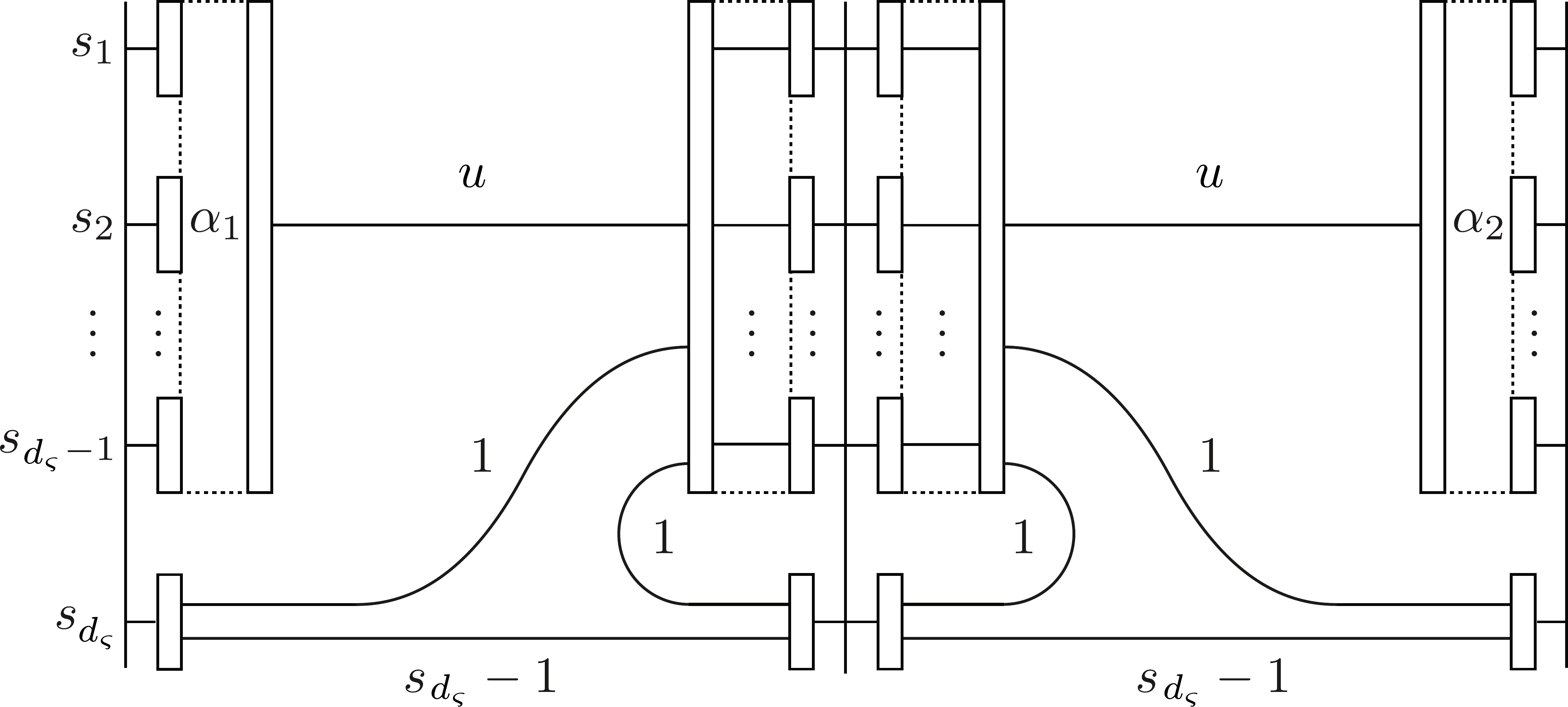} .}}
\end{align} 
Again, we decompose the bottom-middle projector box of this product.  As before, we have exactly two nonzero terms: 
\begin{align} 
\label{MirrorDec1} 
& \vcenter{\hbox{\includegraphics[scale=0.275]{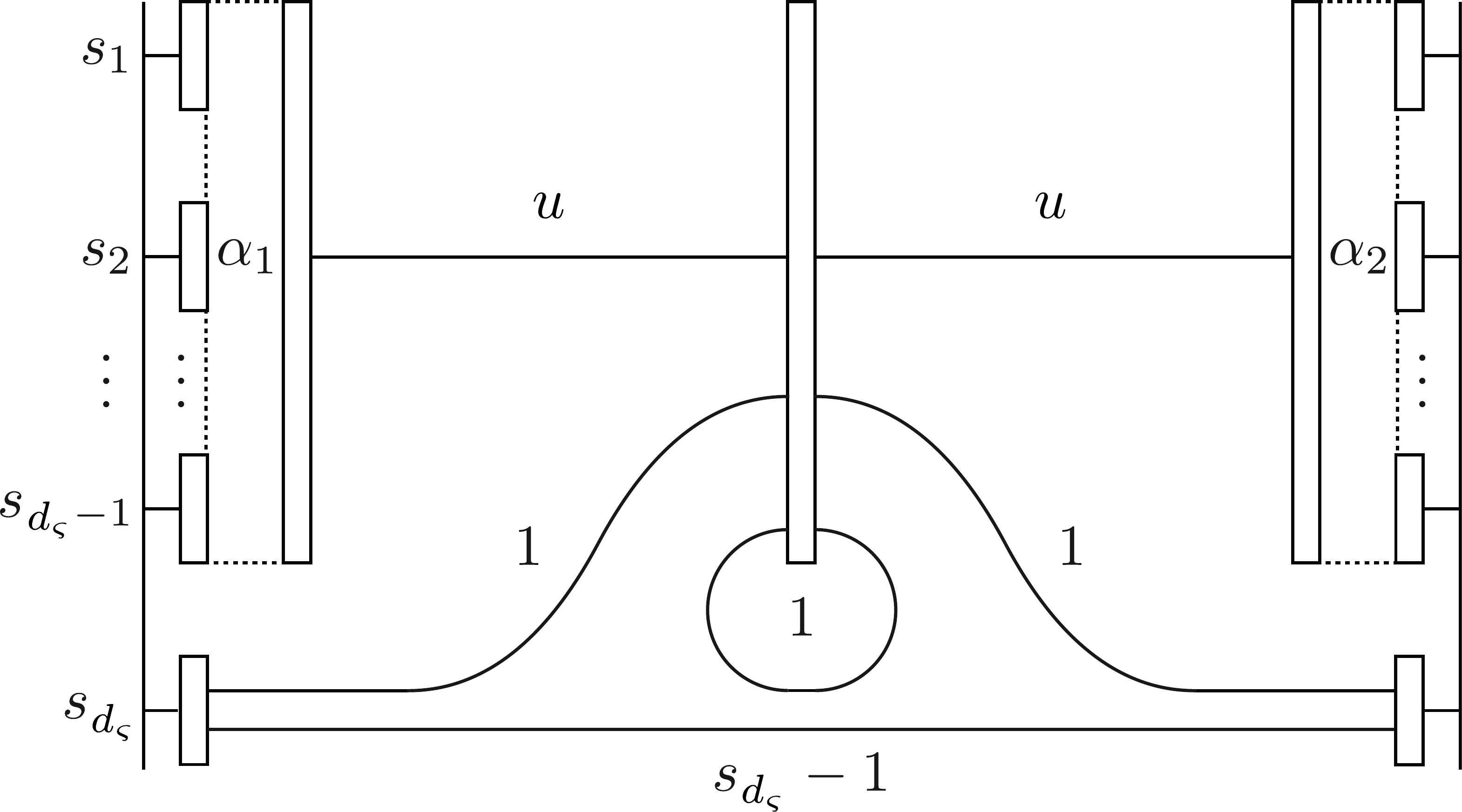}}} \\[1em]
\label{MirrorDec2} 
+ \quad \frac{[\sIndex_{\np_\multii}-1]}{[\sIndex_{\np_\multii}]} \,\, \times \,\, & \vcenter{\hbox{\includegraphics[scale=0.275]{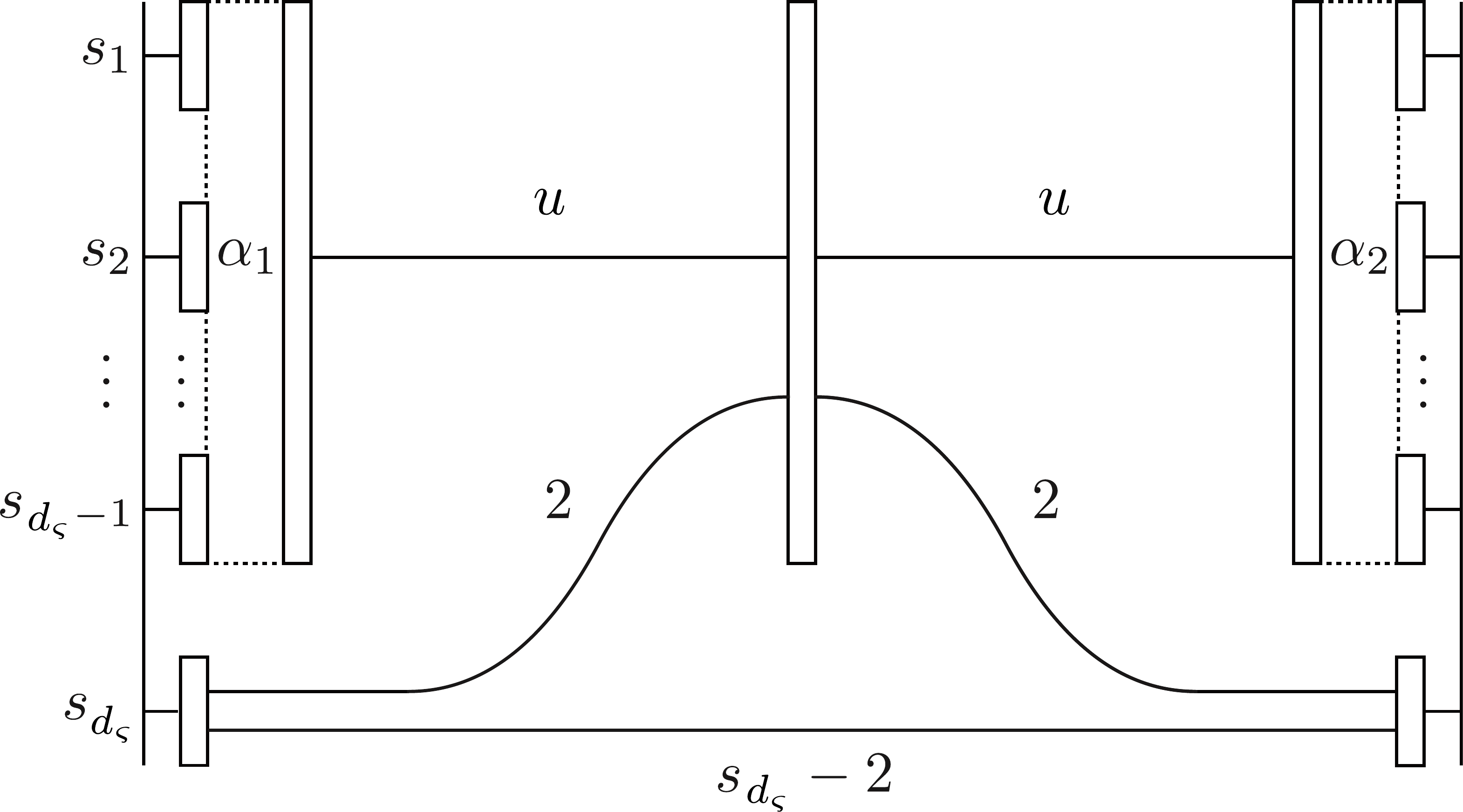} .}} 
\end{align}
Next, we use identity~\eqref{DeltaTangleGen} from 
lemma~\ref{CollectionLem}
to remove the loop passing through the middle projector box of~\eqref{MirrorDec1}, and we decompose that 
box. Once again, there are only two nonvanishing terms.  Using~\eqref{SpecialT}, we find that
\begin{alignat}{2} \label{strt1} 
\eqref{MirrorDec1} \quad
\overset{\eqref{SpecialT}}{=} \; & \quad - \frac{[\Defect+3]}{[\Defect+2]} \,\, \times \,\, && 
\vcenter{\hbox{\includegraphics[scale=0.275]{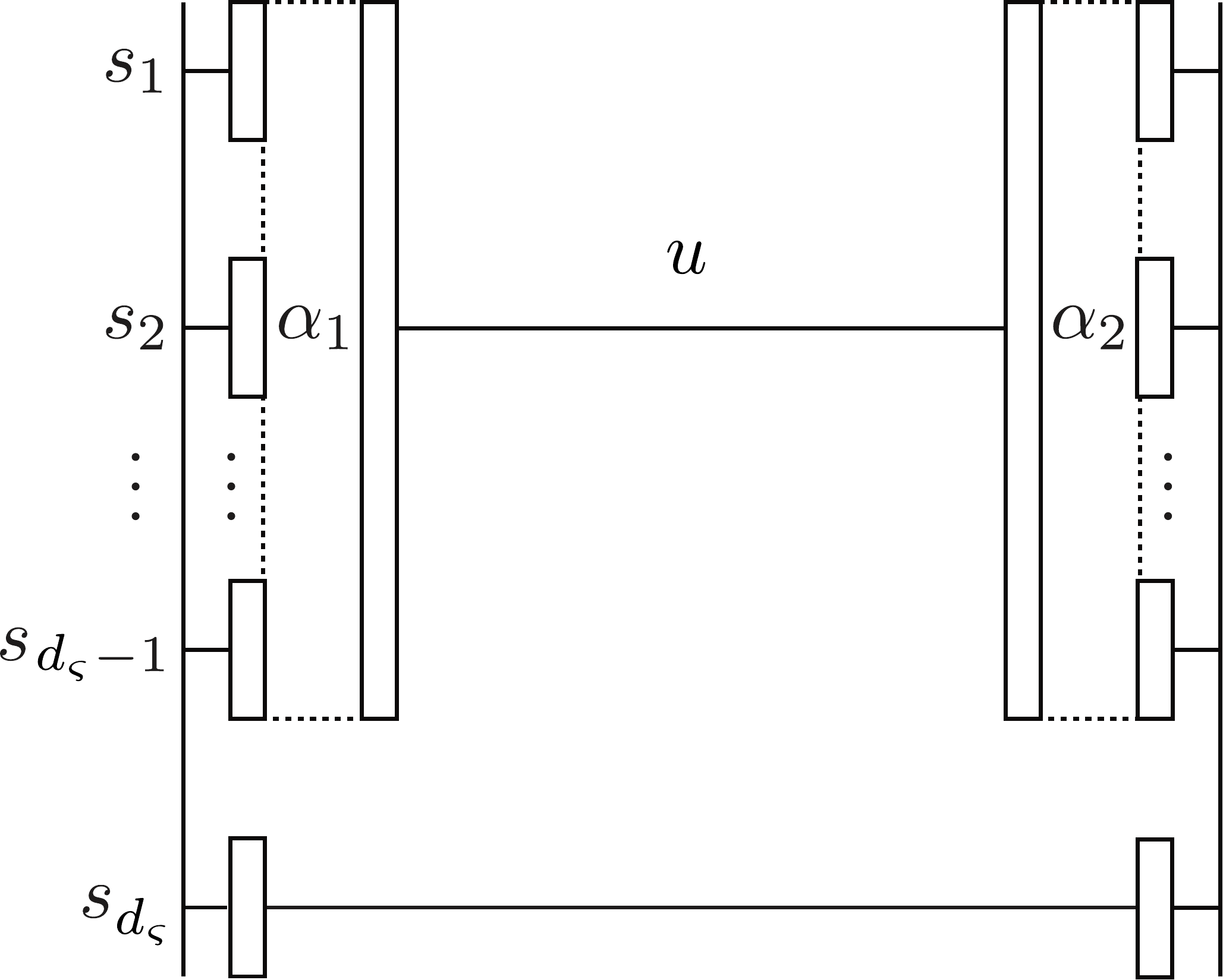}}} \\[1em] 
\label{end1} 
\; &  \; -\frac{[\Defect+3][\Defect]}{[\Defect+2][\Defect+1]} \,\, \times \,\, 
&& \vcenter{\hbox{\includegraphics[scale=0.275]{e-Generators134General.pdf} .}} 
\end{alignat}
Then, we decompose the middle projector box in the tangle of~\eqref{MirrorDec2}, obtaining exactly three nonzero terms:
\begin{alignat}{2}
\label{strt2}
\eqref{MirrorDec2} \quad
\overset{\eqref{SpecialT}}{=} \; & \quad \quad \quad \frac{[\sIndex_{\np_\multii}-1]}{[\sIndex_{\np_\multii}]} \,\, \times \,\, && \vcenter{\hbox{\includegraphics[scale=0.275]{e-Generators140General.pdf}}} \\[1em] 
 \; &  \; + \quad \frac{[\sIndex_{\np_\multii}-1][\Defect][2]}{[\sIndex_{\np_\multii}][\Defect+2]} \,\, \times \,\, && \vcenter{\hbox{\includegraphics[scale=0.275]{e-Generators134General.pdf}}} \\[1em] 
\label{end2} 
 \; &  \; + \quad \frac{[\sIndex_{\np_\multii}-1][\Defect][\Defect-1]}{[\Defect+2][\Defect+1]} \,\, \times \,\, && \vcenter{\hbox{\includegraphics[scale=0.275]{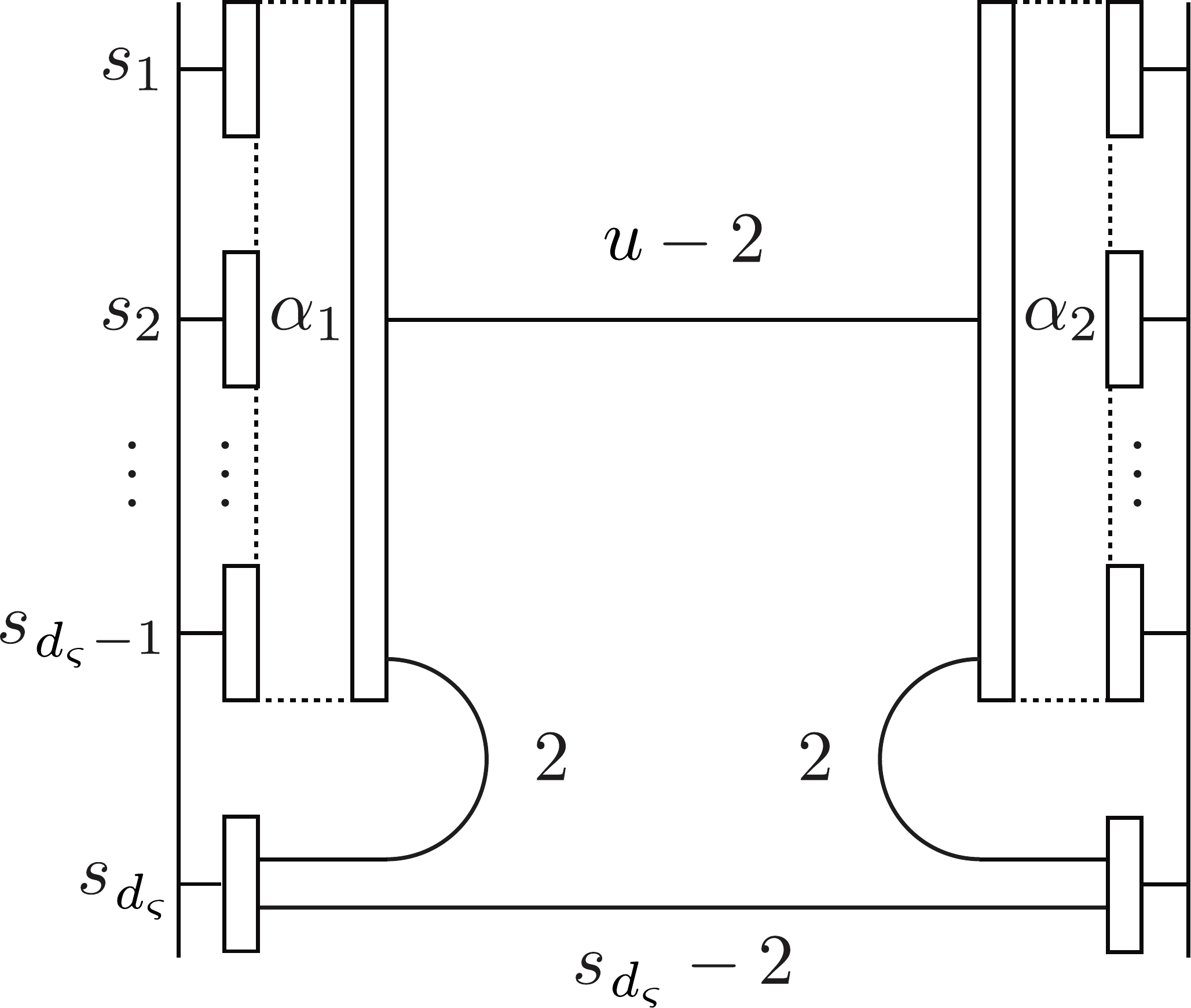} .}}
\end{alignat}
After replacing~\eqref{MirrorDec1} 
by (\ref{strt1},~\ref{end1}) and 
\eqref{MirrorDec2} by (\ref{strt2}--\ref{end2}), we obtain 
\begin{alignat}{2} 
\label{Sys2} 
\eqref{Mirror2NoAlpha} \quad 
= \; & \quad \quad \left(\frac{[\sIndex_{\np_\multii}-1]}{[\sIndex_{\np_\multii}]} - \frac{[\Defect+3]}{[\Defect+2]} \right) \,\, \times \,\,
&& \vcenter{\hbox{\includegraphics[scale=0.275]{e-Generators140General.pdf}}} \\[1em]
\label{Sys3} 
+ \; & \left( \frac{[\sIndex_{\np_\multii}-1][2][\Defect]}{[\sIndex_{\np_\multii}][\Defect+2]} - \frac{[\Defect+3][\Defect]}{[\Defect+2][\Defect+1]} \right) \,\, \times \,\,
&& \vcenter{\hbox{\includegraphics[scale=0.275]{e-Generators134General.pdf}}} \\[1em]
\label{Sys4} 
\; & + \quad \quad \quad \frac{[\sIndex_{\np_\multii}-1][\Defect][\Defect-1]}{[\sIndex_{\np_\multii}][\Defect+2][\Defect+1]} \,\, \times \,\,
&& \vcenter{\hbox{\includegraphics[scale=0.275]{e-Generators146.pdf} .}}
\end{alignat}

Now, to finish, we solve the tangle in~\eqref{Sys3} from the two equations~\eqref{Sys1} and (\ref{Sys2}--\ref{Sys4}):
\begin{align}
\label{Solved} 
\; & \bigg(  \frac{[\Defect][\Defect-1]}{[\Defect+2][\Defect+1]} 
+ \frac{[\Defect+3][\Defect]}{[\Defect+2][\Defect+1]} - \frac{[\sIndex_{\np_\multii}-1][2][\Defect]}{[\sIndex_{\np_\multii}][\Defect+2]} \bigg) \,\, \times \,\,
\vcenter{\hbox{\includegraphics[scale=0.275]{e-Generators134General.pdf}}} 
\end{align} 
\begin{alignat}{2}
= \; & \quad \left(\frac{[\sIndex_{\np_\multii}-1]}{[\sIndex_{\np_\multii}]} - \frac{[\Defect+3]}{[\Defect+2]}\right) \,\, \times \,\,
&&\vcenter{\hbox{\includegraphics[scale=0.275]{e-Generators140General.pdf}}} \\[1em]
\; & + \quad \frac{[\Defect][\Defect-1]}{[\Defect+2][\Defect+1]} \,\, \times \,\,
&& \vcenter{\hbox{\includegraphics[scale=0.275]{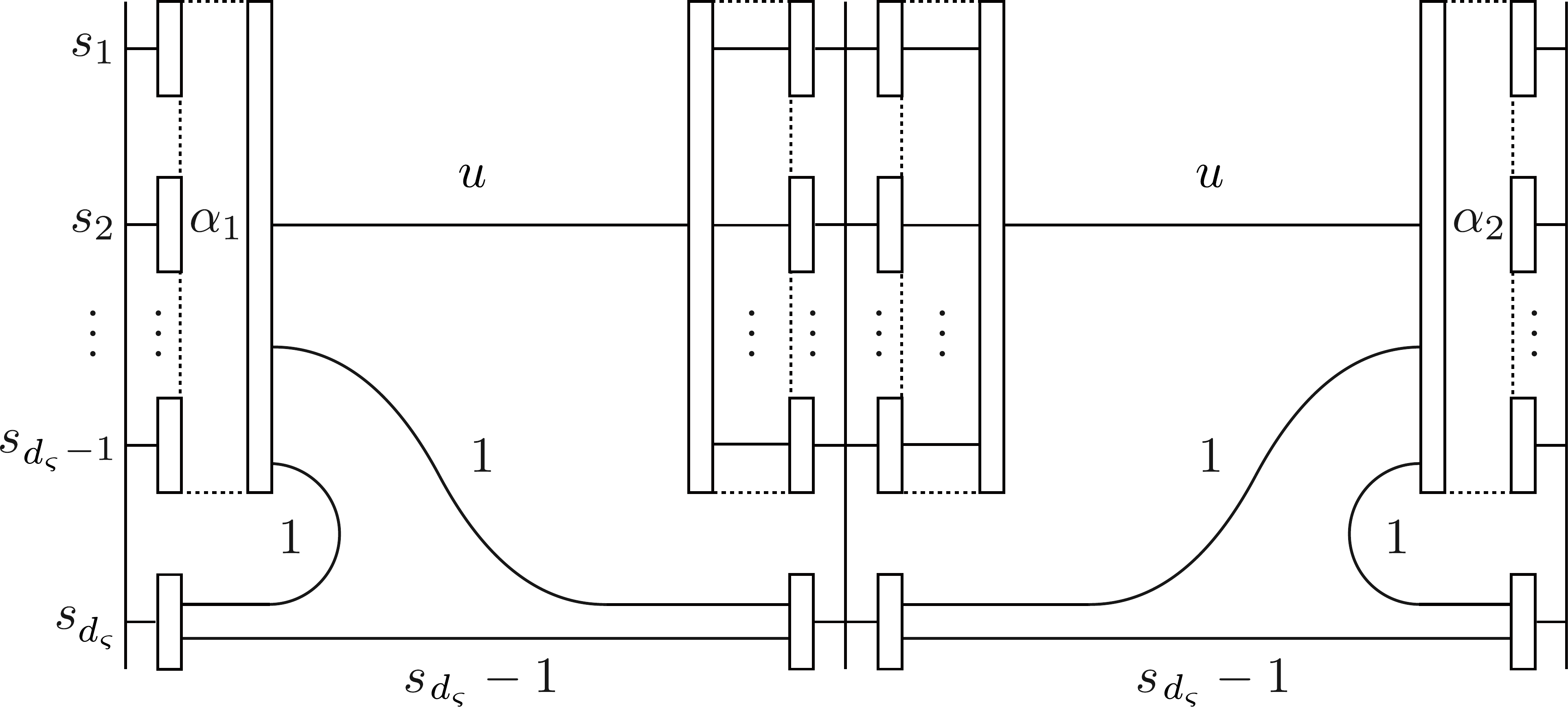}}} \\[1em]
\; & \quad \quad \quad \quad - 1 \,\, \times \,\, && \vcenter{\hbox{\includegraphics[scale=0.275]{e-Generators148.pdf} .}} 
\end{alignat} 
By induction hypothesis~\ref{IndAss1}, lemma~\ref{2stepLem}, and corollary~\ref{2stepCor}, all tangles on the right side are polynomials 
in the elements of $\mathsf{G}_\multii$.  
Using identity~\eqref{QintID} from lemma~\ref{CollectionLem},
we simplify the coefficient on the left side of~\eqref{Solved} to
\begin{align} \label{CoefSimpl} 
\frac{[\Defect][\Defect-1]}{[\Defect+2][\Defect+1]} + \frac{[\Defect+3][\Defect]}{[\Defect+2][\Defect+1]} 
- \frac{[\sIndex_{\np_\multii}-1][2][\Defect]}{[\sIndex_{\np_\multii}][\Defect+2]} 
\overset{\eqref{QintID}}{=} \frac{[\Defect][2 \sIndex_{\np_\multii} + 2]}{[\Defect+2][\sIndex_{\np_\multii}][\sIndex_{\np_\multii}+1]} . 
\end{align} 
Finally, provided that this coefficient does not vanish, we may
conclude that the tangle on the left side of~\eqref{Solved} is indeed
a polynomial in the elements of $\mathsf{G}_\multii$ too, thus proving the lemma.  
To show that this is the case, we observe that
\begin{align} \label{MidIne} 
\begin{cases}  
\text{$\sIndex_j \geq 1$ for all $j = 1,2,\ldots,\np_\multii$,} \\
\text{$\sIndex_2 > 1$ if $\np_\multii = 3$ by the first paragraph of this proof,}
\end{cases}  
\qquad \Longrightarrow \qquad \sIndex_1 + \sIndex_{\np_\multii} + 2 \leq \Summed_\multii < \ppmin(q) . 
\end{align} 
Therefore, using the assumption $\sIndex_{\np_\multii} \leq \sIndex_1$ of the lemma, we see that
\begin{align} 
2\sIndex_{\np_\multii} + 2 \leq \sIndex_1 + \sIndex_{\np_\multii} + 2 &\overset{\eqref{MidIne}}{<} \ppmin(q) 
\qquad \Longrightarrow \qquad [2 \sIndex_{\np_\multii} + 2] \neq 0 
\quad \text{and} \quad [\sIndex_{\np_\multii}], [\sIndex_{\np_\multii}+1] \neq 0 .
\end{align}
Also, assumption~\eqref{SameSideLemDefectNumber} shows that 
\begin{align}
\Defect \leq \smax(\lds) - 2 = \Summed_\multii - \sIndex_{\np_\multii} - 2 < \ppmin(q) - 3 
\qquad \Longrightarrow \qquad [\Defect] \neq 0 \quad \text{and} \quad [\Defect+2] \neq 0 .
\end{align}
We conclude that the coefficient~\eqref{CoefSimpl} indeed does not vanish (nor blow up).  This proves the lemma.
\end{proof}

\begin{remark} \label{FlipRemark}
The assumption $\sIndex_{\np_\multii} \leq \sIndex_1$ in 
lemma~\ref{SameSideLem} is not restrictive.  Indeed, if $\sIndex_1 \leq \sIndex_{\np_\multii}$ instead, then we 
may horizontally flip all tangles and repeat our work with $\smash{\lds} \mapsto \smash{\fds}$ and $\sIndex_{\np_\multii} \mapsto \sIndex_1$ for the flipped tangles. 
\end{remark}

\begin{cor} \label{SameSideCor} 
Suppose $\Summed_\multii < \ppmin(q)$. 
If induction hypothesis~\ref{IndAss1} holds, then for all 
Jones-Wenzl link states $\alpha_1, \alpha_2 \in \PS_\lds$,
the following tangle is a polynomial in the elements of $\mathsf{G}_\multii$\textnormal{:}
\begin{align} 
\vcenter{\hbox{\includegraphics[scale=0.275]{e-Generators134General.pdf} .}} 
\end{align} 
\end{cor} 

\begin{proof} 
This immediately follows from combining lemmas~\ref{SameSideLem1},~\ref{SameSideLem2},
and~\ref{SameSideLem} and remark~\ref{FlipRemark}.
\end{proof} 
\begin{center}
\bf Constructing basis tangles without diagonal cables
\end{center}

We continue by constructing the basis tangles of type $\PD1_\multii$~\eqref{InsertTwoBoxes} with $\Defect_{\alpha_1} = \Defect_{\alpha_2}$, 
and $v = 0$, and $r = w \geq 2$.

\begin{lem} \label{BigSameSideLem} 
Suppose $\Summed_\multii < \ppmin(q)$. 
If induction hypothesis~\ref{IndAss1} holds, then for all 
Jones-Wenzl link states $\alpha_1, \alpha_2 \in \PS_\lds$,
the following tangle is a polynomial in the elements of $\mathsf{G}_\multii$\textnormal{:}
\begin{align} \label{BigSameSideTangle} 
\vcenter{\hbox{\includegraphics[scale=0.275]{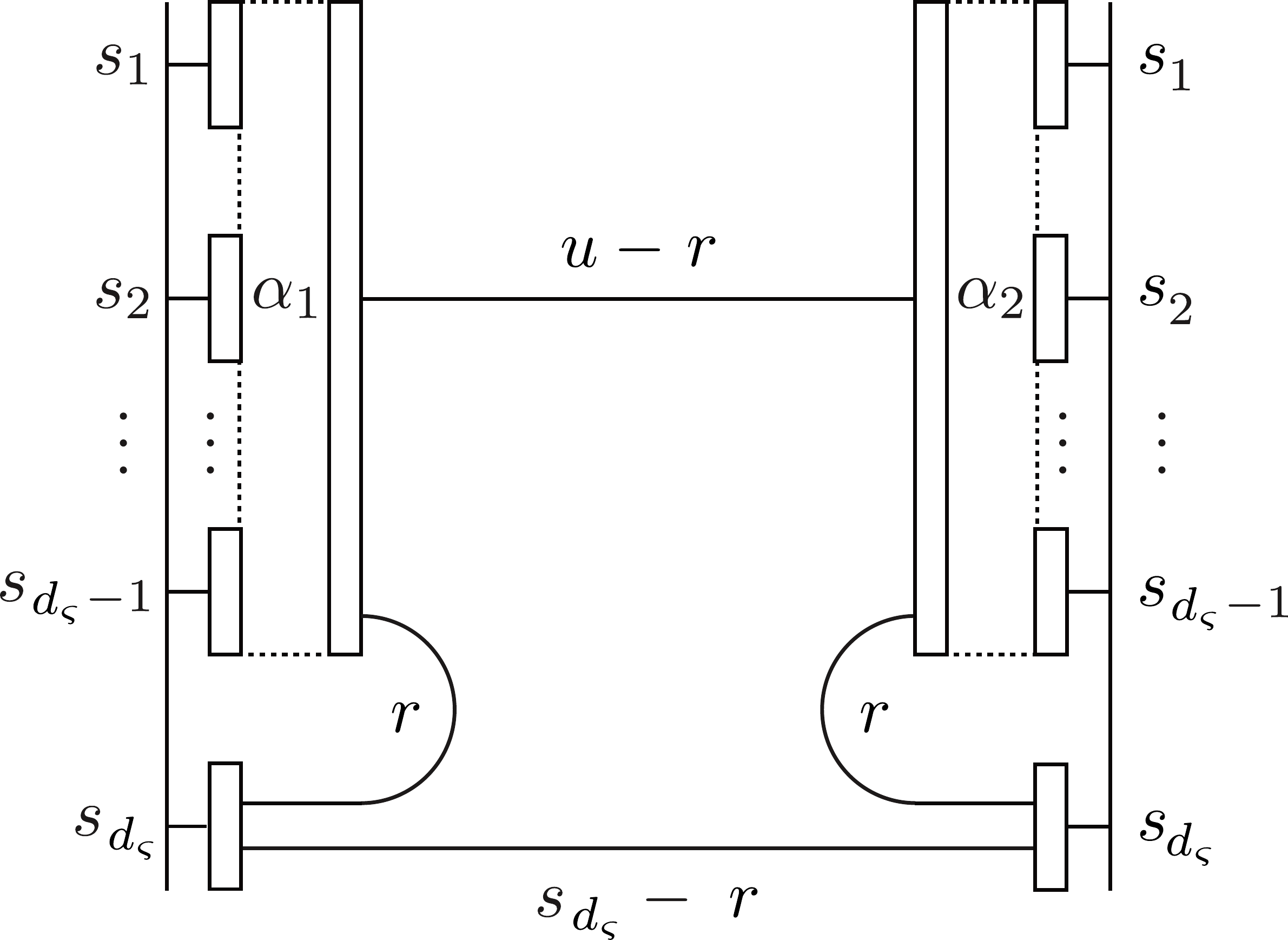} ,}} 
\end{align} 
where $r \in \{0,1,\ldots, \min(\sIndex_{\np_\multii}, \Defect)\}$.
\end{lem} 
\begin{proof} 
By remark~\ref{FlipRemark}, we may assume that $\sIndex_{\np_\multii} \leq \sIndex_1$ without loss of generality.
We prove the claim by induction on 
$r \geq 0$. 
Induction hypothesis~\ref{IndAss1} gives the initial case $r = 0$, and corollary~\ref{SameSideCor} the case $r=1$.
Assuming then that the claim holds if the cables joining the upper and lower boxes in~\eqref{BigSameSideTangle} 
have size $r$, we form the product
\begin{align} \label{ProdOfTangless} 
\vcenter{\hbox{\includegraphics[scale=0.275]{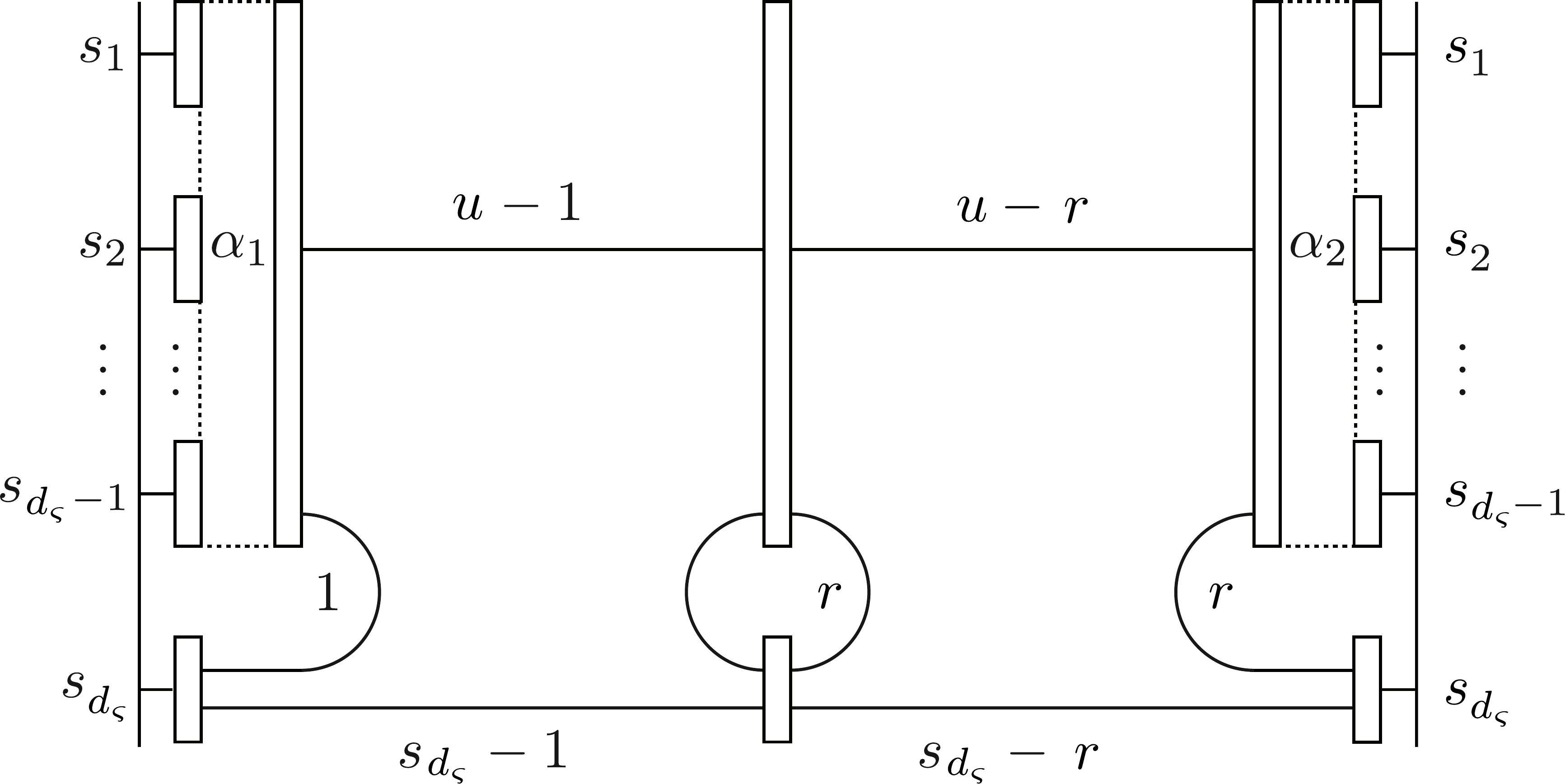} ,}} 
\end{align} 
and proceed exactly as in the proof of lemma~\ref{InitialCaseLem}, starting from~\eqref{MultiplyDiagrams}.  
\end{proof}

\begin{center}
\bf Constructing basis tangles with diagonal cables
\end{center}

In the next lemma~\ref{DifferentSideLem} and corollary~\ref{DifferentSideCor}, 
we finally construct general basis tangles of type $\PD1_\multii$~\eqref{InsertTwoBoxes}.

\begin{lem} \label{DifferentSideLem} 
Suppose $\Summed_\multii < \ppmin(q)$. 
If induction hypothesis~\ref{IndAss1} holds, then for all 
Jones-Wenzl link states $\alpha_1, \alpha_2 \in \PS_\lds$ 
such that $\Defect_{\alpha_1} \leq \Defect_{\alpha_2}$, the following tangle is a polynomial in the elements of $\mathsf{G}_\multii$\textnormal{:}
\begin{align} \label{DifferentSideTangle} 
\vcenter{\hbox{\includegraphics[scale=0.275]{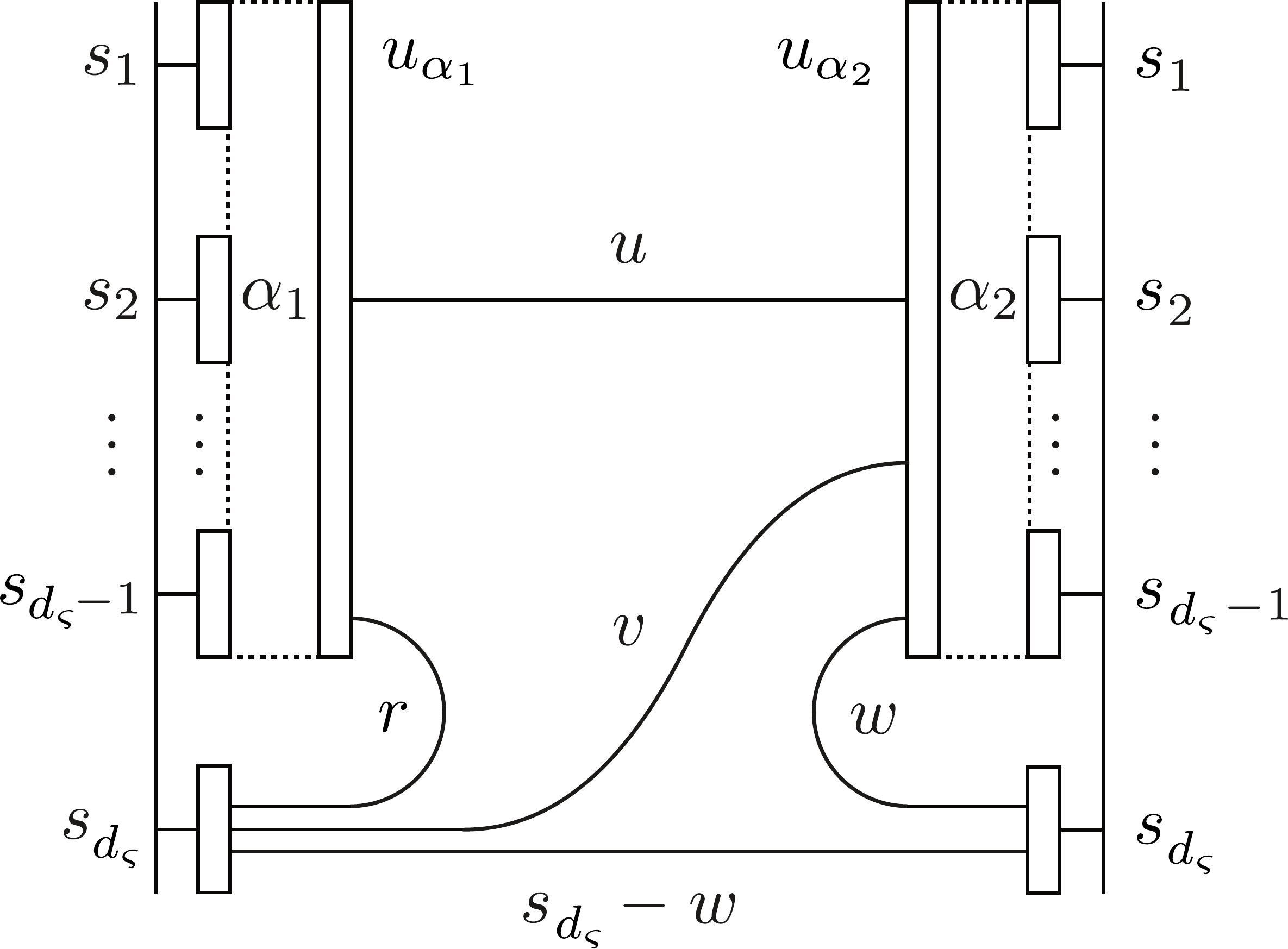} ,}}
\end{align}
where $r \in \mathsf{R}_{\alpha_1,\alpha_2}$,
$\Defect_{\alpha_2} = \Defect_{\alpha_1} + 2v$, $u + r = \Defect_{\alpha_1}$, and $w = r + v$. 
\end{lem} 

\begin{proof} 
By remark~\ref{FlipRemark}, we may assume that $\sIndex_{\np_\multii} \leq \sIndex_1$ without loss of generality.
We prove the claim by induction on 
$v \geq 0$.  Lemma~\ref{BigSameSideLem} gives the initial case $v=0$.  
Then, assuming that the claim holds if the diagonal cable in~\eqref{DifferentSideTangle} 
has size $v$, we form the product
\begin{align} \label{ProdOfTangles3} 
& \vcenter{\hbox{\includegraphics[scale=0.275]{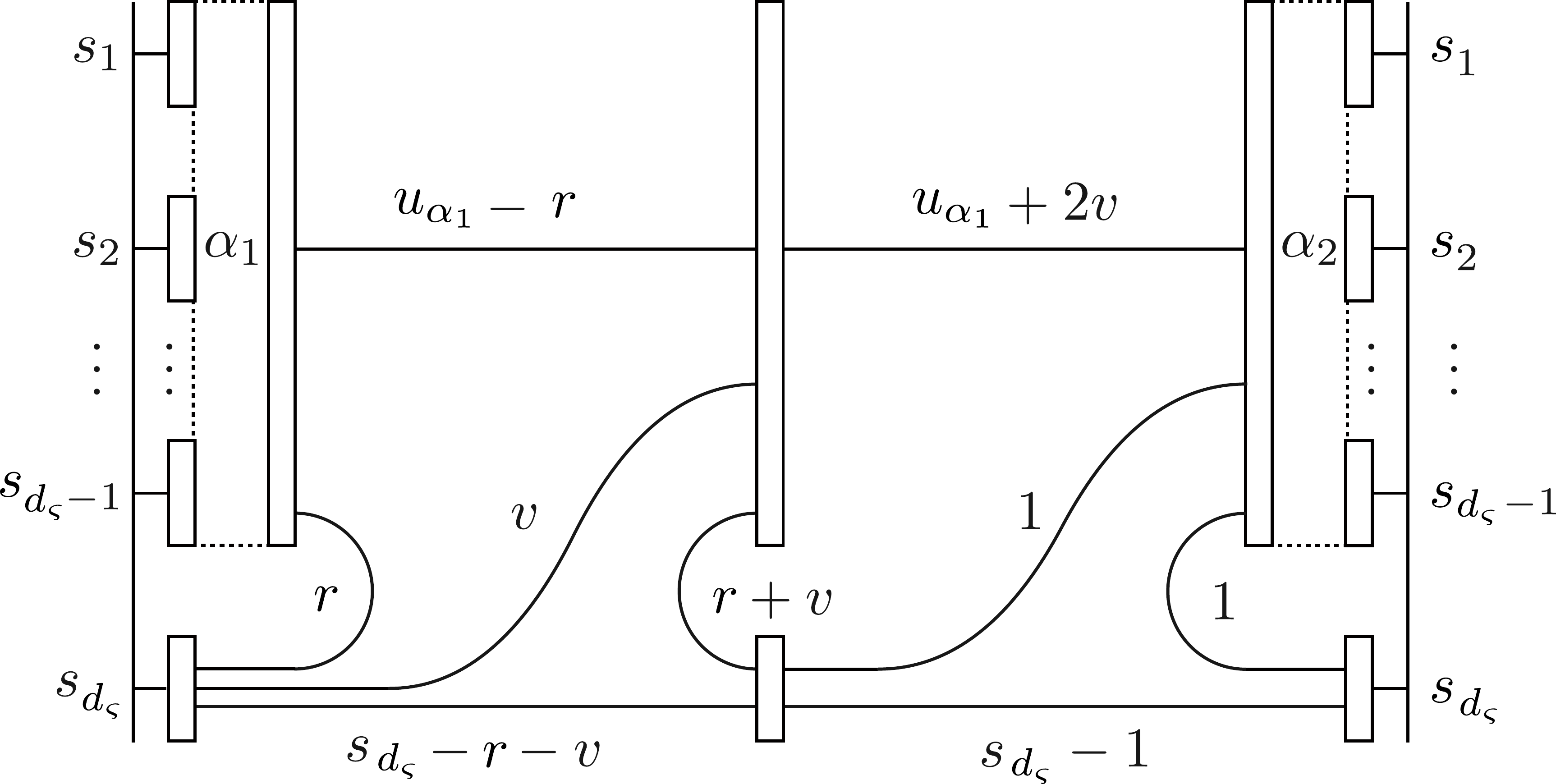}}} \\[1em] 
\label{ProdOfTangles4} 
\overset{\eqref{ProjectorID1}}{=} \quad &\vcenter{\hbox{\includegraphics[scale=0.275]{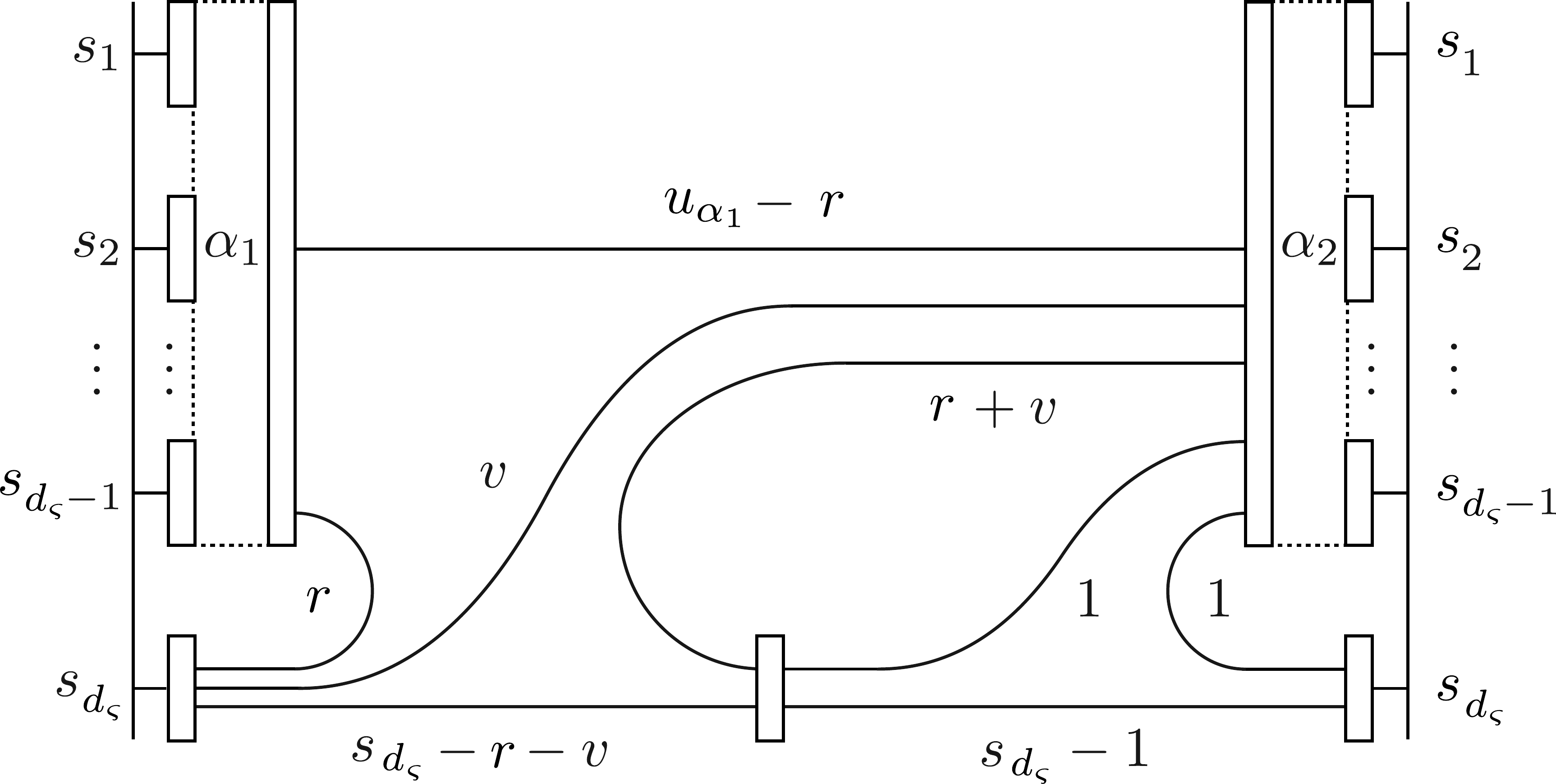} .}} 
\end{align} 
We decompose the middle projector box.
By property~\eqref{ProjectorID2}, only one tangle in this decomposition is nonzero: the one of the form~\eqref{SpecialTDiag} with 
$i = 1$, $j = r + v - 1$, $m = \sIndex_{\np_\multii} - r - v - 1$, and $k = 0$.
Hence, using~\eqref{SpecialT}, we find that 
\begin{align} \label{ProdOfTangles5} 
\eqref{ProdOfTangles4} \quad \overset{\eqref{SpecialT}}{=} 
\quad \frac{[\sIndex_{\np_\multii} - r - v]}{[\sIndex_{\np_\multii}]} \,\, \times \,\,
\vcenter{\hbox{\includegraphics[scale=0.275]{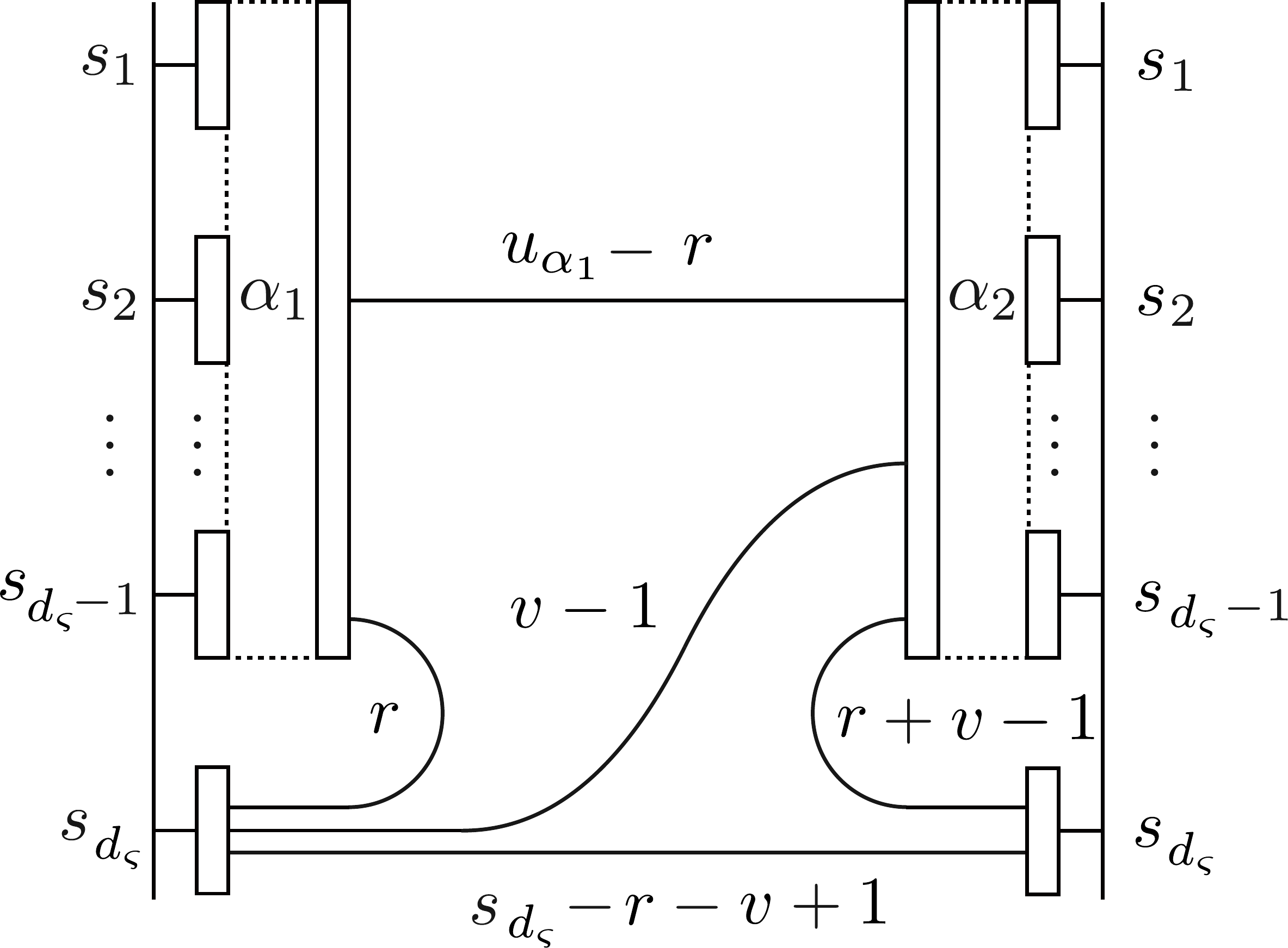} .}} 
\end{align} 
The tangles~(\ref{ProdOfTangles3}--\ref{ProdOfTangles4}) are all equal and, by the induction hypothesis applied to~\eqref{ProdOfTangles3}, they
equal a product of tangles that are polynomials in the elements of $\mathsf{G}_\multii$. Also, with 
$\sIndex_{\np_\multii} - r - v < \sIndex_{\np_\multii}  < \ppmin(q)$,
the coefficient 
in~\eqref{ProdOfTangles5} does not vanish (or blow up).
Hence, the tangle in~\eqref{ProdOfTangles5} 
is such a polynomial too. 
This finishes the induction step. 
\end{proof}

\begin{cor} \label{DifferentSideCor} 
Suppose $\Summed_\multii < \ppmin(q)$. 
If induction hypothesis~\ref{IndAss1} holds, then for all 
Jones-Wenzl link states $\alpha_1, \alpha_2 \in \PS_\lds$
such that $\Defect_{\alpha_2} \leq \Defect_{\alpha_1}$, the following tangle is a polynomial in the elements of $\mathsf{G}_\multii$\textnormal{:}
\begin{align} \label{FlipDifferentSideTangle} 
\vcenter{\hbox{\includegraphics[scale=0.275]{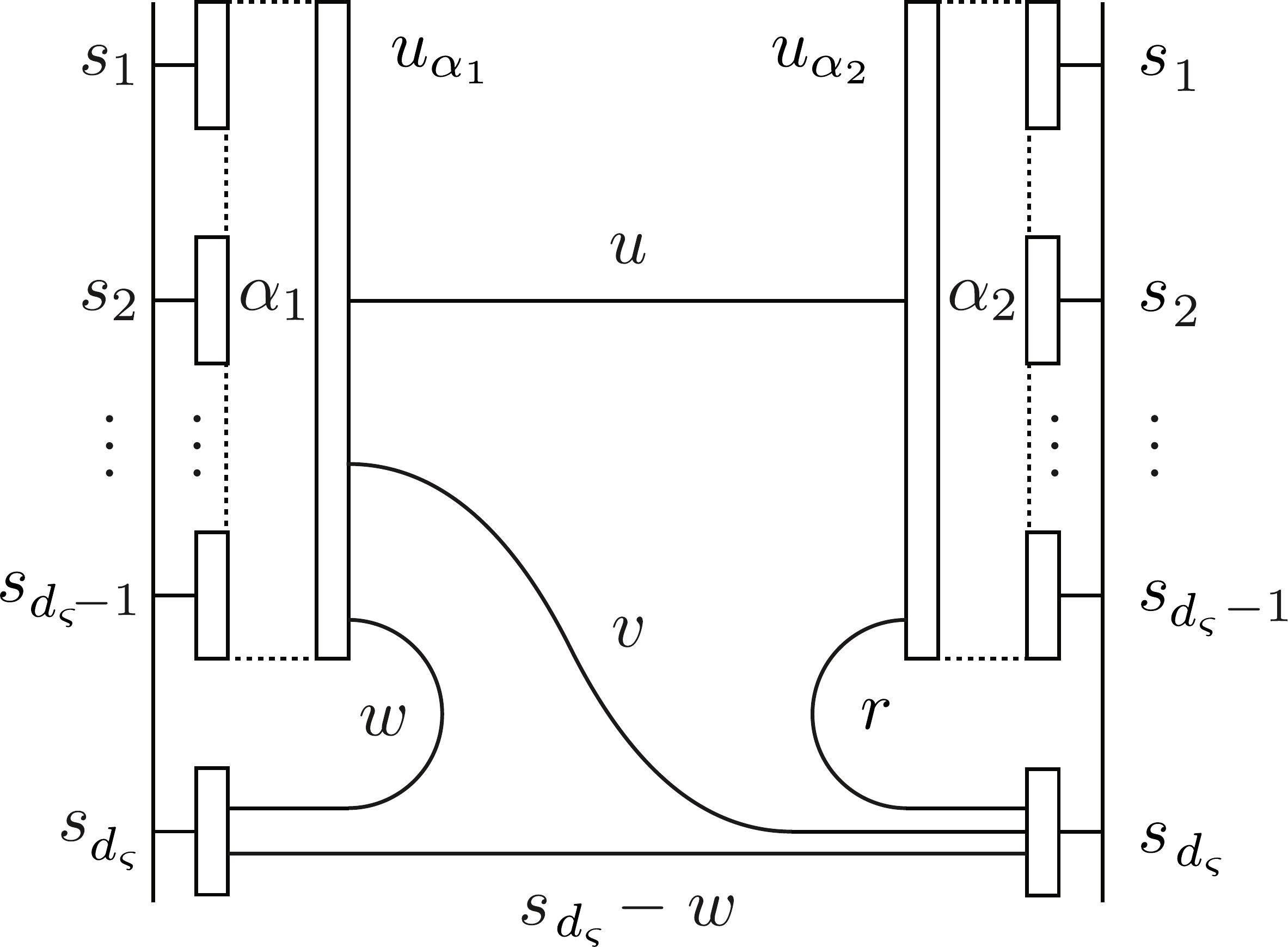} ,}} 
\end{align}
where $r \in \mathsf{R}_{\alpha_1,\alpha_2}$,
$\Defect_{\alpha_1} = \Defect_{\alpha_2} + 2v$, $u + r = \Defect_{\alpha_2}$, and $w = r + v$. 
\end{cor} 

\begin{proof} 
We obtain this result by vertically reflecting tangle~\eqref{DifferentSideTangle} of lemma~\ref{DifferentSideLem} 
and exchanging $\alpha_1$ and $ \alpha_2$.
\end{proof}

\subsection{Finishing the proof of the generator theorem} \label{TheProofSec}

Now we are ready to finish the induction step, i.e., claim~\ref{IndClaim}, and then the proof of theorem~\ref{GeneratorThm}.

\begin{cor} \label{FinalCor} 
Suppose $\Summed_\multii < \ppmin(q)$. 
If induction hypothesis~\ref{IndAss1} holds,
then claim~\ref{IndClaim} holds. 
\end{cor}

\begin{proof} 
By lemma~\ref{DifferentSideLem} and corollary~\ref{DifferentSideCor}, every tangle in the collection $\PD1_\multii$~\eqref{InsertTwoBoxes} is a polynomial in the 
elements of $\mathsf{G}_\multii$.   Because this set is a basis for $\WJ_\multii(\nu)$ by lemma~\ref{ManyBasesLem}, this property linearly extends from the basis elements 
to all tangles in $\WJ_\multii(\nu)$.  This proves claim~\ref{IndClaim}.
\end{proof}

\begin{proof}[Proof of theorem~\ref{GeneratorThm}] 
We prove items~\ref{GeneratorThmItem1} and~\ref{GeneratorThmItem2} as follows:
\begin{enumerate}[leftmargin=*]
\itemcolor{red}

\item
That $\mathsf{G}_\multii$ generates $\WJ_\multii(\nu)$ immediately follows by induction on 
$\np_\multii \geq 1$: corollary~\ref{InitialCaseCor} gives 
the initial case $\np_\multii = 2$, and corollary~\ref{FinalCor} (i.e., claim~\ref{IndClaim}) completes the induction step.  

\item
To prove that also the collection of tangles of the form 
\begin{align} \label{ClaimedGenerators}
\vcenter{\hbox{\includegraphics[scale=0.275]{e-Generators_3Vertex.pdf} ,}}
\end{align} 
where $s \in \DefectSet\sub{\sIndex_i,\sIndex_{i+1}}$ and $i \in \{ 1, 2, \ldots, \np_\multii - 1 \}$, generates $\WJ_\multii(\nu)$, 
we use definition~\eqref{3vertex1}, 
properties (\ref{ProjectorID2},~\ref{ProjectorID1}), and proposition~\ref{SpecialTProp} to 
obtain the following upper-triangular system of equations for each $i \in \{1,2,\ldots,\np_\multii - 1\}$:
\begin{align} \label{UpperTri} 
\vcenter{\hbox{\includegraphics[scale=0.275]{e-Generators_3Vertex.pdf}}} \quad
=  & \; \quad \sum_{k=\frac{1}{2}(\sIndex_i+\sIndex_{i+1}-s)}^{\min(\sIndex_i,\sIndex_{i+1})} \text{coef}_{s,k} 
\,\, \times \,\, \vcenter{\hbox{\includegraphics[scale=0.275]{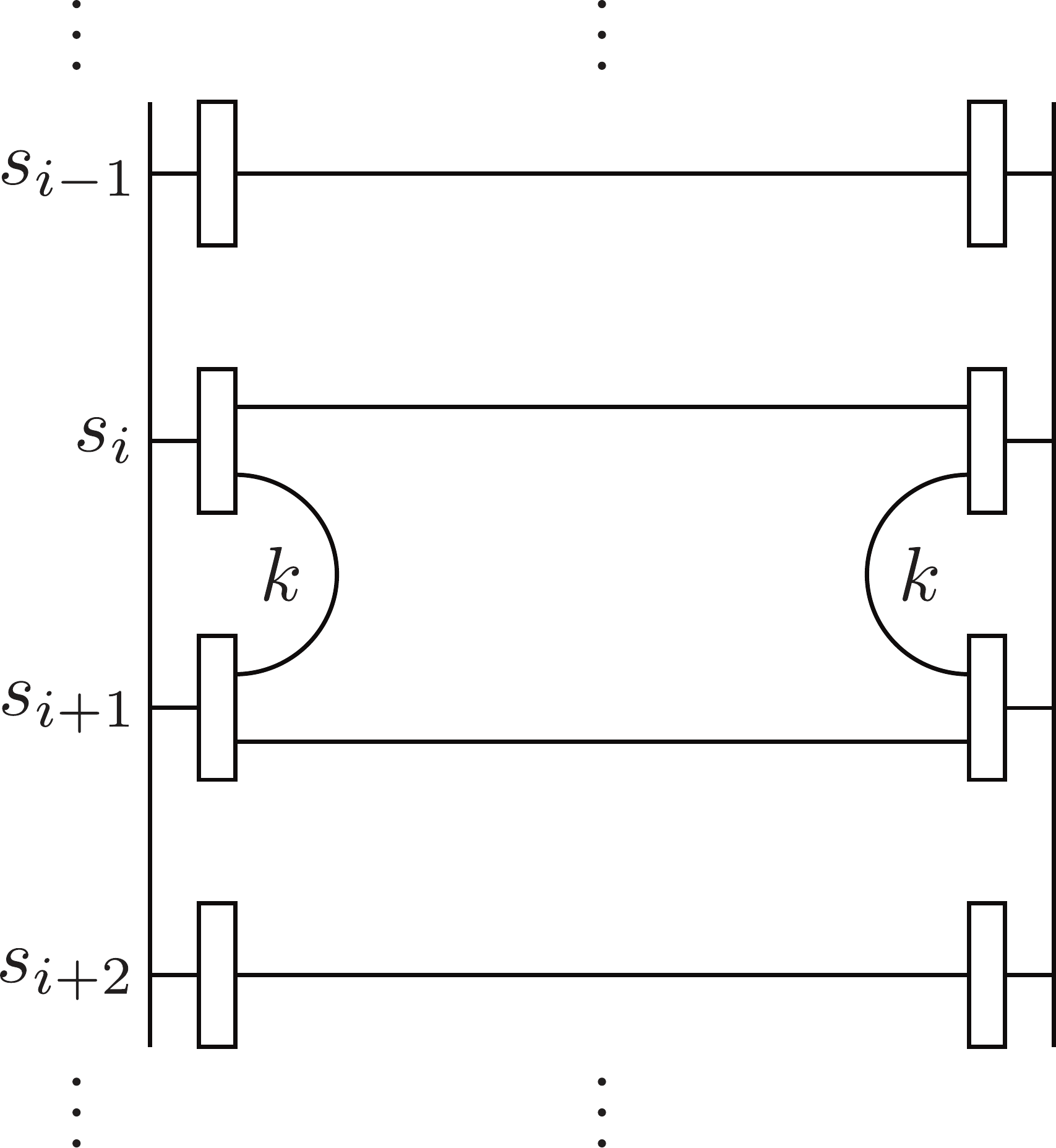}, }}  \\[1em]
\text{where} \qquad \quad \text{coef}_{s,k} := & \; 
\frac{\left[\frac{\sIndex_i+s-\sIndex_{i+1}}{2}\right]!\left[\frac{\sIndex_{i+1}+s-\sIndex_i}{2}\right]!\left[\frac{\sIndex_i+\sIndex_{i+1}+s}{2}-k\right]!}{[s]!\left[k-\frac{\sIndex_i+\sIndex_{i+1}-s}{2}\right]![\sIndex_i-k]![\sIndex_{i+1}-k]!} . 
\end{align}
Now, with $s \in \DefectSet\sub{\sIndex_i,\sIndex_{i+1}}$ and $\Summed_\multii < \ppmin(q)$, we see that
all of  the coefficients $\text{coef}_{s,k}$ are finite and none of the coefficients 
of the diagonal terms, i.e., terms with $k = \smash{\frac{1}{2}(\sIndex_i+\sIndex_{i+1}-s)}$, vanish.
Therefore, the system~\eqref{UpperTri} is invertible for each $i \in \{1,2,\ldots,\np_\multii - 1\}$. 
Because the collection of tangles appearing on the right side of~\eqref{UpperTri} generates $\WJ_\multii(\nu)$, 
so too does the collection of tangles of the form~\eqref{ClaimedGenerators}, appearing on the left side of~\eqref{UpperTri}.
\end{enumerate}
To finish, both generating sets~(\ref{EGenerators-00},~\ref{MasterDiagramsWJ-00}) (with the unit~\eqref{WJCompProj} added to the first one) 
are minimal because if we remove an element from either of them, then we cannot generate all of the Jones-Wenzl tangles in $\WJ_\multii(\nu)$.
\end{proof}

\section{Relations in the Jones-Wenzl algebra} \label{RelationLemProofSect}

It is well-known~\cite{vj, lk, rsa} that the Temperley-Lieb algebra of diagrams is isomorphic 
to the abstract associative algebra with generating set 
$\{ \mathbf{1}, \Gen_1, \Gen_2, \ldots, \Gen_{n-1} \}$, where $\mathbf{1}$ 
is the unit and the other generators satisfy exclusively the relations~(\ref{WordRelations1}--\ref{WordRelations3}).
We readily note the resemblance of the generating set~\eqref{EGenerators-00} 
of the Jones-Wenzl algebra $\WJ_\multii(\nu)$ with the $\TL_n(\nu)$-generators~\eqref{ExtMe}.
For a general multiindex $\multii = (\sIndex_1, \sIndex_2,\ldots, \sIndex_{\np_\multii}) \in \smash{\bZpos^\#}$,
the number of generators of type~\eqref{EGenerators-00} for the algebra $\WJ_\multii(\nu)$ equals 
\begin{align}
\sum_{i = 1}^{\np_\multii-1} (\# \DefectSet\sub{\sIndex_i,\sIndex_{i+1}} - 1) \overset{\eqref{SpecialDefSet}}{=}
\sum_{i = 1}^{\np_\multii-1} \min(\sIndex_i,\sIndex_{i+1}) ,
\end{align}
which is usually different from the number $\np_\multii-1$ of generators of type~\eqref{ExtMe} for $\TL_{\np_\multii}(\nu)$.
However, analogues and generalizations of the Temperley-Lieb relations can be obtained for generators~\eqref{EGenerators-00} as well.

In this section, we find some relations in the Jones-Wenzl algebra
and discuss special cases when all of the relations are known. 
In the end of this section, we prove proposition~\ref{RelationProp}.

\subsection{Case of two projectors} \label{TwoNodeCase}

In this section, we strengthen theorem~\ref{GeneratorThm} for $\multii = (\sIndex_1,\sIndex_2)$.

\begin{theorem} \label{GeneratorThmTwoNode}
Suppose $\max (\sIndex_1, \sIndex_2) < \ppmin(q)$. Then the Jones-Wenzl algebra $\WJ\sub{\sIndex_1,\sIndex_2}(\nu)$
has the following presentations in terms of generators and relations:
\begin{enumerate} 
\itemcolor{red}
\item \label{GeneratorThmTwoNodeItem1}
It has the two generators~\eqref{EGeneratorsTwonode}, 
which satisfy exclusively the relations
\begin{align} 
\label{EGeneratorsTwonodeUnitrel0}
\WJProj\sub{\sIndex_1,\sIndex_2}^2 & = \WJProj\sub{\sIndex_1,\sIndex_2} , \\
\label{EGeneratorsTwonodeUnitrel1}
\WJProj\sub{\sIndex_1,\sIndex_2} \ValGenWJ_1 & = \ValGenWJ_1 = \ValGenWJ_1 \WJProj\sub{\sIndex_1,\sIndex_2}  ,\\
\label{EGeneratorsTwonodeRel}
\ValGenWJ_1^{\min (\sIndex_1,\sIndex_2) + 1} 
& = - \frac{[\max (\sIndex_1,\sIndex_2)+1]}{[\sIndex_1][\sIndex_2]} \, \ValGenWJ_1^{\min (\sIndex_1,\sIndex_2)} .
\end{align}
Relations~(\ref{EGeneratorsTwonodeUnitrel0},~\ref{EGeneratorsTwonodeUnitrel1}) imply that 
$\WJProj\sub{\sIndex_1,\sIndex_2} = \mathbf{1}_{\WJ\sub{\sIndex_1,\sIndex_2}}$ is the unit in $\WJ\sub{\sIndex_1,\sIndex_2}(\nu)$.

\item \label{GeneratorThmTwoNodeItem2}
It has the $\min (\sIndex_1,\sIndex_2) + 1$ generators 
\begin{align}  \label{MasterDiagramsWJTwonode}
\ValGenMWJ\super{s}_1 \quad = \quad \vcenter{\hbox{\includegraphics[scale=0.275]{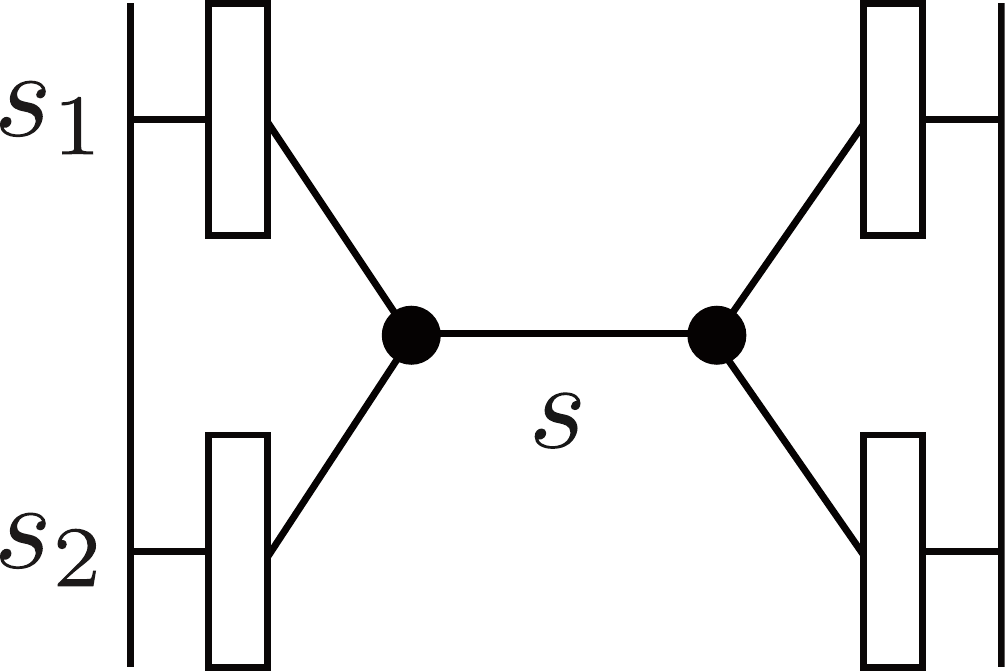} ,}}  
\end{align}  
with $s \in \DefectSet\sub{\sIndex_1, \sIndex_2} = \{ |\sIndex_1 -\sIndex_2|, |\sIndex_1 -\sIndex_2| + 2 , \ldots , \sIndex_1 + \sIndex_2 \}$,
%
which satisfy exclusively the relations
\begin{align} 
\label{MasterDiagramsWJTwonodeRel}
\ValGenMWJ\super{s}_1 \ValGenMWJ\super{s'}_1 
& =
\delta_{s, s'} \frac{ \ThetaNet(\sIndex_1,\sIndex_2,s) }{(-1)^s [s+1]} \, \ValGenMWJ\super{s}_1 , \qquad 
\textnormal{for all } s, s' \, \in \, \DefectSet\sub{\sIndex_1,\sIndex_2} .
\end{align}
The unit in $\WJ\sub{\sIndex_1,\sIndex_2}(\nu)$ is given by
\begin{align}
\label{MasterDiagramsWJTwonodeUnitrel}
\mathbf{1}_{\WJ\sub{\sIndex_1,\sIndex_2}}
& = \sum_{s \, \in \, \DefectSet\sub{\sIndex_1,\sIndex_2}} \frac{(-1)^s [s+1]}{\ThetaNet(\sIndex_1,\sIndex_2,s)} \, \ValGenMWJ\super{s}_1  .
\end{align}
\end{enumerate}
In particular, each of the following sets forms a basis for $\WJ\sub{\sIndex_1,\sIndex_2}(\nu)$:
\begin{align}
\label{Basis1Twonode}
& \big\{ \mathbf{1}_{\WJ\sub{\sIndex_1,\sIndex_2}} \big\}  \cup \big\{ \ValGenWJ_1^s \, \big| \, s = 1,2,\ldots,\min (\sIndex_1,\sIndex_2)\big\} , \\
\label{Basis2Twonode}
& \big\{ \ValGenMWJ\super{s} \, \big| \, s = |\sIndex_1 - \sIndex_2|, \, |\sIndex_1 - \sIndex_2| + 2 , \,\ldots ,\,\sIndex_1 + \sIndex_2 \big\} , \\
\label{Basis3Twonode}
& \big\{ \mathbf{1}_{\WJ\sub{\sIndex_1,\sIndex_2}} \big\}  \cup \big\{
 \ValGenMWJ\super{s} \, \big| \, s = |\sIndex_1 - \sIndex_2|, \, |\sIndex_1 - \sIndex_2| + 2 , \ldots , \, \sIndex_1 + \sIndex_2 - 2 \big\} .
\end{align}

\end{theorem}

\begin{proof} 
We first record that~(\ref{SpecialDefSet},~\ref{DimOfWJ},~\ref{PreRecursion2}) 
give the dimension of $\WJ\sub{\sIndex_1,\sIndex_2}(\nu)$:
\begin{align} \label{DimOfWJTwoNode}
\dim \WJ\sub{\sIndex_1,\sIndex_2}(\nu) \overset{\eqref{DimOfWJ}}{=} 
\sum_{s \, \in \, \DefectSet\sub{\sIndex_1,\sIndex_2}} \big(\Dim\sub{\sIndex_1,\sIndex_2}\super{s}\big)^2 
\overset{\eqref{PreRecursion2}}{=} \# \DefectSet\sub{\sIndex_1,\sIndex_2}
\overset{\eqref{SpecialDefSet}}{=} \min (\sIndex_1,\sIndex_2) + 1 .
\end{align}
Then, we prove items~\ref{GeneratorThmTwoNodeItem1}--\ref{GeneratorThmTwoNodeItem2} as follows:
\begin{enumerate}[leftmargin=*]
\itemcolor{red}
\item Lemma~\ref{InitialCaseLem} shows that~\eqref{EGeneratorsTwonode} forms a minimal generating set for $\WJ\sub{\sIndex_1,\sIndex_2}(\nu)$.
Relations~(\ref{EGeneratorsTwonodeUnitrel0}--\ref{EGeneratorsTwonodeUnitrel1}) follow from idempotent property~\eqref{ProjectorID0}.
The calculation leading to~\eqref{TopLineTangles} in the proof of lemma~\ref{InitialCaseLem} gives the other relations~\eqref{EGeneratorsTwonodeRel}. 

\item 
The upper-triangular system of equations found in~\eqref{UpperTri} in the end of section~\ref{TheProofSec}
enables to write every tangle in $\WJ\sub{\sIndex_1,\sIndex_2}(\nu)$ as a polynomial in the elements of the generating set~\eqref{MasterDiagramsWJTwonode}.
(The technical lemma~\ref{FiniteAndNonzeroLem} in appendix~\ref{TLRecouplingSect} shows that the coefficients in this expression are finite.)
Identity~\eqref{LoopErasure1} from lemma~\ref{CollectionLem} in appendix~\ref{TLRecouplingSect} gives relations~\eqref{MasterDiagramsWJTwonodeRel}.
(The technical lemma~\ref{FiniteAndNonzeroLem2} in appendix~\ref{TLRecouplingSect} shows that the coefficients in these relations are finite.)


Finally, we prove identity~\eqref{MasterDiagramsWJTwonodeUnitrel}. 
Because the left side of~\eqref{MasterDiagramsWJTwonodeUnitrel} is the unique unit element in $\WJ\sub{\sIndex_1,\sIndex_2}(\nu)$, 
we only need to prove that the right side of~\eqref{MasterDiagramsWJTwonodeUnitrel} is a unit in $\WJ\sub{\sIndex_1,\sIndex_2}(\nu)$ too. 
For this, it suffices to show that multiplication by it, either from the left, or from the right, 
of the generator tangles $\smash{\ValGenMWJ\super{s}_1}$ leaves them invariant. Indeed, we have 
\begin{align}
\left( \sum_{s \, \in \, \DefectSet\sub{\sIndex_1,\sIndex_2}} \frac{(-1)^s [s+1]}{\ThetaNet(\sIndex_1,\sIndex_2,s)} \, \ValGenMWJ\super{s}_1 \right)
\ValGenMWJ\super{s'}_1
\overset{\eqref{orthogonalidem}}{=} \ValGenMWJ\super{s'}_1 
\overset{\eqref{orthogonalidem}}{=} \ValGenMWJ\super{s'}_1 
\left( \sum_{s \, \in \, \DefectSet\sub{\sIndex_1,\sIndex_2}} \frac{(-1)^s [s+1]}{\ThetaNet(\sIndex_1,\sIndex_2,s)} \, \ValGenMWJ\super{s}_1 \right) ,
\end{align}
for all $s' \in \DefectSet\sub{\sIndex_1,\sIndex_2}$. This proves identity~\eqref{MasterDiagramsWJTwonodeUnitrel}.  
\end{enumerate}
From the explicit dimension~\eqref{DimOfWJTwoNode} of $\WJ\sub{\sIndex_1,\sIndex_2}(\nu)$
it is now clear that generators~\eqref{EGeneratorsTwonode} (resp.~\eqref{MasterDiagramsWJTwonode}) satisfy no other relations
and that each of the collections~(\ref{Basis1Twonode},~\ref{Basis2Twonode},~\ref{Basis3Twonode}) 
is a basis for $\WJ\sub{\sIndex_1,\sIndex_2}(\nu)$. This concludes the proof.
\end{proof}

We note that identity~\eqref{MasterDiagramsWJTwonodeUnitrel} is also a special case of~\cite[equation~\red{9.15}, page~\red{99}]{kl}.

\subsection{Case of three projectors where two indices equal one} \label{ThreeNodeCase11}

In this section, we consider $\multii = (\sIndex_1,\sIndex_2, \sIndex_3)$.
We further restrict two of the indices in $\multii$ to equal one.
In this special case, generators~\eqref{MasterDiagramsWJ-00} are particularly simple:
\begin{enumerate}[leftmargin=*]
\itemcolor{red}
\item \label{case1gen} If $\sIndex_1 = \sIndex_2 = 1$, then generators~\eqref{MasterDiagramsWJ-00} are 
\begin{align}
\label{Gen3nodefirst1}
\ValGenMWJ\super{0}_1 \quad  
= & \quad \vcenter{\hbox{\includegraphics[scale=0.275]{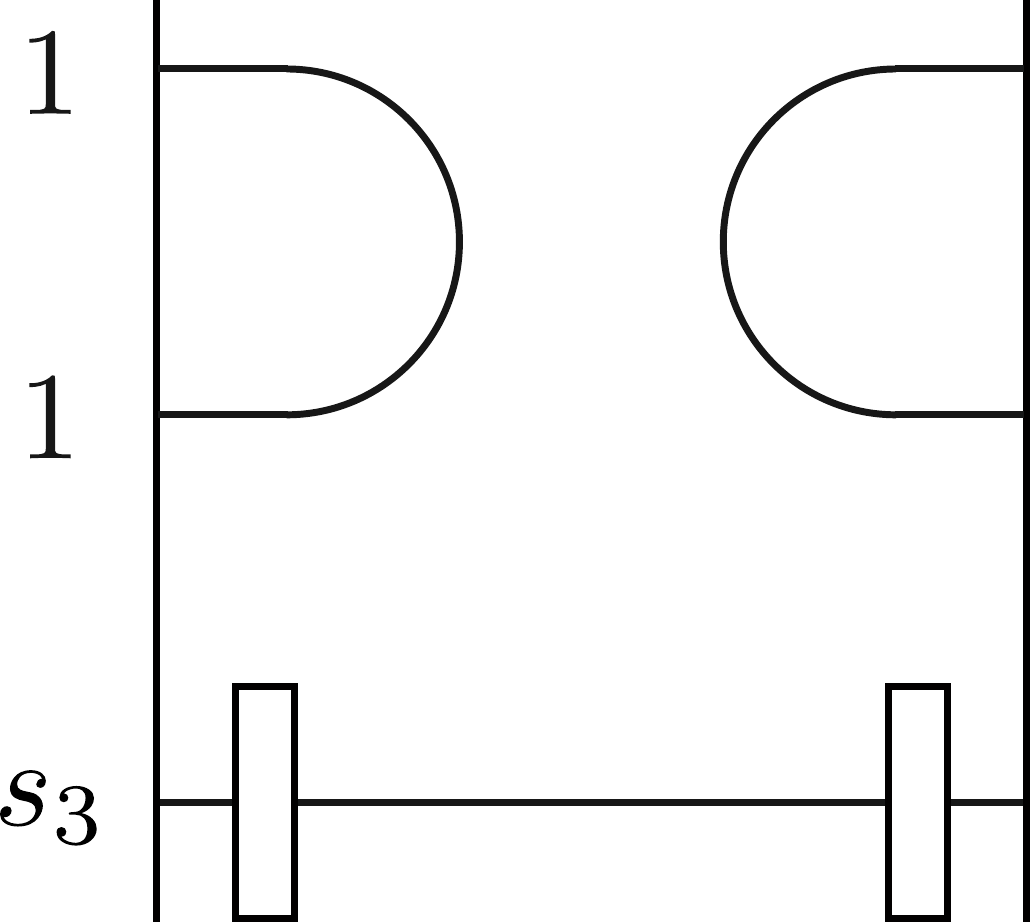} ,}} 
&& \ValGenMWJ\super{2}_1 \hphantom{{.}^{+1}} \quad
= \quad \vcenter{\hbox{\includegraphics[scale=0.275]{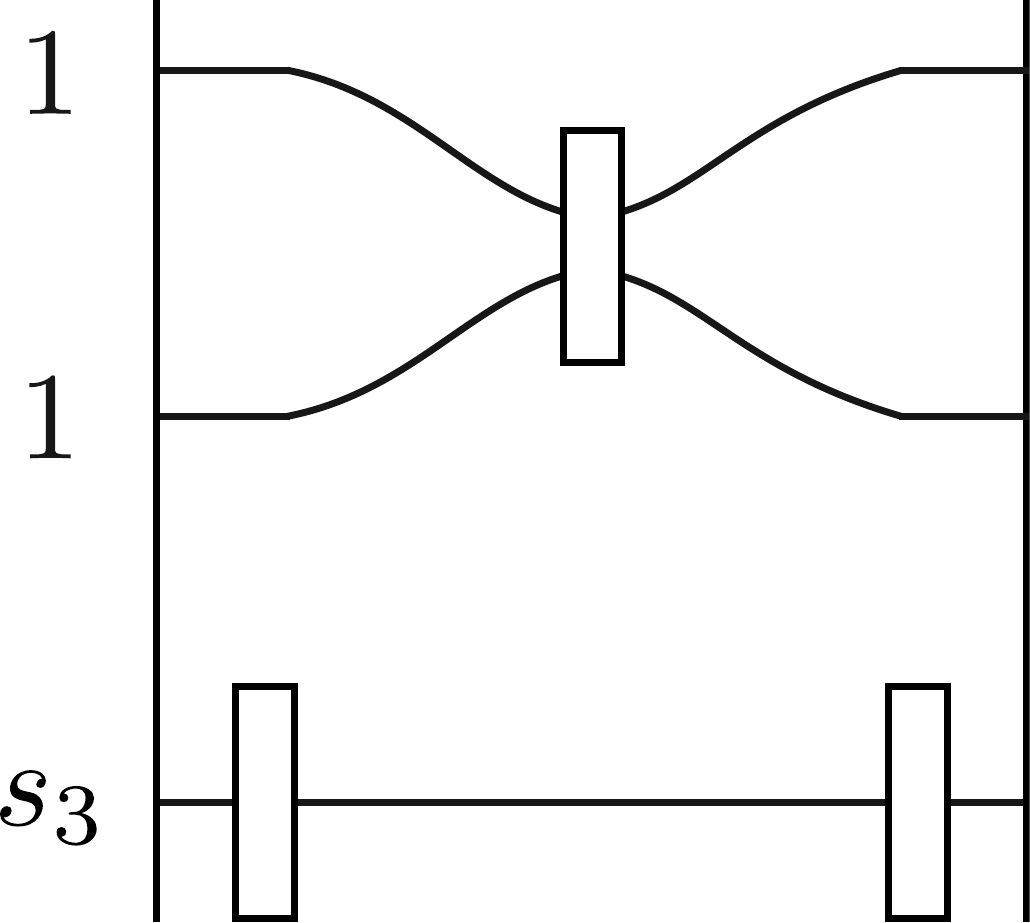} ,}} \\[1em]
\label{Gen3nodefirst2}
\ValGenMWJ\super{\sIndex_3-1}_2 \quad
= & \quad \vcenter{\hbox{\includegraphics[scale=0.275]{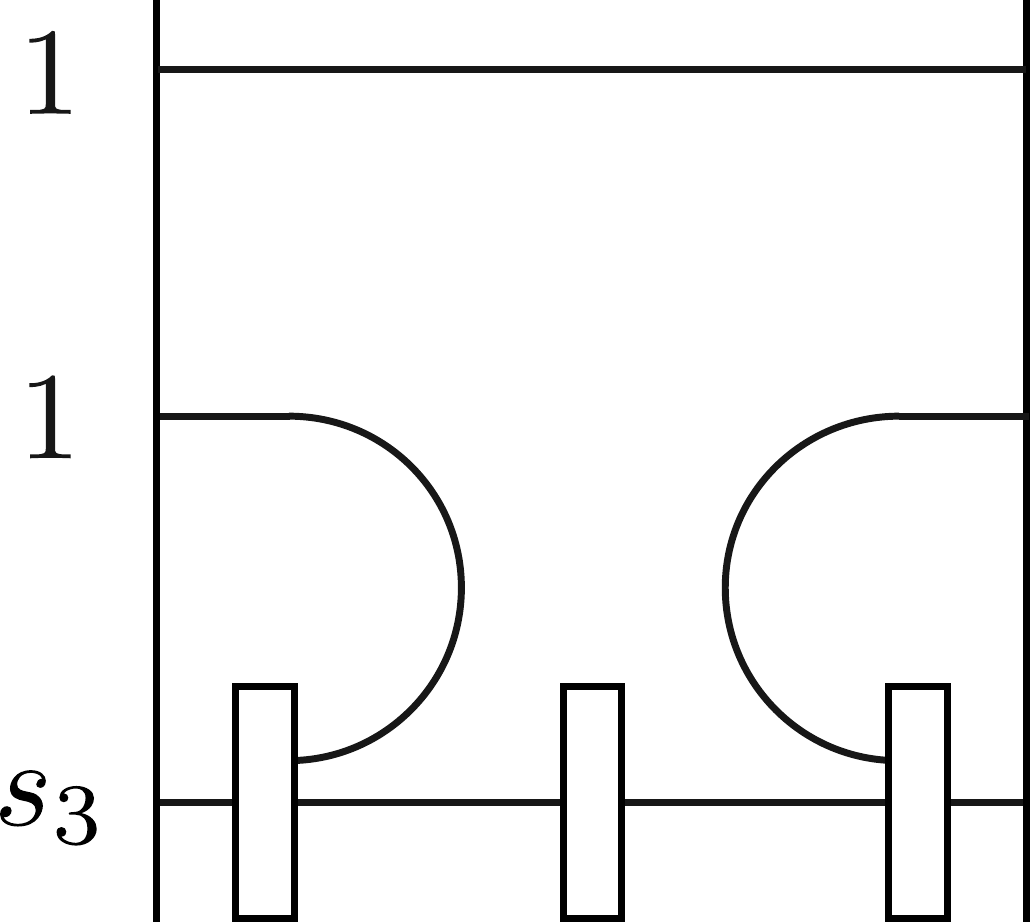} ,}}
&& \ValGenMWJ\super{\sIndex_3+1}_2 \quad
= \quad \vcenter{\hbox{\includegraphics[scale=0.275]{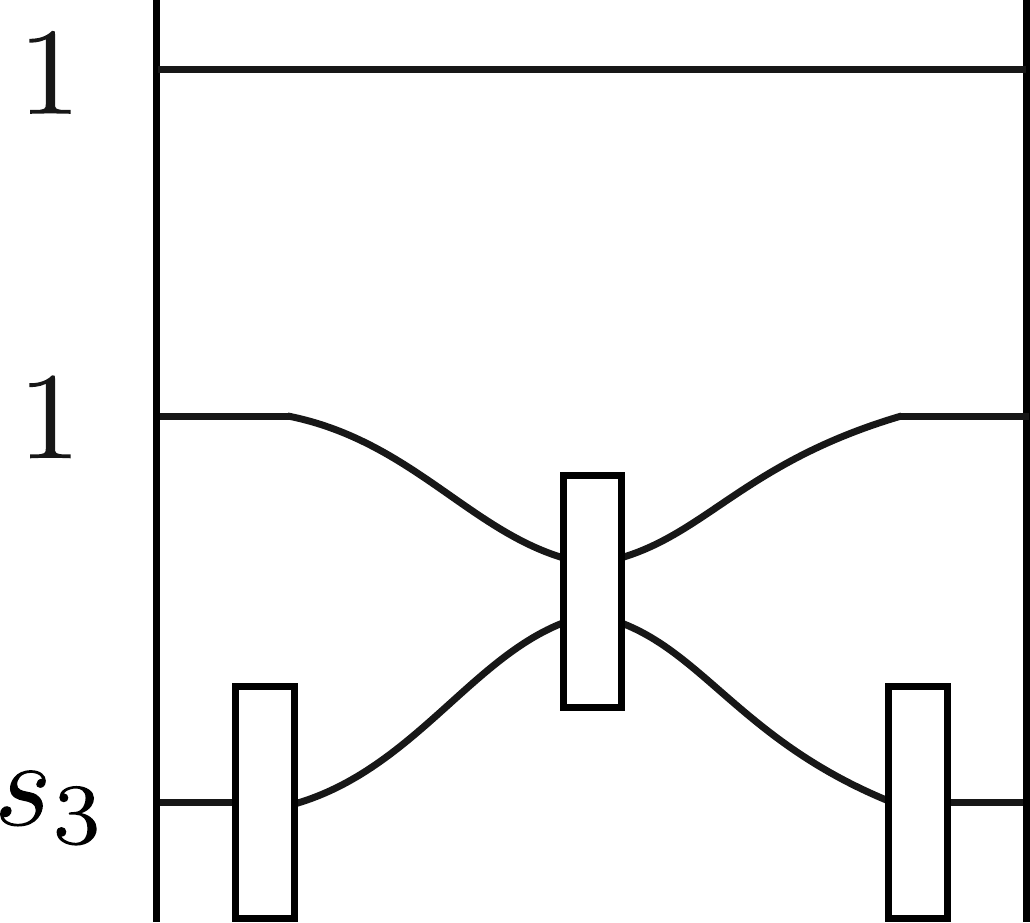} .}}
\end{align}

\item If $\sIndex_1 = \sIndex_3 = 1$, then generators~\eqref{MasterDiagramsWJ-00} are 
\begin{align} 
\label{Gen3nodemiddle1}
& \ValGenMWJ\super{\sIndex_2-1}_1 \quad  
= \quad \vcenter{\hbox{\includegraphics[scale=0.275]{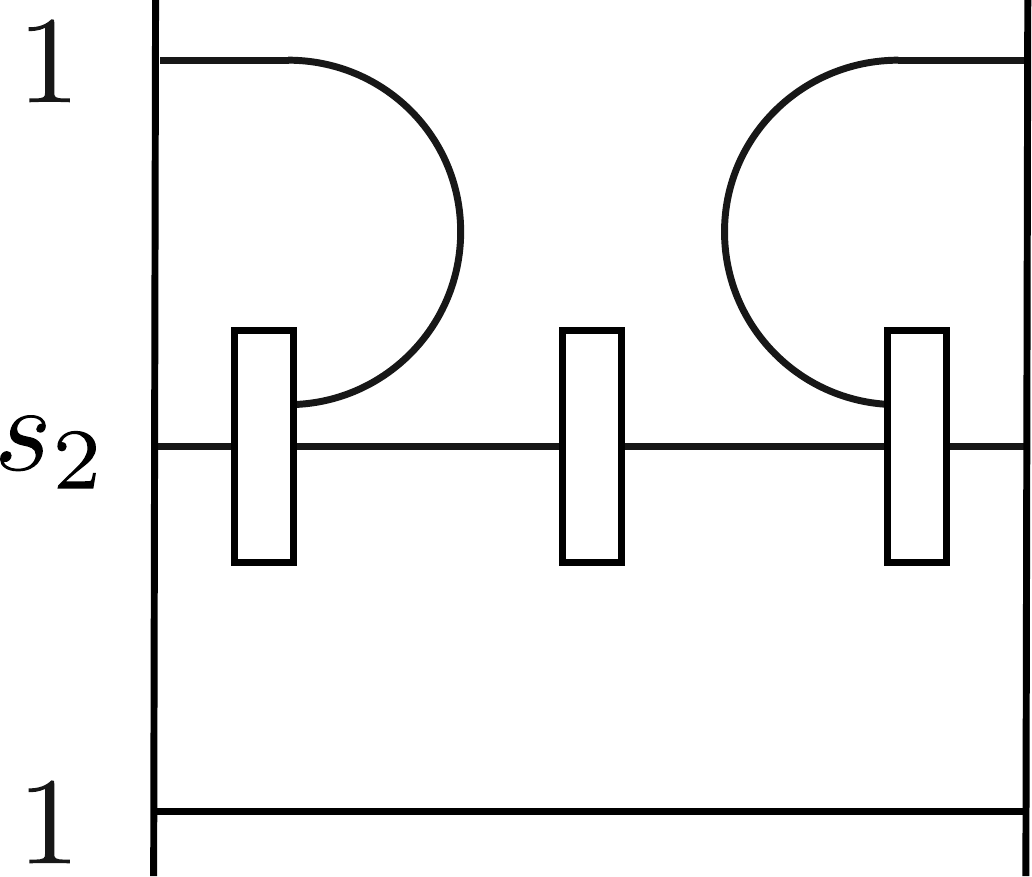} ,}} 
&& \ValGenMWJ\super{\sIndex_2+1}_1 \quad
= \quad \vcenter{\hbox{\includegraphics[scale=0.275]{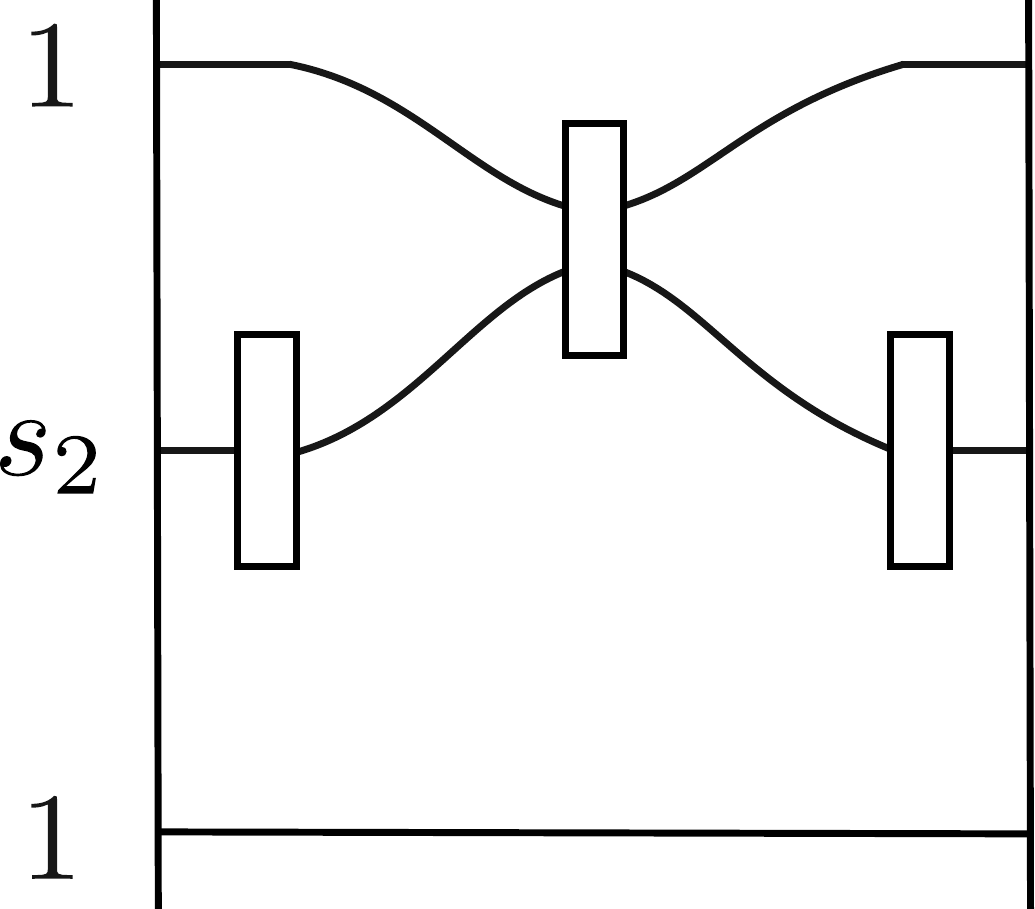} ,}} \\[1em]
\label{Gen3nodemiddle2}
& \ValGenMWJ\super{\sIndex_2-1}_2 \quad
= \quad \vcenter{\hbox{\includegraphics[scale=0.275]{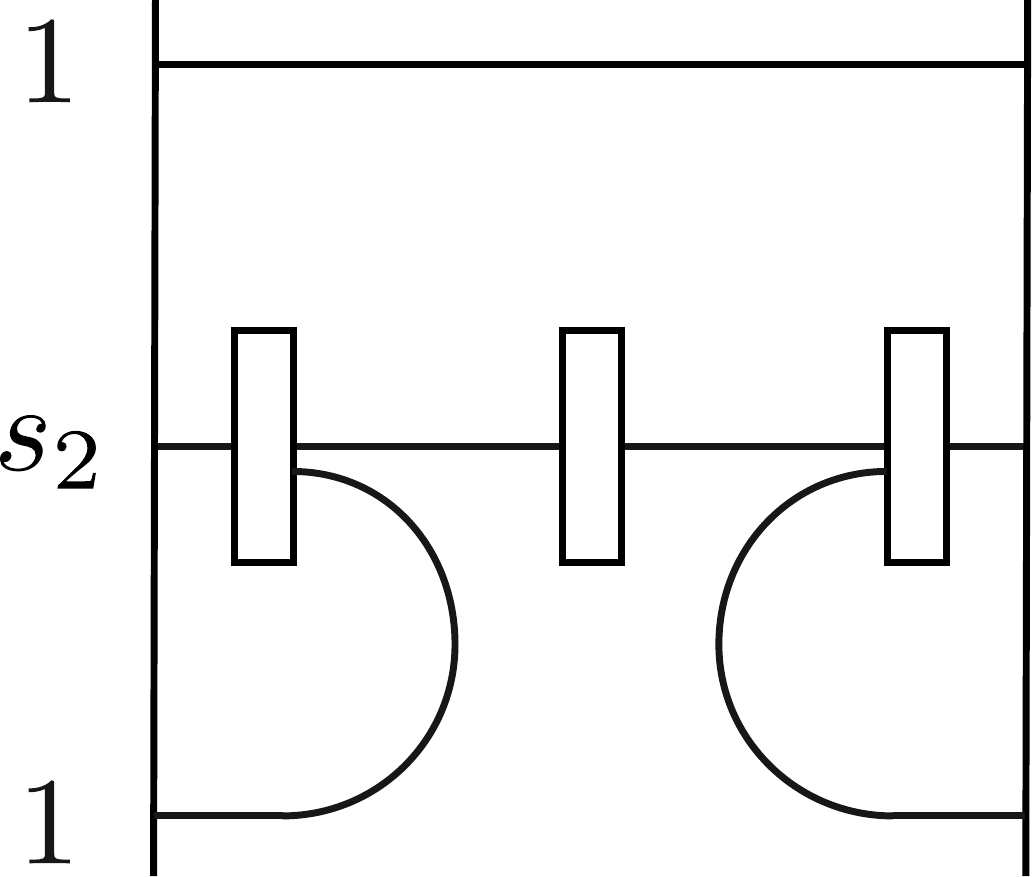} ,}}
&& \ValGenMWJ\super{\sIndex_2+1}_2 \quad
= \quad \vcenter{\hbox{\includegraphics[scale=0.275]{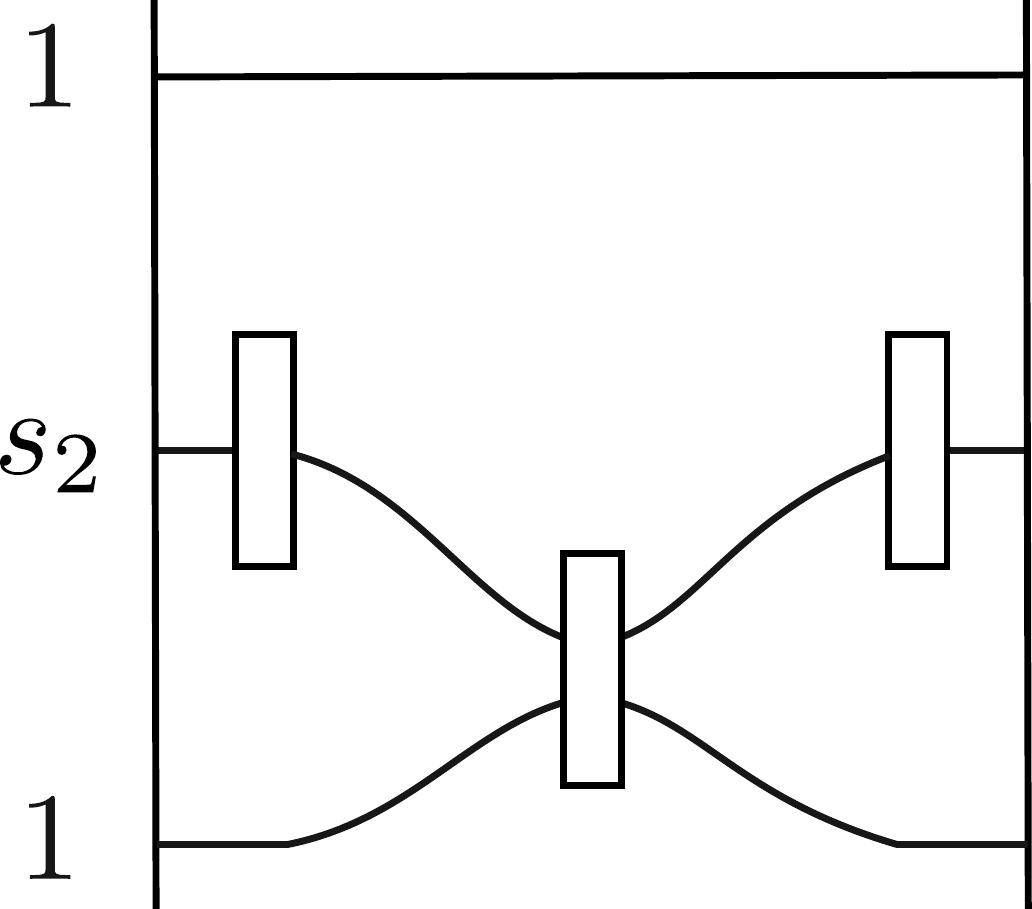} .}}
\end{align}


\item If $\sIndex_2 = \sIndex_3 = 1$, then generators~\eqref{MasterDiagramsWJ-00} are similar to those in case~\ref{case1gen},
but reflected with across a horizontal axis.
\end{enumerate}

We give a complete set of relations for these generators.

\begin{prop} \label{GeneratorThmThreeNode}
Suppose $\multii = (\sIndex_1,\sIndex_2, \sIndex_3) \in \bZpos^3$ and $\Summed_\multii < \ppmin(q)$.
If exactly two of the indices 
in $\multii$
equal one, then the Jones-Wenzl algebra $\WJ\sub{\sIndex_1,\sIndex_2,\sIndex_3}(\nu)$
has the following presentation in terms of generators and relations: 
it has four generators
\begin{align} \label{GeneratorThmThreeNodeGensabs}
\big\{ \ValGenMWJ\super{s}_1 \, \big| \, s \in \DefectSet\sub{\sIndex_1,\sIndex_2}\big\}
\cup
\big\{ \ValGenMWJ\super{s}_2 \, \big| \, s \in \DefectSet\sub{\sIndex_2,\sIndex_3} \big\}
\end{align}
which satisfy exclusively the following relations:
\begin{align}
\label{MasterDiagramsWJThreenodeUnitrel}
\mathbf{1}_{\WJ\sub{\sIndex_1,\sIndex_2,\sIndex_3}} 
& = \sum_{s \, \in \, \DefectSet\sub{\sIndex_i,\sIndex_{i+1}}} \frac{(-1)^s [s+1]}{\ThetaNet(\sIndex_i,\sIndex_{i+1},s)} \, \ValGenMWJ\super{s}_i , \\
\label{MasterDiagramsWJThreenodeSquareRel}
\ValGenMWJ\super{s}_i \ValGenMWJ\super{s'}_i 
& =
\delta_{s, s'} \frac{ \ThetaNet(\sIndex_i,\sIndex_{i+1},s) }{(-1)^s [s+1]} \, \ValGenMWJ\super{s}_i , \qquad \textnormal{for all } s, s' \, \in \, \DefectSet\sub{\sIndex_i,\sIndex_{i+1}} ,
\end{align}
for $i=1,2$, and
\begin{enumerate} 
\itemcolor{red}
\item \label{MasterDiagramsWJThreenodeItem1}
if $\sIndex_1 = \sIndex_2 = 1$, 
\begin{align} \label{MasterDiagramsWJThreenodeRel1}
\ValGenMWJ\super{0}_1 \ValGenMWJ\super{\sIndex_3-1}_2 \ValGenMWJ\super{0}_1 & = \ValGenMWJ\super{0}_1 ,
\end{align}

\item \label{MasterDiagramsWJThreenodeItem2}
if $\sIndex_1 = \sIndex_3 = 1$, 
\begin{align} \label{MasterDiagramsWJThreenodeRel2}
\ValGenMWJ\super{\sIndex_2-1}_1 \ValGenMWJ\super{\sIndex_2-1}_2 \ValGenMWJ\super{\sIndex_2-1}_1
- \ValGenMWJ\super{\sIndex_2-1}_2 \ValGenMWJ\super{\sIndex_2-1}_1 \ValGenMWJ\super{\sIndex_2-1}_2
& = \frac{1}{[\sIndex_2]^2}  \, \big( \ValGenMWJ\super{\sIndex_2-1}_1 - \ValGenMWJ\super{\sIndex_2-1}_2 \big) ,
\end{align}

\item \label{MasterDiagramsWJThreenodeItem3}
if $\sIndex_2 = \sIndex_3 = 1$, 
\begin{align} \label{MasterDiagramsWJThreenodeRel3}
\ValGenMWJ\super{0}_2 \ValGenMWJ\super{\sIndex_1-1}_1 \ValGenMWJ\super{0}_2 & = \ValGenMWJ\super{0}_2 .
\end{align}
\end{enumerate} 
In particular, each of the following sets forms a basis for $\WJ\sub{\sIndex_1,\sIndex_2,\sIndex_3}(\nu)$:
\begin{align} 
\label{Basis1Threenode}
& \big\{ \mathbf{1}_{\WJ\sub{\sIndex_1,\sIndex_2,\sIndex_3}}  , \; \ValGenMWJ\super{0}_1 , \; \ValGenMWJ\super{\sIndex_3-1}_2 , \;
\ValGenMWJ\super{0}_1 \ValGenMWJ\super{\sIndex_3-1}_2 , \;
\ValGenMWJ\super{\sIndex_3-1}_2 \ValGenMWJ\super{0}_1 , \;
\ValGenMWJ\super{\sIndex_3-1}_2 \ValGenMWJ\super{0}_1 \ValGenMWJ\super{\sIndex_3-1}_2  \big\} , \quad && \textnormal{if } \sIndex_1 = \sIndex_2 = 1 , \\
\label{Basis2Threenode}
& \big\{ \mathbf{1}_{\WJ\sub{\sIndex_1,\sIndex_2,\sIndex_3}}  , \; \ValGenMWJ\super{\sIndex_2-1}_1 , \; \ValGenMWJ\super{\sIndex_2-1}_2 , \;
\ValGenMWJ\super{\sIndex_2-1}_1 \ValGenMWJ\super{\sIndex_2-1}_2 , \;
\ValGenMWJ\super{\sIndex_2-1}_2 \ValGenMWJ\super{\sIndex_2-1}_1 , \;
\ValGenMWJ\super{\sIndex_2-1}_2 \ValGenMWJ\super{\sIndex_2-1}_1 \ValGenMWJ\super{\sIndex_2-1}_2 \big\} , \quad && \textnormal{if } \sIndex_1 = \sIndex_3 = 1 , \\
\label{Basis3Threenode}
& \big\{ \mathbf{1}_{\WJ\sub{\sIndex_1,\sIndex_2,\sIndex_3}}  , \; \ValGenMWJ\super{\sIndex_2-1}_1 , \; \ValGenMWJ\super{\sIndex_2-1}_2 , \;
\ValGenMWJ\super{\sIndex_2-1}_1 \ValGenMWJ\super{\sIndex_2-1}_2 , \;
\ValGenMWJ\super{\sIndex_2-1}_2 \ValGenMWJ\super{\sIndex_2-1}_1 , \;
\ValGenMWJ\super{\sIndex_2-1}_2 \ValGenMWJ\super{\sIndex_2-1}_1 \ValGenMWJ\super{\sIndex_2-1}_2  \big\} , \quad && \textnormal{if } \sIndex_2 = \sIndex_3 = 1 .
\end{align}
\end{prop}

\begin{proof}
We first find the dimension of $\WJ\sub{\sIndex_1,\sIndex_2,\sIndex_3}(\nu)$.
Because two of the indices $\sIndex_1, \sIndex_2, \sIndex_3$ equal one, we have
\begin{align} \label{DefSetThreeNode}
\DefectSet\sub{\sIndex_1,\sIndex_2,\sIndex_3} = \{ s-2, s, s+2 \} ,
\end{align}
with $s = \max(\sIndex_1,\sIndex_2,\sIndex_3)$. 
Therefore, a simple calculation with~(\ref{SpecialDefSet},~\ref{DimOfWJ},~\ref{PreRecursion2}) gives
\begin{align} \label{DimOfWJThreeNode}
\dim \WJ\sub{\sIndex_1,\sIndex_2,\sIndex_3}(\nu) 
= 6.
\end{align}

Item~\ref{GeneratorThmItem2} of theorem~\ref{GeneratorThm} shows that the asserted generators do generate 
the Jones-Wenzl algebra $\WJ\sub{\sIndex_1,\sIndex_2,\sIndex_3}(\nu)$. Next, we prove the asserted relations for the generators.
Relation~\eqref{MasterDiagramsWJThreenodeSquareRel} follows already from proposition~\ref{GeneratorThmTwoNode}. 
Relation~\eqref{MasterDiagramsWJThreenodeUnitrel} follows from~\cite[equation~\red{9.15}, page~\red{99}]{kl},
or can also be proven easily by a direct calculation.

For the other relations, relation~\eqref{MasterDiagramsWJThreenodeRel1}  in case~\ref{MasterDiagramsWJThreenodeItem1}
is straightforward:
\begin{align}
\ValGenMWJ\super{0}_1 \ValGenMWJ\super{\sIndex_3-1}_2 \ValGenMWJ\super{0}_1 \quad
\overset{(\ref{Gen3nodefirst1}-\ref{Gen3nodefirst2})}{=} \; & \quad \vcenter{\hbox{\includegraphics[scale=0.275]{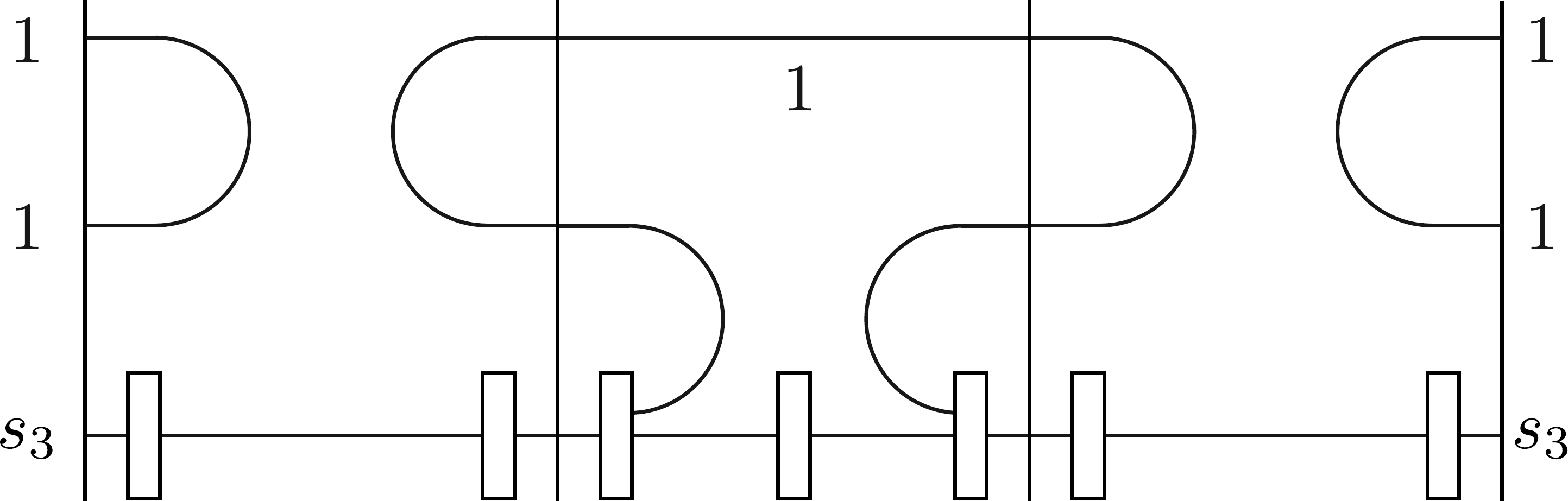}}} 
\quad \overset{\eqref{ProjectorID1}}{=} \quad \ValGenMWJ\super{0}_1 ,
\end{align}
and case~\ref{MasterDiagramsWJThreenodeItem3} is similar. 
In the remaining case~\ref{MasterDiagramsWJThreenodeItem2}, we first expand the product
$\ValGenMWJ\super{\sIndex_2-1}_1 \ValGenMWJ\super{\sIndex_2-1}_2 \ValGenMWJ\super{\sIndex_2-1}_1$ as
\begin{align}
\; &  \ValGenMWJ\super{\sIndex_2-1}_1 \ValGenMWJ\super{\sIndex_2-1}_2 \ValGenMWJ\super{\sIndex_2-1}_1 
\quad \overset{(\ref{Gen3nodemiddle1}-\ref{Gen3nodemiddle2})}{=} \quad \vcenter{\hbox{\includegraphics[scale=0.275]{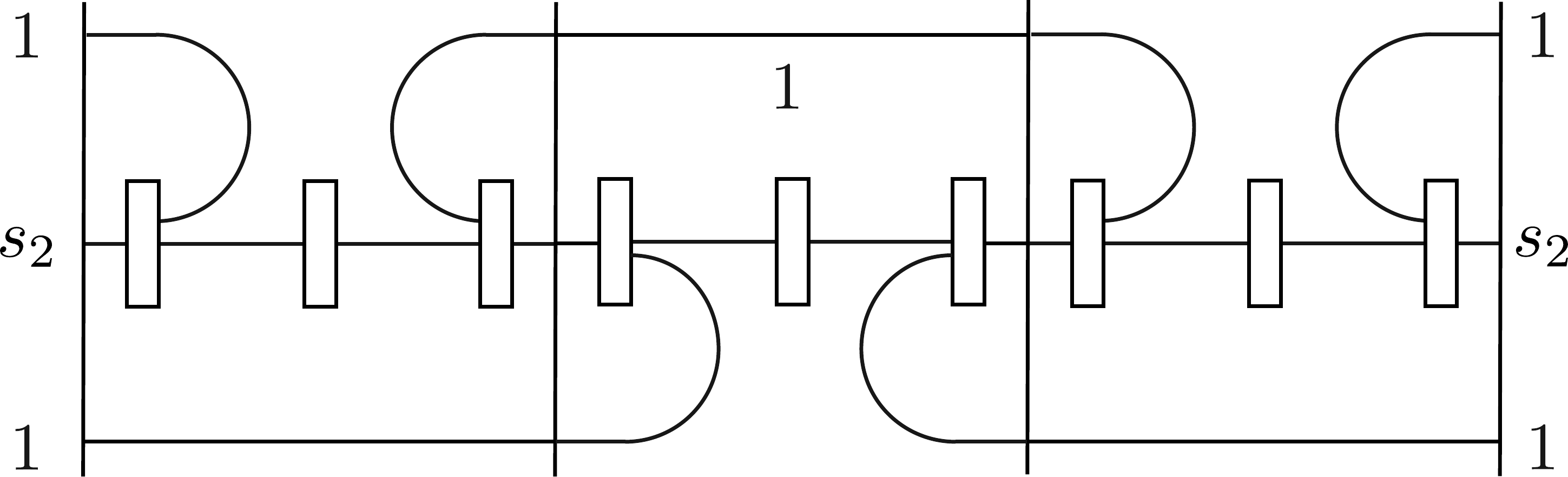}}} \\[1em]
\quad \underset{\eqref{SpecialT}}{\overset{\eqref{ProjectorID1}}{=}} \; & \quad
\vcenter{\hbox{\includegraphics[scale=0.275]{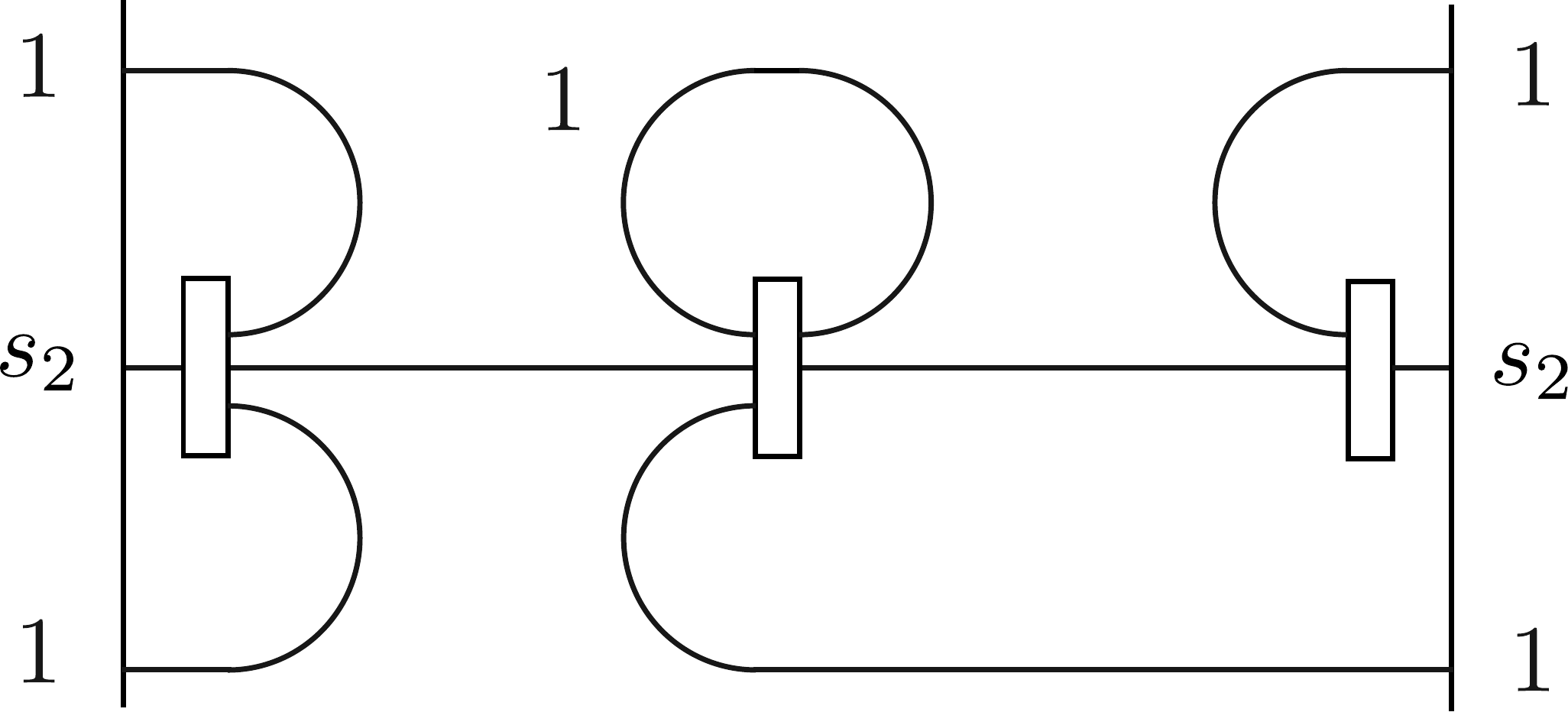}}} \quad + \quad
\frac{1}{[\sIndex_2]} \,\, \times \,\, \vcenter{\hbox{\includegraphics[scale=0.275]{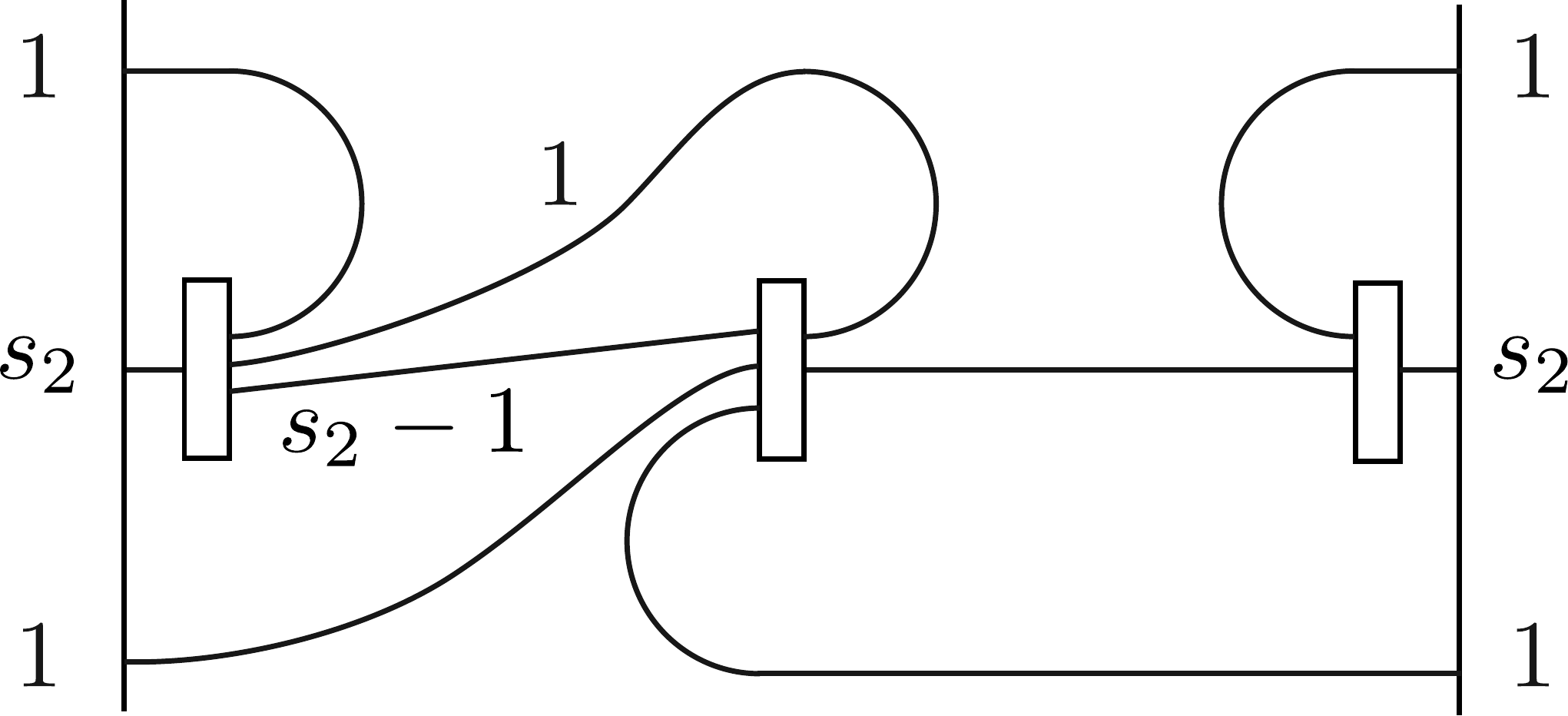} ,}}
\label{collection1}
\end{align}
where we 
used proposition~\ref{SpecialTProp} from appendix~\ref{WJProjApp} to find the coefficient of the second term.
Then, we use identity~\eqref{DeltaTangleGen} from lemma~\ref{CollectionLem} to simplify the first term in~\eqref{collection1},
and decompose in the second term in~\eqref{collection1} the remaining box into its internal link diagrams, 
with coefficients from proposition~\ref{SpecialTProp}. Collecting all of the terms, we obtain
\begin{align}
\frac{1}{[\sIndex_2]^2} \,\, \times \,\, \vcenter{\hbox{\includegraphics[scale=0.275]{e-Relations3node5.pdf}}} \quad + \quad
\left( -\frac{[\sIndex_2+1]}{[\sIndex_2]} + \frac{[2]}{[\sIndex_2][\sIndex_2+1]} \right) 
\,\, \times \,\, \vcenter{\hbox{\includegraphics[scale=0.275]{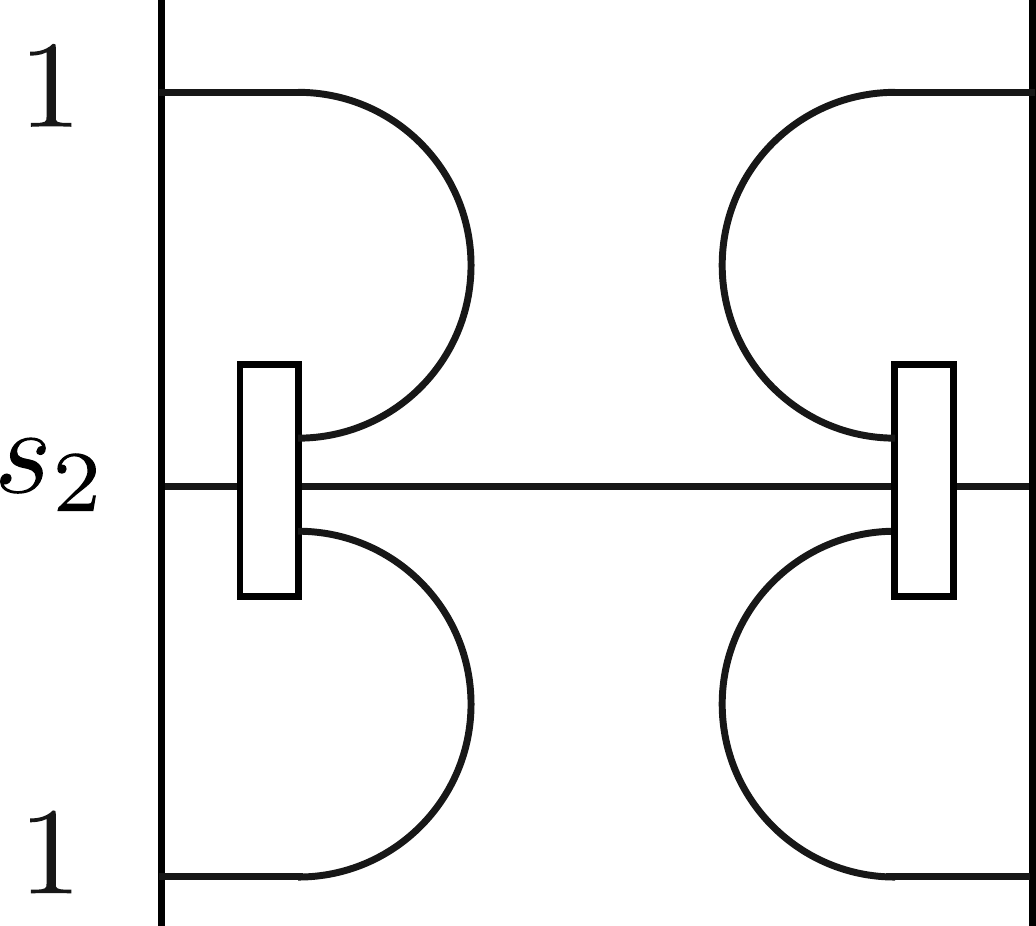} .}}
\label{collection2}
\end{align}
Expanding 
$\ValGenMWJ\super{\sIndex_2-1}_2 \ValGenMWJ\super{\sIndex_2-1}_1 \ValGenMWJ\super{\sIndex_2-1}_2$
similarly, we obtain asserted relation~\eqref{MasterDiagramsWJThreenodeRel2}:
\begin{align}
\nonumber
\; & \ValGenMWJ\super{\sIndex_2-1}_1 \ValGenMWJ\super{\sIndex_2-1}_2 \ValGenMWJ\super{\sIndex_2-1}_1
- \ValGenMWJ\super{\sIndex_2-1}_2 \ValGenMWJ\super{\sIndex_2-1}_1 \ValGenMWJ\super{\sIndex_2-1}_2  \\[1em]
\quad \overset{(\ref{collection1}-\ref{collection2})}{=} \; & \quad \frac{1}{[\sIndex_2]^2} \,\, \times \,\, \vcenter{\hbox{\includegraphics[scale=0.275]{e-Relations3node5.pdf}}} \quad - \quad
\frac{1}{[\sIndex_2]^2} \,\, \times \,\, \vcenter{\hbox{\includegraphics[scale=0.275]{e-Relations3node7.pdf} .}} 
\end{align}

Finally, because $\dim \WJ\sub{\sIndex_1,\sIndex_2,\sIndex_3}(\nu) = 6$ by~\eqref{DimOfWJThreeNode},
it is now clear that generators~\eqref{GeneratorThmThreeNodeGensabs} satisfy no other relations
and that each of the collections~(\ref{Basis1Threenode},~\ref{Basis2Threenode},~\ref{Basis3Threenode}) 
is a basis for $\WJ\sub{\sIndex_1,\sIndex_2,\sIndex_3}(\nu)$. This concludes the proof.
\end{proof}

\begin{proof}[Proof of proposition~\ref{RelationProp}]
Relations~\eqref{WordRelationsWJ04} and~\eqref{WordRelationsWJ06} 
are obvious from the definition of tangle multiplication. 
Relation~\eqref{WordRelationsWJ05} follows from~\eqref{MasterDiagramsWJTwonodeRel} 
in theorem~\ref{GeneratorThmTwoNode}, and relation~\eqref{WordRelationsWJ07}, e.g., from~\cite[equation~\red{9.15}, page~\red{99}]{kl}.
The other asserted relations 
are consequences of proposition~\ref{GeneratorThmThreeNode}:
If $\sIndex_1 = \sIndex_2 = 1$, then $\smash{\ValGenMWJ\super{0}_1} = \ValGenWJ_1$ and 
$\smash{\ValGenMWJ\super{\sIndex_3-1}_2} = \ValGenWJ_2$,
and a similar observation holds if $\sIndex_2 = \sIndex_3 = 1$.
Also, if $\sIndex_1 = \sIndex_3 = 1$, then
$\smash{\ValGenMWJ\super{\sIndex_2-1}_1} = \ValGenWJ_1$ and $\smash{\ValGenMWJ\super{\sIndex_2-1}_2} = \ValGenWJ_2$.
Therefore, asserted relations~(\ref{WordRelationsWJ01}--\ref{WordRelationsWJ03})
indeed follow from proposition~\ref{GeneratorThmThreeNode}. 
This proves proposition~\ref{RelationProp}.
\end{proof}

%

%

\bigskip
\bigskip

\appendixpage

\begin{appendices}
\renewcommand{\thesection}{\Alph{section}}
\renewcommand{\thesubsection}{\arabic{subsection}}
\renewcommand{\thesubsubsection}{\Alph{subsubsection}}

\section{Coefficients of the Jones-Wenzl projector} \label{WJProjApp}

In this appendix, we derive formulas for the coefficients of the Jones-Wenzl projector in expansion~\eqref{ProjDecomp}:
\begin{align}  \label{WJExpansion} 
\WJProj\sub{n} \overset{\eqref{ProjDecomp}}{=} \sum_{T \, \in \, \LD_n} (\text{coef}_T) \, T. 
\end{align}  
In corollary~\ref{WJCoeffCor}, we determine the coefficients $\text{coef}_T$ of this expansion and in proposition~\ref{SpecialTProp},
we give explicit formulas in important special cases used frequently in this article. 

Using the basis $\smash{\LP_m\super{s}}$ of $(m,s)$-link patterns, we define the Gram matrix $\smash{\Gram_m\super{s}}$ 
of the bilinear form~\eqref{LSBiForm} on $\smash{\LS_n\super{s}}$ as
\begin{align}\label{GramMatrix2}
[ \Gram_m\super{s} ]_{\alpha, \beta} := \BiForm{\alpha}{\beta} , \quad 
\text{for all $\alpha, \beta \in \LP_m\super{s}$.}
\end{align}
When $s = 0$ and $m = 2n$ is even, it is customary to call the Gram matrix $\smash{\Gram_{2n}\super{0}}$ the \emph{meander matrix}~\cite{fgg}.
We will derive formulas for certain entries of the inverse of this meander matrix.
From these formulas, we obtain the coefficients of the Jones-Wenzl projector.
The former appear to us to be new, but the latter have been derived already in the article~\cite{sm} of S.~Morrison.
We also mention that in~\cite[corollary~\red{3.7}]{gl2}, J.~Graham and G.~Lehrer obtain a closed formula for the Jones-Wenzl 
projector in the exceptional cases when $q$ in~\eqref{fugacity} is a root of unity.

In lemma~\ref{InvMeandLem}, we show that the coefficients  $\text{coef}_T$ in~\eqref{WJExpansion}
are given by certain entries of the inverse of the 
meander matrix $\smash{\Gram_{2n}\super{0}}$. 
We determine these matrix entries in lemma~\ref{ExplicitLem}.

\subsection{Preliminary observations}

As a special case of~\cite[lemma~\red{2.6}]{fp0} 
the linear extension of the following map is an isomorphism of vector spaces from 
the Temperley-Lieb algebra $\TL_n(\nu)$ to the space $\smash{\LS_{2n}\super{0}}$ of link states without defects:

\begin{align}  \label{TLBij} 
T \quad = \quad \vcenter{\hbox{\includegraphics[scale=0.275]{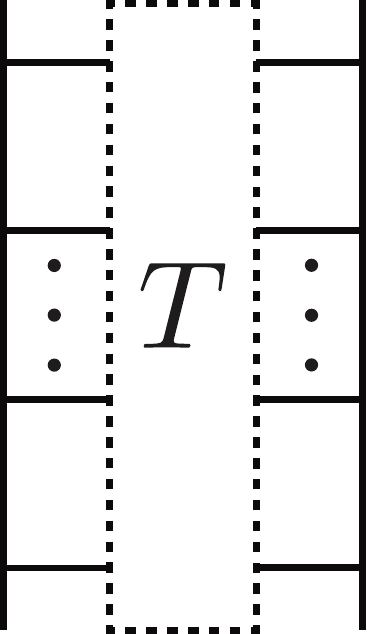}}}
\quad \qquad  \longmapsto \qquad \quad
\alpha_T \quad := \quad \raisebox{-25pt}{\includegraphics[scale=0.275]{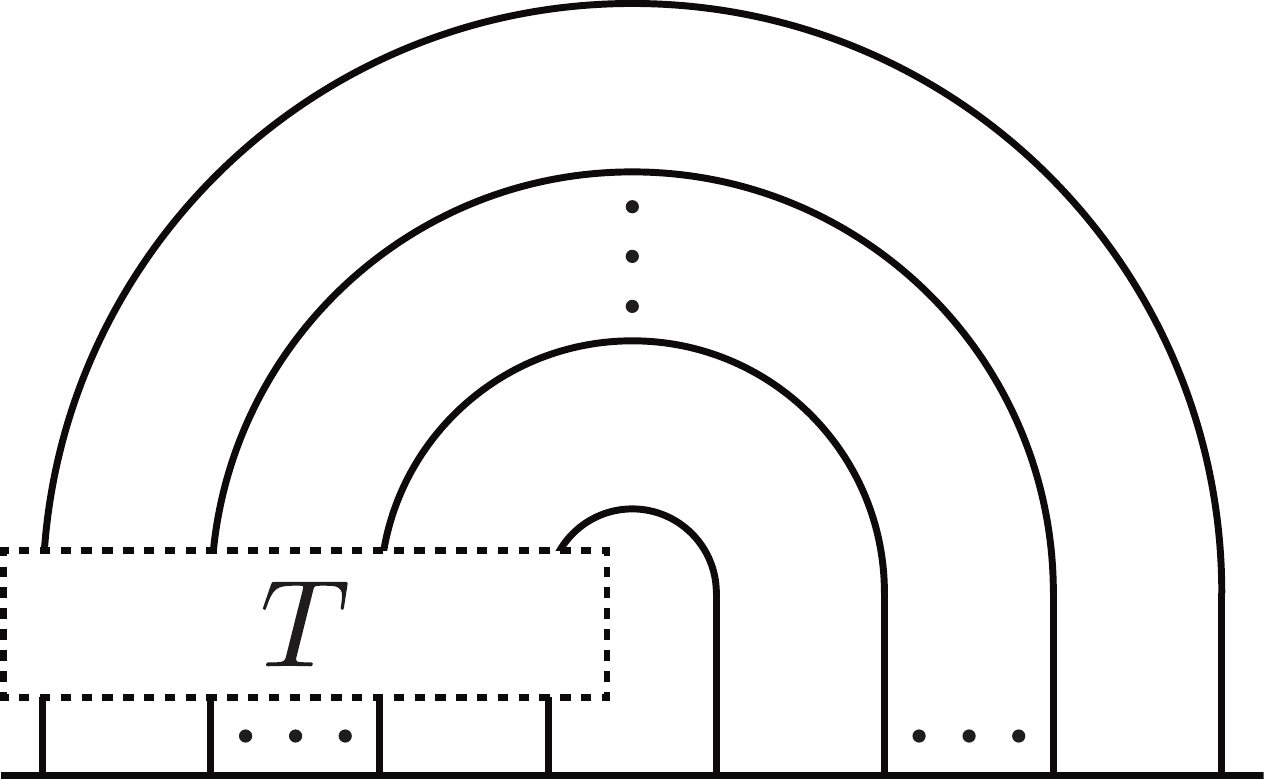} .}
\end{align}  

To begin, we prove that the following \emph{rainbow link pattern} behaves very nicely under the isomorphism~\eqref{TLBij}:
\begin{align}  \label{CapID} 
\Cap_n := \alpha_{\mathbf{1}_{\TL_n}} \quad = \quad \vcenter{\hbox{\includegraphics[scale=0.275]{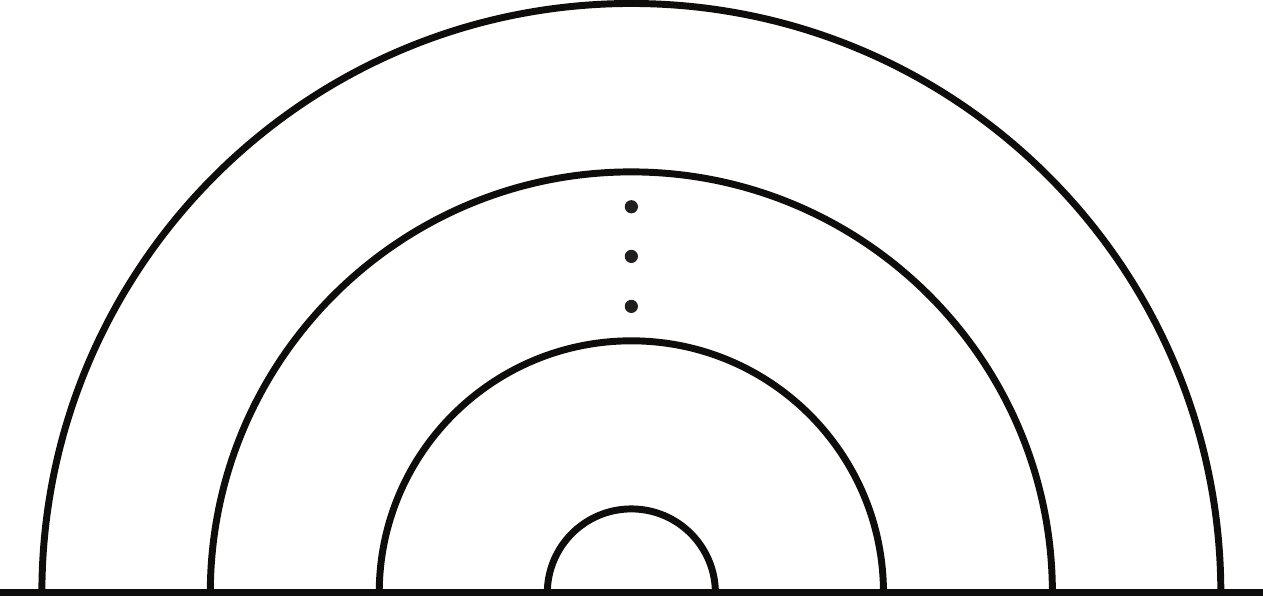} .}} 
\end{align}   

\begin{lem} \label{CapDualLem} 
Suppose $n+1 < \ppmin(q)$.  Then $\rad \LS_{2n} = \{0\}$, and the bijection~\eqref{TLBij} sends 
$\WJProj\sub{n}$ to $ (-1)^n[n+1] \Cap_n\superscr{\cheque}$.  Equivalently, we have
\begin{align}  \label{NoDiagExpressions} 
(-1)^n[n+1]\Cap_n\superscr{\cheque} = (\WJProj\sub{n} \otimes \mathbf{1}_{\TL_n}) \Cap_n 
= ( \mathbf{1}_{\TL_n} \otimes \WJProj\sub{n} ) \Cap_n = ( \WJProj\sub{n} \otimes \WJProj\sub{n} ) \Cap_n , 
\end{align}  
or in terms of diagrams,
\begin{align} 
\label{DiagExpressions-0}
\raisebox{-15pt}{$(-1)^n [n+1]$} \quad \vcenter{\hbox{\includegraphics[scale=0.275]{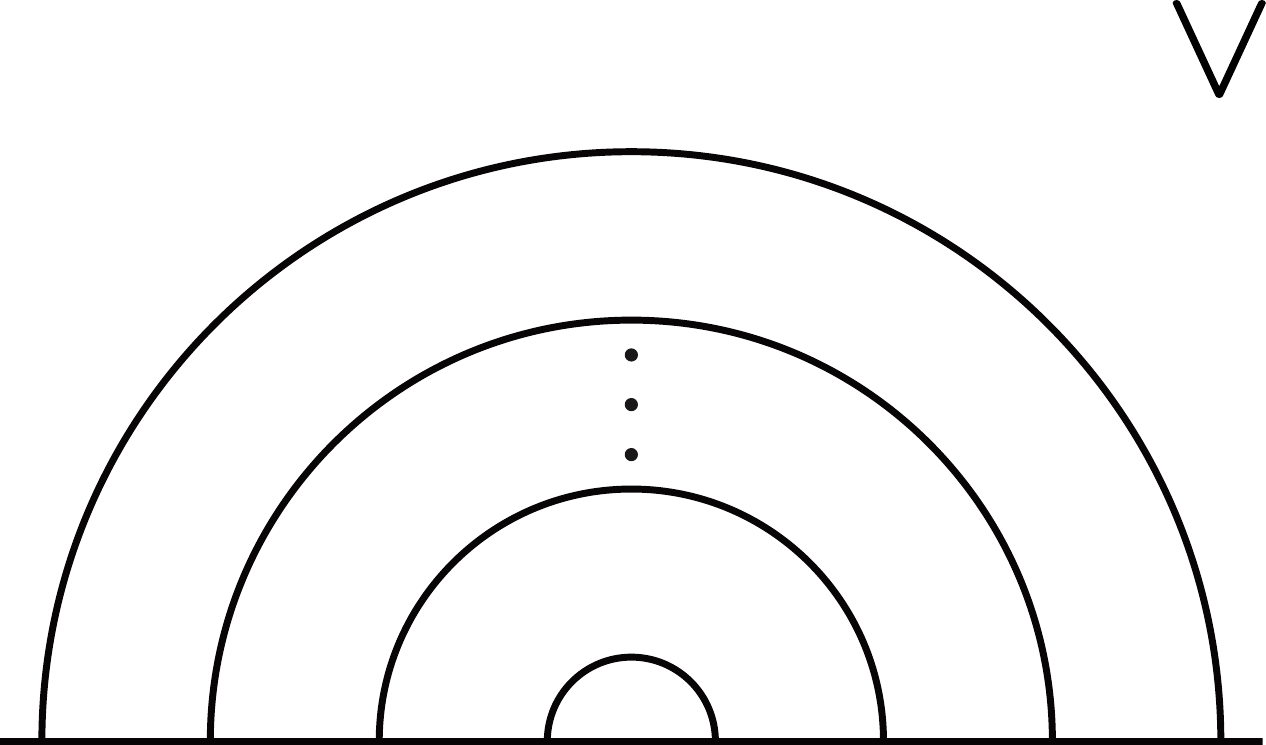}}} 
& \;\quad \raisebox{-15pt}{$=$} \quad\quad \vcenter{\hbox{\includegraphics[scale=0.275]{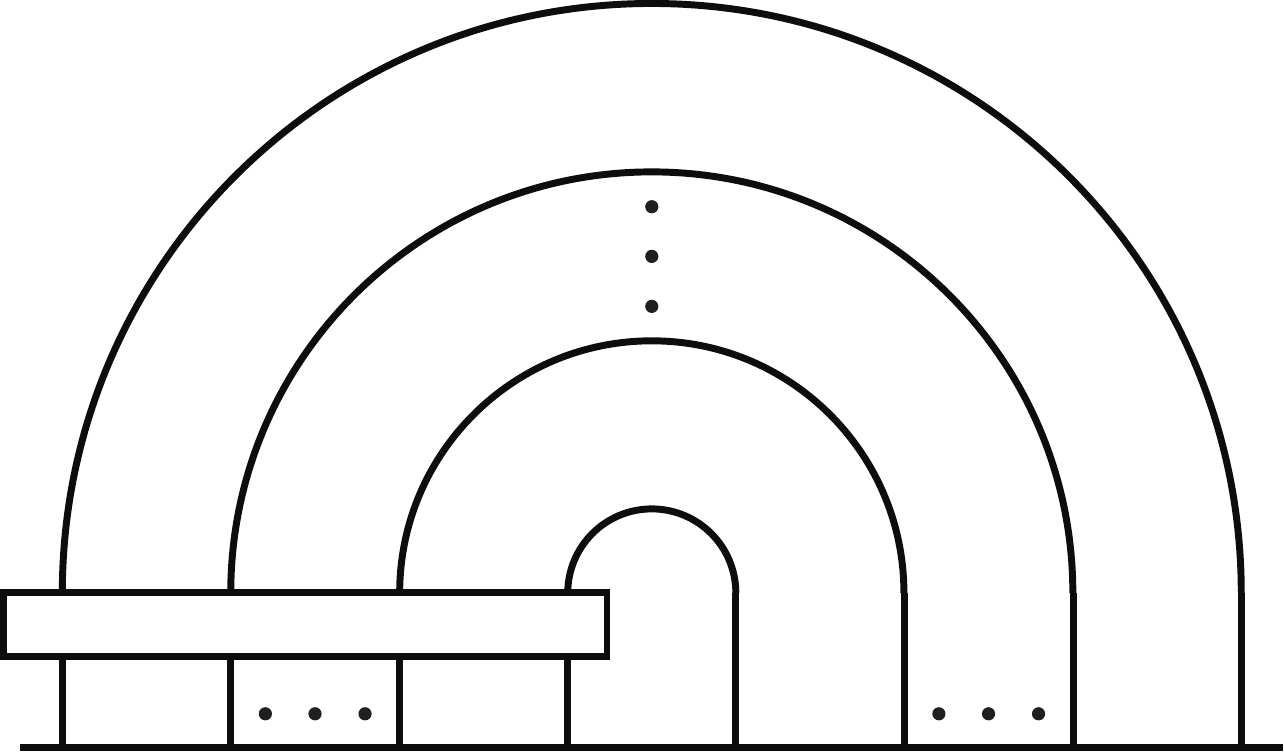}}} \\[1em]
\label{DiagExpressions-1} 
& \quad\quad \raisebox{-15pt}{$=$}  \quad\quad \vcenter{\hbox{\includegraphics[scale=0.275]{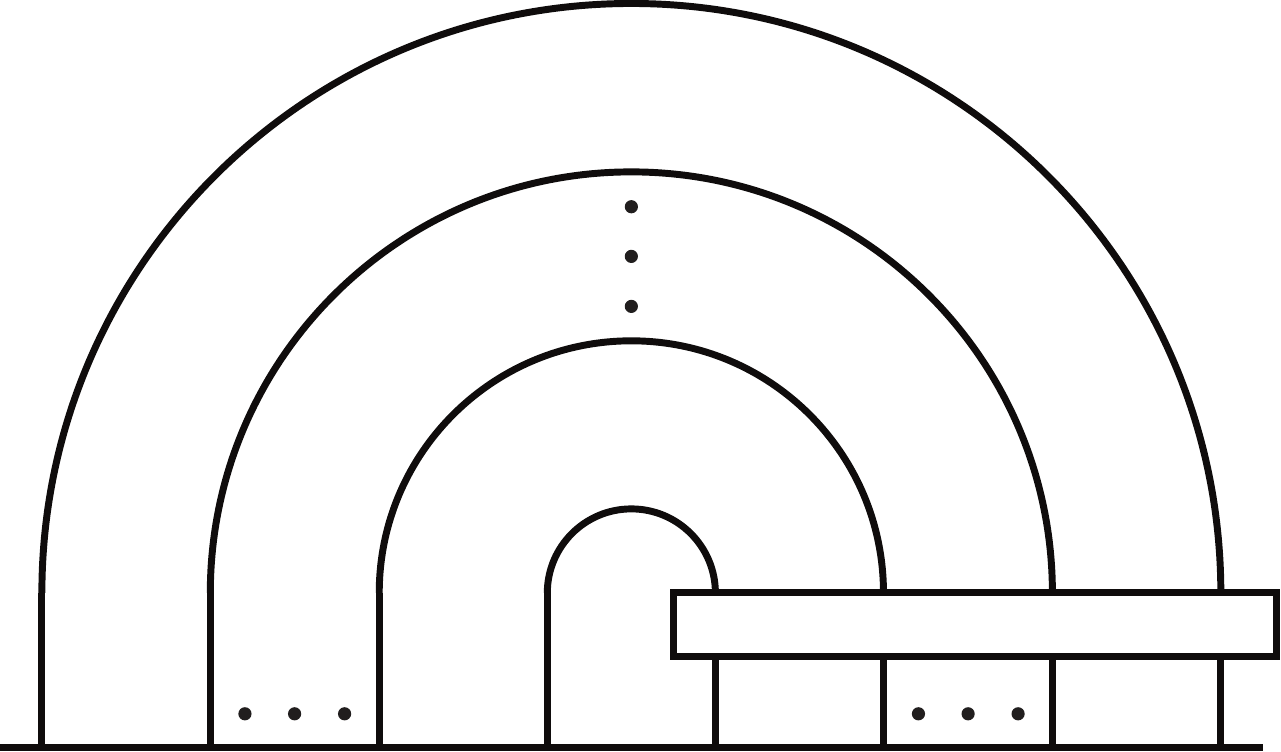}}} 
\quad\quad \raisebox{-15pt}{$=$}  \quad\quad \vcenter{\hbox{\includegraphics[scale=0.275]{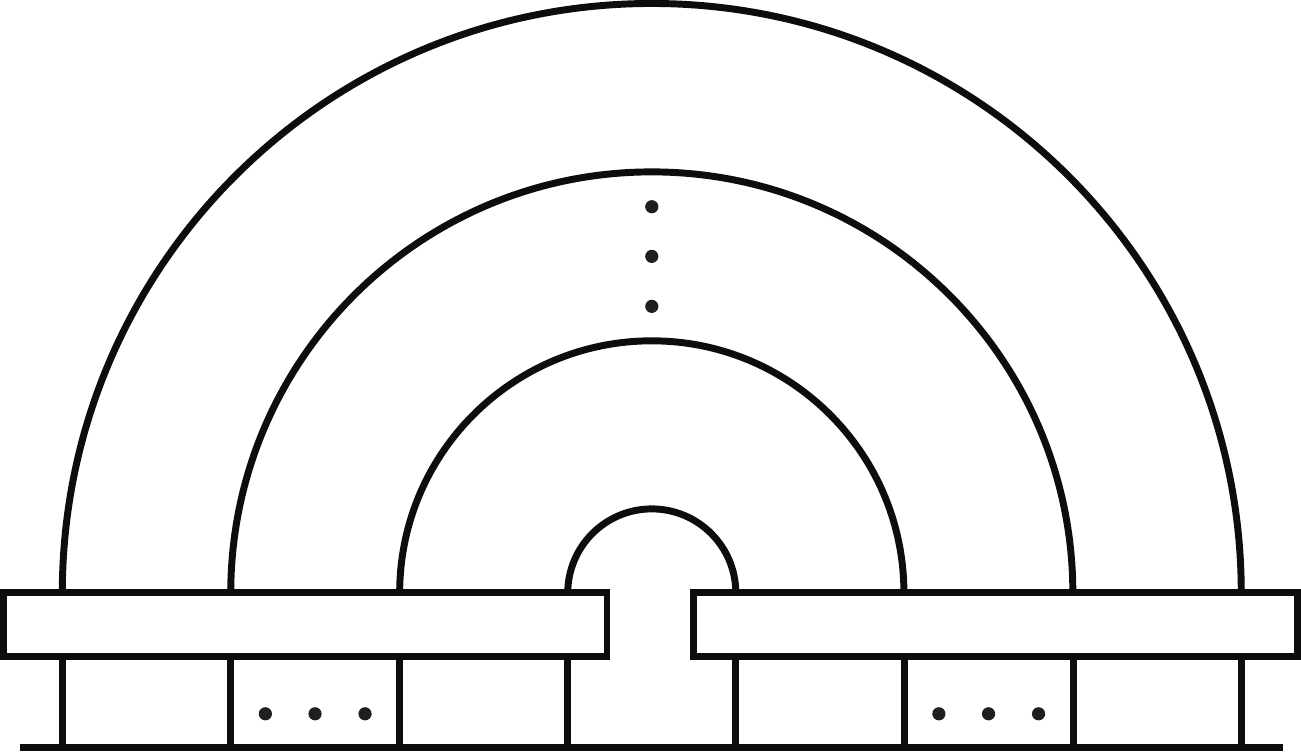} .}} 
\end{align}
\end{lem}

\begin{proof} 
That $\rad \smash{\LS_{2n}\super{0}} = \{0\}$ follows from~\cite[proposition~\red{5.3}]{rsa}
(see also~\cite[corollary~\red{5.1}]{fp0})
With the radical of $\smash{\LS_{2n}\super{0}}$ trivial, the bilinear form is nondegenerate, 
so the dual link pattern $\smash{\Cap_n\superscr{\cheque}}$ is well-defined.

Now, by property~\eqref{ProjectorID0} of the Jones-Wenzl projectors, it is evident
that all three diagrams in~(\ref{DiagExpressions-0}--\ref{DiagExpressions-1}), 
and hence all three expressions in~\eqref{NoDiagExpressions}, are equivalent.  Also, if $\alpha \in \smash{\LP_{2n}\super{0}}$ but $\alpha \neq \Cap_n$, then the network $( \WJProj\sub{n} \otimes \WJProj\sub{n} ) \Cap_n \BarAction \alpha$ must have a link touching two nodes of a projector box and therefore vanishes.  For example,
\begin{align}  
\vcenter{\hbox{\includegraphics[scale=0.275]{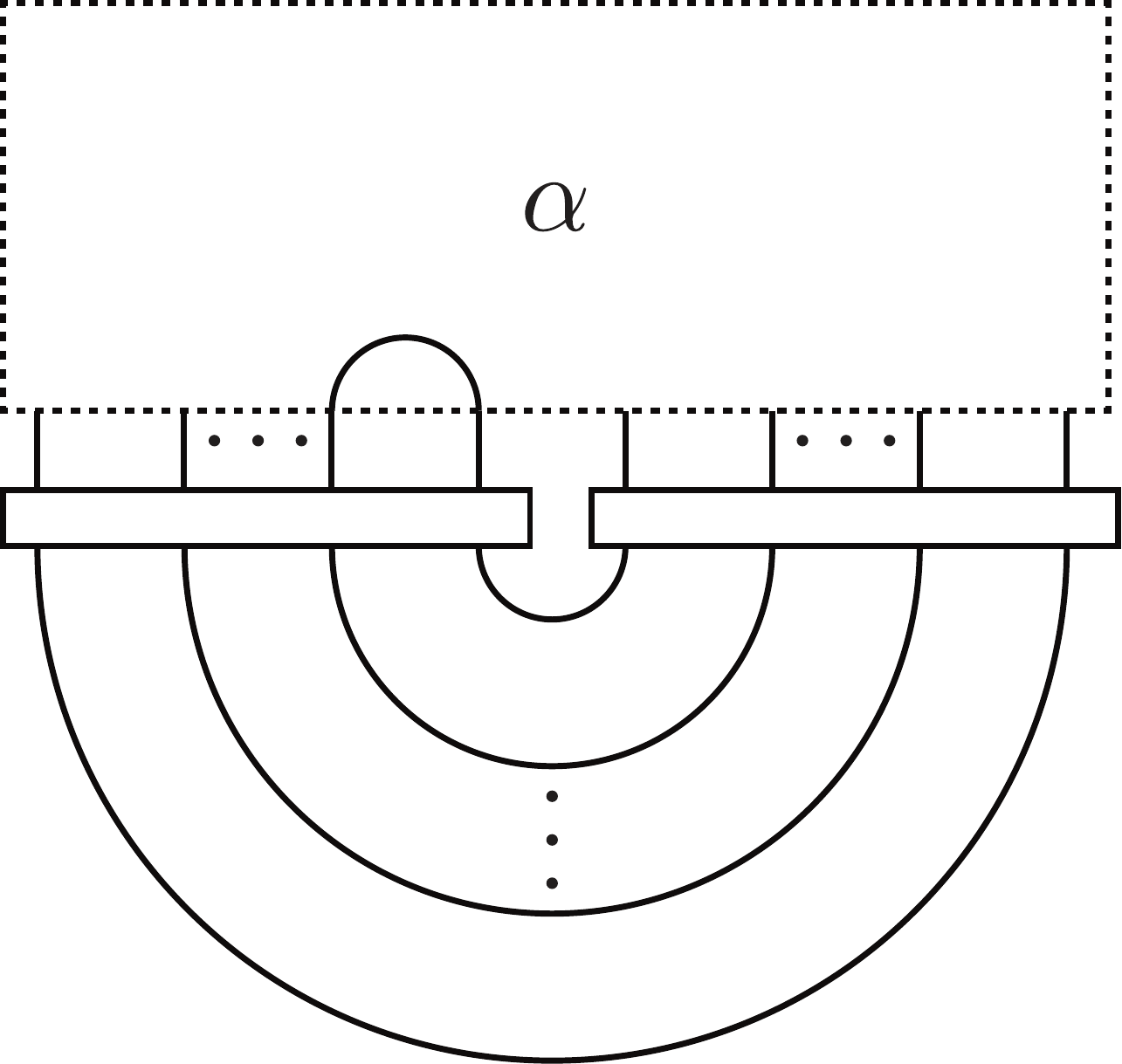}}} \quad = \quad 0 . 
\end{align}  
Therefore, we have $\BiForm{(\WJProj\sub{n} \otimes \WJProj\sub{n} ) \Cap_n}{\alpha} = 0$ if $\alpha \in \smash{\LP_{2n}\super{0}} \setminus \{ \Cap_n \}$.  
Finally, it immediately follows from~\eqref{DeltaTangleGen} of 
lemma~\ref{CollectionLem} that 
$\BiForm{(\WJProj\sub{n} \otimes \mathbf{1}_{\TL_n} ) \Cap_n}{\Cap_n} = (-1)^n[n+1]$.  
Hence, the claim that $\Cap_n\superscr{\cheque}$ is given by~\eqref{NoDiagExpressions} follows.
\end{proof}

\begin{lem} \label{InvMeandLem} 
Suppose $n + 1 < \ppmin(q)$.  Then the 
meander
matrix $\smash{\Gram_{2n}\super{0}}$ is invertible, and we have
\begin{align} \label{ProjExplicit1} 
\WJProj\sub{n} &= \sum_{T \, \in \, \LD_n} (-1)^n[n+1]\Big[\big(\Gram_{2n}\super{0}\big)^{-1}\Big]_{\Cap_n,\alpha_T} T \\
\label{ProjExplicit2} 
&= \sum_{T \, \in \, \LD_n} \left(\frac{\big[\big(\Gram_{2n}\super{0}\big)^{-1}\big]_{\Cap_n,\alpha_T}}{\big[\big(\Gram_{2n}\super{0}\big)^{-1}\big]_{\Cap_n,\Cap_n}}\right) T.
\end{align}
\end{lem}

\begin{proof} We recall that $\smash{\Gram_{2n}\super{0}}$ is invertible if and only if $\rad \smash{ \LS_{2n}\super{0}} = \{0\}$, and the latter condition holds if 
$n +1 < \ppmin(q)$ according to lemma~\ref{CapDualLem}.  Therefore, $\smash{\Gram_{2n}\super{0}}$ is invertible if $n + 1 < \ppmin(q)$.

Next, we prove~\eqref{ProjExplicit1}. 
Using the nondegenerate bilinear form on $ \smash{\LP_{2n}\super{0}}$
and definition~\eqref{GramMatrix2} we see that for all link patterns $\alpha, \beta \in \smash{\LP_{2n}\super{0}}$, the following two equations are both true and equivalent:
\begin{align}  \label{DualEqns} 
\alpha = \sum_{\beta \, \in \, \LP_{2n}\super{0}} \big[ \Gram_{2n}\super{0} \big]_{\alpha,\beta} \, \beta\superscr{\cheque} \qquad \Longleftrightarrow \qquad \beta\superscr{\cheque} = \sum_{\alpha \, \in \, \LP_{2n}\super{0}} \big[ \big( \Gram_{2n}\super{0} \big)^{-1} \big]_{\beta,\alpha} \alpha .
\end{align}   
After substituting $\beta = \Cap_n$ into the second equation of~\eqref{DualEqns}, multiplying both sides of it by $(-1)^n[n+1]$ and indexing the link patterns in $\smash{\LP_{2n}\super{0}}$ by the link diagrams in $\LD_n$ via the bijection~\eqref{TLBij}, we arrive with
\begin{align}  \label{InterResult} 
(-1)^n[n+1]\Cap_n\superscr{\cheque} = \sum_{T \, \in \, \LD_n} (-1)^n[n+1]\big[ \big( \Gram_{2n}\super{0} \big)^{-1} \big]_{\Cap_n, \alpha_T}\, \alpha_T. 
\end{align}  
Now, lemma~\ref{CapDualLem} says that the bijection~\eqref{TLBij} sends $\WJProj\sub{n}$ to $(-1)^n[n+1] \Cap_n$.  Thus after applying the inverse of this bijection~\eqref{TLBij} to both sides of~\eqref{InterResult}, we arrive with the sought result~\eqref{ProjExplicit1}.  

Finally, to obtain~\eqref{ProjExplicit2} from~\eqref{ProjExplicit1}, we examine the term in the sum of~\eqref{ProjExplicit1} with link diagram $T = \id_{\TL_n(\nu)}$.  According to~\eqref{ProjDecomp}, the coefficient of this term necessarily equals one.  Combining this fact with~\eqref{CapID}, we find  
\begin{align}  \label{DiagFormula0} 
(-1)^n[n+1]\big[\big(\Gram_{2n}\super{0}\big)^{-1}\big]_{\Cap_n,\Cap_n} = 1. 
\end{align}  
Using this identity, we immediately obtain~\eqref{ProjExplicit2} from~\eqref{ProjExplicit1}. This finishes the proof.
\end{proof}

\begin{cor} \label{DiagFormulaCor}
Suppose $n + 1 < \ppmin(q)$. Then we have 
\begin{align}  \label{DiagFormula} 
\big[\big(\Gram_{2n}\super{0}\big)^{-1}\big]_{\Cap_n,\Cap_n} = \frac{(-1)^n}{[n+1]}. 
\end{align}  
\end{cor}
\begin{proof}
This follows immediately from~\eqref{DiagFormula0}.
\end{proof}

\subsection{Formulas for entries of the inverse meander matrix: ideas}

In light of~\eqref{ProjExplicit2}, we may determine the coefficients 
of the Jones-Wenzl projector~\eqref{WJExpansion}
by computing the entries of the inverse of the meander matrix $\smash{\Gram_{2n}\super{0}}$ with row index $\Cap_n$.  
In our method, we assume that $q = e^{4\pi \ii/\kappa}$ for some irrational $\kappa \in (4,8)$.
Because the entries of $\smash{\Gram_{2n}\super{0}}$ are analytic functions of $q \in \bC^\times$ with $n+1 < \ppmin(q)$,
we do not lose generality with this assumption.

For any $\kappa \in (4,8)$, any points $x_1 < x_2 < \ldots < x_{2n}$, 
and any collection $\{ \Gamma_1, \Gamma_2, \ldots, \Gamma_{n-1} \}$ of contours in $\bC$, we define a Coulomb gas integral function by
\begin{multline} \label{eulerintegral}
\mathcal{J}\Big(\Gamma_1,\Gamma_2,\ldots,\Gamma_{n-1}\,\Big|\,x_1, x_2, \ldots, x_{2n}\Big) 
= \int_{\Gamma_{n-1}}{\rm d}u_{n-1}\int_{\Gamma_{n-2}}{\rm d}u_{n-2}\dotsm\\
\dotsm\,\,\int_{\Gamma_2}{\rm d}u_2\,\,\int_{\Gamma_1}{\rm d}u_1\,\Bigg(\prod_{l=1}^{2n-1}\prod_{m=1}^{n-1}(x_l-u_m)^{-4/\kappa}\Bigg)
\Bigg( \prod_{m=1}^{n-1} (x_{2n}-u_m)^{12/\kappa-2} 
\Bigg) \Bigg(\prod_{r<s}^{n-1}(u_r - u_s)^{8/\kappa}\Bigg) ,
\end{multline}
where the branch of the multivalued integrand is specified later.

We use two classes of functions employing Coulomb gas integrals~\eqref{eulerintegral}. 
Both classes are linearly independent collections indexed by link patterns, and we denote them by $\smash{\big\{\sF_\alpha \,|\, \alpha \in \LP_{2n}\super{0} \big\}}$
and $\smash{\big\{\Pi_\alpha \,|\, \alpha \in \LP_{2n}\super{0} \big\}}$. To define the function $\sF_\alpha$,
we enumerate all links of the link pattern $\alpha \in \smash{\LP_{2n}\super{0}}$ from one to $n$ and in such a way that the 
link anchored to the $2n$:th node of $\alpha$, is always the $n$:th link. 
Then, the function $\sF_\alpha$ is defined by~\cite[definition~\red{4}]{sfk3}
\begin{align}\label{Fsoln} 
\sF_\alpha( x_1, x_2, \ldots, x_{2n}) = 
\const \times \mathcal{J}\Big( \Gamma_m = \text{$m$:th link of $\alpha$, for all $m \in \{1,2,\ldots,n-1\}$} \,\Big|\,x_1, x_2, \ldots, x_{2n}\Big) , 
\end{align}
where the constant is finite and independent of $\alpha$ (but depends on $\kappa$), and where the branch of the multivalued integrand in~\eqref{eulerintegral}
is fixed in such a way that the integrand is real-valued and positive when
\begin{align}  \label{BranchChoice}
x_1 < x_2 < \cdots < x_{n+1} < u_1 < x_{n+2} < u_2 < \cdots < x_{2n-1} < u_{n-1} < x_{2n} .
\end{align}
The functions $\Pi_\alpha$ in the second aforementioned class are called 
``connectivity weights" or "multiple SLE pure partition functions."
They are related to the functions $\sF_\alpha$ via the meander matrix~\cite[definition~\red{4}]{sfk4}: 
\begin{align}\label{SolvePi} 
\sF_\alpha = \sum_{\beta \, \in \, \smash{\LP_{2n}\super{0}}} \big[\Gram_{2n}\super{0}\big]_{\alpha,\beta} \Pi_\beta \quad \text{ for all $\alpha \in \LP_{2n}\super{0}$}. 
\end{align}
With $\kappa$ irrational, we have $\ppmin(q) = \infty$, so the matrix $\smash{\Gram_{2n}\super{0}}$ is invertible and we may solve this system 
for the rainbow link pattern connectivity weight $\Pi_{\Cap_n}$, finding 
\begin{align}\label{InvertFPi} 
\Pi_{\Cap_n} = \sum_{\alpha \, \in \, \smash{\LP_{2n}\super{0}}} \big[\big( \Gram_{2n}\super{0} \big)^{-1}\big]_{\Cap_n, \alpha} \sF_\alpha. 
\end{align}
These coefficients of the $\sF_\alpha$ on the right side also appear in the formula~\eqref{ProjExplicit2} for the Jones-Wenzl projector coefficients.

We may use~\eqref{InvertFPi} and the explicit formula~\cite[equation~(\red{56})]{fsk}
for the rainbow connectivity weight $\Pi_{\Cap_n}$ 
to solve for all coefficients that appear on the right side of~\eqref{InvertFPi}.  
We give these coefficients in lemma~\ref{ExplicitLem}. Then, upon inserting the result into~\eqref{ProjExplicit2}, in corollary~\ref{WJCoeffCor}
we arrive with explicit formulas for all Jones-Wenzl projector coefficients in~\eqref{WJExpansion}.

To begin, we let $\mathscr{P}(A,B)$ denote the ``Pochhammer contour" surrounding two simply connected regions $A,B \subset \bC$:
\begin{align}\label{Poch} 
\mathscr{P}(A,B) \quad = \quad \vcenter{\hbox{\includegraphics[scale=0.275]{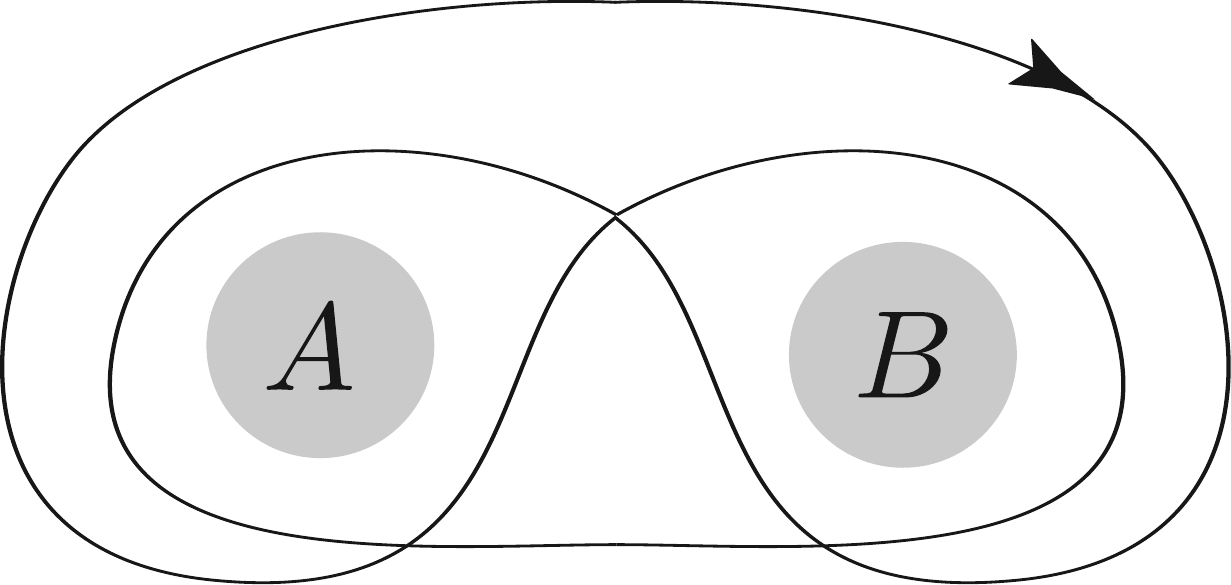} .}} 
\end{align}
We note that
$\mathscr{P}(A,B)$ is a loop on the universal cover of the set $\bC \setminus (A \cup B)$. 
Now, denoting $\Gamma_0 := \{x_{2n}\}$, the formula for $\Pi_{\Cap_n}$ 
according to~\cite[equation~(\red{56})]{fsk} reads
\begin{align}\label{PiSoln} 
\Pi_{\Cap_n}(x_1, x_2, \ldots, x_{2n}) = 
\const \times \mathcal{J}\Big( \Gamma_m = \mathscr{P}(x_{2n - m}, \Gamma_{m-1}), \; \text{for all $m \in \{1,2,\ldots,n-1\}$} \,\Big|\, x_1, x_2, \ldots, x_{2n} \Big) ,  
\end{align}
where again, the $\kappa$-dependent constant is irrelevant to our purpose. 
In our definition~\eqref{PiSoln} of $\Pi_{\Cap_n}$,
the branch of the multivalued integrand in~\eqref{eulerintegral} is fixed in such a way that the integrand is real-valued and positive when
all the integration variables $u_1, u_2, \ldots, u_{n-1}$ are at the points indicated by red dots, as in~\eqref{BranchChoice}:
\begin{align}
\vcenter{\hbox{\includegraphics[scale=0.275]{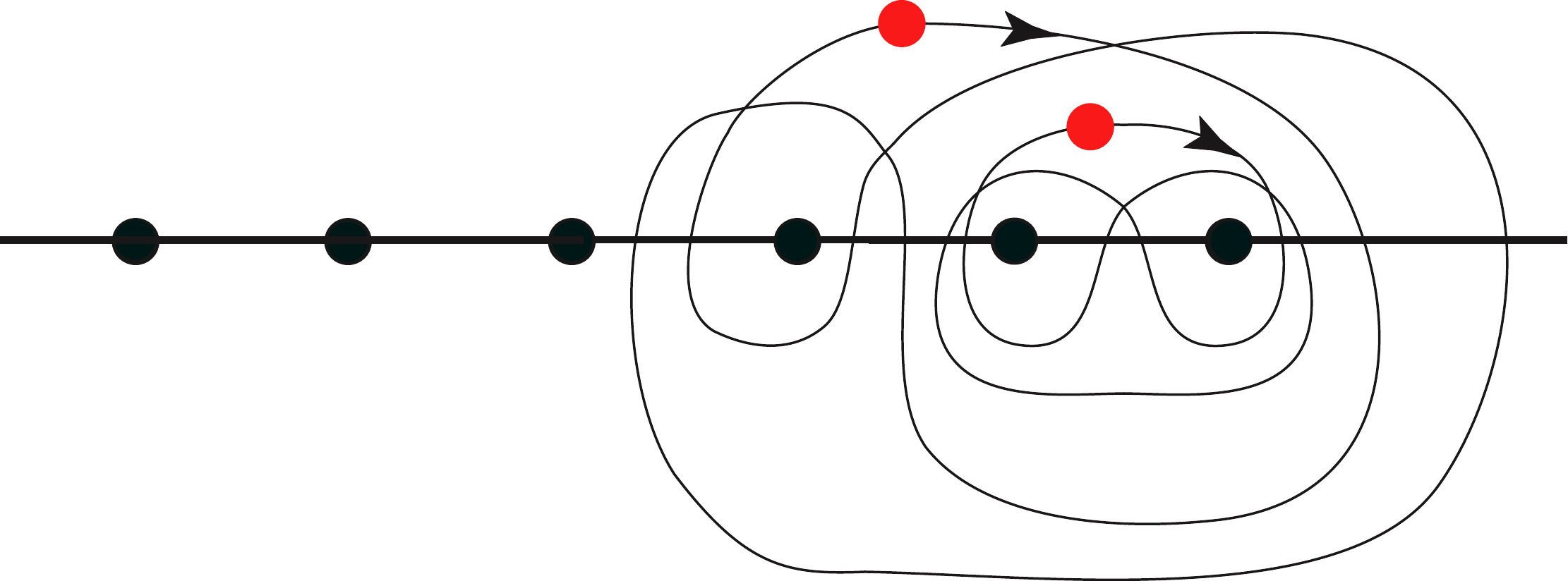} .}} 
\end{align}
Combining (\ref{Fsoln},~\ref{InvertFPi},~\ref{PiSoln}), we infer that for some nonzero $\kappa$-dependent constant, we have
\begin{align} \label{Combined}  
\mathcal{J}\Big( \Gamma_m = \mathscr{P}(x_{2n - m}, \Gamma_{m-1}), \; \text{for all $m \in \{1,2,\ldots,n-1\}$} \Big) 
= & \; \const \times \sum_{\mathclap{\alpha \, \in \, \LP_{2n}\super{0}}} \,\, \big[\big( \Gram_{2n}\super{0} \big)^{-1}\big]_{\Cap_n, \alpha} 
\sF_\alpha . 
\end{align}

Our goal in this appendix 
is to find a formula in terms of $q$ for all entries of the inverse of the meander matrix $\smash{\Gram_{2n}\super{0}}$ that appear in~\eqref{Combined}.
The text between equations~(\red{48}--\red{51}) of~\cite[section~\red{II E}]{fsk} details how to do this, and we survey the highlights here.  
The idea is to deform the Pochhammer contours $\mathscr{P}(x_{2n - m}, \Gamma_{m-1})$ in~\eqref{Combined} and write the integrals~\eqref{eulerintegral} 
around them in terms of the integrals appearing in~\eqref{Combined} with some (nonzero) multiplicative factors depending on $\kappa$ (i.e.,~$q$).
These multiplicative factors arise from the multi-valuedness of the integrand of $\mathcal{J}$ in~\eqref{eulerintegral} when deforming the contours.
The overall multiplicative constants $"\const"$ that were not specified above are fixed by formula~\eqref{DiagFormula} of corollary~\eqref{DiagFormulaCor}.
Collecting all of the constants, we finally obtain the sought entries of the inverse meander matrix.

\subsection{Formulas for entries of the inverse meander matrix: method}

We begin by writing
the outermost Pochhammer contour $\Gamma_{n-1}$, which surrounds the point $x_{n+1}$ and the other integration contours 
$\Gamma_1$, $\Gamma_2, \ldots, \Gamma_{n-2}$ (represented here by a gray disc), in the following form:
\begin{align}\label{paste} 
\vcenter{\hbox{\includegraphics[scale=0.275]{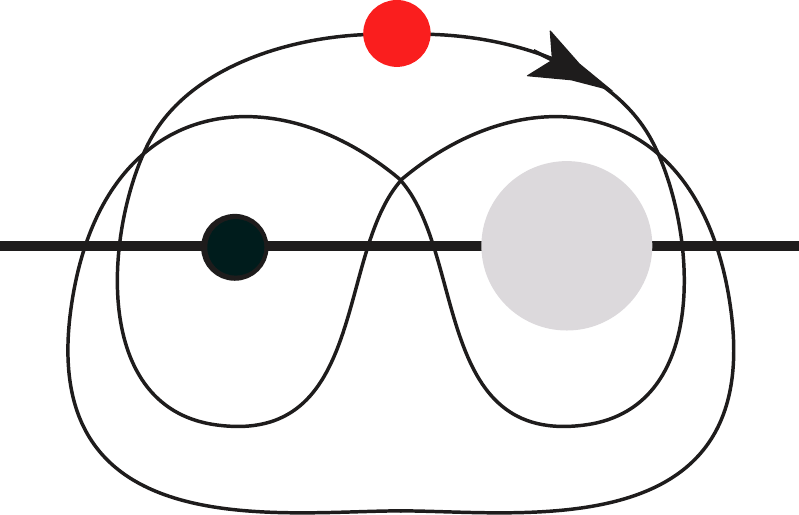}}} \quad = \quad 
(1 - q^{-2}) \quad 
\vcenter{\hbox{\includegraphics[scale=0.275]{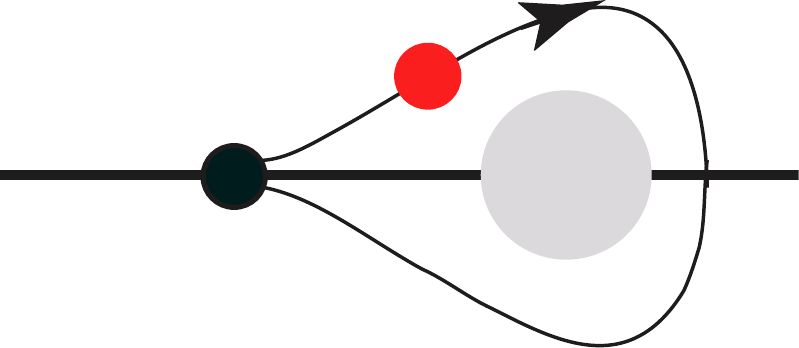} ,}} 
\end{align}  
where the red dot represents the choice of branch of the multivalued integrand of $\mathcal{J}$ in~\eqref{eulerintegral}.
The phase factor in the second term comes from the integration variable $u_{n-1}$ winding around the branch cut at $x_{n+1}$ 
(there is also a branch cut at 
the gray disc, but $u_{n-1}$ winds around it both in the clockwise and 
in the counter-clockwise direction, so the phase factors cancel).

On the other hand, we consider a clockwise-oriented simple loop $\Gamma$ surrounding the interval $[x_1,x_{2n}]$ and the integration contours 
$\Gamma_1$, $\Gamma_2, \ldots, \Gamma_{n-2}$. Because the residue of the integrand of $\mathcal{J}$ in~\eqref{eulerintegral} at infinity equals zero,
the integration around $\Gamma$ gives zero: 
\begin{align}\label{lrloop} 
\vcenter{\hbox{\includegraphics[scale=0.275]{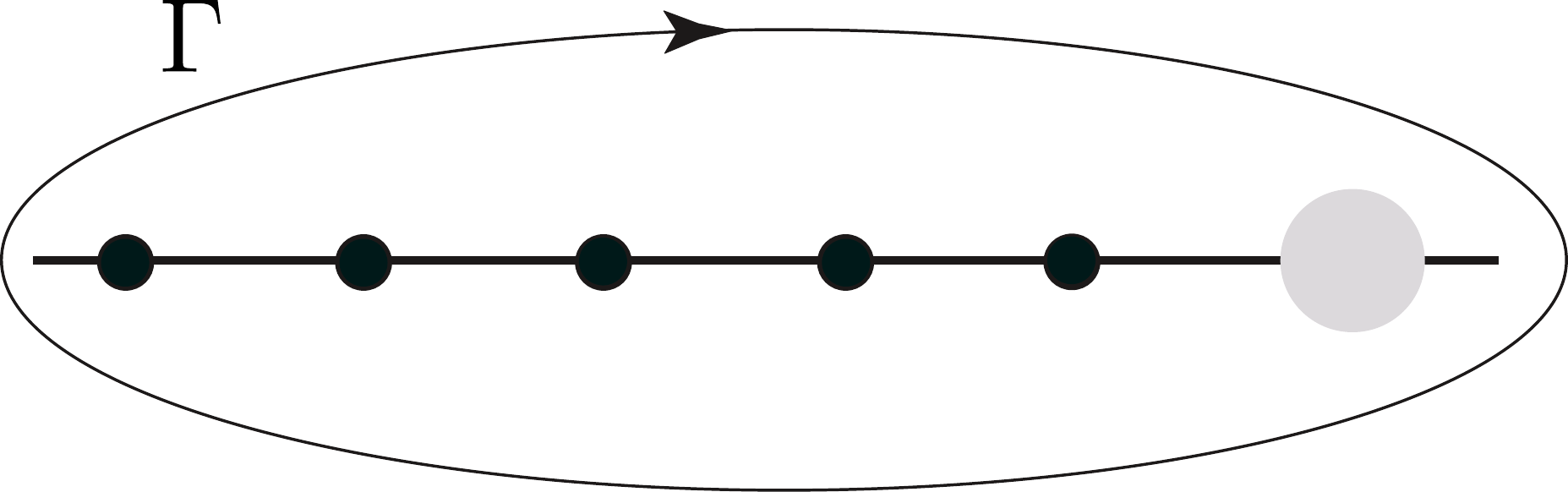}}} 
\quad = \quad 0 .
\end{align}
Pinching $\Gamma$ to itself at the point $x_{n+1}$ and dividing
it into a left half $\Gamma_l$ and a right half $\Gamma_r$ at the pinch point, we can write~\eqref{lrloop} in the form
\begin{align}
0 \quad \overset{\eqref{lrloop}}{=} & \; \quad \vcenter{\hbox{\includegraphics[scale=0.275]{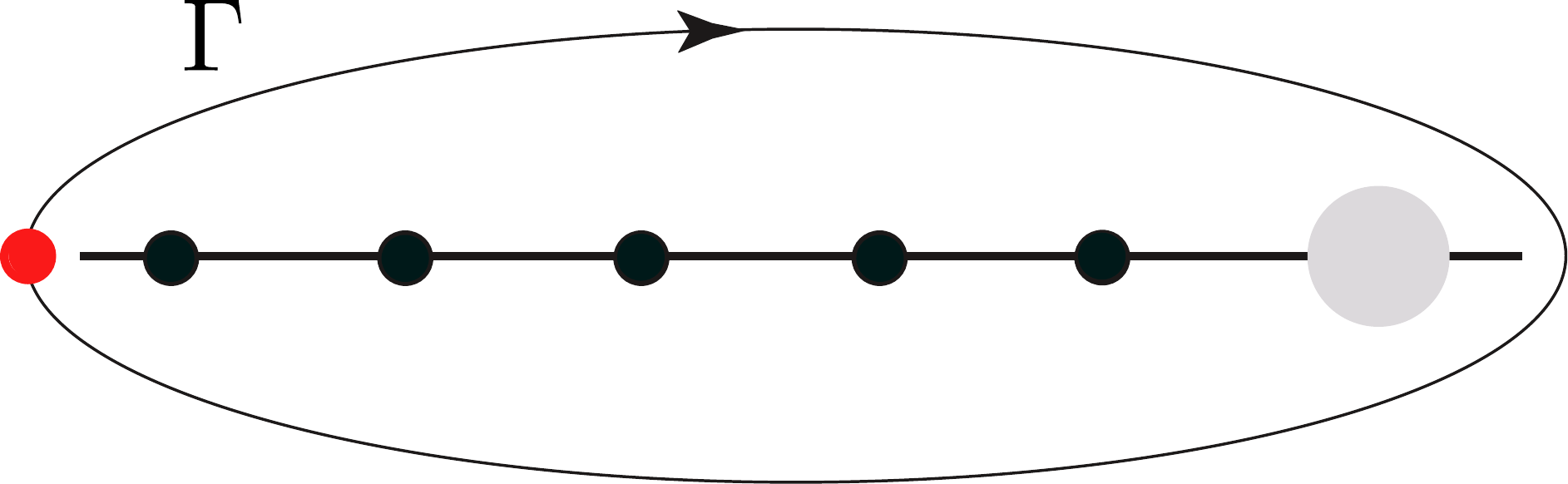}}} \\[1em] 
= & \; \quad
\vcenter{\hbox{\includegraphics[scale=0.275]{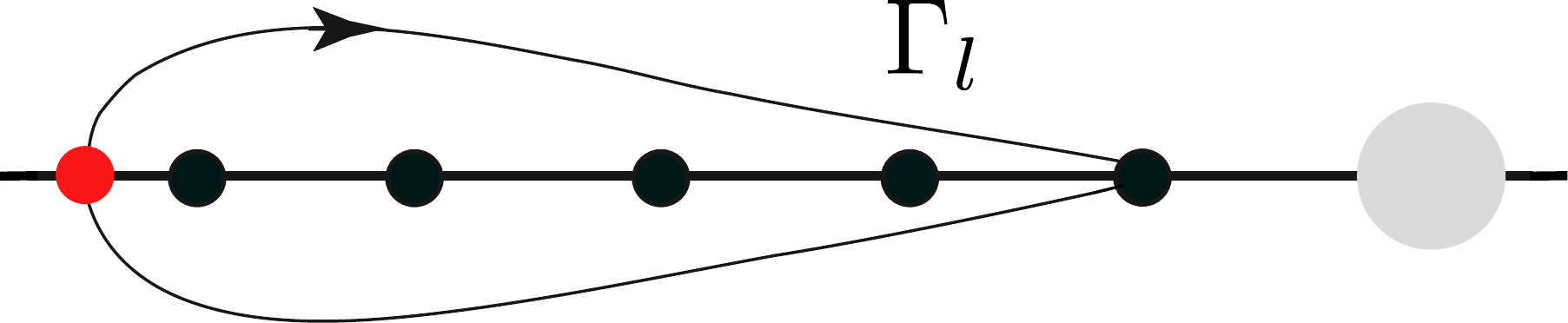}}} \quad \quad + \quad \quad \quad \;
q^{n+1} \quad \vcenter{\hbox{\includegraphics[scale=0.275]{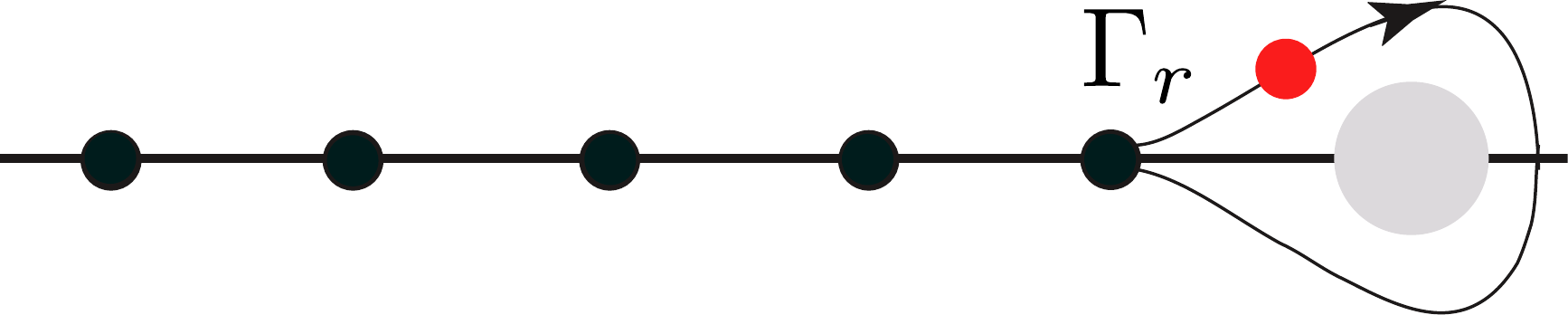}}} \\[1em]  
= & \; \quad
\vcenter{\hbox{\includegraphics[scale=0.275]{e-lrloop3.pdf}}} \quad \quad + \quad \quad
\frac{q^{n+2}}{q - q^{-1}} \quad \vcenter{\hbox{\includegraphics[scale=0.275]{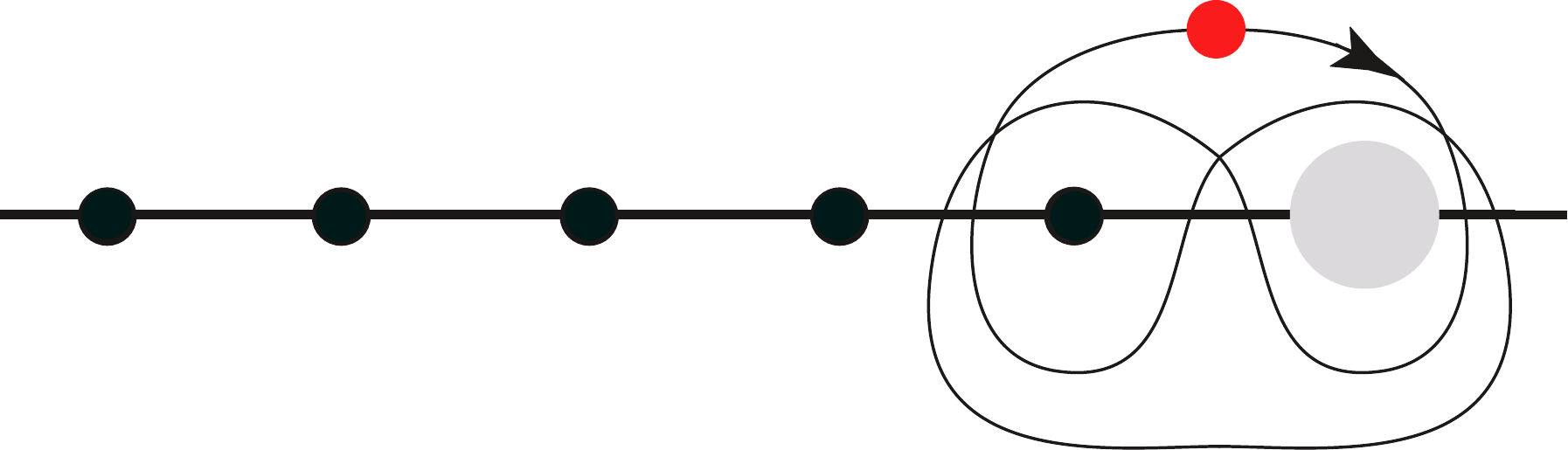} .}} 
\label{lrloop2} 
\end{align}

Next, we decompose $\Gamma_l$ into a linear combination of contours $[x_j,x_{j+1}]^+$, for $j \in \{1,2,\ldots,n-1\}$, 
which are links in the upper half-plane joining the points $x_j$ to $x_{j+1}$:
\begin{align}
[x_j,x_{j+1}]^+ \quad := \quad \vcenter{\hbox{\includegraphics[scale=0.275]{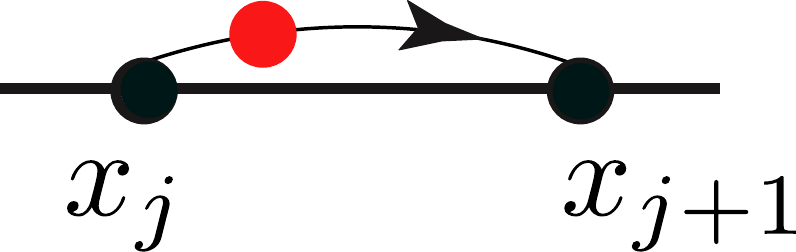} .}} 
\end{align}
We obtain
\begin{align} \label{DecompositionOfLeftLoop}
\vcenter{\hbox{\includegraphics[scale=0.275]{e-lrloop3.pdf}}} \quad
= \quad  (q-q^{-1}) \sum_{j = 1}^n \,[j]\, \quad \vcenter{\hbox{\includegraphics[scale=0.275]{e-segment.pdf} ,}} 
\end{align}
where $[j]$ is the $j$:th quantum integer~\eqref{Qinteger}.
Including the overall multiplicative factor $(q-q^{-1})$ into the (yet unspecified) constant, we write~\eqref{Combined} as
\begin{align} 
\label{Combined1} 
& \; \mathcal{J}\Big(\text{$\Gamma_m = \mathscr{P}(x_{2n - m}, \Gamma_{m-1})$ for all $m \in \{1,2,\ldots,n-1\}$} \Big) \\
\label{Combined2} 
= & \; \const \times 
\sum_{j = 1}^n 
\,[j] \, \mathcal{J}\Big(\text{$\Gamma_m = \mathscr{P}(x_{2n - m}, \Gamma_{m-1})$ for all $m \in \{1,2,\ldots,n-2\}$, and } \Gamma_{n-1} = [x_j,x_{j+1}]^+ \Big) ,
\end{align}
where the $\kappa$-dependent nonzero constant includes not only the factor $(q-q^{-1})$ from~\eqref{DecompositionOfLeftLoop} but also 
the further phase factor $(-q^{-n-2} (q-q^{-1}))$ 
from~\eqref{lrloop2}.

Now, to find the Jones-Wenzl projector coefficients, we repeat this process for each term in~\eqref{Combined2}: 
\begin{itemize}[leftmargin=*]
\item We write the outermost contour $\Gamma_{n-2}$ in the form~\eqref{paste}.

\item Via~(\ref{paste},~\ref{lrloop2}), we identify it with the right half $\Gamma_r$ of a simple clockwise-oriented loop $\Gamma$ surrounding the interval $[x_1, x_{2n}]$ and 
the contours $\Gamma_1$, $\Gamma_2, \ldots, \Gamma_{n-3}$ and $[x_j,x_{j+1}]^+$.

\item With the residue of the integrand of $\mathcal{J}$ in~\eqref{eulerintegral} at infinity equaling zero, the integration around $\Gamma$ gives zero~\eqref{lrloop},
and we may thus replace the integration around $\Gamma_r$ by integration around the left half $\Gamma_l$ of $\Gamma$.

\item We decompose $\Gamma_l$ into a linear combination of link-type contours, analogously to~\eqref{DecompositionOfLeftLoop}. 
We make use of~\cite[lemma~\red{10}]{sfk3}, which implies that the decomposition of $\Gamma_l$ in the $j$:th term 
in~\eqref{Combined2} takes place as if the contour $[x_j,x_{j+1}]^+$ and its endpoints are 
invisible:
thus, only links $[x_k,x_{k+1}]^+$ for $k \in \{1,2,\ldots, j-2, j+3,\ldots,n+2 \}$ and $[x_{j-1},x_{j+2}]^+$ appear in this decomposition.
We clarify this below in section~\ref{ExampleSubSec} with an example. 
\end{itemize}

After repeating the above process for all Pochhammer contours of~\eqref{Combined}, we finally arrive with an equation of the form~\eqref{Combined}.
Using formula~\eqref{DiagFormula} of corollary~\eqref{DiagFormulaCor} to fix the overall multiplicative constant, 
we can identify the entries of the inverse of the meander matrix $\smash{\Gram_{2n}\super{0}}$ as explicit formulas in terms of 
$\kappa$ (i.e.,~$q$). We give the result of this procedure in lemma~\ref{ExplicitLem}, but first, we present an example for illustration.

\subsection{Formulas for entries of the inverse meander matrix: examples} \label{ExampleSubSec}

The contour deformation recipe given above
is somewhat vague, and we invite the reader to investigate the details.  To clarify some of this vagueness, 
we give an example of this recipe's successful implementation.  The example is the simplest one that illustrates all features of this recipe: the case with $n = 3$.  
Then,~\eqref{Combined} reads
\begin{align} 
\label{Combined3}  
\mathcal{J}\Big( \Gamma_1 = \mathscr{P}(x_5, x_6), \;\; \Gamma_2 = \mathscr{P}(x_4, \Gamma_1) \Big) 
= & \; \const \times
\sum_{\mathclap{\alpha \, \in \, \LP_6\super{0}}} \,\, \big[\big( \Gram_6\super{0} \big)^{-1}\big]_{\Cap_3, \alpha} \sF_\alpha . 
\end{align}
To determine the coefficients 
in~\eqref{Combined3}, we begin by deforming
the Pochhammer contour $\Gamma_2 = \mathscr{P}(x_4, \Gamma_1)$. 
According to~(\ref{Combined1}--\ref{Combined2})
we obtain
\begin{align}
\label{Combined4} 
\mathcal{J}\Big( \Gamma_1 = \mathscr{P}(x_5, x_6), \;\; \Gamma_2 = \mathscr{P}(x_4, \Gamma_1) \Big) =
\const \times 
\sum_{j = 1}^{3} \,[j]  \, \mathcal{J}\Big(\Gamma_1 = \mathscr{P}(x_5, x_6), \;\; \Gamma_2 = [x_j,x_{j+1}]^+ \Big) .
\end{align}

Next, we deform the Pochhammer contour $\Gamma_1 = \mathscr{P}(x_5, x_6)$ in each term on the 
right side of~\eqref{Combined4}. 
To arrive with a formula analogous to~(\ref{Combined1}--\ref{Combined2}), we first write $\Gamma_1$ in the form~\eqref{paste} and, 
as in~(\ref{lrloop}--\ref{lrloop2}),
identify it with a constant multiple of the left half $\Gamma_l$ of a simple clockwise-oriented loop $\Gamma$ surrounding the interval $[x_1, x_{6}]$ and 
the contours $[x_j,x_{j+1}]^+$. Then, we decompose $\Gamma_l$ into a linear combination of link-type contours. 
By \cite[lemma~\red{10}]{sfk3}, for the $j$:th term on the right side of~\eqref{Combined4}, the decomposition of $\Gamma_l$ 
takes place as if the contour $[x_j,x_{j+1}]^+$ and its endpoints are 
invisible. 
The terms $j=1,2,3$ are the following (including a factor $(q - q^{-1})$ into the overall multiplicative constants):
\begin{enumerate}[leftmargin=*]
\itemcolor{red}
\item For the $j = 1$ term, we obtain
\begin{align}
\label{Combined5} 
\mathcal{J}\Big(\Gamma_1 = \mathscr{P}(x_5, x_6), \;\; \Gamma_2 = [x_1,x_2]^+ \Big) 
= \const \times 
\sum_{k=1}^2 \,[k] \, \mathcal{J}\Big(\Gamma_1 = [x_{k+2}, x_{k+3}]^+, \;\; \Gamma_2 = [x_1,x_2]^+ \Big) .
\end{align}

\item For the $j = 2$ term, we obtain
\begin{align} 
& \; \mathcal{J}\Big(\Gamma_1 = \mathscr{P}(x_5, x_6), \Gamma_2 = [x_2,x_3]^+ \Big) \\
= & \; \const \times 
\Bigg(
[1] \,  \mathcal{J}\Big(\Gamma_1 = [x_1, x_4]^+, \Gamma_2 = [x_2,x_3]^+ \Big) 
+ [2] \, \mathcal{J}\Big(\Gamma_1 = [x_4, x_5]^+, \Gamma_2 = [x_2,x_3]^+ \Big) \Bigg),
\end{align}
where the contour $[x_1, x_4]^+$ arcs over the contour $[x_2, x_3]^+$ in the upper half-plane.  

\item The $j = 3$ term is similar to the $j = 2$ term.  We 
obtain
\begin{align} 
& \; \mathcal{J}\Big(\Gamma_1 = \mathscr{P}(x_5, x_6), \Gamma_2 = [x_3,x_4]^+ \Big) \\
\label{Combined7}  
= & \; \const \times 
\Bigg(
[1] \, \mathcal{J}\Big(\Gamma_1 = [x_1, x_2]^+, \Gamma_2 = [x_3,x_4]^+ \Big) 
+ [2] \, \mathcal{J}\Big(\Gamma_1 = [x_2, x_5]^+, \Gamma_2 = [x_3,x_4]^+ \Big)  \Bigg),
\end{align}
where again, the contour $[x_2, x_5]^+$ arcs over the contour $[x_3, x_4]^+$ in the upper half-plane.  
\end{enumerate}
Inserting~(\ref{Combined5}--\ref{Combined7}) into~\eqref{Combined4} and recalling~\eqref{Fsoln}, we obtain
\begin{align} 
& \; \mathcal{J}\Big( \Gamma_1 = \mathscr{P}(x_5, x_6), \;\; \Gamma_2 = \mathscr{P}(x_4, \Gamma_1) \Big) \\
=  & \;  \const \times 
\left( [2]  [3]  \, \sF_{\hbox{\includegraphics[scale=0.2]{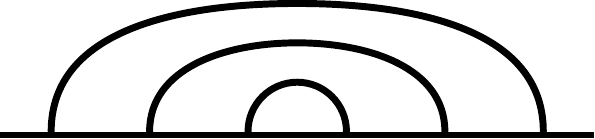}}} 
+ [2]  [2]  \, \sF_{\hbox{\includegraphics[scale=0.2]{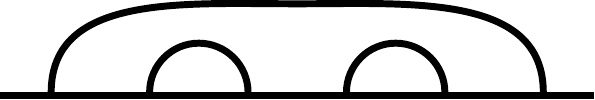}}}
+ [2]  \, \sF_{\hbox{\includegraphics[scale=0.2]{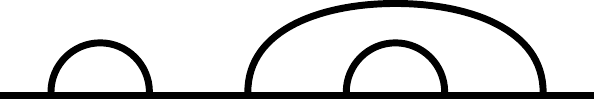}}}
+ [2]  \, \sF_{\hbox{\includegraphics[scale=0.2]{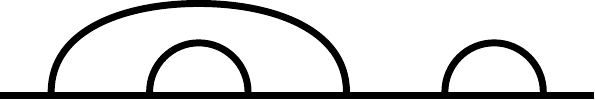}}}
+ ([1]  + [3] ) \, \sF_{\hbox{\includegraphics[scale=0.2]{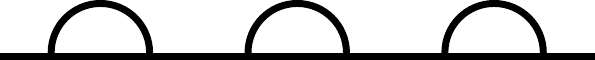}}} \right) .
\end{align}
After identifying this final result with the target equation~\eqref{Combined3}, we find the entries in the row of the inverse of the meander matrix $\smash{\Gram_6\super{0}}$ indexed by 
$\smash{\Cap_3 = \hbox{\includegraphics[scale=0.2]{e-LP7.pdf}}}$, up to a common multiplicative constant:
\begin{align}  
\label{Top1}
\big[\big(\Gram_6\super{0}\big)^{-1}\big]_{\hbox{\includegraphics[scale=0.2]{e-LP7.pdf}}, \hbox{\includegraphics[scale=0.2]{e-LP7.pdf}}}
& = \const \times [2]  [3]  \\
\big[\big(\Gram_6\super{0}\big)^{-1}\big]_{\hbox{\includegraphics[scale=0.2]{e-LP7.pdf}}, \; \hbox{\includegraphics[scale=0.2]{e-LP6.pdf}}}
& = \const \times [2]  [2]  \\
\big[\big(\Gram_6\super{0}\big)^{-1}\big]_{\hbox{\includegraphics[scale=0.2]{e-LP7.pdf}}, \; \hbox{\includegraphics[scale=0.2]{e-LP5.pdf}}}
& = \const \times [2]  \\
\big[\big(\Gram_6\super{0}\big)^{-1}\big]_{\hbox{\includegraphics[scale=0.2]{e-LP7.pdf}}, \; \hbox{\includegraphics[scale=0.2]{e-LP4.pdf}}}
& = \const \times [2]  \\
\big[\big(\Gram_6\super{0}\big)^{-1}\big]_{\hbox{\includegraphics[scale=0.2]{e-LP7.pdf}}, \; \hbox{\includegraphics[scale=0.2]{e-LP3.pdf}}}
& = \const \times ([1]  + [3] ) .
\label{Top5}
\end{align}
Finally, we use~\eqref{DiagFormula} from corollary~\ref{DiagFormulaCor}
with $n = 3$ to solve for the overall multiplicative constant:
\begin{align}
- \frac{1}{[4] } \overset{\eqref{DiagFormula}}{=} \big[\big(\Gram_6\super{0}\big)^{-1}\big]_{ \Cap_3, \Cap_3} 
= \big[\big(\Gram_6\super{0}\big)^{-1}\big]_{\hbox{\includegraphics[scale=0.2]{e-LP7.pdf}}, \; \hbox{\includegraphics[scale=0.2]{e-LP7.pdf}}}
\overset{\eqref{Top1}}{=} \const \times [3]  [2]  \qquad \Longrightarrow \qquad 
\const = -\frac{1}{[4] !} . 
\end{align}
Inserting this normalization into (\ref{Top1}--\ref{Top5}) finally gives us explicit formulas in terms of $q$ of all five entries in the row of the inverse of the meander matrix $\smash{\Gram_6\super{0}}$ indexed by $\smash{\Cap_3 = \hbox{\includegraphics[scale=0.2]{e-LP7.pdf}}}$:
\begin{align}  
\big[\big(\Gram_6\super{0}\big)^{-1}\big]_{\hbox{\includegraphics[scale=0.2]{e-LP7.pdf}}, \; \hbox{\includegraphics[scale=0.2]{e-LP7.pdf}}}
& = - \frac{1}{[4] } \\
\big[\big(\Gram_6\super{0}\big)^{-1}\big]_{\hbox{\includegraphics[scale=0.2]{e-LP7.pdf}}, \; \hbox{\includegraphics[scale=0.2]{e-LP6.pdf}}}
& = - \frac{[2] }{[3]  [4] } \\
\big[\big(\Gram_6\super{0}\big)^{-1}\big]_{\hbox{\includegraphics[scale=0.2]{e-LP7.pdf}}, \; \hbox{\includegraphics[scale=0.2]{e-LP5.pdf}}}
& = - \frac{1}{[3]  [4] } \\
\big[\big(\Gram_6\super{0}\big)^{-1}\big]_{\hbox{\includegraphics[scale=0.2]{e-LP7.pdf}}, \; \hbox{\includegraphics[scale=0.2]{e-LP4.pdf}}}
& = - \frac{1}{[3]  [4] } \\
\big[\big(\Gram_6\super{0}\big)^{-1}\big]_{\hbox{\includegraphics[scale=0.2]{e-LP7.pdf}}, \; \hbox{\includegraphics[scale=0.2]{e-LP3.pdf}}}
& = - \frac{1}{[4] !} ([1]  + [3] ).
\end{align}

\subsection{Formulas for entries of the inverse meander matrix: general case}

Generalizing the above
work to arbitrary $n \in \bZpos$, we obtain an explicit formula in terms of $q$ of all entries in the row of the inverse of 
the meander matrix $\smash{\Gram_{2n}\super{0}}$ indexed by $\Cap_n$.  To write this formula, we use the following recipe:

\begin{recipe} \label{recipe} 
For a link pattern $\alpha \in \smash{\LP_{2n}\super{0}}$, we let $a_i,b_i \in \{1,2,\ldots,2 n\}$ respectively denote the label of the left and right endpoint of the $i$:th 
link of $\alpha$. We enumerate the links of $\alpha$ from $1$ to $n$ in such a way that 
\begin{enumerate} 
\itemcolor{red}
\item \label{EnumIt2} for all $i,j \in \{1,2,\ldots,n\}$, if $i < j$, then the $i$:th link does not nest the $j$:th link, and
\item \label{EnumIt1} $b_i \leq i+n$ for all $i \in \{1,2,\ldots,n\}$ 
or equivalently by item~\ref{EnumIt2},  
$a_i \leq i + n - 1$
for all $i \in \{1,2,\ldots,n\}$.
\end{enumerate}
We denote $\vartheta = (a_1, a_2,\ldots, a_n)$, and we let
\begin{align}  \mathsf{S}(\alpha) = \big\{ \vartheta \in \bZnn^n \, \big| \, \text{$\vartheta$ arises from an enumeration of the links of $\alpha$ satisfying items~\ref{EnumIt2} and~\ref{EnumIt1}} \big\}.
\end{align}  
Finally, for each $\vartheta \in \mathsf{S}(\alpha)$, we let $\gamma(\vartheta)$ be a multiindex with $n$ entries and whose $i$:th entry is given by 
\begin{align}  
\gamma(\vartheta)_i = 2\times \# \big\{1 \leq j \leq i-1 \, \big| \, a_j < a_i \big\} . 
\end{align}  
\end{recipe}

For example, the rainbow link pattern $\Cap_n$ has only one enumeration (here, circled numbers indicate the 
label of the link while uncircled numbers still indicate the size of a cable):
\begin{align}  
\vcenter{\hbox{\includegraphics[scale=0.275]{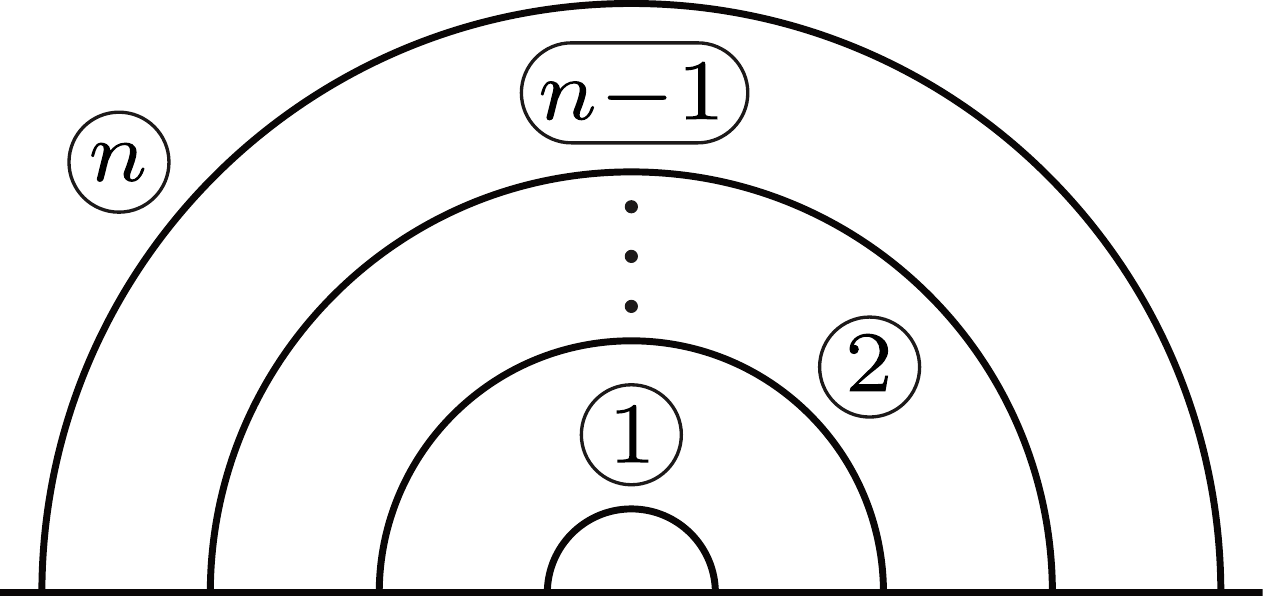}}} \quad 
\qquad \qquad \Longrightarrow \qquad \qquad
\mathsf{S}(\Cap_n) = \{ (n,n-1,n-2,\ldots,2,1) \} . 
\end{align}  
Other link patterns may have several different enumerations.  For example, the links of the following link pattern $\alpha \in \smash{\LS_8\super{0}}$ 
may be enumerated according to recipe~\ref{recipe} in and only in any one of the following three ways:
\begin{gather} 
\vcenter{\hbox{\includegraphics[scale=0.275]{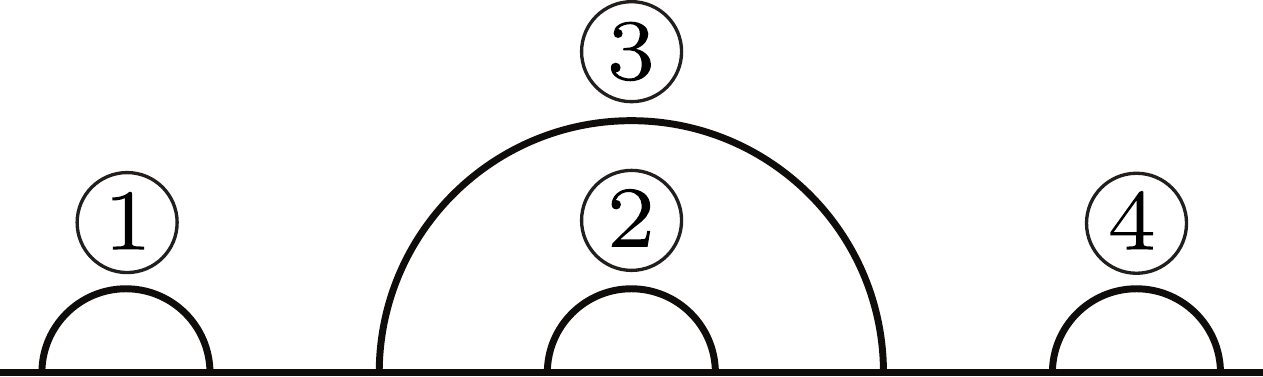} ,}}  \qquad \qquad
\vcenter{\hbox{\includegraphics[scale=0.275]{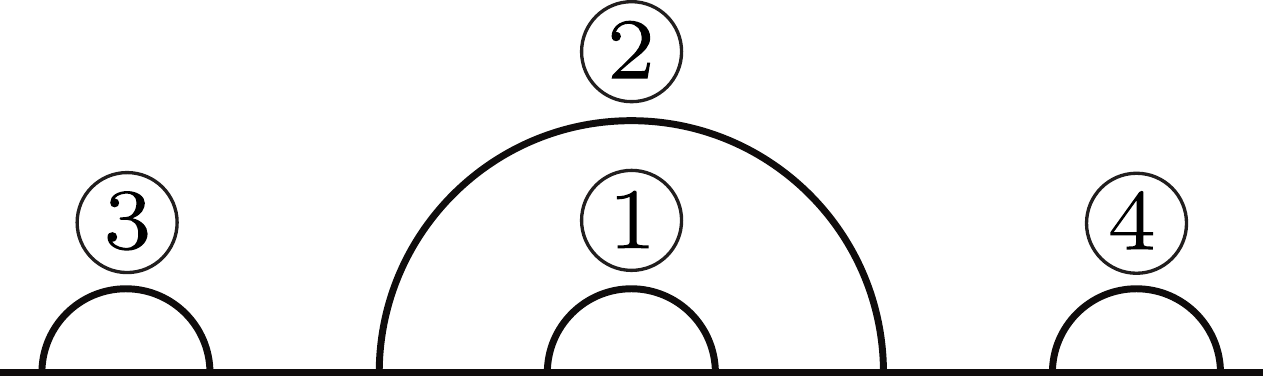} ,}} \qquad \qquad
\vcenter{\hbox{\includegraphics[scale=0.275]{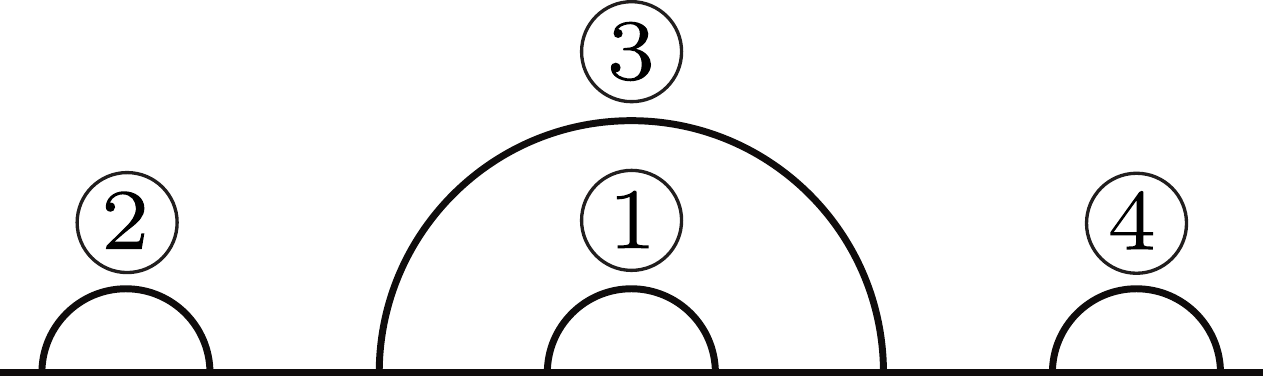}}} \\[1em]
\Longrightarrow \qquad \mathsf{S}(\alpha) = \{ (1,4,3,7), (4,3,1,7), (4,1,3,7) \} . 
\end{gather}

\begin{lem} \label{ExplicitLem} Suppose $n + 1 < \ppmin(q)$, and recall the notations and definitions of recipe~\ref{recipe}.  Then the entries of the inverse of the meander matrix $\smash{\Gram_{2n}\super{0}}$ along the row indexed by $\Cap_n$ are given by the formula
\begin{align}  \label{InvG} 
\big[\big( \Gram_{2n}\super{0} \big)^{-1}\big]_{\Cap_n, \alpha} 
= \frac{(-1)^n}{[n+1]!} \sum_{\vartheta \, \in \, \mathsf{S}(\alpha)} \prod_{ \vphantom{\vartheta \, \in \, \mathsf{S}(\alpha)} i = 1}^{n} \, 
[a_i - \gamma(\vartheta)_i] . 
\end{align}  
\end{lem} 

\begin{proof} By continuing our work from~\eqref{Combined3}, we eventually find the following formula for the coefficients $\smash{\big[\big( \Gram_{2n}\super{0} \big)^{-1}\big]_{\Cap_n, \alpha}}$ appearing in~\eqref{Combined3}, 
\begin{align}  
\big[\big( \Gram_{2n}\super{0} \big)^{-1}\big]_{\Cap_n, \alpha} 
= 
\const \sum_{\vartheta \, \in \, \mathsf{S}(\alpha)} \prod_{  \vphantom{\vartheta \, \in \, \mathsf{S}(\alpha)} i = 1}^{n} \, 
[a_i - \gamma(\vartheta)_i]. 
\end{align}  
We leave the details for the reader.
To find the constant, 
we set $\alpha = \Cap_n$ and use~\eqref{DiagFormula} from corollary~\ref{DiagFormulaCor}.
\end{proof}

\subsection{Formulas for coefficients of the Jones-Wenzl projector}

Using lemma~\ref{ExplicitLem}, we recover the following formula~\cite{sm} for the Jones-Wenzl projector coefficients in~\eqref{WJExpansion}: 

\begin{cor} \label{WJCoeffCor}
\textnormal{\cite[proposition~\red{5.1}]{sm}} 
Suppose $n < \ppmin(q)$.  Then for all $T \in \LD_n$, we have 
\begin{align}  \label{WJCoeff} 
\textnormal{coef}_T 
= \frac{1}{[n]!} \sum_{\vartheta \, \in \, \mathsf{S}(\alpha_T)} \prod_{  \vphantom{\vartheta \, \in \, \mathsf{S}(\alpha_T)} i = 1}^n \, 
[a_i - \gamma(\vartheta)_i]. 
\end{align}  
\end{cor}

\begin{proof}  
For $n + 1 < \ppmin(q)$, formula~\eqref{WJCoeff} follows from combining~\eqref{ProjExplicit1} of lemma~\ref{InvMeandLem} with~\eqref{InvG} of lemma~\ref{ExplicitLem}.  
Because the coefficients $\textnormal{coef}_T$ are analytic functions of $q$ away from the poles $\{q \in \bC^\times \,| \, n < \ppmin(q) \}$, this formula 
extends continuously to cases with $n + 1 = \ppmin(q) $ too.  (Alternatively,~\cite{sm} gives the original proof of this corollary.)
\end{proof}

To finish, we give explicit formulas for certain important special cases of the Jones-Wenzl projector coefficients.
These formulas are needed frequently in the present article. To derive them, we use the following simple identity.

\begin{lem} \label{SumFormulaLem} 
For all $b,k \in \bZnn$, and $ i \in \{1,2,\ldots, b-1\}$, we have the sum formula
\begin{align}  \label{SumFormula} 
\sum_{a = 0}^ k \frac{[a+i-1]! [b - i + a -1]!}{[a]! [a + b]!} = \frac{1}{[i] [b - i]} \frac{[i+k]! [b + k -i]!}{[k]! [b + k]! }. 
\end{align}  
\end{lem}

\begin{proof}  
We prove the claim for any $b \in \bZnn$ by induction on 
$k \geq 0$. It is 
easy to see that~\eqref{SumFormula} holds for $k = 0$.  
Now, we assume that~\eqref{SumFormula} holds with $k = l-1$ for some $l \in \bZpos$.  Denoting by $S_l$ the left side of~\eqref{SumFormula}, we calculate
\begin{align} 
S_l &= S_{l-1} + \frac{[l+i-1]![b-i+l-1]!}{[l]![b+l]!} \\
&= \frac{1}{[i] [b - i]} \frac{[l + i - 1]! [b + l - i - 1]!}{[l - 1]! [b + l - 1]!} 
+ \frac{[l+i-1]![b-i+l-1]!}{[l]![b+l]!} \\
&= \frac{1}{[i] [b - i]} \frac{[l + i - 1]! [b + l - i - 1]!}{[l]! [b + l]!}  
\big( [l][b+l] + [i][b-i] \big) \\
&= \frac{1}{[i] [b - i]} \frac{[l + i]! [b + l -i]!}{[l]! [b + l]! }. 
\end{align}
This proves formula~\eqref{SumFormula} with $k = l$ and concludes the proof.
\end{proof}

\begin{lem} \label{LinkPattLem} 
Suppose $n < \ppmin(q)$. Let $\alpha \in \smash{\LP_{2n}\super{0}}$ be the link pattern
\begin{align}  \label{SpecialLinkPattern} 
\alpha \quad = \quad \vcenter{\hbox{\includegraphics[scale=0.275]{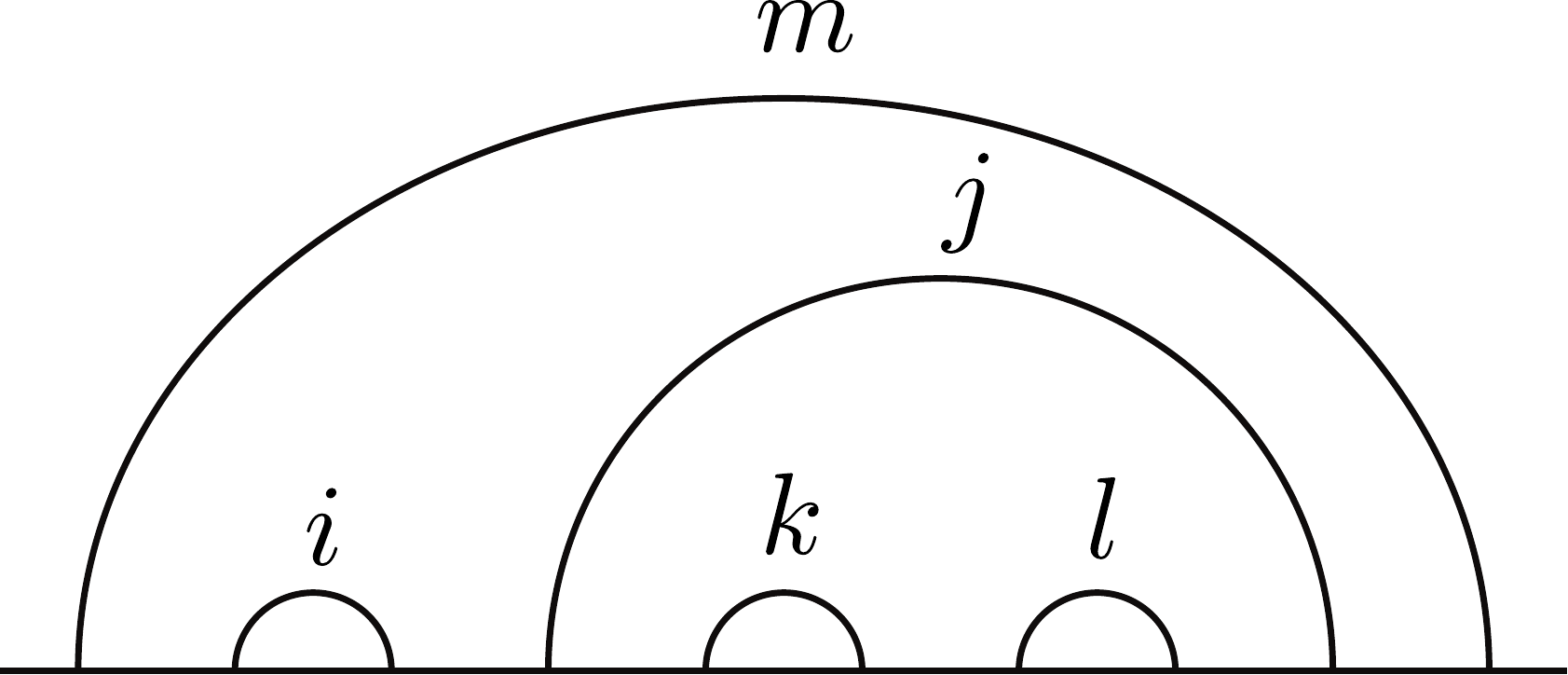} ,}} 
\end{align}  
with $i + j + k + l + m = n$. Then we have 
\begin{align}  \label{InvG2} 
\big[\big( \Gram_{2n}\super{0} \big)^{-1}\big]_{\Cap_n, \alpha} = 
\frac{(-1)^n}{[n+1]} \frac{ [i+k]![i+m]![i+j+k+m]![j+l+m]!}{[i]![k]![i+j+m]![m]!} . 
\end{align}  
\end{lem}

\begin{proof}  
To explicitly show the dependence on the
sizes $i$ and $k$ of the cables of $\alpha$ thus labeled in~\eqref{SpecialLinkPattern}, 
we write $\alpha_{ik}$ for the link pattern $\alpha$ in~\eqref{SpecialLinkPattern}.  We also let
\begin{align}  \label{UpsilonDefn} 
\Upsilon_{ik} = \sum_{\vartheta \, \in \, \mathsf{S}(\alpha_{ik})} \prod_{  \vphantom{\vartheta \, \in \, \mathsf{S}(\alpha_{ik})} t = 1}^n \, 
[a_t - \gamma(\vartheta)_t] 
\qquad \qquad \Longrightarrow \qquad \qquad
\big[\big( \Gram_{2n}\super{0} \big)^{-1}\big]_{\Cap_n, \alpha_{ik}} \overset{\eqref{InvG}}{=} \frac{(-1)^n \Upsilon_{ik}}{[n+1]}. 
\end{align}  
To prove the lemma, we recursively compute the $\Upsilon_{ik}$.  
From item~\ref{EnumIt2} of recipe~\ref{recipe}, we see that the links have to be enumerated in order from innermost to outermost.
Furthermore, by item~\ref{EnumIt1} of recipe~\ref{recipe}, the link having label $1$ has to be either the innermost of the cable of size $i$ or 
the innermost of the cable of size $k$. In either case, for some $a \in \{0,1,\ldots,k\}$, the innermost $a$ links in the cable of size $k$ are enumerated from $1$
to $a$ (in order from innermost to outermost), and the innermost link in the cable of size $i$ is enumerated as $a+1$:
\begin{align}  
\vcenter{\hbox{\includegraphics[scale=0.275]{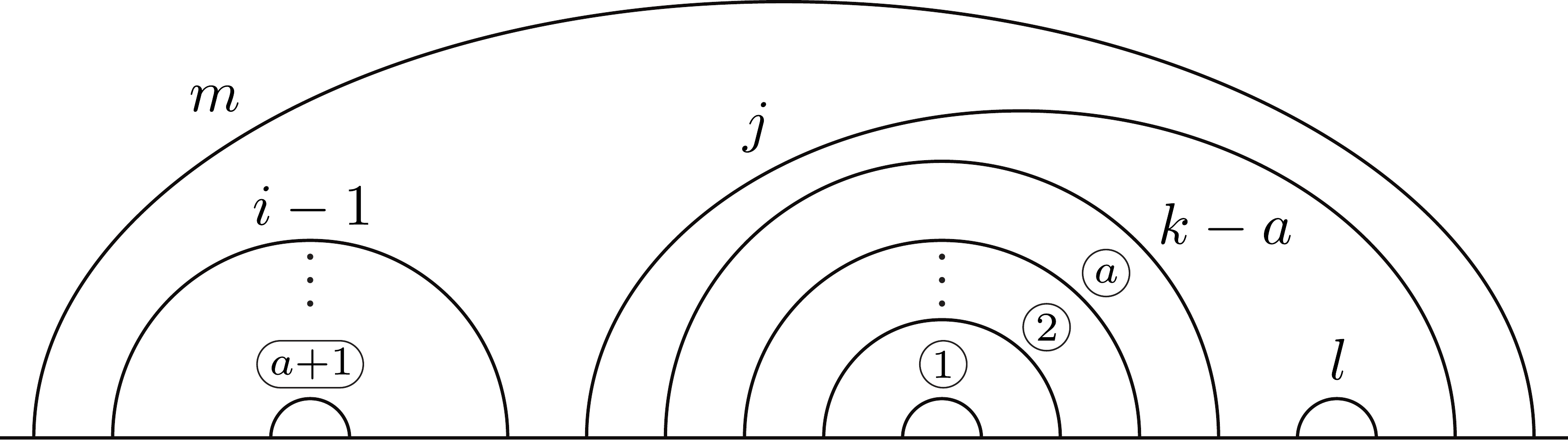} .}} 
\end{align}  
After enumerating the links of $\alpha_{ik}$ in this way, if we drop the links 
with labels $1$ through $a+1$ from $\alpha_{ik}$, 
and reduce the labels of the remaining links by $a+1$, then we obtain an enumeration for links of $\alpha_{i-1,k-a}$. This gives
\begin{align} 
\label{UpsilonRec} 
\Upsilon_{ik} &= \sum_{a = 0}^k \overbrace{[i+m]}^{1.}\overbrace{[2i + j + k + m][2i + j + k + m -1]\dotsm [2i + j + k + m - a + 1]}^{2.} \Upsilon_{i-1,k-a} \\
\label{UpsilonRec2} 
&= [i+m] [2i + j + k + m]!\sum_{a = 0}^k \frac{\Upsilon_{i-1,a}}{[2i + j + a + m]!},
\end{align}
where the factor (resp.~factors) under the first (resp.~second) brace follow (resp.~follows) from the innermost link (resp.~$a$ links) 
dropped from the cable of size $i$ (resp.~$k$) in $\alpha_{ik}$.  
Recursion~(\ref{UpsilonRec}--\ref{UpsilonRec2}) gives the closed formula
\begin{align}  \label{UpsilonFormula} 
\Upsilon_{ik} = \frac{[i+k]![i+m]![i+j+k+m]![j+l+m]!}{[i]![k]![i+j+m]![m]!} . 
\end{align}  
We prove this formula by induction on 
$i \geq 0$. For $i = 0$, we may compute $\Upsilon_{0k}$ explicitly from~\eqref{UpsilonDefn}.  
With the cable of size $i$ empty, the links in the cable of size $k$ in $\alpha_{0k}$ are necessarily enumerated one through $a$.  Thus,
\begin{align} 
\Upsilon_{0k} &= \overbrace{[ j + m ]!}^{1.} \, \overbrace{[j + m +1][j + m + 2]\dotsm[j + m +k]}^{2.} \, \overbrace{[j + m +1][j + m + 2]\dotsm[j + m + l]}^{3.} \\
&= \frac{[ j + k + m]![ j + l + m]!}{ [j + m]!},
\end{align}
where the factors under the first (resp.~second, resp.~third) brace follow from the cable of size $(j+m)$ (resp.~$k$, resp.~$l$) in $\alpha_{0k}$. 
This confirms formula~\eqref{UpsilonFormula} for the case $i = 0$.  
Now, assuming that~\eqref{UpsilonFormula} holds for $\Upsilon_{i-1,k}$, 
we prove that it holds for $\Upsilon_{ik}$ too.  Indeed, using lemma~\ref{SumFormulaLem} with $b = 2i + m + j$, we have 
\begin{align} 
\Upsilon_{ik} & 
\overset{\textnormal{(\ref{UpsilonRec}-\ref{UpsilonRec2})}}{=} 
[i+m] [2i + j + k + m]!\sum_{a = 0}^k \frac{[i + a -1]![i+m-1]![i+j+a+m-1]![j+l+m]!}{[i-1]![a]![i+j+m-1]![m]![2i + j + a + m]!} \\
& 
\underset{\hphantom{\textnormal{(\ref{UpsilonRec}-\ref{UpsilonRec2})})}}{\overset{\eqref{SumFormula}}{=}} 
\left(\frac{[i+m]![2i + j + k + m]![i+m-1]![j+l+m]!}{[i-1]![i+j+m-1]![m]!}\right) \left(\frac{[i+k]![m+i+j+k]!}{[i][k]![2i + j + k + m]![m+i+j]}\right) \\
& \underset{\hphantom{\textnormal{(\ref{UpsilonRec}-\ref{UpsilonRec2})})}}{\overset{\eqref{SumFormula}}{=}} 
\frac{[i+k]![i+m]![i+j+k+m]![j+l+m]!}{[i]![k]![i+j+m]![m]!}.
\end{align}
This proves formula~\eqref{UpsilonFormula} for $\Upsilon_{ik}$.  After inserting this into~\eqref{UpsilonDefn}, we finally obtain~\eqref{InvG2}.
\end{proof}

\begin{prop} \label{SpecialTProp}
\textnormal{\cite[proposition~\red{3.10}]{fk}} 
Suppose $n < \ppmin(q)$. Let $T \in \TL_n(\nu)$ be the tangle
\begin{align}  \label{SpecialTDiag} 
T \quad = \quad \vcenter{\hbox{\includegraphics[scale=0.275]{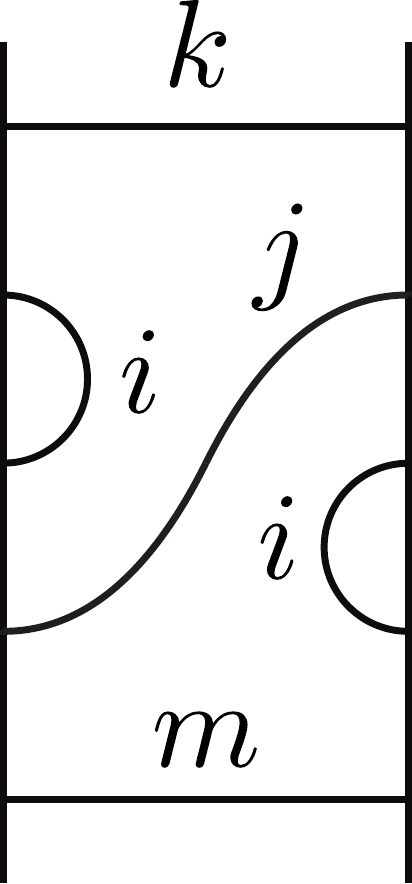}}} \quad 
\qquad \textnormal{or} \qquad \quad \vcenter{\hbox{\includegraphics[scale=0.275]{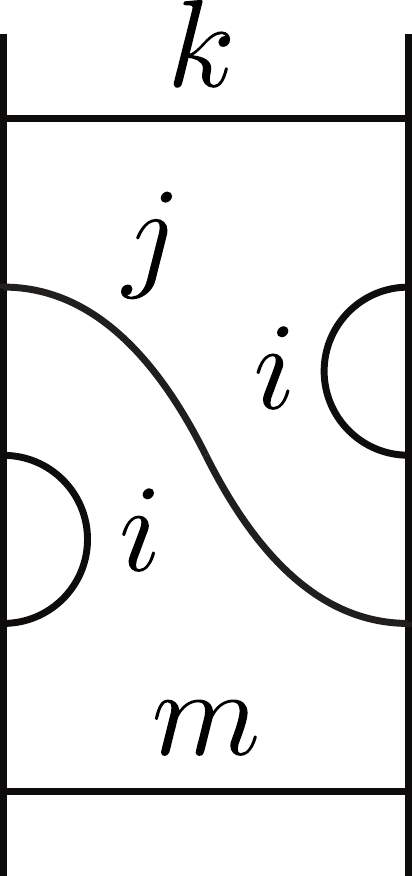} ,}}  
\end{align}  
with $2i + j + k + m = n$. Then we have
\begin{align} \label{SpecialT} 
\textnormal{coef}_T = \frac{[i+k]![i+m]![i+j+k+m]!}{[n]![i]![k]![m]!}. 
\end{align}   
\end{prop} 
\begin{proof} 
The sum in formula~\eqref{WJCoeff} for $\textnormal{coef}_T$ equals the sum $\Upsilon_{ik}$, defined in~\eqref{UpsilonDefn} in the proof of lemma~\ref{LinkPattLem}, 
with $l = i$.  In that proof, we computed a closed formula~\eqref{UpsilonFormula} for this sum, and after inserting this formula into~\eqref{WJCoeff}, we obtain~\eqref{SpecialT}.  (Alternatively, to prove the lemma, we may use~\eqref{ProjExplicit2} of lemma~\ref{InvMeandLem} with~\eqref{InvG2} of lemma~\ref{LinkPattLem}.) 
\end{proof}

\section{Diagram simplifications and evaluations} \label{TLRecouplingSect}

In this appendix, we collect auxiliary results concerning ``networks'' and their evaluations.
Such notions appear in Temperley-Lieb recoupling theory~\cite{pen, kl, cfs}, and are extensively used in our companion article~\cite{fp0}.

First, we define the \emph{Theta network}~\cite{kl} to be the tangle
\begin{align} \label{ThetaDefinition} 
\vcenter{\hbox{\includegraphics[scale=0.275]{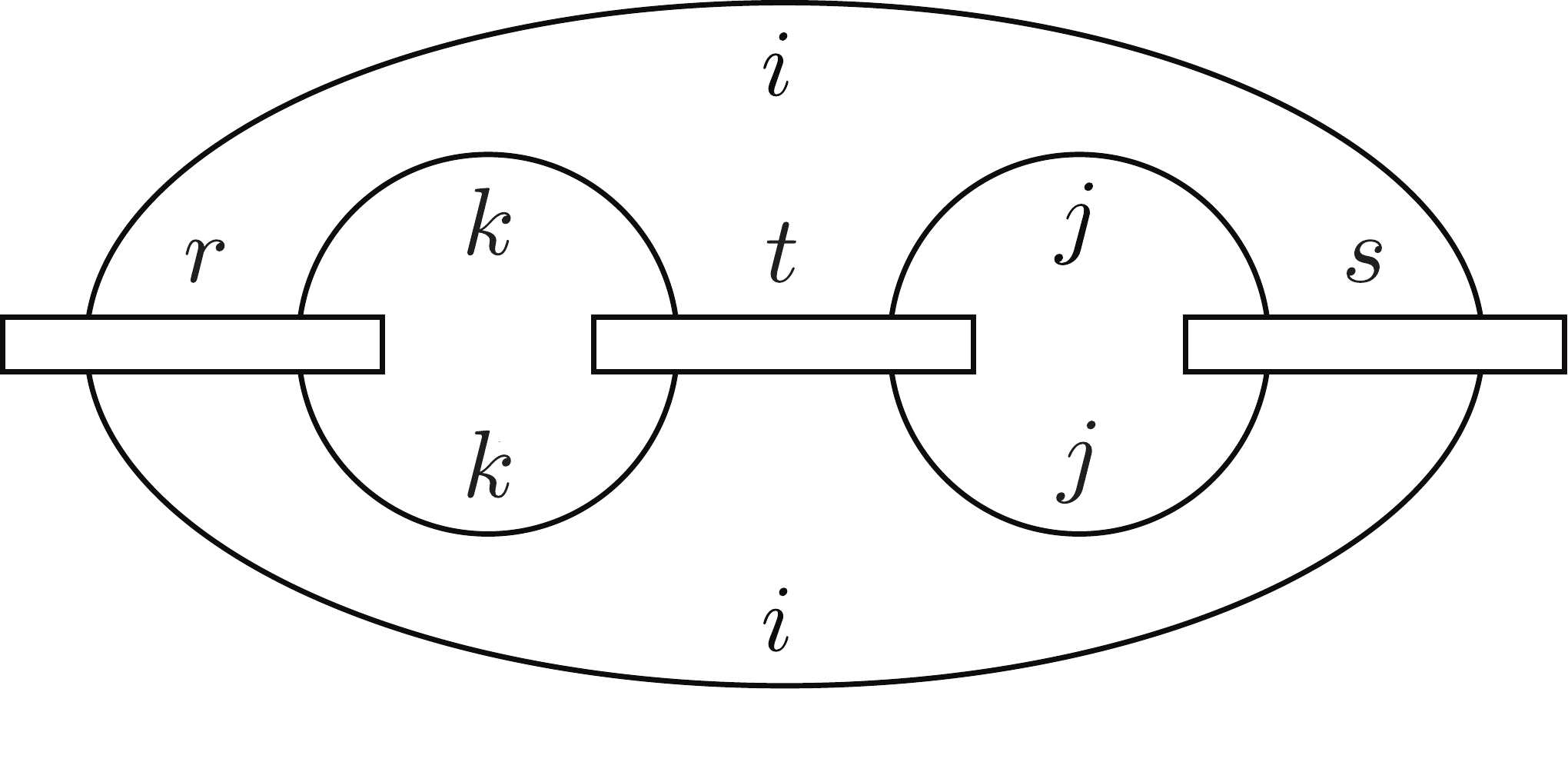}}} \quad 
& = \quad \vcenter{\hbox{\includegraphics[scale=0.275]{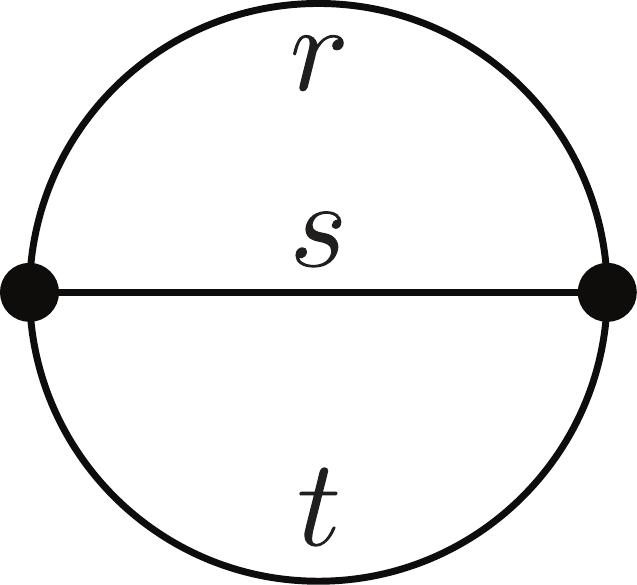} ,}} \qquad \qquad
\begin{aligned} 
i & = \frac{r + s - t}{2} , \\[.7em] 
j & = \frac{s + t - r}{2} , \\[.7em] 
k & = \frac{t + r - s}{2} .
\end{aligned}
\end{align} 
We denote the evaluation~\eqref{evT} of the Theta network by $\ThetaNet(r,s,t)$. It has the following explicit formula:

\begin{lem} \label{ThetaLem} \textnormal{\cite[lemma~\red{A.7}]{fp0}} 
Suppose $\max(r,s,t) < \ppmin(q)$. Then we have 
\begin{align} \label{ThetaFormula1}
\ThetaNet(r,s,t)
= \frac{(-1)^{\frac{r + s + t}{2}} \left[ \frac{r + s + t}{2} + 1 \right]! \left[ \frac{ r + s - t }{2} \right]! \left[ \frac{ s + t - r}{2} \right]! \left[ \frac{t + r - s}{2} \right]! }{[ r ]! [s ]! [ t ]!} . 
\end{align}
\hfill \qed
\end{lem}

We also define the \emph{Tetrahedral network}~\cite{kl} and its evaluation~\eqref{evT} as 
\begin{align} \label{TetraDefinition}
T \quad = \quad 
\vcenter{\hbox{\includegraphics[scale=0.275]{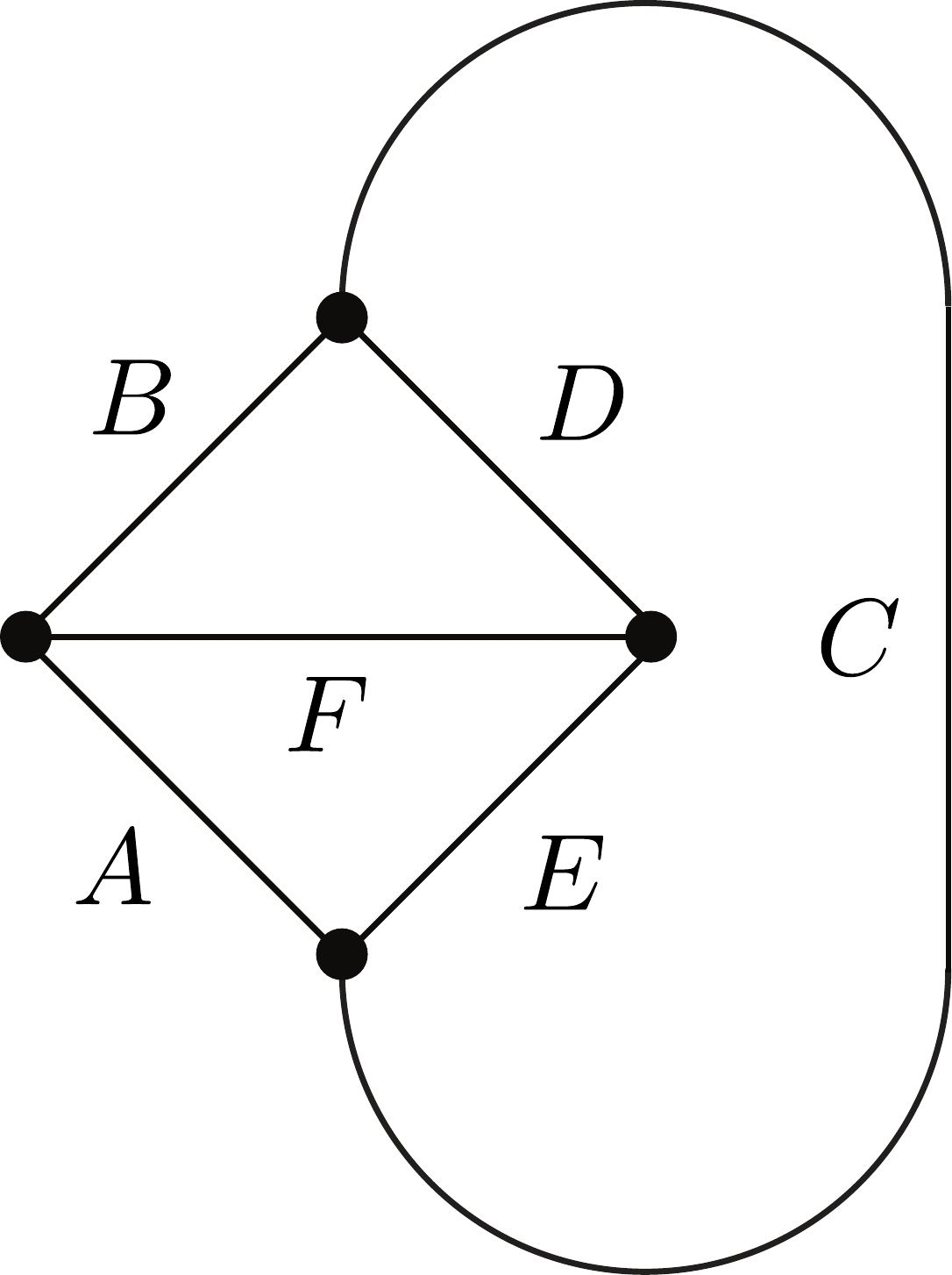}}} \quad 
\qquad \qquad \Longrightarrow \qquad \qquad
(T) \;  = \; \TetraNet\left[ 
\begin{array}{ccc} 
A & B & C \\ 
D & E &F 
\end{array} \right] \; .
\end{align}

The following extraction rules, even though simple, are very useful.

\begin{lem} \label{CollectionLem} \
\begin{enumerate} 
\itemcolor{red}
\item \label{LoopLemGenItem} \textnormal{\cite[lemma~\red{A.1}]{fp0}}
Suppose $s + r < \ppmin(q)$.  Then we have the following extraction rule:
\begin{align} \label{DeltaTangleGen} 
\vcenter{\hbox{\includegraphics[scale=0.275]{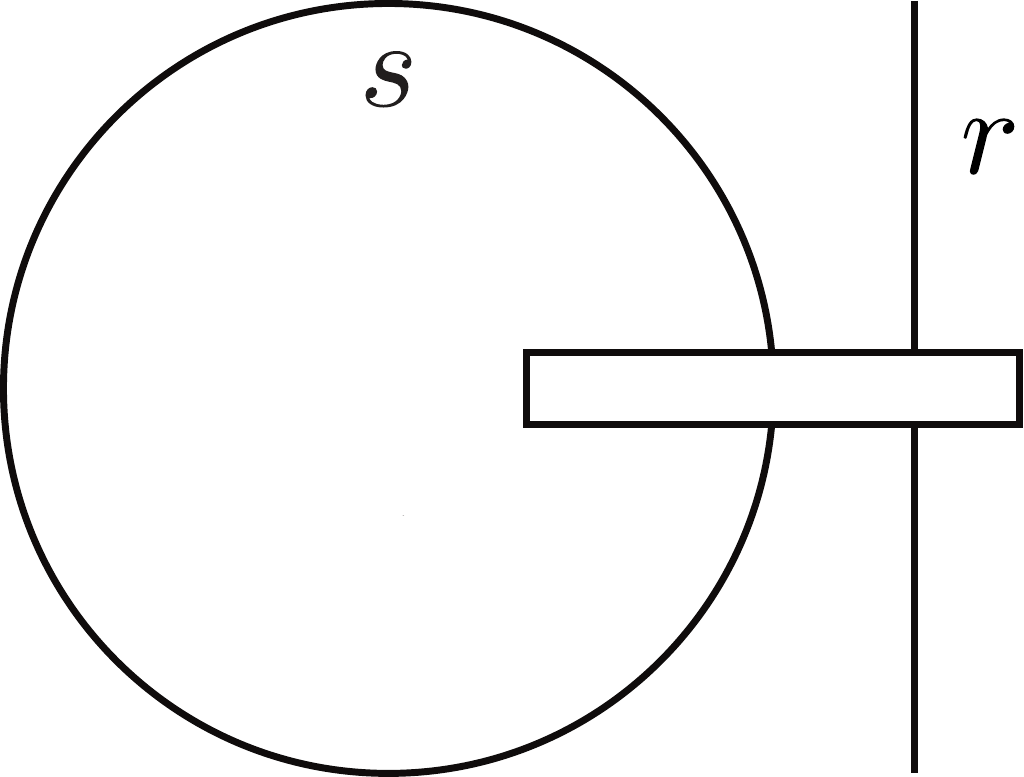}}}
\quad = \quad (-1)^s \frac{[r+s+1]}{[r+1]} 
 \,\, \times \,\,  \vcenter{\hbox{\includegraphics[scale=0.275]{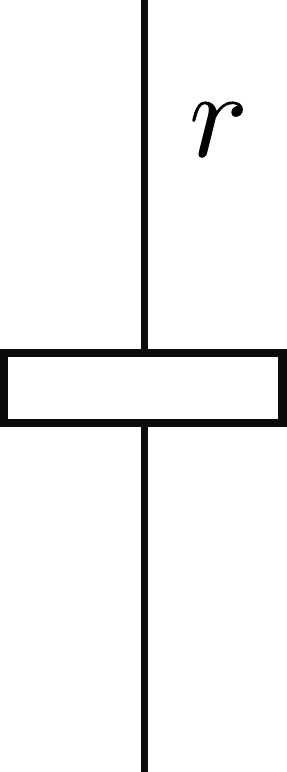} .}} 
\end{align}

\item \label{ExtractLemItem} \textnormal{\cite[lemma~\red{A.4}]{fp0}}
Suppose $s < \ppmin(q)$.  Then, for any network $T$ contained between two projector boxes within a larger network, 
we have the following extraction rule:
\begin{align} \label{ExtractID} 
\vcenter{\hbox{\includegraphics[scale=0.275]{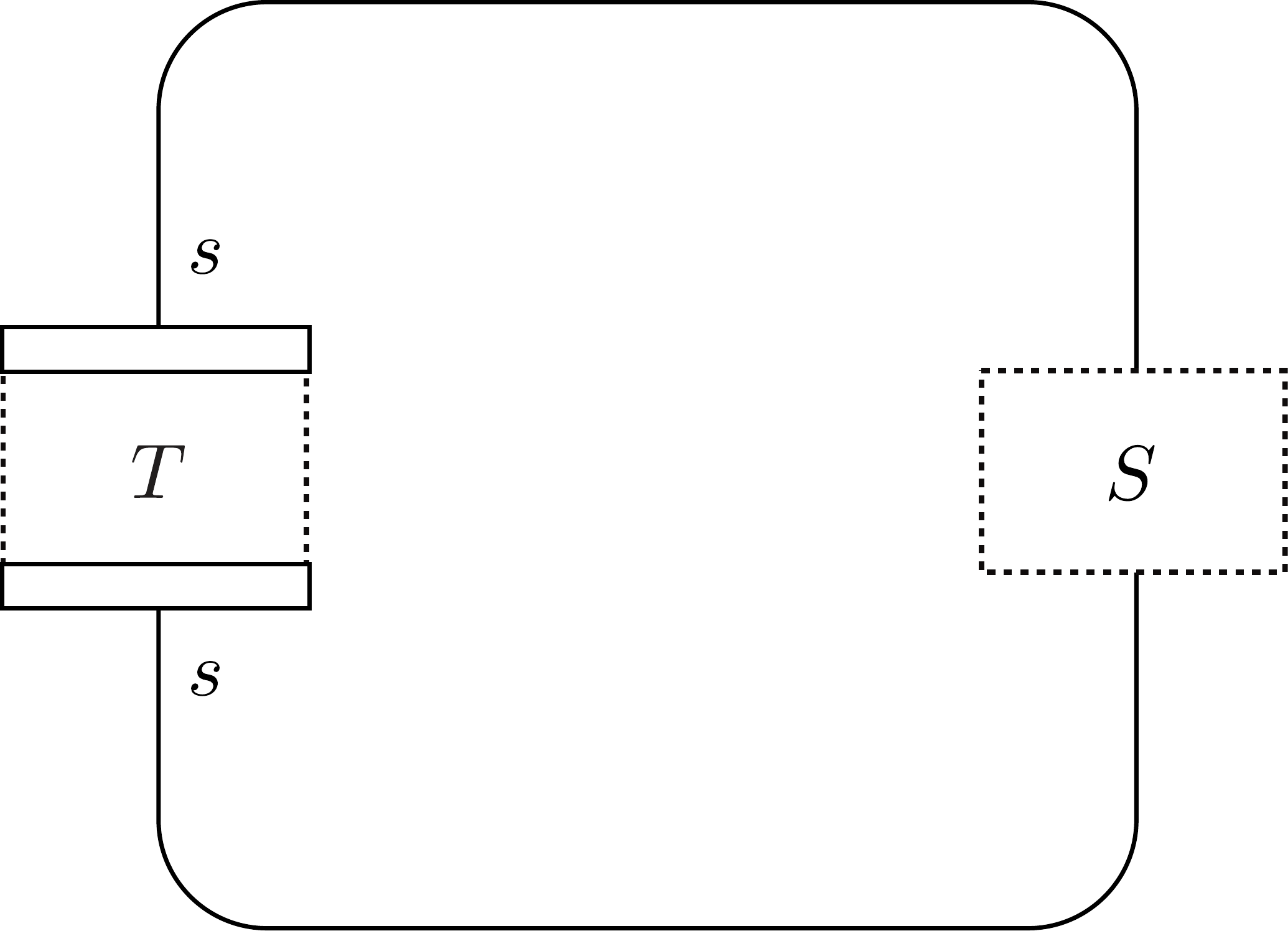}}} 
\quad = \quad (T) \,\, \times \,\, \vcenter{\hbox{\includegraphics[scale=0.275]{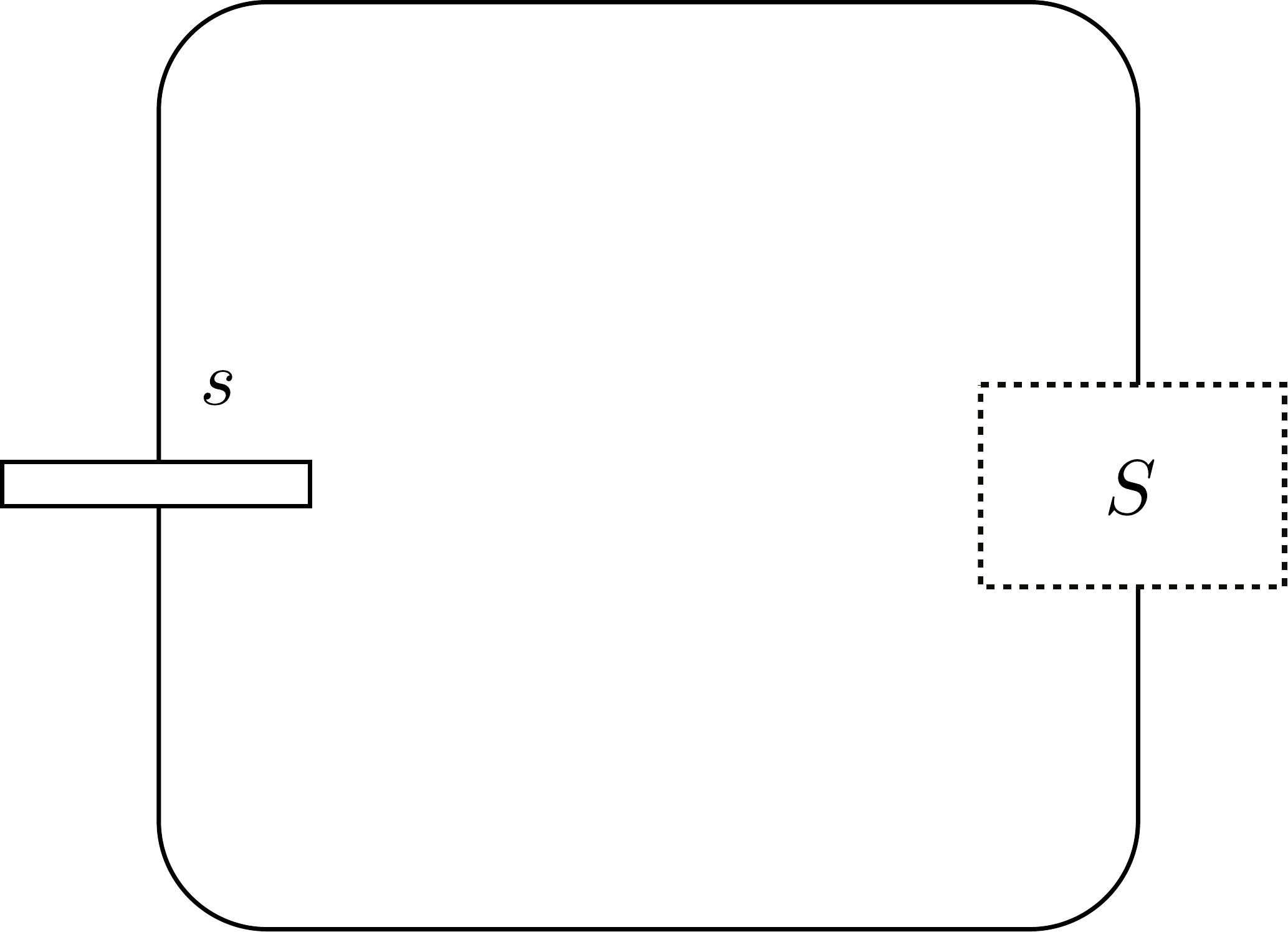} .}}
\end{align} 

\item \label{LoopErasureLemItem} \textnormal{\cite[lemma~\red{A.5}]{fp0}}
Suppose $\max(r,s,s',t) < \ppmin(q)$.  Then we have 
\begin{align}\label{LoopErasure1} 
 \left( \; \vcenter{\hbox{\includegraphics[scale=0.275]{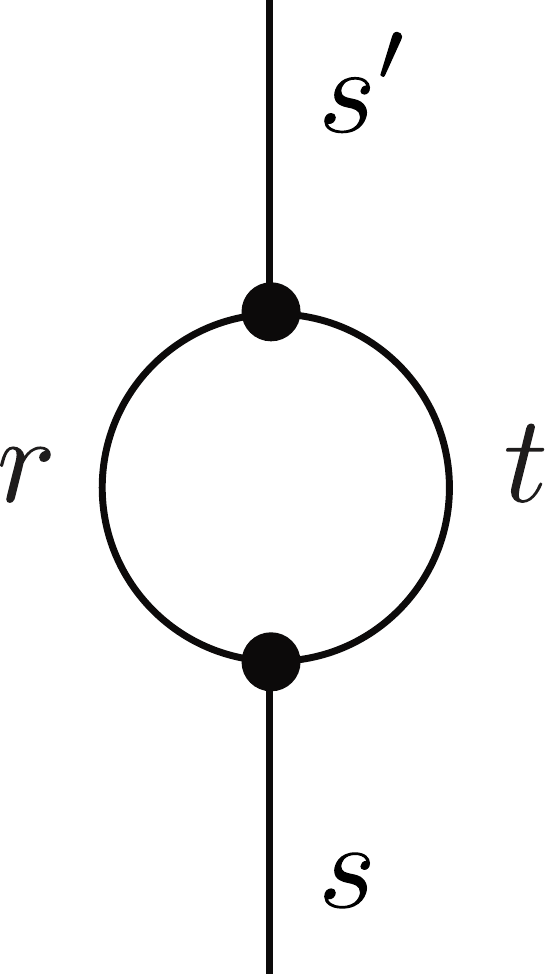}}} \;  \right) \quad 
= \quad \delta_{s, s'} \frac{ \ThetaNet(r,s,t) }{(-1)^s [s+1]} . 
\end{align} 

\item \label{QintIDItem}  \textnormal{\cite[item~\red{1} of lemma~\red{A.6}]{fp0}}
The following identity holds for all $i,j,k \in \bZ$\textnormal{:}
\begin{align} \label{QintID} 
[i] [ j - k]+[j] [ k - i ]+[k] [ i - j ] = 0. 
\end{align} 
\end{enumerate}
\hfill \qed
\end{lem}

Next, we prove some further diagram identities that we need in section~\ref{GeneratorLemProofSect}.

\begin{lem} \label{NetEvalLem} 
We have the following network evaluations:
\begin{enumerate}
\itemcolor{red}
\item \label{NetworkIt1} 
For $s < \ppmin(q)$, we have
\begin{align} \label{PreLoopBox} 
 \left( \; \vcenter{\hbox{\includegraphics[scale=0.275]{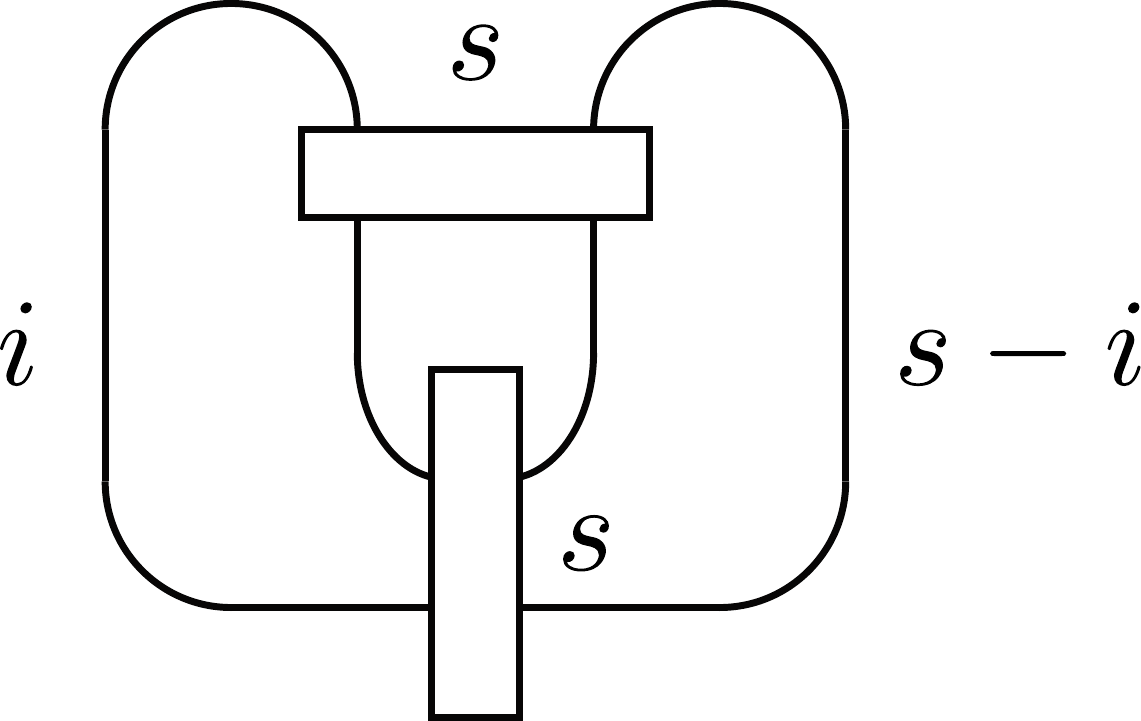}}} \; \right)
\quad = \quad (-1)^s[s+1] \frac{[i]! [s-i]!}{[s]!} .
\end{align} 

\item \label{NetworkIt2} 
For $s < \ppmin(q)$, we have
\begin{align}  \label{LoopBox} 
 \left( \; \vcenter{\hbox{\includegraphics[scale=0.275]{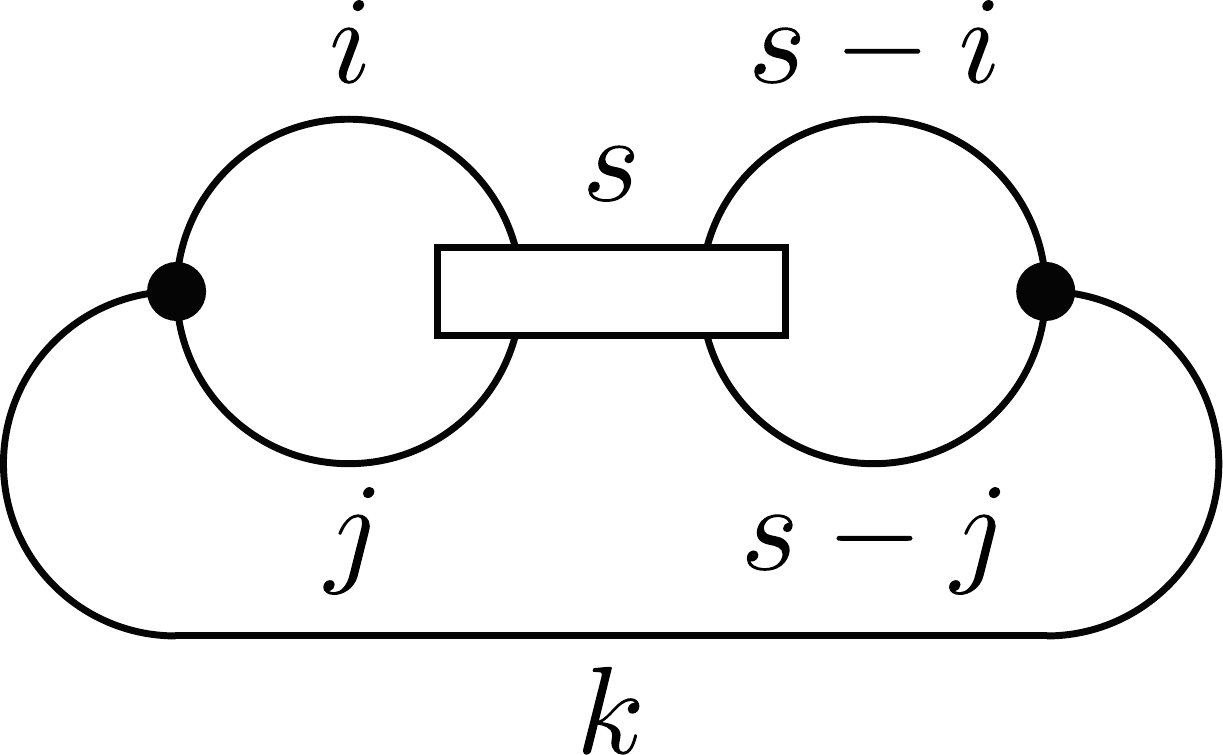}}} \; \right)
\quad = \quad \frac{(-1)^k[s+1]^2}{[k+1]} \frac{[\frac{i-j+k}{2}]! [s-\frac{i-j+k}{2}]!}{[s]!} .
\end{align} 
\item \label{NetworkIt3} 
For $\max(A,B+2,F) < \ppmin(q)$, we have 
\begin{align} \label{SpecialTetNetwork} 
\TetraNet\left[ 
\begin{array}{ccc} 
A & B & 2 \\ 
B+2 & A &F 
\end{array} \right] 
= \frac{1}{[A]} \left[\frac{A - B + F}{2}\right] \ThetaNet(A,F,B+2). 
\end{align} 
\end{enumerate}
\end{lem}

\begin{proof} 
We prove formulas~\eqref{PreLoopBox}--\eqref{SpecialTetNetwork} in items~\ref{NetworkIt1}--\ref{NetworkIt3} as follows:
\begin{enumerate}[leftmargin=*]
\itemcolor{red}
\item 
Decomposing the upper projector box over all internal link diagrams as in~\eqref{ProjDecomp}, we see by rule~\eqref{ProjectorID2} that only 
tangle~\eqref{SpecialTDiag} in proposition~\ref{SpecialTProp} with $m = k = 0$ contributes a nonvanishing term. 
Thus, using~\eqref{SpecialT}, we get
\begin{align} 
\vcenter{\hbox{\includegraphics[scale=0.275]{e-NeatTangle1.pdf}}} \quad
& \overset{\eqref{SpecialT}}{=} 
\quad  \frac{[i]! [s-i]!}{[s]!} \,\, \times \,\,
\begin{cases} 
\quad \vcenter{\hbox{\includegraphics[scale=0.275]{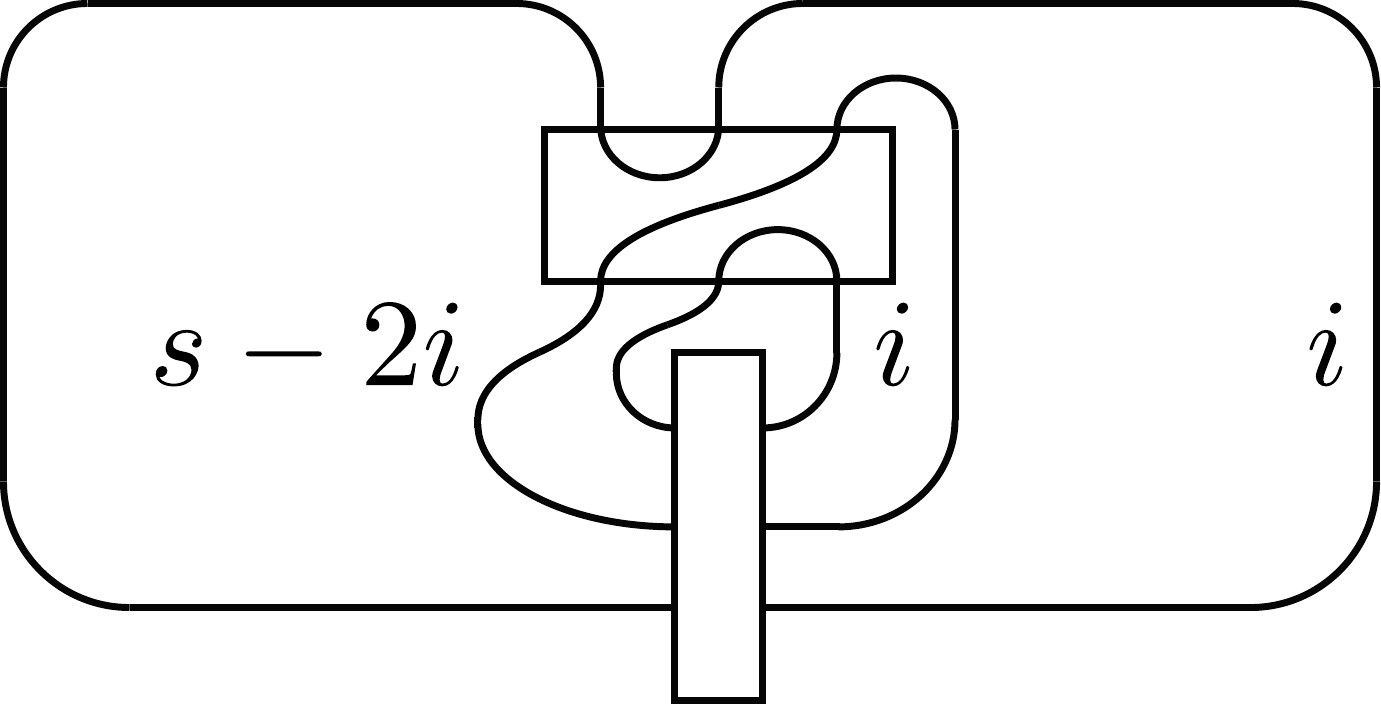} ,}} & \qquad s \geq 2i \\
\quad \vcenter{\hbox{\includegraphics[scale=0.275]{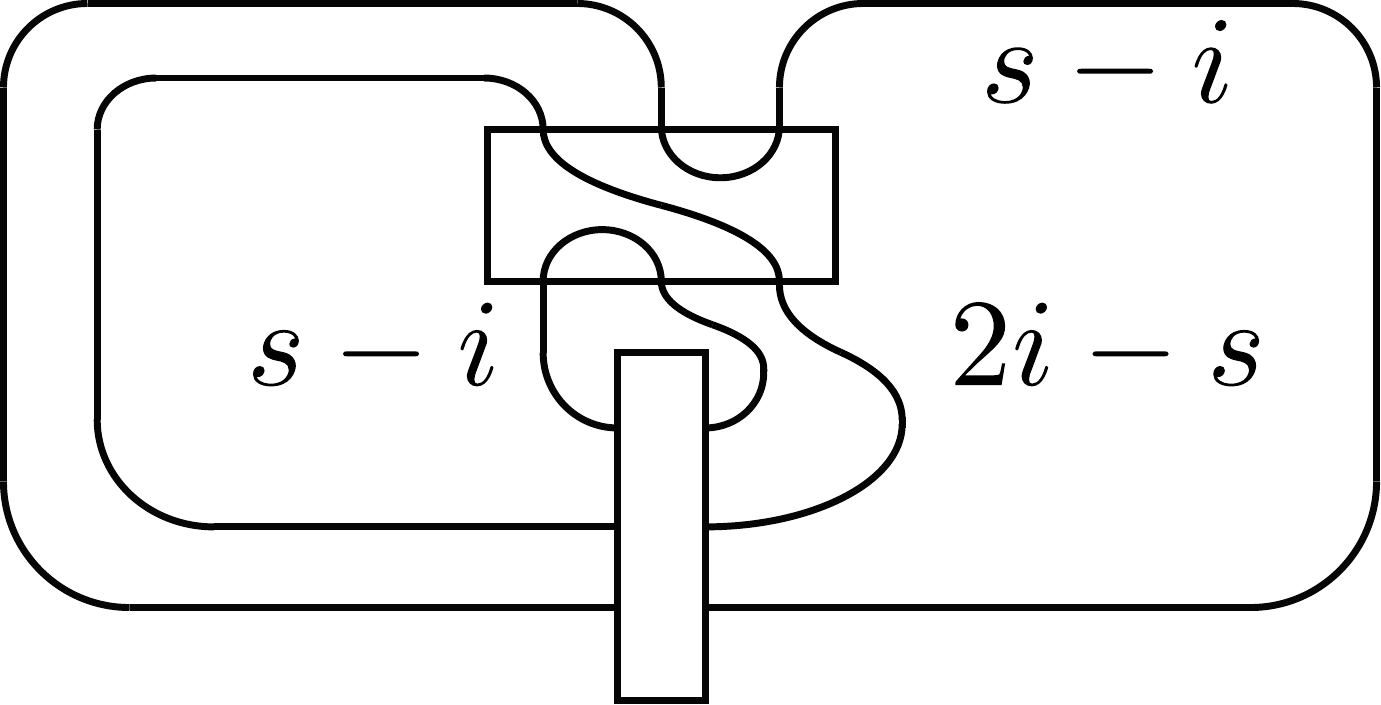} ,}}  & \qquad s \leq 2i. 
\end{cases} 
\end{align}
Now, simplification rule~\eqref{DeltaTangleGen} from 
lemma~\ref{CollectionLem} with $r=0$ gives asserted formula~\eqref{PreLoopBox}.

\item 
We decompose the three-vertices via definition~\eqref{3vertex1}, obtaining
\begin{align} 
\vcenter{\hbox{\includegraphics[scale=0.275]{e-NeatTangle4.pdf}}} \quad
\overset{\eqref{3vertex1}}{=} \quad \vcenter{\hbox{\includegraphics[scale=0.275]{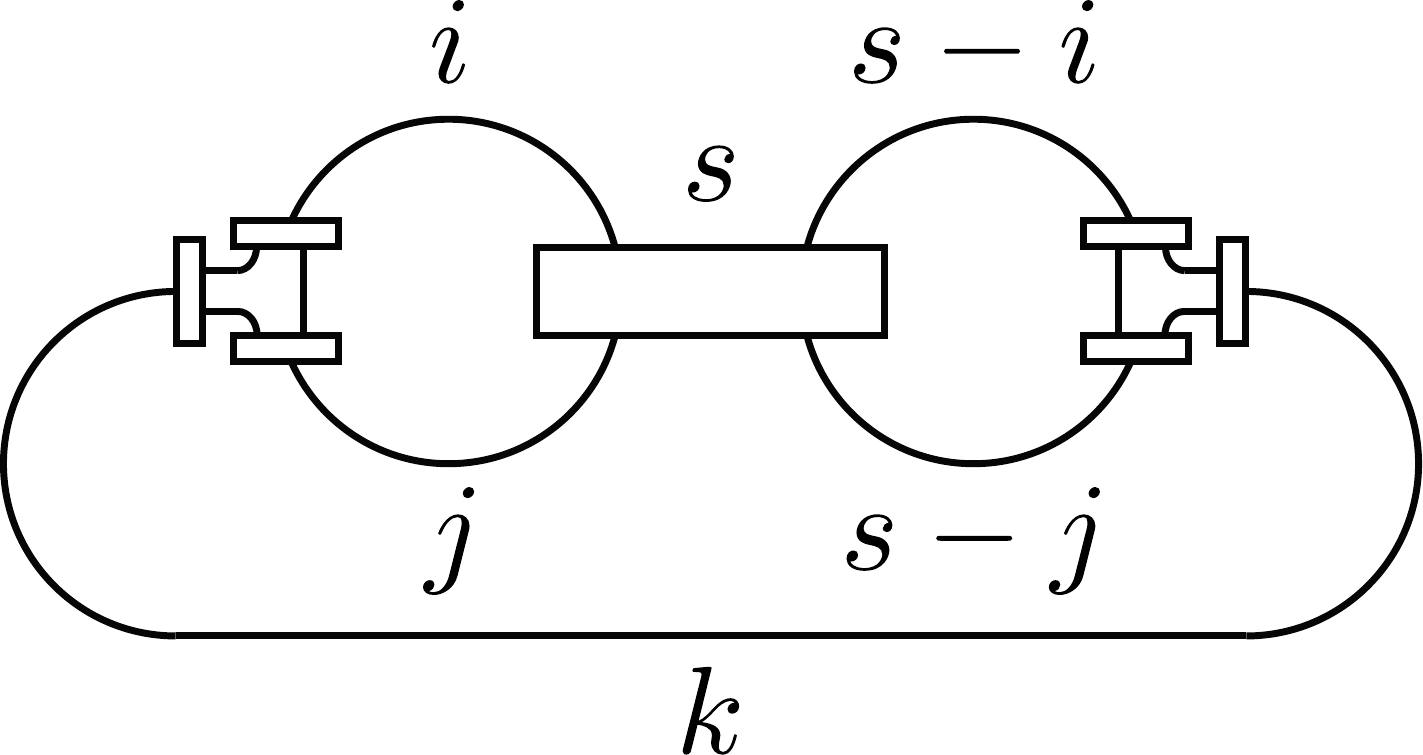}}} \quad
\overset{\eqref{ProjectorID1}}{=} \quad \vcenter{\hbox{\includegraphics[scale=0.275]{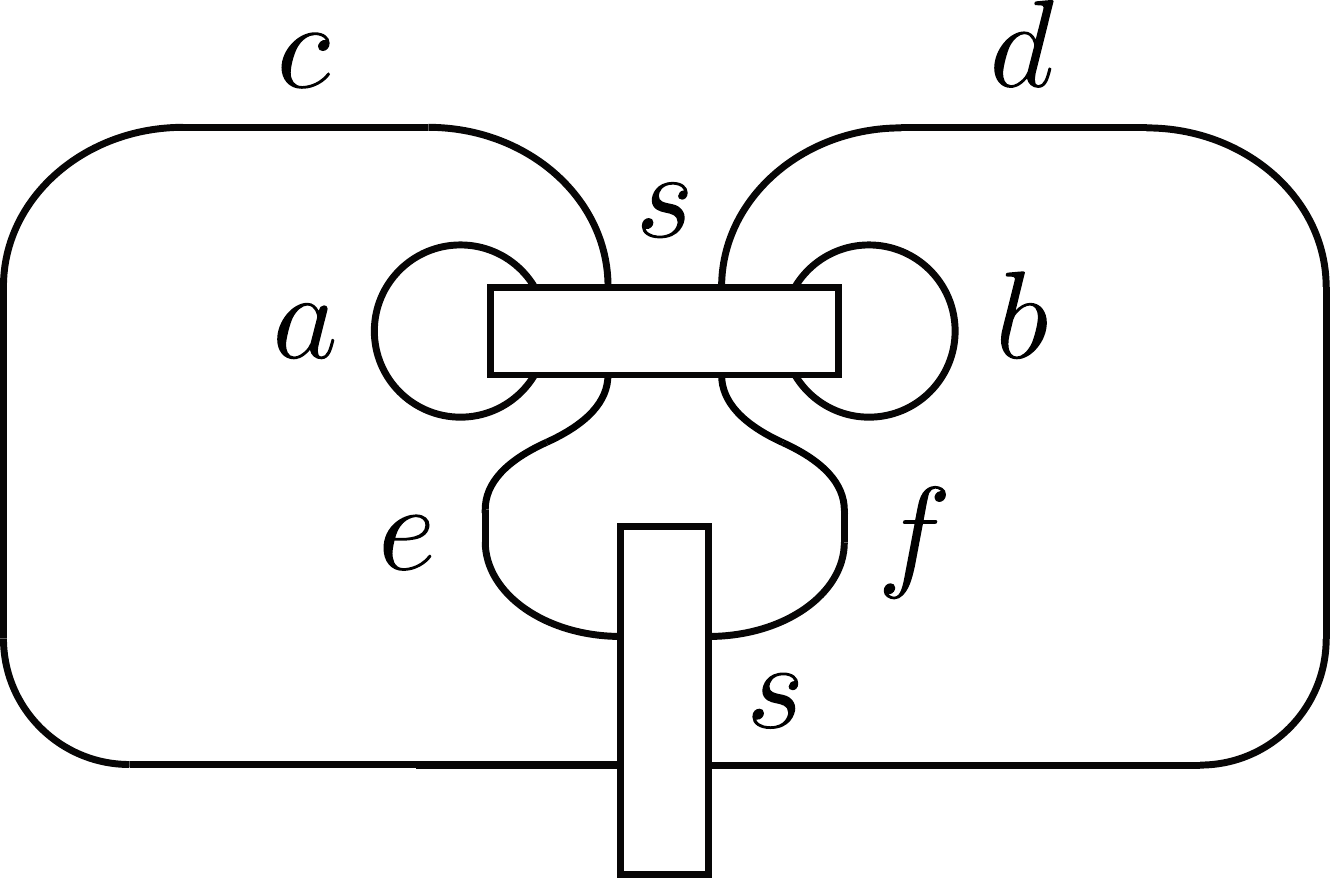} ,}} 
\end{align} 
where 
\begin{align} 
\begin{aligned} 
a = \; & \frac{i+j-k}{2} , \qquad \qquad b = \frac{2s-i-j-k}{2} , \qquad \quad & c = \frac{i + k - j}{2} , \\
d = \; & \frac{j-i+k}{2}, \qquad \qquad e = \frac{k + j - i}{2}, \qquad \quad \qquad  & f = \frac{i + k - j}{2} .
\end{aligned}
\end{align}
Using identity~\eqref{DeltaTangleGen} from 
lemma~\ref{CollectionLem}, we find 
\begin{align} \label{PrePreLoopBox}
\vcenter{\hbox{\includegraphics[scale=0.275]{e-NeatTangle6.pdf}}} \quad
\overset{\eqref{DeltaTangleGen}}{=} \quad 
\frac{(-1)^{a+b}[s+1]}{[s-a-b+1]}  \,\, \times \,\, \vcenter{\hbox{\includegraphics[scale=0.275]{e-NeatTangle1.pdf} .}} 
\end{align} 
After inserting evaluation~\eqref{PreLoopBox} with $i = c$ of the tangle on the right side into~\eqref{PrePreLoopBox}, we obtain~\eqref{LoopBox}.  

\item Finally, to prove item~\ref{NetworkIt3}, we first write asserted equality~\eqref{SpecialTetNetwork} in terms of networks:
\begin{align} \label{SpecialTetNetwork0}
\vcenter{\hbox{\includegraphics[scale=0.275]{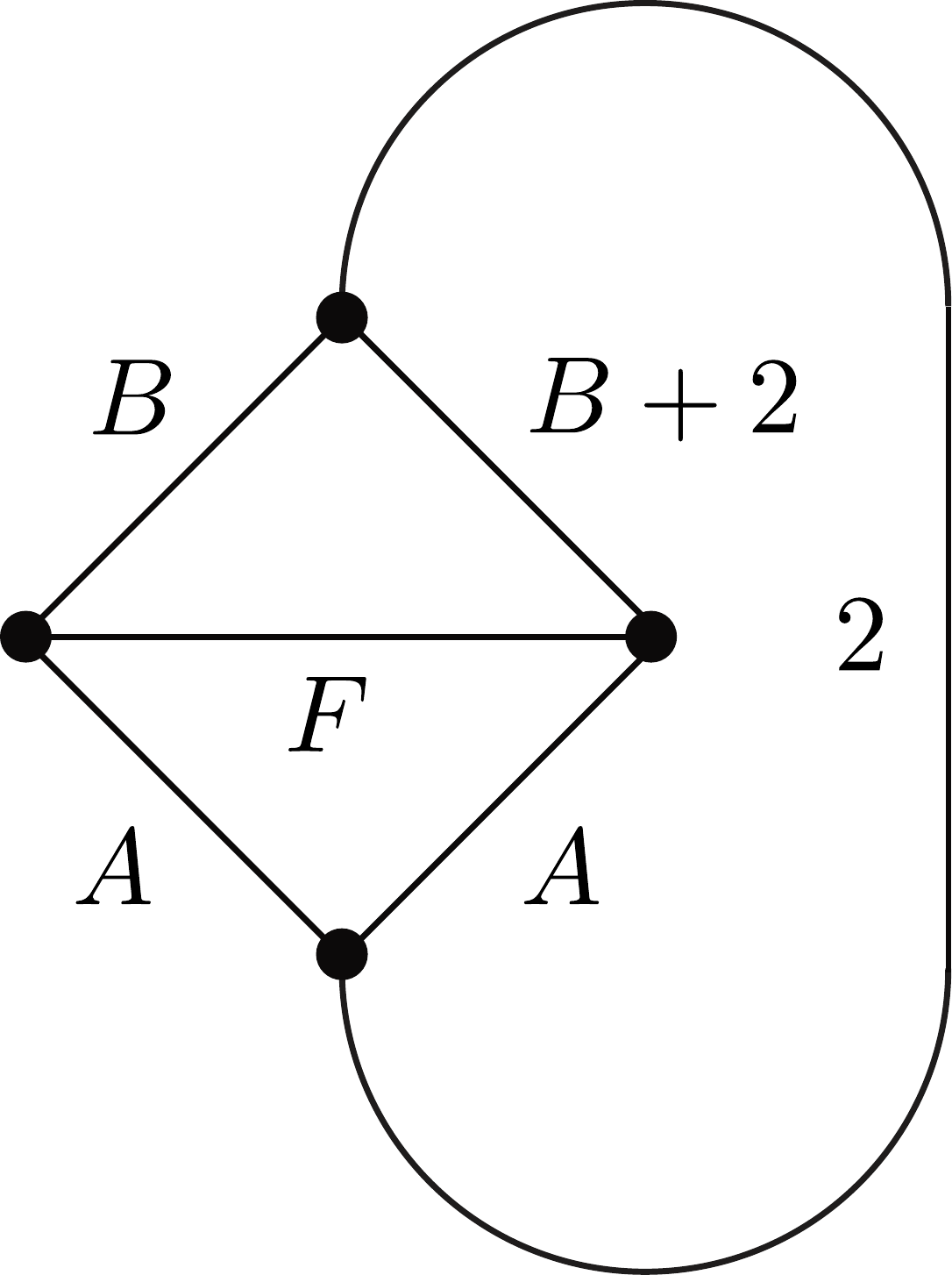}}}
\quad = \quad \frac{1}{[A]} \left[\frac{A-B+F}{2}\right]
 \,\, \times \,\, \vcenter{\hbox{\includegraphics[scale=0.275]{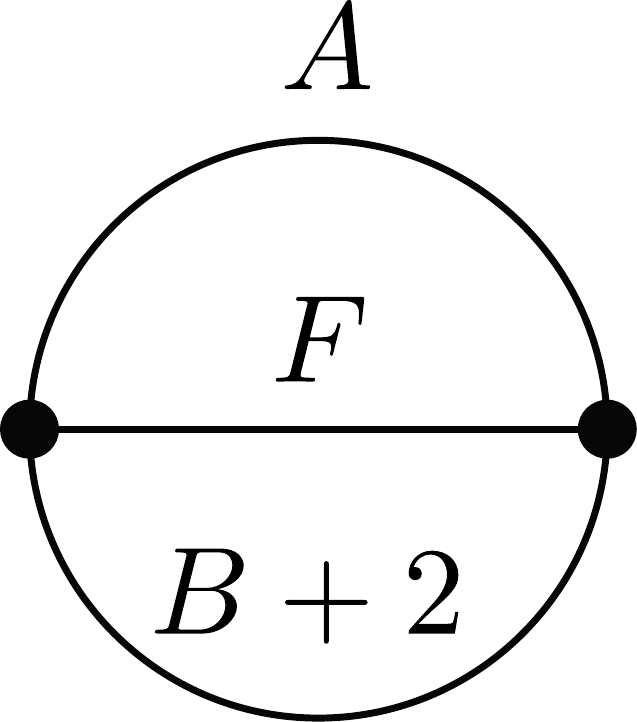} .}}
\end{align} 
By definition~\eqref{3vertex1} the top and bottom three-vertices of the Tetrahedral network 
respectively simplify to
\begin{align} 
\label{SimpleVertex1} 
& \vcenter{\hbox{\includegraphics[scale=0.275]{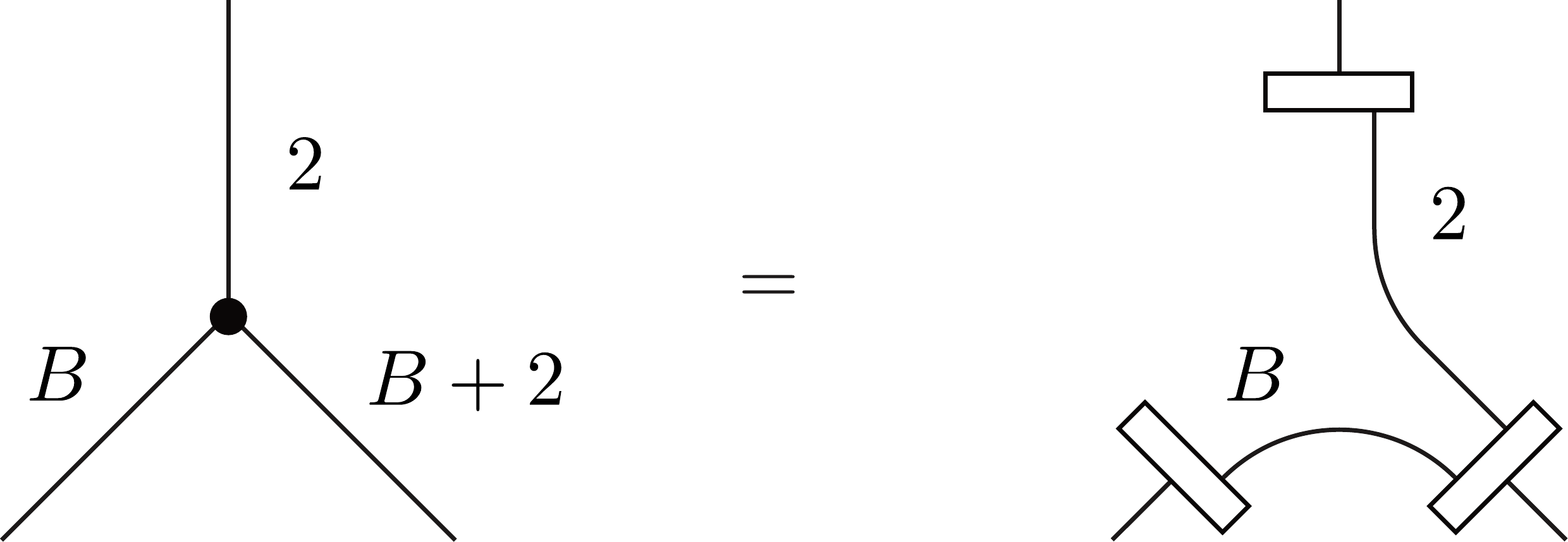} ,}}  \\[1.5em]
\label{SimpleVertex2} 
& \vcenter{\hbox{\includegraphics[scale=0.275]{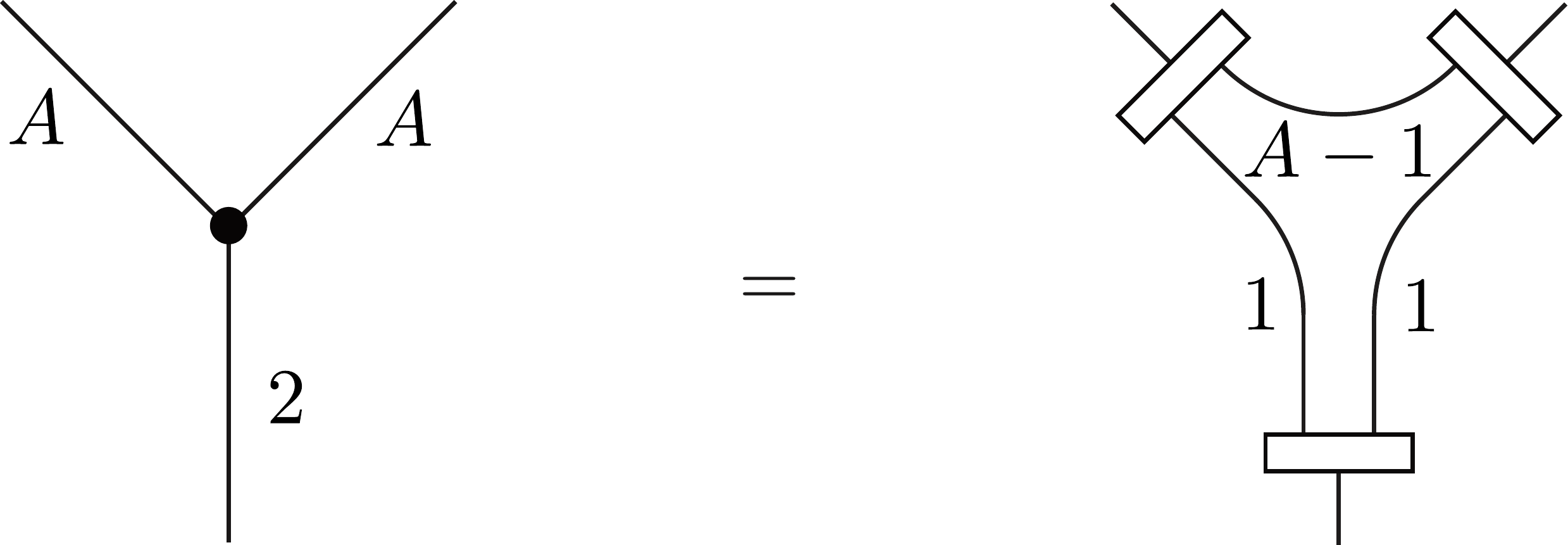} .}} 
\end{align} 
Using~(\ref{SimpleVertex1}--\ref{SimpleVertex2}) and rule~\eqref{ProjectorID0}, 
we simplify the Tetrahedral network on the left side of~\eqref{SpecialTetNetwork0} to  
\begin{align} \label{SimpleSpecialTetNetwork}
\vcenter{\hbox{\includegraphics[scale=0.275]{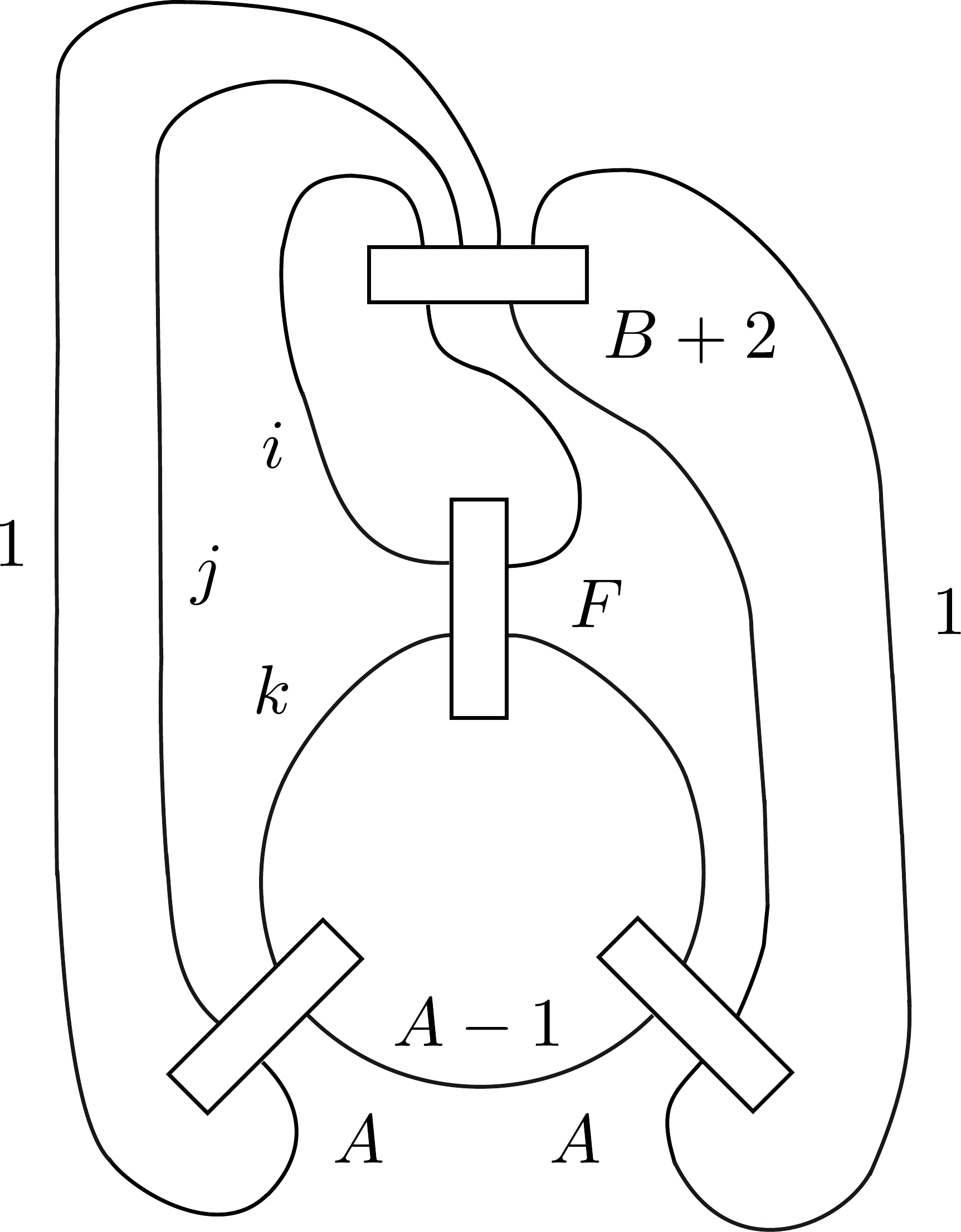} ,}}
\end{align} 
where
\begin{align} 
i = \frac{B + F - A}{2} , \qquad \quad
j = \frac{B + A - F}{2} , \qquad \quad
k = \frac{A + F - B}{2} .
\end{align}
Now, when decomposing the bottom left projector box of size $A$ in~\eqref{SimpleSpecialTetNetwork}, 
rule~\eqref{ProjectorID2} shows that the only nonzero contribution comes from the term with exactly one turn-back link.
Again, we find the coefficient of this term using~\eqref{SpecialT} of proposition~\ref{SpecialTProp}.
Thus, network~\eqref{SimpleSpecialTetNetwork} further simplifies to
\begin{align} \label{SimpleSpecialTetNetwork2}
\frac{1}{[A]} \left[\frac{A-B+F}{2}\right]  \,\, \times \,\, \vcenter{\hbox{\includegraphics[scale=0.275]{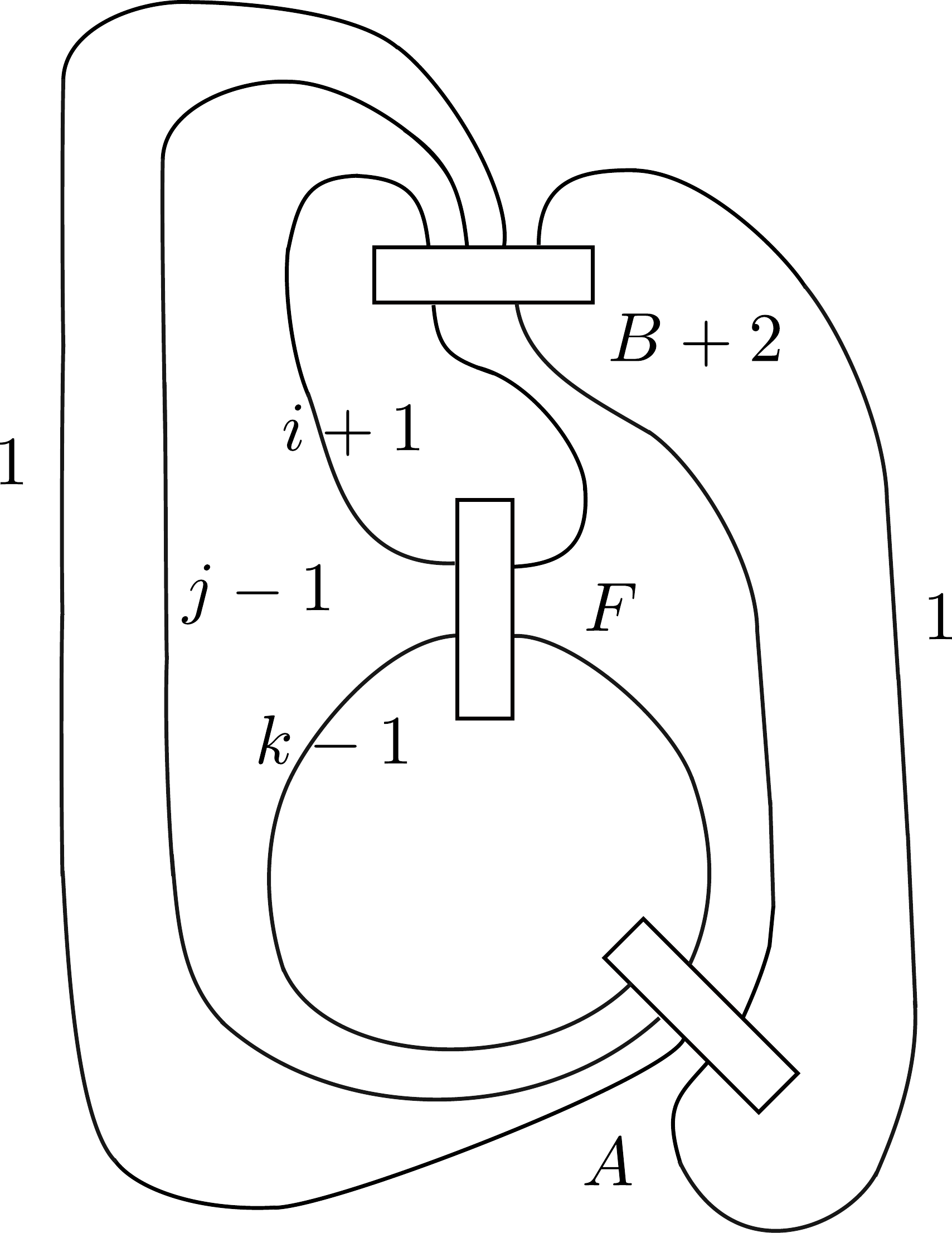} .}} 
\end{align} 
It remains to note that the network in~\eqref{SimpleSpecialTetNetwork2} is exactly the Theta network 
appearing on the right side of~\eqref{SpecialTetNetwork0}.  
This proves item~\ref{NetworkIt3} and concludes the proof.
\end{enumerate}
\end{proof}

To finish, we settle two technical details concerning the evaluation of the Theta network,
for use in sections~\ref{GeneratorLemProofSect}--\ref{RelationLemProofSect}.

\begin{lem} \label{FiniteAndNonzeroLem} 
Suppose $\max (r, t) < \ppmin(q)$. Then for all indices
$s \in \DefectSet\sub{r,t} = \{ |r -t|, |r -t| + 2 , \ldots , r + t \}$ 
and for all indices
$k \in \big\{\frac{r+t-s}{2}, \frac{r+t-s}{2} + 1 , \ldots , \min(r,t) \big\}$, we have
\begin{align} \label{ShouldBeFiniteAndNonzero}
0 < \bigg| \frac{\left[\frac{r+s-t}{2}\right]!\left[\frac{t+s-r}{2}\right]!\left[\frac{r+t+s}{2}-k\right]!}{[s]!\left[k-\frac{r+t-s}{2}\right]![r-k]![t-k]!} \bigg| < \infty. 
\end{align}
\end{lem}
\begin{proof}
First, we note that 
\begin{align}
\label{condi1General2ndoe-f}
\frac{r+s-t}{2} & \leq r \leq \max (r, t) < \ppmin(q) \\
\frac{t+s-r}{2} & \leq t \leq \max (r, t) < \ppmin(q) \\
k-\frac{r+t-s}{2} & \leq \min(r,t)-\frac{1}{2}(\max(r,t)+\min(r,t)-s)
\leq \min(r,t) \leq \max(r,t) < \ppmin(q) \\
r-k & \leq r- \frac{1}{2}(r+t-s) = \frac{1}{2}(r+s-t) < \ppmin(q) \\
t-k & \leq t- \frac{1}{2}(r+t-s) = \frac{1}{2}(t+s-r)  < \ppmin(q) .
\label{condi1General2ndoe-l}
\end{align}
By definition~\eqref{Qinteger}, the quantum integer $[k]$ does not vanish if $k \in \{ 0, 1, \ldots, \ppmin(q) - 1 \}$, 
so~(\ref{condi1General2ndoe-f}--\ref{condi1General2ndoe-l}) 
show that all of the quantum factorials in~\eqref{ShouldBeFiniteAndNonzero}, except possibly
$[s]!$ and $\left[\frac{r+t+s}{2}-k\right]!$, are nonzero.

Second, we show that the ratio of the remaining factors in~\eqref{ShouldBeFiniteAndNonzero} is also finite and nonzero:
by the assumption on $k$, we have, for any fixed $s \in \DefectSet\sub{r,t}$,
\begin{align} 
\left[\frac{1}{2}(r+t+s-k)\right]!
& = \left[\frac{1}{2}(r+t+s)-\frac{1}{2}(r+t-s)\right] \\
& = \left[\frac{1}{2}(r+t+s)-\frac{1}{2}(r+t-s) + 1\right]
\cdots \left[\frac{1}{2}(r+t+s)-\min(r,t)\right] \\
& = \left[\frac{1}{2}(|r -t| +s)\right]
\left[\frac{1}{2}(|r -t| +s) + 1\right]
\cdots \left[s\right] \\
\label{Ratio}
\Longrightarrow \qquad
\frac{\left[\frac{r+t+s}{2}-k\right]!}{[s]!}
& = \left(\left[\frac{|r -t| +s}{2} - 1\right]!\right)^{-1} .
\end{align}
Now, because $s \in \DefectSet\sub{r,t}$, we have
\begin{align}
\frac{|r -t| +s}{2} - 1 \, \leq \, \frac{|r -t| + r + t}{2} - 1
\leq \max (r, t) - 1 \, < \, \ppmin(q) ,
\end{align}
so the ratio in~\eqref{Ratio} is indeed finite. This concludes the proof.
\end{proof}

\begin{lem} \label{FiniteAndNonzeroLem2} 
Suppose $\max (r, t) < \ppmin(q)$. Then for all indices
$s \in \DefectSet\sub{r,t} = \{ |r -t|, |r -t| + 2 , \ldots , r + t \}$, we have
\begin{align} \label{ThetaNetShouldBeFiniteAndNonzero0}
0 < \bigg| \frac{\ThetaNet(r,t,s)}{(-1)^s [s+1]} \bigg| < \infty. 
\end{align}
\end{lem}
\begin{proof}
Using lemma~\ref{ThetaLem}, we write the quantity of interest in~\eqref{ThetaNetShouldBeFiniteAndNonzero0} in the form
\begin{align}  \label{ThetaNetShouldBeFiniteAndNonzero}
\frac{\ThetaNet(r,t,s)}{[s+1]} 
\overset{\eqref{ThetaFormula1}}{=}
\frac{(-1)^{\frac{r + s + t}{2}} \left[ \frac{r + s + t}{2} + 1 \right]! \left[ \frac{ r + s - t }{2} \right]! \left[ \frac{ s + t - r}{2} \right]! \left[ \frac{t + r - s}{2} \right]! }{[s + 1]!  [ r ]! [ t ]!} . 
\end{align}
From~(\ref{condi1General2ndoe-f}--\ref{condi1General2ndoe-l}) and the assumption $\max (r, t) < \ppmin(q)$, 
we see that all of the quantum factorials in~\eqref{ThetaNetShouldBeFiniteAndNonzero}, except possibly
$\left[ \frac{t + r - s}{2} \right]!$, $[s+1]!$, and $\left[\frac{r+t+s}{2}+1\right]!$, are nonzero. 
The first of these is also nonzero because
\begin{align}
\frac{t + r - s}{2} \leq \frac{t + r - |r-t|}{2} = \min(r,t) < \ppmin(q) .
\end{align}
Finally, we show that the ratio of the remaining factors in~\eqref{ThetaNetShouldBeFiniteAndNonzero} is also finite and nonzero.
We note that
\begin{alignat}{2}
0 & \leq |r-t| \leq s 
\qquad  && \text{and}  \qquad
s \leq r + t \leq 2\max(r,t) < 2\ppmin(q) , \\
0 & \leq \max(r,t) + 1 \leq
\frac{r+t+s}{2}+1 
\qquad  && \text{and}  \qquad
\frac{r+t+s}{2} \leq \frac{r+t+r+t}{2} = r + t \leq 2\max(r,t)  < 2\ppmin(q) .
\end{alignat}
In particular, the factorial $[s+1]!$ or $\left[\frac{r+t+s}{2}+1\right]!$ can only be zero if 
$\ppmin(q) \leq s+1 \leq 2\ppmin(q)$ or $\ppmin(q) \leq \frac{r+t+s}{2}+1 \leq 2\ppmin(q)$, respectively.
We consider two cases:
\begin{enumerate}[leftmargin=*]
\itemcolor{red}
\item If $s = 2\ppmin(q) - 1$, then we necessarily have $s = r+t$, and $\max(r,t) = \ppmin(q)$, and $\min(r,t) = \ppmin(q) - 1$.
In this case, we have $\frac{1}{2}(r+t+s)+1 = 2\ppmin(q)$, so the ratio of $[s+1]!$ and $\left[\frac{r+t+s}{2}+1\right]!$ equals one.

\item If $\ppmin(q) \leq s + 1 < 2\ppmin(q)$, then we necessarily have $r + t > \ppmin(q)$, and hence,
\begin{align}
\ppmin(q) = \frac{\ppmin(q)-1+\ppmin(q)-1}{2} + 1 \leq 
\frac{r+t+s}{2}+1 & \, < \, \frac{2\ppmin(q) - 1 + 2\ppmin(q) - 1}{2} + 1 = 2\ppmin(q) .
\end{align}
From definition~\eqref{Qinteger}, we see that for any $k \in \{ 0, 1, \ldots, \ppmin(q) - 1 \}$, we have
$[k]_q = 0$ if and only if $\ppmin(q) \,|\, k$, and these zeros of $q \mapsto [k]_q$ 
are of first order. This shows that the zeros in the ratio of $[s+1]!$ and $\left[\frac{r+t+s}{2}+1\right]!$ cancel, so 
this ratio is finite and nonzero.
\end{enumerate}
This concludes the proof.
\end{proof}

\section{Further auxiliary results} \label{LemmaApp}

In this appendix, we collect auxiliary results needed in this article. 
Lemmas~\ref{WJSandwichLem},~\ref{ChangeOfBasisLem}, and~\ref{DefectLem} 
follow from their analogues proved in~\cite{fp0} for the valenced Temperley-Lieb algebra $\TL_\multii(\nu)$
together with the fact from~\cite[lemmas~\red{B.2} and~\red{B.3}]{fp0}
that $\TL_\multii(\nu)$ and $\WJ_\multii(\nu)$ as well as their standard modules $\smash{\LS_\multii\super{s}}$ 
and $\smash{\PS_\multii\super{s}}$ are isomorphic.
Lemma~\ref{WJSandwichLem2} is similar to lemma~\ref{WJSandwichLem}.
We use lemmas~\ref{WJSandwichLem}--\ref{WJSandwichLem2} in section~\ref{CellSec},
and lemmas~\ref{ChangeOfBasisLem}--\ref{DefectLem} in section~\ref{GeneratorLemProofSect}.


Throughout, we let $\multii$ and $\multiii$ be two multiindices with $\np_\multii$ and $\np_\multiii$ nonnegative integer entries respectively, 
\begin{align}
\multii = (\sIndex_1, \sIndex_2,\ldots, \sIndex_{\np_\multii}) \in \{ (0) \} \cup \bZpos^\#, \qquad \qquad \multiii = (p_1, p_2, \ldots, p_{\np_\multiii}) \in \{ (0) \} \cup \bZpos^\# , 
\end{align}
and such that $\Summed_\multii + \Summed_\multiii = 0 \Mod 2$.
Then, we extend the definition of Jones-Wenzl tangles to include projector boxes of different sizes on the left and right sides of the tangles,
given by the two multiindices $\multii$ and $\multiii$: we denote by 
\begin{align}
\WJ_\multii^\multiii := \Span \PD_\multii^\multiii
\end{align}
the space of \emph{$(\multii, \multiii)$-Jones-Wenzl tangles}, spanned by \emph{$(\multii, \multiii)$-Jones-Wenzl link diagrams}.
A generic element in $\smash{\WJ_\multii^\multiii}$
is depicted in~\eqref{SandwichmapGen} in lemma~\ref{WJSandwichLem} below.
We distinguish Jones-Wenzl tangles with given number of crossing links, writing
\begin{align}\label{WJDirSum2} 
\WJ_\multii^{\multiii} = \; & \bigoplus_{s \, \in \, \DefectSet_\multii^{\multiii}} \smash{\WJ_\multii^{\multiii; \scaleobj{0.85}{(s)}}} , \\
\WJ_\multii^{\multiii; \scaleobj{0.85}{(s)}} := \; & \Span \{ \text{all Jones-Wenzl link diagrams in 
$\smash{\PD_\multii^\multiii}$ with exactly $s$ crossing links} \} ,
\end{align}
where $\smash{\DefectSet_\multii^{\multiii}}$ denotes the set of all integers $s \geq 0$ such that 
$\smash{\WJ_\multii^{\multiii; \scaleobj{0.85}{(s)}}}$ is not empty. We also note that, with $\multiii = \multii$,
we have
\begin{align} \label{smallerthans}
\WJ_\multii(\nu)\super{< s} =  \bigoplus_{r \, < \, s} \smash{\WJ_\multii^{\multiii; \scaleobj{0.85}{(r)}}} ,
\end{align}
for all $s \in \DefectSet_\multii$, where the left side is defined as 
$\smash{\WJ_\multii(\nu)\super{< s} := \Span \{\BarAction \epsilon \quad \eta \BarAction \, | \, \epsilon, \eta \in \PP_{\multii}\super{r} \text{ and } r < s \}}$, 
c.f.~(\ref{Sandwichmap},~\ref{SandwichmapGen}).

For a $\multii$-Jones-Wenzl link state $\alpha$,
we let $\tilde{\alpha}$ denote the link state obtained by reflecting $\alpha$ about a vertical axis:
\begin{align} \label{Flip} 
\alpha \quad = \quad \vcenter{\hbox{\includegraphics[scale=0.275]{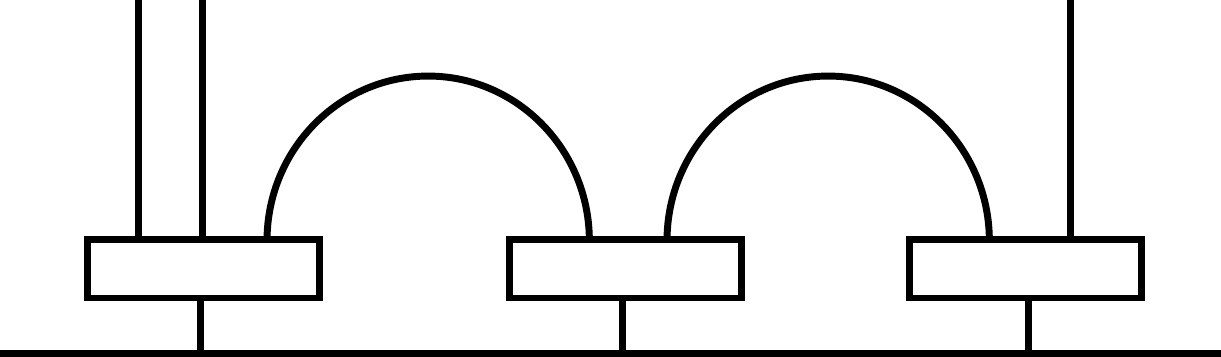}}}
\qquad \qquad \Longrightarrow \qquad \qquad
\tilde{\alpha} \quad = \quad \vcenter{\hbox{\includegraphics[scale=0.275]{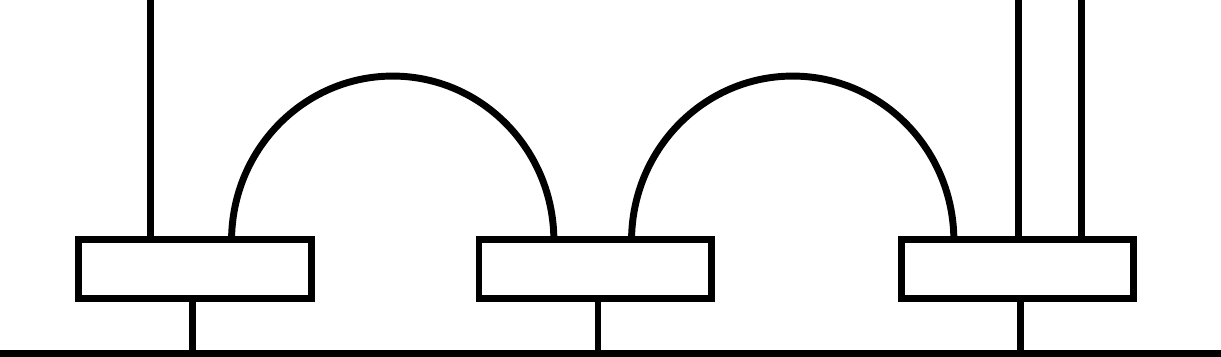} .}} 
\end{align}
Thus, if $\alpha \in \PS_\multii$, then $\tilde{\alpha} \in \PS_{\tilde{\multii}}$, where the multiindex $\tilde{\multii}$ is given by
\begin{align}
\multii = (\sIndex_1, \sIndex_2, \ldots, \sIndex_{\np_\multii-1}, \sIndex_{\np_\multii}) 
\qquad \qquad \Longrightarrow \qquad \qquad
\tilde{\multii} := (\sIndex_{\np_\multii}, \sIndex_{\np_\multii-1}, \ldots, \sIndex_2, \sIndex_1). 
\end{align}

\bigskip

The first two lemmas~\ref{WJSandwichLem}--\ref{WJSandwichLem2} give useful bases for the space
$\WJ_\multii^\multiii$ of Jones-Wenzl tangles. When $\multiii = \multii$, these bases give rise to cellular structures for the 
Jones-Wenzl algebra $\WJ_\multii(\nu)$ (proposition~\ref{CellPropo}).

\begin{lem}  \textnormal{\cite[lemma~\red{2.5}]{fp0}}
\label{WJSandwichLem}  
The map $\BarAction \, \cdot \quad \cdot \, \BarAction \colon 
\smash{\underset{s \, \in \, \DefectSet_\multii^\multiii}{\bigoplus}
\PS_\multii\super{s}} \otimes \smash{\PS_\multiii\super{s}} \longrightarrow \WJ_\multii^\multiii$ 
defined by linear extension~of
\begin{align} \label{SandwichmapGen}
\alpha \otimes \beta \qquad \longmapsto \qquad \BarAction \alpha \quad \beta \BarAction 
\quad := \quad \vcenter{\hbox{\includegraphics[scale=0.275]{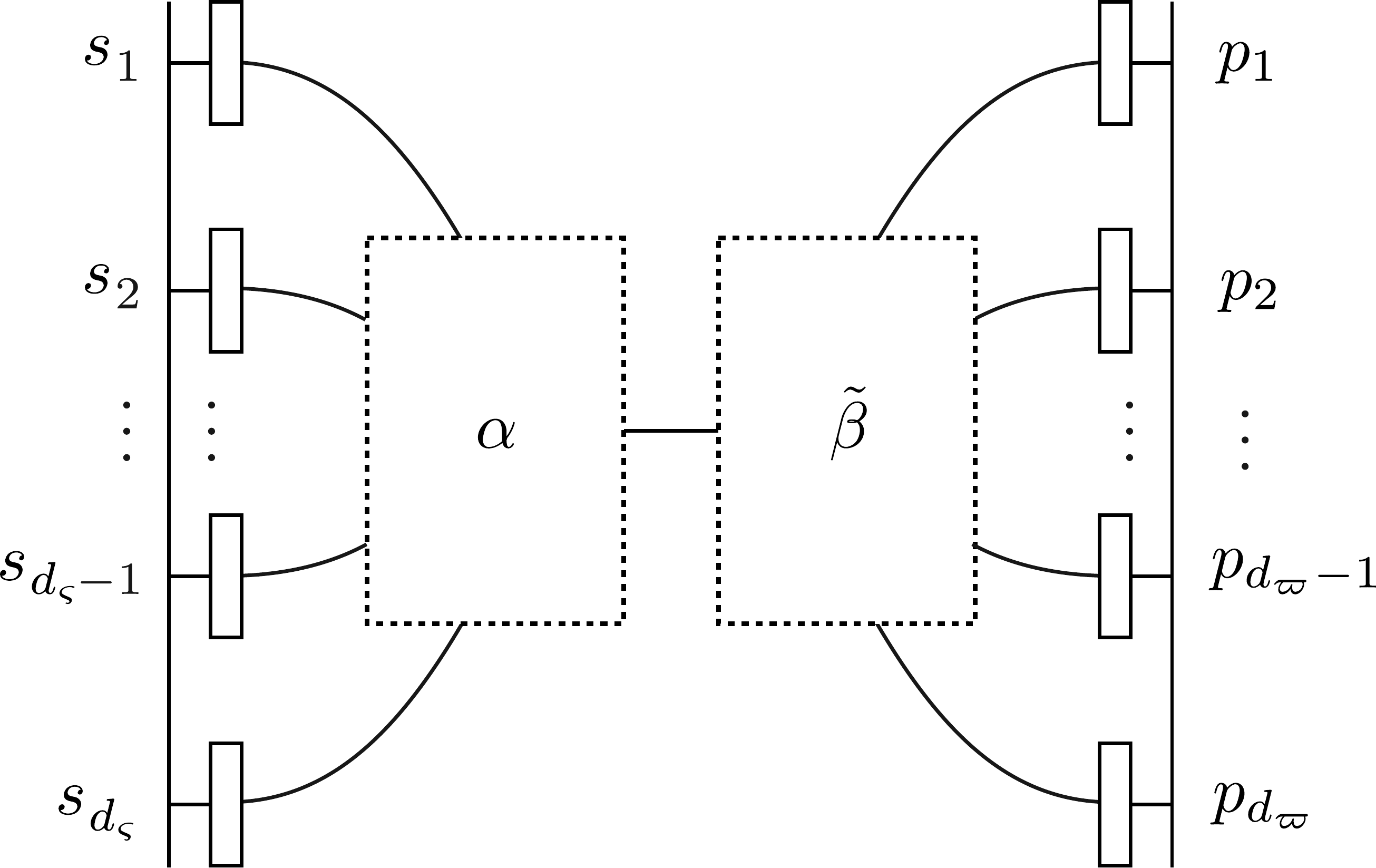} ,}}
\end{align}
for all Jones-Wenzl  link patterns $\alpha \in \smash{\PP_\multii\super{s}}$ and 
$\beta \in \smash{\PP_\multiii\super{s}}$, is an isomorphism of vector spaces.
\hfill \qed
\end{lem}

\begin{lem} \label{WJSandwichLem2} 
Suppose $\min(\Summed_\multii, \Summed_\multiii) < \ppmin(q)$. Then the 
map 
$\BarAction \, \cdot \quad \ProjBox \quad \cdot \, \BarAction \colon \smash{\underset{s \, \in \, \DefectSet_\multii^\multiii}{\bigoplus}
\PS_\multii\super{s}} \otimes \smash{\PS_\multiii\super{s}} \longrightarrow \WJ_\multii^\multiii$ 
defined by linear extension of
\begin{align} \label{InsertBarInTheMiddleMap}
 \alpha \otimes \beta \qquad \longmapsto \qquad \BarAction \alpha \quad \ProjBox \quad \beta \BarAction 
 \quad := \quad \vcenter{\hbox{\includegraphics[scale=0.275]{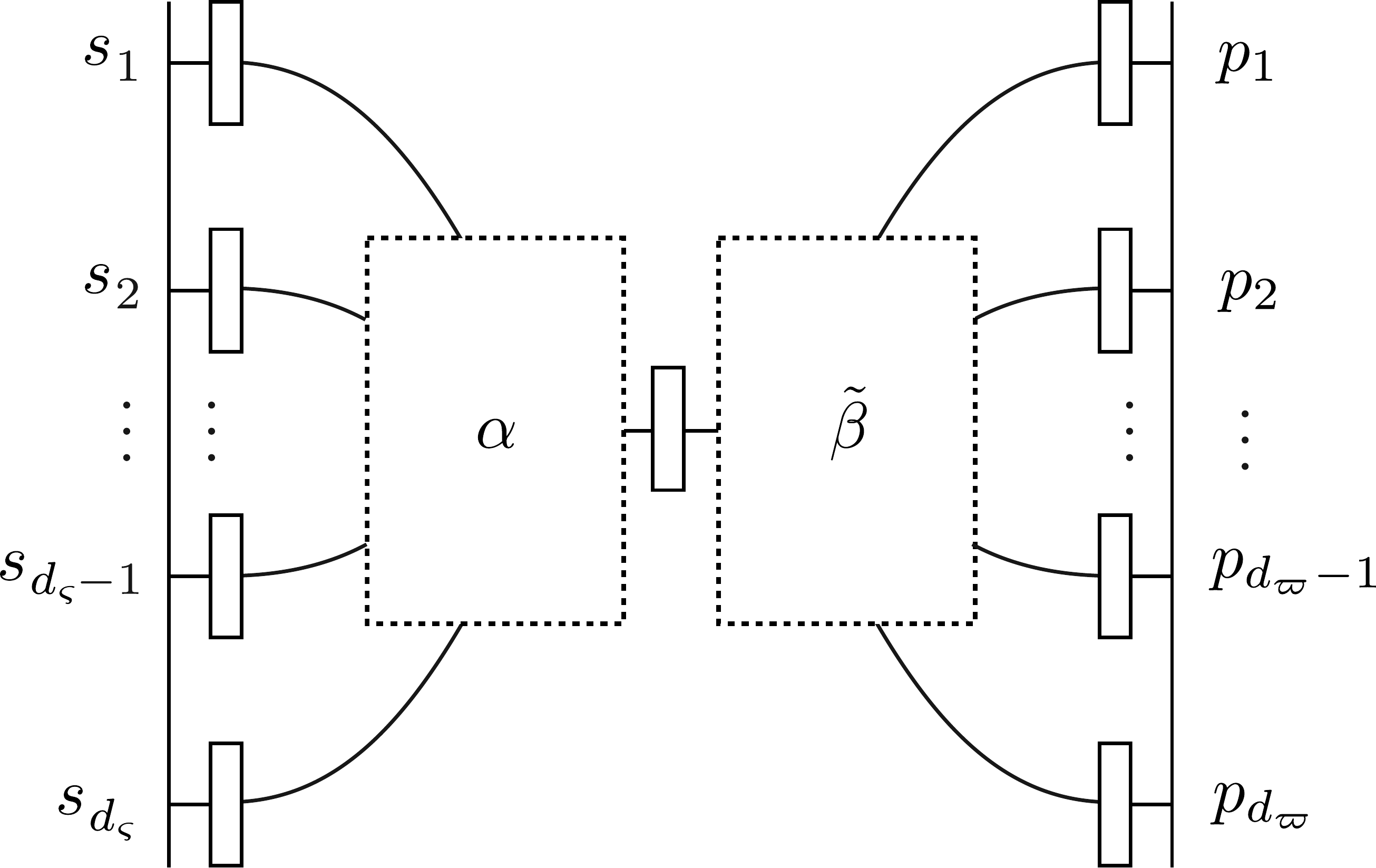} ,}}
\end{align}
for all Jones-Wenzl link patterns 
$\alpha \in \smash{\PP_\multii\super{s}}$ and $\beta \in \smash{\PP_\multiii\super{s}}$, 
is an isomorphism of vector spaces.
\end{lem}

\begin{proof} 
The map in the statement is well-defined because the projector box inserted in~\eqref{InsertBarInTheMiddleMap} 
is well-defined, having size at most 
$\max \smash{\DefectSet_\multii^{\multiii}} = \min(\Summed_\multii, \Summed_\multiii) < \ppmin(q)$.
This map sends the collection
$\smash{ \big\{ \alpha \otimes \beta \, \big|\, s \in \DefectSet_\multii^\multiii, \; \alpha \in \PP_\multii\super{s}, \; \beta \in \PP_\multiii\super{s} \big\}}$, 
which is a basis for its domain, to 
\begin{align}\label{BigSet} 
\big\{ \BarAction \alpha \quad \ProjBox \quad \beta \BarAction \, \big| \, s \in \DefectSet_\multii^\multiii, \; \alpha \in \PS_\multii\super{s}, 
\; \beta \in \PS_\multiii\super{s} \big\}.
\end{align}
We need to to prove that collection~\eqref{BigSet} is a basis for $\smash{\WJ_\multii^{\multiii}}$.
By direct-sum decomposition~\eqref{WJDirSum2}, it suffices to prove that for each 
$s \in \smash{\DefectSet_\multii^{\multiii}}$, the following smaller collection is a basis for 
the subspace $\smash{\WJ_\multii^{\multiii; \scaleobj{0.85}{(s)}}}$:
\begin{align}\label{2ndColl} 
\smash{\big\{\BarAction \alpha \quad \ProjBox \quad \beta \BarAction \big| \,\alpha \in \smash{\PP_\multii\super{s}}, 
\; \beta \in \smash{\PP_\multiii\super{s}}\big\}} . 
\end{align}
To prove this, we decompose the projector box between $\alpha$ and $\beta$.  
By recursion property~\eqref{wjrecursion} of the Jones-Wenzl projectors, we find that 
$\BarAction \alpha \quad \ProjBox \quad \beta \BarAction$ can be expanded in the form
\begin{align}\label{ProjDecBasis} 
\BarAction \alpha \quad \ProjBox \quad \beta \BarAction 
= \BarAction \alpha \quad \beta \BarAction 
+ \sum_{\substack{ r \, \in \, \DefectSet_\multii^{\multiii} \\ r \, < \, s}} T_{\alpha,\beta}\super{r} , 
\quad\text{for some $T_{\alpha,\beta}\super{r} \in \smash{\WJ_\multii^{\multiii; \scaleobj{0.85}{(r)}}}$} . 
\end{align}

First, we prove that~\eqref{2ndColl} is a linearly independent set.  For this, we note that any vanishing linear combination of the elements in this set, with coefficients $\smash{c_{\alpha,\beta}\super{s}} \in \bC$,
implies the following linear relation:
\begin{align}\label{LinCmb0} 
\sum_{\substack{\alpha \, \in \, \PP_\multii\super{s} \\ \beta \, \in \, \PP_\multiii\super{s}}} 
c_{\alpha,\beta}\super{s} \BarAction \alpha \quad \ProjBox \quad \beta \BarAction = 0 
\qquad \overset{\eqref{ProjDecBasis}}{\Longrightarrow} \qquad 
\sum_{\substack{\alpha \, \in \, \PP_\multii\super{s} \\ \beta \, \in \, \PP_\multiii\super{s}}} c_{\alpha,\beta}\super{s} \BarAction \alpha \quad \beta \BarAction 
= -\sum_{\substack{\alpha \, \in \, \PP_\multii\super{s} \\ \beta \, \in \, \PP_\multiii\super{s}}} c_{\alpha,\beta}\super{s} 
\sum_{\substack{ r \, \in \, \DefectSet_\multii^{\multiii} \\ r \, < \, s}}  T_{\alpha,\beta}\super{r} . 
\end{align}
The second sum of~\eqref{LinCmb0} is a tangle in $\smash{\WJ_\multii^{\multiii; \scaleobj{0.85}{(s)}}}$, 
and the last sum of~\eqref{LinCmb0}  is a tangle in $\smash{\WJ_\multii(\nu)\super{< s}}$~\eqref{smallerthans}.
Because the intersection of these two spaces is zero, the equality of these sums implies that both sums vanish.  Hence, we have 
\begin{align}\label{BothGone}  
\sum_{\substack{\alpha \, \in \, \PP_\multii\super{s} \\ \beta \, \in \, \PP_\multiii\super{s}}} c_{\alpha,\beta}\super{s} \BarAction \alpha \quad \ProjBox \quad \beta \BarAction = 0 
\qquad \Longrightarrow \qquad 
\sum_{\substack{\alpha \, \in \, \PP_\multii\super{s} \\ \beta \, \in \, \PP_\multiii\super{s}}} c_{\alpha,\beta}\super{s} \BarAction \alpha \quad \beta \BarAction = 0.
\end{align}
Now, lemma~\ref{WJSandwichLem} implies that the collection 
\begin{align}\label{1stColl} 
\smash{\big\{\BarAction \alpha \quad \beta \BarAction \big| \,\alpha \in \smash{\PP_\multii\super{s}}, \; \beta \in \smash{\PP_\multiii\super{s}}\big\}} 
\end{align}
is a basis for the subspace $\smash{\WJ_\multii^{\multiii; \scaleobj{0.85}{(s)}}}$, and as such, it is linearly independent.
Hence, we infer from~\eqref{BothGone} that $\smash{c_{\alpha,\beta}\super{s}} = 0$ for all $(\multii, \multiii)$-Jones-Wenzl link patterns 
$\alpha \in \smash{\PP_{\multiii}\super{s}}$ and $\beta \in \smash{\PP_\multii\super{s}}$, so~\eqref{2ndColl} is linearly independent too.  

Finally, we observe that the cardinalities of the sets~\eqref{2ndColl} and~\eqref{1stColl} are equal, so with the former set linearly independent and
the latter a basis for $\smash{\WJ_\multii^{\multiii; \scaleobj{0.85}{(s)}}}$, we conclude that the former set is a basis for this space too.
\end{proof}

\bigskip

The next two lemmas~\ref{ChangeOfBasisLem}--\ref{DefectLem} are needed in section~\ref{GeneratorLemProofSect}.

\begin{lem} \textnormal{\cite[lemma~\red{4.4}]{fp0}} \label{ChangeOfBasisLem}
Suppose $\max \multii < \ppmin(q)$. Let $\mathsf{B}_\multii$ be a basis for $\PS_\multii$, all of whose elements $\alpha$ may be written in the form
\begin{align}\label{LPrep} 
\vcenter{\hbox{\includegraphics[scale=0.275]{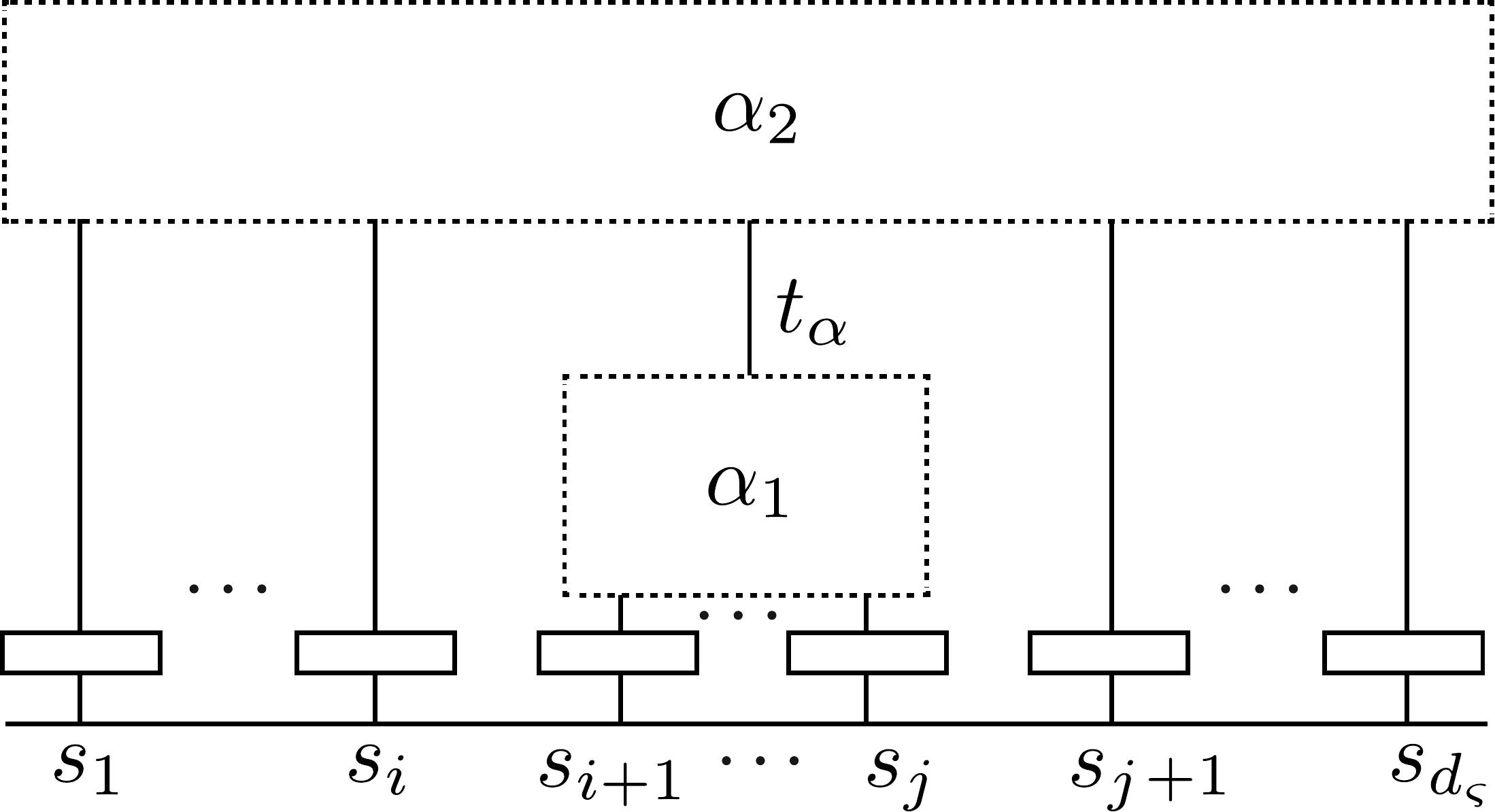} ,}} 
\end{align}
for some integers $i \in \{1,2,\ldots,\np_\multii-1\}$ and $j \in \{i+1, i+2, \ldots, \np_\multii\}$ common to all elements of $\mathsf{B}_\multii$, and 
\textnormal{(}with $\np = \np_\multii$\textnormal{)}
\begin{align} 
t_\alpha \in \DefectSet_{(\sIndex_{i+1}, \sIndex_{i+2}, \ldots, \sIndex_j)}, 
\qquad \alpha_1 \in \PS_{(\sIndex_{i+1}, \sIndex_{i+2}, \ldots, \sIndex_j)}^{(t_\alpha)}, 
\qquad \alpha_2 \in \PS_{(\sIndex_1, \sIndex_2, \ldots, \sIndex_i, t_\alpha, \sIndex_{j+1}, \sIndex_{j+2}, \ldots, \sIndex_{\np})} .
\end{align}
Also, let $T \colon \PS_\multii \longrightarrow \PS_\multii$ be the linear extension of the map sending each element $\alpha \in \mathsf{B}_\multii$, 
represented as~\eqref{LPrep}, to
\begin{align}\label{LPrepIns} 
\vcenter{\hbox{\includegraphics[scale=0.275]{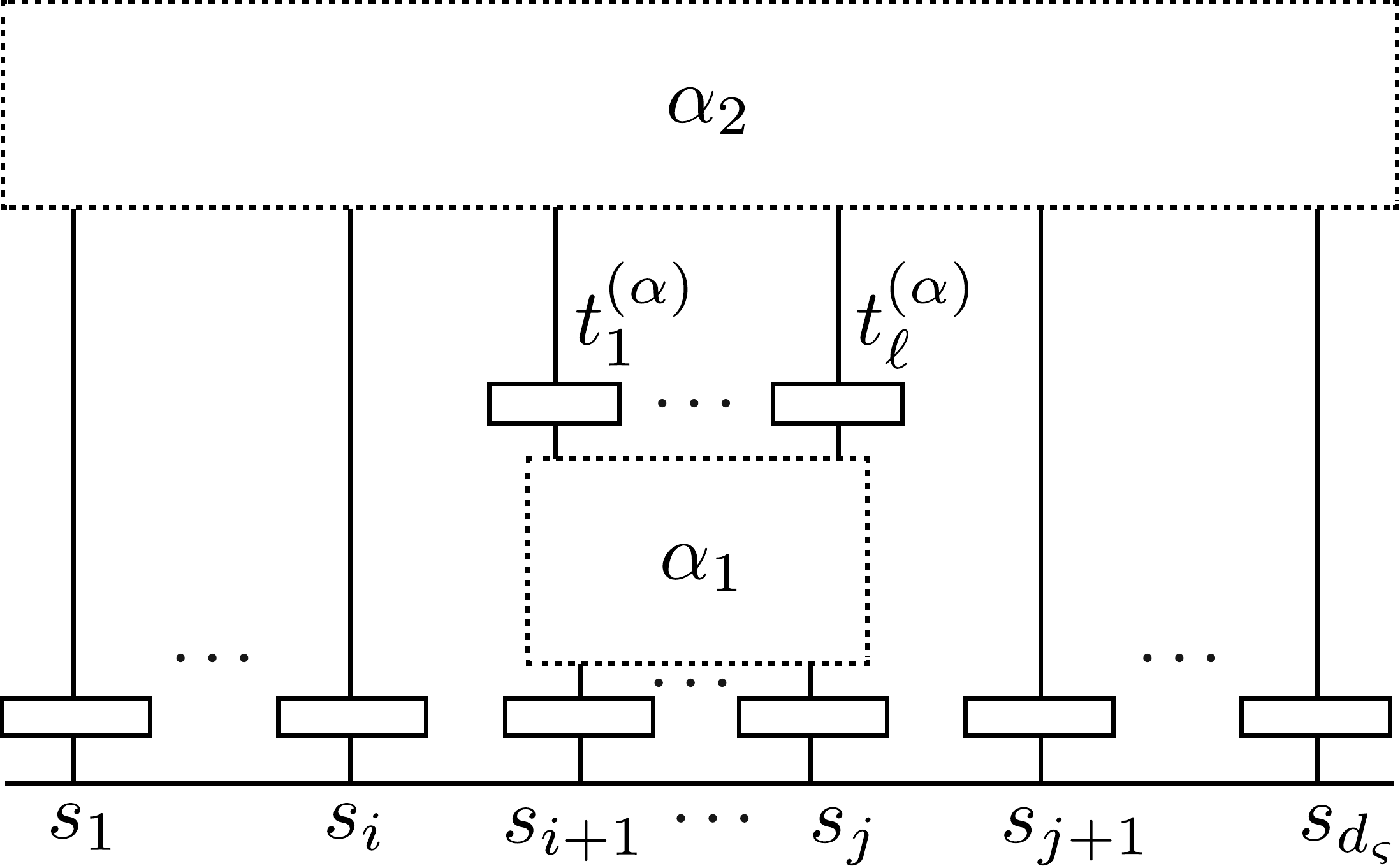} ,}} 
\end{align}
for some integers $\ell_\alpha, \smash{t_1^\alpha}$, $\smash{t_2^\alpha, \ldots, t_\ell^\alpha} \in \bZnn$ 
depending on $\alpha$ and such that $\smash{t_1^\alpha + t_2^\alpha + \dotsm + t_\ell^\alpha = t_\alpha}$,
with $\ell = \ell_\alpha$ vanishing only if $t_\alpha = 0$. 
Then $T$ has an upper-unitriangular matrix representation.
\hfill \qed
\end{lem}


\begin{lem} \textnormal{\cite[lemma~\red{2.4}]{fp0}}  \label{DefectLem} 
If $\alpha \in \smash{\PP_\multii\super{s}}$ with $s = \smin(\multii)$, then 
\begin{enumerate}
\itemcolor{red}
\item \label{DefIt1} all $\smin(\multii)$ defects of $\alpha$ attach to a common projector box of $\alpha$, and
\item \label{DefIt2} if all defects of $\alpha$ attach to its $i$:th projector box, 
then $\smin(\multii) < \sIndex_i$ if $\np_\multii > 1$, and $\smin(\multii) = \sIndex_1$ if $\np_\multii = 1$.
\end{enumerate}
In particular, items~\ref{DefIt1} and~\ref{DefIt2} together imply that
\begin{align} \label{sminineq} 
\smin(\multii) \quad 
\begin{cases} 
< \max \multii, & \np_\multii > 1, \\ = \max \multii, & \np_\multii = 1. 
\end{cases} 
\end{align}
\hfill \qed
\end{lem} 


\end{appendices}

\endgroup

\bigskip
\bigskip


\newcommand{\etalchar}[1]{$^{#1}$}

\renewcommand{\bibnumfmt}[1]{\makebox[5.3em][l]{[#1]}}


\end{document}